\documentclass[12pt]{article}
\usepackage{amsthm,amsmath,amsfonts,amssymb}
\usepackage{graphicx}
\usepackage{natbib}
\usepackage[figuresright]{rotating}
\usepackage{array}

\newtheorem{theorem}{Theorem}
\newtheorem{lemma}{Lemma}
\theoremstyle{plain}
\newtheorem{example}{Example}

\newcommand{\PreserveBackslash}[1]{\let\temp=\\#1\let\\=\temp}
\newcolumntype{C}[1]{>{\PreserveBackslash\centering}p{#1}}
\newcolumntype{L}[1]{>{\PreserveBackslash\raggedright}p{#1}}
\graphicspath{{Figures/}}

\begin{document}

\title{\textbf{Marginal treatment effects in the absence of instrumental
variables}}
\author{Zhewen Pan$^{1,2}$ ~ Zhengxin Wang$^{1,2}$ ~ Junsen Zhang$^3$\thanks{%
Corresponding address: School of Economics, Zhejiang University, 866
Yuhangtang Rd, Hangzhou 310058, P.R. China. Email: jszhang@cuhk.edu.hk
(Junsen Zhang).} ~ Yahong Zhou$^4$ \\
$^1$\textit{\small Center for Quantitative Economic Research, Zhejiang
University of Finance \& Economics} \\
$^2$\textit{\small School of Economics, Zhejiang University of Finance \&
Economics} \\
$^3$\textit{\small School of Economics, Zhejiang University} \\
$^4$\textit{\small School of Economics, Shanghai University of Finance \&
Economics}}
\maketitle

\begin{abstract}
We propose a method for defining, identifying, and estimating the marginal
treatment effect (MTE) without imposing the instrumental variable (IV)
assumptions of independence, exclusion, and separability (or monotonicity).
Under a new definition of the MTE based on reduced-form treatment error that
is statistically independent of the covariates, we find that the
relationship between the MTE and standard treatment parameters holds in the
absence of IVs. We provide a set of sufficient conditions ensuring the
identification of the defined MTE in an environment of essential
heterogeneity. The key conditions include a linear restriction on potential
outcome regression functions, a nonlinear restriction on the propensity
score, and a conditional mean independence restriction which will lead to
additive separability. We prove this identification using the notion of
semiparametric identification based on functional forms. And we provide an
empirical application for the Head Start program to illustrate the
usefulness of the proposed method in analyzing heterogenous causal effects
when IVs are elusive.
\end{abstract}



\newpage

\section{Introduction}

\label{sec:introduction}The marginal treatment effect (MTE), which was
developed in a series of seminal papers by %
\citet{heckman1999local,heckman2001local,heckman2005structural}, has become
one of the most popular tools in social sciences for describing,
interpreting, and analyzing the heterogeneity of the effect of a nonrandom
treatment. The MTE is defined as the expected treatment effect conditional
on observed covariates and a normalized error term that consists of
unobservable determinants of the treatment status. Although the definition
of the MTE doesn't necessitate an instrumental variable (IV), nearly all the
existing identification strategies for the MTE rely heavily on valid
instrumental variation which can induce otherwise similar individuals into
different treatment choices. This is mainly because the MTE has been
regarded as an extension of and supplement to the standard IV method in the
causal inference literature. However, in practical applications, the
validity of IVs is usually difficult to justify. In the most common case of
just-identification in which only one IV is available, the exogeneity or
randomness of the IV is fundamentally untestable. The exclusion restriction,
which is another condition underlying the validity of IVs, has also been
challenged more often than not 
\citep[e.g.,][]{van2007economic,
jones2015economics}.

Motivated by the potential invalidity of the available instruments, and
inspired by the observation that instruments are unnecessary in the
definition of the MTE, we attempt to model, identify, and estimate the MTE
without imposing the standard IV assumptions of conditional independence,
exclusion, and separability in this paper. Specifically, we allow all the
observed covariates to be statistically correlated to the error term and
have a direct impact on the outcome in addition to an indirect impact
through the treatment variable. The cornerstone of our model is a normalized
treatment equation, which determines treatment participation by the
propensity score crossing a reduced-form error term that is statistically
independent of (although functionally dependent on) all the covariates. This
is comparable to the conventional IV-based model, in which normalization is
performed with respect to only non-IV covariates. The independence of the
reduced-form treatment error from all the covariates can partly justify a
conditional mean independence assumption on the potential outcome residuals,
which can ensure the additive separability of the MTE into observables and
unobservables. This separability then facilitates the semiparametric
identification of MTE, given a linear restriction on the potential outcome
regression functions and a nonlinear restriction on the propensity score
function. We prove the identification by representing the MTE as a known
function of conditional moments of observed variables. Intuitively, the
identifying power comes from the excluded nonlinear variation in the
propensity score, which plays the role of the IV in exogenously perturbing
the treatment status.

Compared with imposing the IV assumptions, the main cost of the proposed
IV-free method is the linear restriction imposed on the outcome equations
and the nonlinearity imposed on the treatment equation, which makes our
identification strategy for the MTE closely resemble the semiparametric
counterpart of identification based on functional forms 
\citep{dong2010endogenous,
escanciano2016identification, lewbel2019identification, pan2022semiparametric}%
. We build our result on the notion of identification based on functional
forms mainly because (i) it can lead to point identification and estimation,
(ii) the required assumptions are regular and partly testable, and (iii) the
resulting estimation procedures are the most compatible with the standard
procedures in IV-based MTE models. Although the primary contribution of this
work is on the identification of MTE in the absence of IVs, we also propose
a semiparametric MTE estimator tailored to the identification result by
implementing the kernel-weighted pairwise difference regression %
\citep{ahn1993semiparametric} for each treatment status, and establish its
consistency and asymptotic normality that take into account the
nonparametrically estimated propensity score in the first step.

The main value of our IV-free framework of the MTE is threefold. First, it
provides an approach for consistently estimating heterogeneous causal
effects when a valid IV is hard to find. Second, it suggests an
easy-to-implement means for testing the validity of a candidate IV, because
it nests the models that assume exclusion restrictions. For instance, in
studies on returns to education, the validity of most instruments for
educational attainment is suspect, such as parental education, distance to
the school, and local labor market conditions 
\citep{kedagni2020generalized,
mourifie2020sharp}. Our framework enables a reliable evaluation of returns
to education without imposing IV assumptions on the candidate instruments
and implies a simple test for exclusion restrictions by $t$-testing the
instruments' coefficients in potential outcome equations. Third, even though
the validity of the IV is verified, identification based on functional forms
can be invoked as a way to increase the efficiency of estimation or to check
the robustness of the results to alternative identifying assumptions.

Early explorations of the identification of endogenous regression models
when IVs are unavailable focused on linear systems of simultaneous
equations, in which the lack of IVs is associated with hardly justifiable
exclusion restrictions and insufficient moment conditions. The typical
strategy for addressing this underidentification problem is to impose
second- or higher-order moment restrictions to construct instruments by
using the available exogenous covariates 
\citep[e.g.,][]{fisher1966identification, lewbel1997constructing,
klein2010estimating}. In particular, imposing restrictions on the
correlation between the covariates and the variance matrix of the vector of
model errors leads to identification based on heteroscedasticity %
\citep{lewbel2012using}. \cite{d2021testing} considered a more general
nonseparable, nonparametric triangular model without the exclusion
restriction, and established identification for the control function
approach based on a local irrelevance condition. However, these results are
restricted to the case of continuous endogenous variable or treatment. \cite%
{lewbel2018identification} showed that, in linear regression models, the
moment conditions of identification based on heteroscedasticity can be
satisfied when the endogenous treatment is binary, at the expense of strong
restrictions on the model errors. Alternatively, \cite{conley2012plausibly}
and \cite{van2018beyond} presented a sensitivity analysis approach for
performing inference for linear IV models while relaxing the exclusion
restriction to a certain extent. \cite{kang2016instrumental} proposed an
identification and estimation method for linear treatment effect models when
the exclusion restriction of some IVs may not hold, and \cite%
{windmeijer2019use} further proposed a robust method to detect invalid IVs
that enter the outcome equation as explanatory variables. More recently, 
\cite{masten2021salvaging} suggested a set identification method for linear
IV models under violations of the exclusion restriction. A common limitation
of this body of literature is that linear models implicitly impose an
undesirable homogeneity assumption on treatment effects. \cite{carl2025tsci}
considered invalid IVs in a treatment effect model with nonlinearity in both
the treatment and outcome equations, and established the identification
based on different functional forms in the two equations, which is similar
to our identification strategy in terms of ideas. Moreover, they adopted
machine learning methods in estimating the treatment equation to cope with
the potentially high nonlinearity. However, their model implicitly assumes
homogenous treatment effects as well. As a complement, our model relaxes the
restrictive homogeneity assumption and allows for general heterogeneities in
treatment effects, which is evidently more realistic.

The literature on sample selection models without the exclusion restriction
is also closely related. A general solution to the problem of lack of
excluded variables is partial identification and set estimation 
\citep{honore2020selection,
honore2022sample, westphal2022marginal}. The limitation of this approach is
that the identified or estimated set may be too wide to be informative.
Another solution is the approach of identification at infinity that uses the
data only for large values of a special regressor 
\citep{chamberlain1986asymptotic,
lewbel2007endogenous} or of the outcome \citep{d2013another}. Although
identification at infinity can lead to point identification, it is typically
featured as irregular identification \citep{khan2010irregular}, and the
derived estimate will converge at a rate slower than $n^{-1/2}$, where $n$
is the sample size. \cite{heckman1979sample} exploited nonlinearity in the
selectivity correction function to achieve point identification and root-$n$
consistent estimation for the linear coefficients of a parametric sample
selection model, which is the original version of identification based on
functional forms. However, Heckman's model imposes a restrictive bivariate
normality assumption on the error terms and thus poses the risk of model
misspecification. \cite{escanciano2016identification} extended Heckman's
approach to a general semiparametric model and established identification of
the linear coefficients by exploiting nonlinearity elsewhere in the model.
At the expense of the generality of their model, \cite%
{escanciano2016identification}'s identification result is only up to scale
and thus requires a scale normalization assumption. This normalization would
be innocuous if we know the sign of the normalized coefficient and are
interested in only the sign, but not the magnitude, of the other
coefficients. However, the magnitude of the linear coefficients is
indispensable to the evaluation of a program or a policy. In addition, the
identifying assumption about nonlinearity developed by \cite%
{escanciano2016identification} implicitly rules out the case of the
existence of two continuous covariates. We adapt the result of \cite%
{escanciano2016identification} to the MTE model by amending the two defects.
First, we take advantage of the specific model structure to relax the scale
normalization and identify the magnitude of the linear coefficients.
Moreover, the identifying assumption about nonlinearity will be simplified
owing to the specific model structure. We give an intuitive interpretation
of the new nonlinearity assumption. Second, we combine the nonlinearity
assumption with a local irrelevance condition to allow for an arbitrary
number of continuous covariates.

The rest of this paper is organized as follows. In Section \ref{sec:model},
we introduce our model and definition of the MTE without IVs. In Section \ref%
{sec:identification}, we propose a possible set of sufficient conditions for
the identification of MTE, in place of the standard IV assumptions. The key
conditions include semiparametric functional form restrictions and the
conditional mean independence assumption, of which the latter implies the
additive separability of the MTE into observed and unobserved components.
Section \ref{sec:estimation} suggests consistent estimation procedures, and
Section \ref{sec:application} provides an empirical application to Head
Start. Section \ref{sec:conclusion} concludes. Online appendices comprise of
(A) additional results of the empirical application,
(B) a proof of the main identification result, (C) a discussion on variants
of the nonlinearity assumption, (D) an identification result for limited
valued outcomes which entails a slight modification of the assumptions, (E)
a proof of the asymptotic normality of the semiparametric IV-free MTE
estimator, (F) an implementation guidance, and (G) a simulation study.

\section{Model}

\label{sec:model}In the following, we denote random variables or random
vectors by capital letters, such as $U$, and their possible realizations by
the corresponding lowercase letters, such as $u$. We denote $F_{U}\left(
\cdot \right) $ as the cumulative distribution function (CDF) of $U$ and $%
F_{U\left\vert X\right. }\left( \cdot \left\vert x\right. \right) $ as the
conditional CDF of $U$ given $X=x$. Our analysis builds on the potential
outcomes framework developed by \cite{roy1951some} in econometrics and \cite%
{rubin1973matching,rubin1973use} in statistics. Specifically, we consider a
binary treatment, denoted by $D$, and let $Y_{1}$ and $Y_{0}$ denote the
potential outcomes if the individuals are treated ($D=1$) or not treated ($%
D=0$), respectively. The observed outcome is%
\begin{equation*}
Y=DY_{1}+\left( 1-D\right) Y_{0},
\end{equation*}%
and the quantity of interest is the counterfactual treatment effect $%
Y_{1}-Y_{0}$. We suppose that the treatment status is determined by the
following threshold crossing rule:%
\begin{equation}
D=1\left\{ \mu \left( X\right) \geq U\right\} ,  \label{treatment}
\end{equation}%
where $X$ is a vector of pretreatment covariates, $1\left\{ A\right\} $ is
the indicator function of event $A$, $\mu \left( \cdot \right) $ is an
unknown function, and $U$ is a structural error term containing unobserved
characteristics that may affect participation in the treatment, such as the
opportunity costs or intangible benefits of the treatment.

Compared with the IV-based MTE model, we relax the independence and
separability assumptions in the treatment equation (\ref{treatment}) to
account for the absence of IVs. First, we don't assume that $X$ is
stochastically independent of $U$; that is, no exogenous covariate is
needed. Second, we allow the treatment decision rule to be intrinsically
nonseparable in observed and unobserved characteristics, which is equivalent
to relaxing the monotonicity assumption in the IV model 
\citep{vytlacil2002independence,
vytlacil2006note}. Specifically, let $U_{x}$ be the counterfactual error
term denoting what $U$ would have been if $X$ had been externally set to $x$%
, then the nonseparability of (\ref{treatment}) means that at least two
vectors $x$ and $\tilde{x}$ exist in the support of $X$ such that $U_{x}\neq
U_{\tilde{x}}$ with positive probability. In particular, $U$ is allowed to
depend functionally on $X$, as in the following example.

\begin{example}
\label{example:1} We consider a latent index rule for treatment
participation:%
\begin{equation}
D=1\left\{ m\left( X,\varepsilon \right) \geq 0\right\} ,
\label{latent_index}
\end{equation}%
where the observed $X$ can be statistically correlated to the unobserved $%
\varepsilon $, and no restriction is imposed on the cross-partials of the
index function $m$. Without independence and additive separability, model (%
\ref{latent_index}) is completely vacuous, imposing no restrictions on the
observed or counterfactual outcomes \citep{heckman2001local}. This general
latent index rule fits into the treatment equation (\ref{treatment}) by
taking $\mu \left( X\right) =E\left[ m\left( X,\varepsilon \right)
\left\vert X\right. \right] $ and $U=\mu \left( X\right) -m\left(
X,\varepsilon \right) $.
\end{example}

Example \ref{example:1} illustrates that no generality is lost by modeling
the treatment variable as equation (\ref{treatment}) without imposing the
independence and separability assumptions. We define the propensity score
function as the conditional probability of receiving the treatment given the
covariates, 
\begin{equation*}
\pi \left( x\right) \equiv E\left[ D\left\vert X=x\right. \right]
=F_{U\left\vert X\right. }\left( \mu \left( x\right) \left\vert x\right.
\right) ,
\end{equation*}%
and define the propensity score variable as $P\equiv \pi \left( X\right) $.
Under the regularity condition that $F_{U\left\vert X\right. }\left( \cdot
\left\vert x\right. \right) $ is absolutely continuous with respect to the
Lebesgue measure for all $x$, the structural treatment equation (\ref%
{treatment}) can be innocuously normalized into a reduced form:%
\begin{equation}
D=1\left\{ P\geq V\right\} ,  \label{normalized}
\end{equation}%
where%
\begin{equation}
V=F_{U\left\vert X\right. }\left( U\left\vert X\right. \right)  \label{VV}
\end{equation}%
is a normalized error term, which by definition follows standard uniform
distribution conditional on $X$ and thus is stochastically independent of $X$
and $P$ \citep{heckman2005structural, chernozhukov2005iv}. This
independence, which seems counterintuitive due to the functional dependence
of $V$ on $X$, will be lost if we consider $V_{x}=F_{U\left\vert X\right.
}\left( U_{x}\left\vert x\right. \right) $, the counterfactual variable of $%
V $ when $X$ is set to $x$. In general, the conditional distribution of $%
V_{x}$ given $X=\tilde{x}$ for $\tilde{x}\neq x$ depends on $\tilde{x}$, and
the unconditional distribution of $V_{x}$ is not uniform.

\begin{example}
\label{example:Vx} We suppose that $X$ is a scalar, and%
\begin{equation*}
\left( 
\begin{array}{l}
U \\ 
X%
\end{array}%
\right) \sim N\left( \left( 
\begin{array}{l}
0 \\ 
0%
\end{array}%
\right) ,\left( 
\begin{array}{cc}
1 & \sigma _{UX} \\ 
\sigma _{UX} & \sigma _{X}^{2}%
\end{array}%
\right) \right) .
\end{equation*}%
By the property of bivariate normal distribution, we obtain $U\left\vert
\left( X=x\right) \right. \sim N\left( \mu _{\left. U\right\vert X}\left(
x\right) ,\sigma _{U\left\vert X\right. }^{2}\right) $ and $F_{U\left\vert
X\right. }\left( u\left\vert x\right. \right) =\Phi \left( \left. \left[
u-\mu _{U\left\vert X\right. }\left( x\right) \right] \right/ \sigma
_{U\left\vert X\right. }\right) $, where $\mu _{U\left\vert X\right. }\left(
x\right) =\left( \sigma _{UX}\left/ \sigma _{X}^{2}\right. \right) x$, $%
\sigma _{U\left\vert X\right. }^{2}=1-\left( \sigma _{UX}^{2}\left/ \sigma
_{X}^{2}\right. \right) $, and $\Phi \left( \cdot \right) $ denotes the
standard normal CDF. Hence,%
\begin{equation*}
V=F_{U\left\vert X\right. }\left( U\left\vert X\right. \right) =\Phi \left( 
\frac{U-\mu _{U\left\vert X\right. }\left( X\right) }{\sigma _{U\left\vert
X\right. }}\right) \text{, and }V_{x}=\Phi \left( \frac{U-\mu _{U\left\vert
X\right. }\left( x\right) }{\sigma _{U\left\vert X\right. }}\right) .
\end{equation*}%
We observe that $V\perp \!\!\!\!\perp X$ because $F_{V\left\vert X\right.
}\left( v\left\vert x\right. \right) =v$, but $V_{x}$ is not independent of $%
X$ because%
\begin{equation*}
F_{V_{x}\left\vert X\right. }\left( v\left\vert \tilde{x}\right. \right)
=\Phi \left( \frac{\left( \sigma _{UX}\left/ \sigma _{X}^{2}\right. \right)
\left( x-\tilde{x}\right) }{\sigma _{U\left\vert X\right. }}+\Phi
^{-1}\left( v\right) \right),
\end{equation*}%
and $V_{x}$ is not uniformly distributed because $F_{V_{x}}\left( v\right)
=\Phi \left( \mu _{U\left\vert X\right. }\left( x\right) +\sigma
_{U\left\vert X\right. }\Phi ^{-1}\left( v\right) \right) $.
\end{example}

Given this independence, we may consider the reduced-form treatment error $V$
as the orthogonalized unobservables with respect to the observables, or the
unobservables that are projected onto the subspace orthogonal to the one
spanned by the observables. Another interpretation of $V$ is the ranking of
the structural error $U$ conditional on $X$. For instance, $V=0.2$
represents a typical individual whose $U$ value ranks above 20\% individuals
with identical covariates. $V$ enters the normalized crossing rule on the
right, making an individual less likely to receive treatment; thus, it
refers to resistance or distaste regarding the treatment in the MTE
literature. If $V$ is large, then the propensity score $P$ should be large
to induce the individual to participate in the treatment. However, an
individual with a $V$ value close to zero will participate even though $P$\
is small.

In the above instrument-free model, we define the MTE as the expected
treatment effect conditional on the observed and unobserved characteristics:%
\begin{equation*}
\Delta ^{\text{MTE}}\left( x,v\right) \equiv E\left[ \left.
Y_{1}-Y_{0}\right\vert X=x,V=v\right] .
\end{equation*}%
$\Delta ^{\text{MTE}}\left( x,v\right) $ captures all the treatment effect
heterogeneity that is consequential for selection bias by conditioning on
the orthogonal coordinates of the observable and unobservable dimensions.
Given $X$ and $V$, the treatment status $D$ is fixed and thus independent of
the treatment effect $Y_{1}-Y_{0}$. Note that if the assumptions of
independence and separability hold for some covariates, the normalized error
term $V$ defined in (\ref{VV}) will be exactly the same as that in the
IV-based MTE model, and if further the independent covariates have no direct
effect on the potential outcomes, our defined MTE will reduce to the MTE
that is defined in the IV model. Similar to the MTE in the IV model, the
IV-free $\Delta ^{\text{MTE}}\left( x,v\right) $ can be used as a building
block for constructing the commonly used causal parameters, such as the
average treatment effect (ATE), the treatment effect on the treated (TT),
the treatment effect on the untreated (TUT), and the local average treatment
effect (LATE), which can be expressed as weighted averages of $\Delta ^{%
\text{MTE}}\left( x,v\right) $ as follows:%
\begin{eqnarray*}
\Delta ^{\text{ATE}}\left( x\right) &\equiv &E\left[ \left.
Y_{1}-Y_{0}\right\vert X=x\right] =\int_{0}^{1}\Delta ^{\text{MTE}}\left(
x,v\right) dv, \\
\Delta ^{\text{TT}}\left( x\right) &\equiv &E\left[ \left.
Y_{1}-Y_{0}\right\vert X=x,D=1\right] =\frac{1}{\pi \left( x\right) }%
\int_{0}^{\pi \left( x\right) }\Delta ^{\text{MTE}}\left( x,v\right) dv, \\
\Delta ^{\text{TUT}}\left( x\right) &\equiv &E\left[ \left.
Y_{1}-Y_{0}\right\vert X=x,D=0\right] =\frac{1}{1-\pi \left( x\right) }%
\int_{\pi \left( x\right) }^{1}\Delta ^{\text{MTE}}\left( x,v\right) dv, \\
\Delta ^{\text{LATE}}\left( x,v_{1},v_{2}\right) &\equiv &E\left[ \left.
Y_{1}-Y_{0}\right\vert X=x,v_{1}\leq V\leq v_{2}\right] =\frac{1}{v_{2}-v_{1}%
}\int_{v_{1}}^{v_{2}}\Delta ^{\text{MTE}}\left( x,v\right) dv.
\end{eqnarray*}%
Somewhat surprisingly, compared with the weights on the IV-based MTE %
\citep[e.g.,][Table IB]{heckman2005structural}, which are generally
difficult to estimate in practice, the weights on $\Delta ^{\text{MTE}%
}\left( x,v\right) $ are simpler, more intuitive, and easier to compute.

\cite{heckman2001local,heckman2005structural} considered defining the MTE in
a similar nonseparable setting by conditioning on the structural error, and
showed how to integrate to generate other causal parameters. However, such
an MTE cannot be identified even in the presence of IVs. Our definition,
based instead on the reduced-form error, can effectively exploit the
statistical independence between the observed and unobserved variables to
facilitate identification of the MTE. \cite{zhou2019marginal} proposed to
redefine the MTE by conditioning on $P$ rather than covariates, which is a
more parsimonious specification of all the relevant treatment effect
heterogeneity for selection bias. The extension of our identification and
estimation procedures to this alternative definition is straightforward.

\section{Identification}

\label{sec:identification}Our identification strategy relies on a linearity
restriction on the potential outcome equations and a nonlinearity
restriction on the propensity score function. The intuition is that the
propensity score minus the linear outcome index will provide the excluded
variation that perturbs treatment status, which plays the role of a
continuous IV. Furthermore, to ensure the exogeneity of the excluded
variation, it is necessary to impose a certain form of independence between
the potential outcome residuals and covariates. We denote $h_{d}\left(
x\right) =E\left[ \left. Y_{d}\right\vert X=x\right] $ and $%
U_{d}=Y_{d}-h_{d}\left( X\right) $, $d=0,1$, as the regression functions and
residuals of potential outcomes.

\medskip

\textbf{Assumption CMI} (Conditional Mean Independence). \textit{Assume that 
}$E\left[ U_{d}\left\vert V,X\right. \right] =E\left[ U_{d}\left\vert
V\right. \right] $, $d=0,1$,\textit{\ with probability one.}

\medskip

Assumption CMI is standard in the MTE literature and commonly referred to as
separability or additive separability assumption %
\citep[e.g.,][]{brinch2017beyond, mogstad2018using, zhou2019marginal},
because it's a necessary and sufficient condition for the MTE to be
additively separable in observables and unobservables %
\citep[p.3074]{zhou2019marginal}:%
\begin{eqnarray*}
\Delta ^{\text{MTE}}\left( x,v\right) &=&h_{1}\left( x\right) -h_{0}\left(
x\right) +E\left[ U_{1}-U_{0}\left\vert X=x,V=v\right. \right] \\
&=&h_{1}\left( x\right) -h_{0}\left( x\right) +E\left[ U_{1}-U_{0}\left\vert
V=v\right. \right] .
\end{eqnarray*}%
Namely, under Assumption CMI, the shape of the MTE curve with respect to $v$
will not vary with covariates. Assumption CMI is partly justified by $E\left[
U_{d}\left\vert X\right. \right] =0$ and $V\perp \!\!\!\!\perp X$, which
come directly from the definition. On this basis, it is sufficient for
Assumption CMI to hold if the conditional covariance of $U_{d}$ and $V$ is
independent of $X$, which is a key assumption in the
heteroscedasticity-based identification method as well %
\citep{lewbel2012using}. Assumption CMI is implied by and much weaker than
the full independence $\left( U_{d},U\right) \perp \!\!\!\!\perp X$
frequently imposed (often implicitly) in applied work, where $U$ is the
structural treatment error in (\ref{treatment}). In particular, Assumption
CMI doesn't rule out the marginal dependence of $U_{d}$ or $U$ on $X$, as
illustrated by Example \ref{example:depend}.

\begin{example}
\label{example:depend} Suppose that $X$ is a scalar and%
\begin{equation*}
\left( 
\begin{array}{c}
U_{d} \\ 
U \\ 
X%
\end{array}%
\right) \sim N\left( \left( 
\begin{array}{c}
0 \\ 
0 \\ 
0%
\end{array}%
\right) ,\left( 
\begin{array}{ccc}
\sigma _{d}^{2} & \sigma _{dU} & 0 \\ 
\sigma _{dU} & 1 & \sigma _{UX} \\ 
0 & \sigma _{UX} & \sigma _{X}^{2}%
\end{array}%
\right) \right) .
\end{equation*}%
In this setting, $U$ is correlated to $X$ with an unconstrained correlation
coefficient. Through a property of the multivariate normal distribution, we
obtain%
\begin{equation}
\left. \left( 
\begin{array}{c}
U_{d} \\ 
U%
\end{array}%
\right) \right\vert \left( X=x\right) \sim N\left( \left( 
\begin{array}{c}
0 \\ 
\mu _{\left. U\right\vert X}\left( x\right)%
\end{array}%
\right) ,\left( 
\begin{array}{cc}
\sigma _{d}^{2} & \sigma _{dU} \\ 
\sigma _{dU} & \sigma _{U\left\vert X\right. }^{2}%
\end{array}%
\right) \right) ,  \label{nonhetero}
\end{equation}%
where $\mu _{U\left\vert X\right. }\left( x\right) =\left( \sigma
_{UX}\left/ \sigma _{X}^{2}\right. \right) x$ and $\sigma _{U\left\vert
X\right. }^{2}=1-\left( \sigma _{UX}^{2}\left/ \sigma _{X}^{2}\right.
\right) $. Hence,%
\begin{equation*}
E\left[ \left. U_{d}\right\vert U=u,X=x\right] =\frac{\sigma _{dU}}{\sigma
_{U\left\vert X\right. }^{2}}\left( u-\mu _{\left. U\right\vert X}\left(
x\right) \right) .
\end{equation*}%
By Example \ref{example:Vx}, we have $V=\Phi \left( \left. \left[ U-\mu
_{U\left\vert X\right. }\left( X\right) \right] \right/ \sigma _{U\left\vert
X\right. }\right) $, so that $U=\sigma _{U\left\vert X\right. }\Phi
^{-1}\left( V\right) +\mu _{U\left\vert X\right. }\left( X\right) $.
Consequently,%
\begin{equation*}
E\left[ \left. U_{d}\right\vert V=v,X=x\right] =E\left[ \left.
U_{d}\right\vert U=\sigma _{U\left\vert X\right. }\Phi ^{-1}\left( v\right)
+\mu _{U\left\vert X\right. }\left( x\right) ,X=x\right] =\frac{\sigma _{dU}%
}{\sigma _{U\left\vert X\right. }}\Phi ^{-1}\left( v\right) ,
\end{equation*}%
and Assumption CMI holds. More generally, to allow the dependence of $U_{d}$
on $X$ as well, we can set%
\begin{equation*}
\left. \left( 
\begin{array}{c}
U_{d} \\ 
U%
\end{array}%
\right) \right\vert \left( X=x\right) \sim N\left( \left( 
\begin{array}{c}
0 \\ 
\mu _{\left. U\right\vert X}\left( x\right)%
\end{array}%
\right) ,\left( 
\begin{array}{cc}
\sigma _{d}^{2}\left( x\right) & \sigma _{dU} \\ 
\sigma _{dU} & \sigma _{U\left\vert X\right. }^{2}%
\end{array}%
\right) \right)
\end{equation*}%
in place of (\ref{nonhetero}), where $\sigma _{d}^{2}\left( x\right) $ is
the conditional variance of $U_{d}$ given $X=x$. Since $E\left[ \left.
U_{d}\right\vert V=v,X=x\right] $ is irrelevant to the variance of $U_{d}$
according to the above analysis, Assumption CMI still holds in the presence
of such heteroscedastic $U_{d}$.
\end{example}

For the purpose of identifying the MTE, we first establish the relationship
between $\Delta ^{\text{MTE}}\left( x,v\right) $ and observed regression
functions. Under Assumption CMI, the observed regression functions for
treated and untreated individuals are additively separable in a similar way:%
\begin{eqnarray}
E\left[ Y\left\vert X,D=d\right. \right] &=&E\left[ Y_{d}\left\vert
X,D=d\right. \right]  \notag \\
&=&h_{d}\left( X\right) +E\left[ U_{d}\left\vert X,D=d\right. \right]  \notag
\\
&=&h_{d}\left( X\right) +g_{d}\left( P\right) ,\text{ }d=0,1,  \label{md0}
\end{eqnarray}%
where%
\begin{eqnarray*}
g_{1}\left( p\right) &=&E\left[ \left. U_{1}\right\vert V\leq p\right] =%
\frac{1}{p}\int_{0}^{p}E\left[ \left. U_{1}\right\vert V=v\right] dv, \\
g_{0}\left( p\right) &=&E\left[ \left. U_{0}\right\vert V>p\right] =\frac{1}{%
1-p}\int_{p}^{1}E\left[ \left. U_{0}\right\vert V=v\right] dv.
\end{eqnarray*}%
The function $g_{d}\left( \cdot \right) $ is usually referred to as
selectivity bias correction term and plays a key part in the literature on
sample selection models. We introduce $g_{d}\left( \cdot \right) $ into the
identification of MTE because $g_{d}\left( \cdot \right) $ is closely
related to the MTE and it is generally easier to identify $g_{d}\left( \cdot
\right) $ than directly identify the MTE. From its definition, $g_{d}\left(
\cdot \right) $ can be viewed as an aggregating function of the unobservable
heterogeneity of the MTE. By multiplying $g_{d}\left( p\right) $ with $p$ or 
$1-p$ and then differentiating with respect to $p$, we can retrieve the
unobservable part of the MTE as%
\begin{eqnarray*}
E\left[ \left. U_{1}\right\vert V=p\right] &=&g_{1}\left( p\right)
+pg_{1}^{\left( 1\right) }\left( p\right) , \\
E\left[ \left. U_{0}\right\vert V=p\right] &=&g_{0}\left( p\right) -\left(
1-p\right) g_{0}^{\left( 1\right) }\left( p\right) ,
\end{eqnarray*}%
where $g_{d}^{\left( 1\right) }$ denotes the derivative function of $g_{d}$.
Therefore, the MTE can be represented as%
\begin{equation}
\Delta ^{\text{MTE}}\left( x,v\right) =h_{1}\left( x\right) -h_{0}\left(
x\right) +\left[ g_{1}\left( v\right) -g_{0}\left( v\right) \right]
+vg_{1}^{\left( 1\right) }\left( v\right) +\left( 1-v\right) g_{0}^{\left(
1\right) }\left( v\right) ,  \label{MTEiden0}
\end{equation}%
and the identification of it can be achieved by identifying $h_{d}$ and $%
g_{d}$ from the regression equation (\ref{md0}).

Recall that $P=\pi \left( X\right) $, where $\pi $ is a deterministic
function, albeit unknown. Without having an IV, especially without having
the exclusion restriction that assumes some element of $X$ to drop out of $%
h_{d}\left( X\right) $, the argument of $g_{d}$ in (\ref{md0}) exhibits no
extra variation than $h_{d}$. Therefore, it is impossible to distinguish
between $g_{d}$ and $h_{d}$. For example, one may choose to let $h_{d}\left(
\cdot \right) $ absorb $g_{d}\left( \pi \left( \cdot \right) \right) $ and
set $g_{d}=0$. Inspired by the identification result of \cite%
{escanciano2016identification} for a double index model, functional form
restrictions that differentiate $h_{d}\left( \cdot \right) $ and $%
g_{d}\left( \pi \left( \cdot \right) \right) $ from one another can address
this issue and facilitate the identification of both $h_{d}$ and $g_{d}$. To
this end, we impose linearity on $h_{d}$ and nonlinearity on $\pi $.

\medskip

\textbf{Assumption L} (Linearity). \textit{Assume that }$E\left[ \left.
Y_{d}\right\vert X\right] =\alpha _{d}+X^{\prime }\beta _{d}$\textit{\ with
probability one for some fixed }$\alpha _{d}$\textit{\ and }$\beta _{d}$%
\textit{, }$d=0,1$\textit{.}

\medskip

Assumption L imposes a linear restriction on the potential outcome
regression functions, which is a common practice in empirical studies. The
linear restriction may seem too strong compared with the existing
identification strategies for the MTE, which generally allow highly flexible
(especially nonparametric) specifications on the potential outcomes.
However, when the identification is accomplished and the estimation comes
into consideration, the linear specification will turn to be nearly
universally adopted due to its tractability and interpretability 
\citep[e.g.,][]{kirkeboen2016field,
kline2016evaluating, brinch2017beyond, heckman2018returns,mogstad2021causal,
aryal2022signaling, mountjoy2022community}. In addition to the empirical
convenience, the linear regression relationship may also be derived from a
multivariate normal distribution of $\left( Y_{d},X\right) $, or justified
by economic models such as the Cobb-Douglas production function which
implies that the logarithm of an output is linearly related to the logarithm
of inputs, the capital asset pricing model (CAPM) which implies that the
expected return of a portfolio is linearly related to the portfolio risk,
and the concavely quadratic utility function which implies a linear demand
system for differentiated products \citep{amir2017microeconomic}.

With a little abuse of notation, we redenote $U_{d}$ to absorb the intercept 
$\alpha _{d}$ so that%
\begin{equation*}
Y_{d}=X^{\prime }\beta _{d}+U_{d},
\end{equation*}%
and thus that the regression equation (\ref{md0}) and the MTE become%
\begin{eqnarray}
E\left[ Y\left\vert X,D=d\right. \right] &=&X^{\prime }\beta
_{d}+g_{d}\left( P\right) ,\text{ }d=0,1,  \label{md} \\
\Delta ^{\text{MTE}}\left( x,v\right) &=&x^{\prime }\left( \beta _{1}-\beta
_{0}\right) +\left[ g_{1}\left( v\right) -g_{0}\left( v\right) \right]
+vg_{1}^{\left( 1\right) }\left( v\right) +\left( 1-v\right) g_{0}^{\left(
1\right) }\left( v\right) .  \label{MTEiden}
\end{eqnarray}%
Noting that Assumption L implicitly requires the potential outcomes to be
continuously distributed and supported on the entire real line, we would
like to point out that our identification strategy can be adapted to the
case of limited valued outcomes as well by specifying a linear latent index.
A detailed discussion is left to Appendix \ref{appendix:LDV}.

To present the nonlinearity assumption, we let $X$ be partitioned as $\left(
X^{C},X^{D}\right) $, where $X^{C}$ and $X^{D}$ consist of covariates that
are continuously and discretely distributed, respectively. We denote $X_{k}$%
, $X_{k}^{C}$, and $X_{k}^{D}$ as the $k$-th coordinates of $X$, $X^{C}$,
and $X^{D}$, respectively. And we denote $x^{C}$ as a generic element in the
support of $X^{C}$; likewise for $x^{D}$, $x_{k}^{C}$, and $x_{k}^{D}$. The
nonlinearity assumption requires the propensity score function $\pi $ to be
nonlinear in $X^{C}$ given a benchmark value of $X^{D}$. Without loss of
generality, we suppose that the vector of zeros is in the support of $X^{D}$
and is the benchmark value. We denote $\pi _{0}\left( x^{C}\right) =\pi
\left( x^{C},0\right) $ for notational convenience, where the discrete
covariates are equal to zero.

\medskip

\textbf{Assumption NL} (Non-Linearity). \textit{Assume that} $E\left[
Y\left\vert X^{C}=x^{C},X^{D}=0,D=d\right. \right] $ and\textit{\ }$\pi
_{0}\left( x^{C}\right) $\textit{, }$d=0,1$\textit{, are differentiable with
respect to }$x^{C}$,\textit{\ and that the derivative of }$\pi _{0}$ \textit{%
satisfies the following NL2 when }$\dim \left( X^{C}\right) \geq 2$\textit{\
or NL1 when }$\dim \left( X^{C}\right) =1.$

\textit{-- NL2 (}$\dim \left( X^{C}\right) \geq 2$\textit{): there exist two
vectors }$x^{C}$\textit{, }$\tilde{x}^{C}$\textit{\ in the support of }$%
X^{C} $\textit{\ and two elements }$k$\textit{, }$j$\textit{\ in set }$%
\left\{ 1,2,\cdots ,\dim \left( X^{C}\right) \right\} $\textit{\ such that
(i) }$\partial _{k}\pi _{0}\left( x^{C}\right) \neq 0$\textit{, (ii) }$%
\partial _{j}\pi _{0}\left( x^{C}\right) \neq 0$\textit{, (iii) }$\partial
_{k}\pi _{0}\left( \tilde{x}^{C}\right) \neq 0$\textit{, (iv) }$\partial
_{j}\pi _{0}\left( \tilde{x}^{C}\right) \neq 0$\textit{, and (v) }$\left.
\partial _{k}\pi _{0}\left( x^{C}\right) \right/ \partial _{j}\pi _{0}\left(
x^{C}\right) \neq \left. \partial _{k}\pi _{0}\left( \tilde{x}^{C}\right)
\right/ \partial _{j}\pi _{0}\left( \tilde{x}^{C}\right) $\textit{, where }$%
\partial _{k}\pi _{0}\left( x^{C}\right) =\left. \partial \pi _{0}\left(
x^{C}\right) \right/ \partial x_{k}^{C}$\textit{\ is the partial derivative
of }$\pi _{0}$\textit{\ with respect to the }$k$\textit{-th argument.}

\textit{-- NL1 (}$\dim \left( X^{C}\right) =1$\textit{): there exists a
constant }$\tilde{x}^{C}$\textit{\ in the support of }$X^{C}$\textit{\ such
that }$\pi _{0}^{\left( 1\right) }\left( \tilde{x}^{C}\right) =0$,\textit{\
where }$\pi _{0}^{\left( 1\right) }\left( x^{C}\right) =\left. d\pi
_{0}\left( x^{C}\right) \right/ dx^{C}$\textit{\ is the univariate
derivative of }$\pi _{0}$\textit{.}

\medskip

Assumption NL requires the propensity score function to be nonlinear in
continuous covariates in a generalized sense. Note that Assumption NL is
sufficient but not necessary for nonlinearity of $\pi _{0}$. Appendix \ref%
{appendix:NL} provides other possible forms of sufficient nonlinearity
conditions. The combination of Assumptions L and NL will enable
identification based on functional forms in a semiparametric version, which
can realize the identification of linear coefficients by exploiting
nonlinearity elsewhere in the model. Actually, the nonlinear or
nonparametric setting of the treatment equation is typically adopted in the
applied literature when semiparametric approaches to estimating the MTE are
taken, as illustrated in \cite{carneiro2011estimating} and summarized in 
\cite{cornelissen2016late}.\footnote{%
We thank an anonymous referee for pointing out this observation.} The
following example provides an economic intuition that underlies the
nonlinearity assumption.

\begin{example}
Consider the generalized Roy model where individuals choose treatment if the
perceived benefit from treatment is greater than the subjective cost of
treatment \citep{eisenhauer2015generalized}. Specifically,%
\begin{equation}
D=1\left\{ Y_{1}-Y_{0}\geq c\left( X\right) +\epsilon \right\} ,
\label{benefit-cost}
\end{equation}%
where $c\left( X\right) $ and $\epsilon $ are the observed and unobserved
parts of the cost, respectively. Under Assumption L, the treatment equation (%
\ref{benefit-cost}) becomes $D=1\left\{ X^{\prime }\left( \beta _{1}-\beta
_{0}\right) -c\left( X\right) \geq U\right\} $, where $U=\epsilon
-U_{1}+U_{0}$, therefore $\pi \left( x\right) =F_{U\left\vert X\right.
}\left( \left. x^{\prime }\left( \beta _{1}-\beta _{0}\right) -c\left(
x\right) \right\vert x\right) $, which will be a nonlinear function of $x$
if $c\left( x\right) $ is nonlinear or the conditional distribution function
of $U$ given $X=x$ is nonlinear. In economic applications, the cost function 
$c\left( x\right) $ is typically nonlinear in individual and family
characteristics. For instance, the subjective cost of higher education would
be U-shaped with respect to the family income, because of the relatively
high tuition fee for students with disadvantaged backgrounds and high
opportunity cost for students with superior backgrounds. On the other hand,
the conditional distribution function $F_{U\left\vert X\right. }\left(
\left. u\right\vert x\right) $ will be nonlinear in $x$ in general if $U$ is
not independent of $X$, as shown by Example \ref{example:hetero}. More
generally, the treatment equation (\ref{benefit-cost}) is a special case of
a utility maximization model of treatment defined as $D=1\left\{ \mathfrak{u}%
_{1}\left( Y_{1},X,\epsilon _{1}\right) \geq \mathfrak{u}_{0}\left(
Y_{0},X,\epsilon _{0}\right) \right\} $, where the utility functions $%
\mathfrak{u}_{1}$ and $\mathfrak{u}_{0}$ are quasi-linear with respect to
the first argument. An early version of \cite{lee2023nonparametric} gave
further examples illustrating why the nonlinearity of the utility functions
makes sense in economic studies.
\end{example}

The nonlinearity assumption has two different forms, depending on the number
of continuous covariates. When two or more continuous covariates are
available, Assumption NL2 will require some variation in $\pi _{0}$ to
distinguish it from a single-index function. Concretely, NL2 will not hold
if $\pi _{0}\left( x^{C}\right) =f\left( x^{C\prime }\gamma \right) $ for a
smooth function $f$, because in this case, both sides of the inequality in
(v) are equal to $\left. \gamma _{k}\right/ \gamma _{j}$. Otherwise,
however, it is difficult to construct examples that violate (v). Therefore,
we can informally test NL2 by constructing a test for the single-index
specification of the propensity score against a general nonparametric
alternative 
\citep[e.g.,][]{fan1996consistent, stute2005nonparametric,
xia2009model, maistre2019nonparametric}. In practice, NL2 can be fulfilled
even when $\pi _{0}$ is single-index specified, if interaction or/and
quadratic terms are added, as shown by the following Example.

\begin{example}
\label{example:nonlinear} Consider the case of two continuous covariates.
Suppose that for a smooth function $f$, $\pi _{0}\left( x^{C}\right)
=f\left( \gamma _{1}x_{1}^{C}+\gamma _{2}x_{2}^{C}+\gamma
_{3}x_{1}^{C}x_{2}^{C}\right) $ or $\pi _{0}\left( x^{C}\right) =f\left(
\gamma _{1}x_{1}^{C}+\gamma _{2}x_{2}^{C}+\gamma _{3}\left( x_{1}^{C}\right)
^{2}\right) $. Then, we obtain $\left. \partial _{1}\pi _{0}\left(
x^{C}\right) \right/ \partial _{2}\pi _{0}\left( x^{C}\right) =\left. \left(
\gamma _{1}+\gamma _{3}x_{2}^{C}\right) \right/ \left( \gamma _{2}+\gamma
_{3}x_{1}^{C}\right) $ for the interaction case, or $\left. \partial _{1}\pi
_{0}\left( x^{C}\right) \right/ \partial _{2}\pi _{0}\left( x^{C}\right)
=\left. \left( \gamma _{1}+2\gamma _{3}x_{1}^{C}\right) \right/ \gamma _{2}$
for the quadratic case. In both cases, Assumption NL2.(v) will generally
hold for $x^{C}$ and $\tilde{x}^{C}$ satisfying $x_{1}^{C}\neq \tilde{x}%
_{1}^{C}$.
\end{example}

Assumption NL2 is a form of the nonlinearity assumption of \cite%
{escanciano2016identification} that is specific to our model, and thus
inherits a demerit that it requires the existence of at least two continuous
covariates. However, in empirical studies based on survey data, most of the
demographic characteristics are documented as discrete or categorical
variables, such as age, gender, race, marital status, educational
attainment, and so on. Therefore, we also impose Assumption NL1, as a
supplement to NL2, to take into account the situation in which only one
continuous covariate is available. NL1 requires $\pi _{0}$ to have at least
one stationary point. Namely, NL1 holds if the probability of receiving
treatment is unaffected by some local change in the continuous covariate.
It's a differential version of the local irrelevance assumption imposed in
nonseparable models for attaining point identification %
\citep[e.g.,][]{torgovitsky2015identification, d2021testing}. Note that NL1
will hold if $\pi _{0}$ is flat everywhere, i.e., if the continuous
covariate has no influence on the treatment choice conditional on the
benchmark value of $X^{D}$. However, this extreme case will be ruled out by
the subsequent Assumption S. Given that, NL1 implies that $\pi _{0}$ is
necessarily nonlinear. NL1 can be extended to the case of no continuous
covariate, as discussed in Appendix \ref{appendix:NL}. However, using only
discrete covariates provides not much identifying variation, which may lead
to poor performance in the subsequent model estimation %
\citep{garlick2022quasi}. Hence, we focus on the case of at least one
continuous covariate. NL1 can be visually tested by plotting the estimated
propensity score as a function of the continuous covariate given a benchmark
value of $X^{D}$, as illustrated in Figure \ref{Fig:AssumptionNL1} in the
subsequent empirical application.

We note that Assumption NL doesn't exclude the widely specified linear-index
treatment equation, as long as the structural error is not independent of
covariates.

\begin{example}
\label{example:hetero} Suppose the treatment status is determined by a
linear-index threshold crossing rule:%
\begin{equation*}
D=1\left\{ X^{\prime }\gamma \geq U\right\} ,
\end{equation*}%
where $U$ is not independent of $X$. Namely, suppose $\mu \left( X\right)
=X^{\prime }\gamma $ in the structural treatment equation (\ref{treatment}).
Consider a multiplicatively heteroscedastic $U$ such that $U=\sigma \left(
X\right) \tilde{U}$, where $\sigma \left( X\right) $ is a positive function
and $\tilde{U}$ is an idiosyncratic error independent of $X$. The
normalization derives%
\begin{equation*}
D=1\left\{ \frac{X^{\prime }\gamma }{\sigma \left( X\right) }\geq \tilde{U}%
\right\} =1\left\{ F_{\tilde{U}}\left(\frac{X^{\prime }\gamma }{\sigma
\left( X\right) }\right)\geq V\right\},
\end{equation*}%
that is, the reduced-form error is defined as $V=F_{\tilde{U}}\left(\tilde{U}%
\right)$ and the propensity score function is equal to $\pi \left(x\right) =
F_{\tilde{U}}\left(\left. \left( x^{\prime }\gamma \right) \right/ \sigma
\left( x\right) \right)$. In general, $\pi \left( x\right) $ is a nonlinear
function. Another specification we consider is a linear-index model with
endogeneity in a certain component $X^{e}$ of $X^{C}$, in which $\pi \left(
x\right) =F_{U\left\vert X^{e}\right. }\left( \left. x^{\prime }\gamma
\right\vert x^{e}\right) $. Given that $\pi \left( x\right) $ is generally
highly nonlinear in $x^{e}$, Assumption NL1 holds straightforwardly, and NL2
holds with $x_{k}^{C}=x^{e}$.
\end{example}

Unlike in the binary response model, the ubiquitous heteroscedasticity and
endogeneity benefit our results while inducing no trouble, because the
identification of MTE is irrelevant to the structural coefficients in $%
\gamma $. Our MTE is defined by the reduced-form treatment error; thus, all
we need from the treatment equation is the propensity score, which has a
reduced-form nature.

Assumption NL ensures identification of coefficients of continuous
covariates. In order to identify coefficients of discrete covariates, we
impose a mild support condition. For any $x_{k}^{D}\neq 0$, we denote $%
x^{Dk} $ as the $\dim \left( X^{D}\right) \times 1$ vector with the $k$-th
coordinate being equal to $x_{k}^{D}$ and all the other coordinates being
equal to zero.

\medskip

\textbf{Assumption S} (Support). \textit{For each }$k\in \left\{ 1,2,\cdots
,\dim \left( X^{D}\right) \right\} $\textit{, assume for some }$%
x_{k}^{D}\neq 0$\textit{\ in the support of }$X_{k}^{D}$\textit{\ that there
exists }$x^{C}\left( k\right) $\textit{\ in the support of }$X^{C}$\textit{\
such that }$\pi \left( x^{C}\left( k\right) ,x^{Dk}\right) $\textit{\ is in
the support of }$\pi _{0}\left( X^{C}\right) $\textit{.}

\medskip

A sufficient condition for Assumption S to hold is that $\pi _{0}\left(
X^{C}\right) $ has a full support on the unit interval, or, more generally,
that the support of $\pi _{0}\left( X^{C}\right) $ is overlapped with those
of $\pi \left( X^{C},x^{D}\right) $ for all $x^{D}\neq 0$. Otherwise, we
have to find an $x_{k}^{D}\neq 0$ for each $k$ such that the support of $\pi
\left( X^{C},x^{Dk}\right) $ is overlapped with that of $\pi _{0}\left(
X^{C}\right) $.\ The following theorem establishes our main identification
result.

\begin{theorem}
\label{theorem:main} If Assumptions CMI, L, NL, and S hold, then $\beta _{d}$
and $g_{d}\left( p\right) $ at all $p$ in the support of the propensity
score $P$ are identified for $d=0,1$.
\end{theorem}

Theorem \ref{theorem:main} is an adaption of the identification result of 
\cite{escanciano2016identification}. Since \cite%
{escanciano2016identification} considered a general double index model with
an unknown link function, their identification result is only up to scale
and thus requires a scale normalization assumption. In comparison, Theorem %
\ref{theorem:main} can identify the magnitude of the linear coefficients
without any normalization. The proof of Theorem \ref{theorem:main} is
delegated to Appendix \ref{appendix:proof}. Briefly speaking, the proof is
grounded on the observed regression functions (\ref{md}), which summarize
the information from the data. As $g_{d}$ is unknown, we need to eliminate
it through some subtraction to realize the identification of $\beta _{d}$.
When only one continuous covariate exists, the subtraction can be carried
out locally around the point satisfying Assumption NL1. Otherwise, our
strategy is to perturb two continuous covariates, such as $x_{k}^{C}$ and $%
x_{j}^{C}$, in such a way that $\pi \left( x\right) $ remains unchanged.
Specifically, for each group $d$, we increase $x_{k}^{C}$ by a small $%
\epsilon $ and simultaneously change $x_{j}^{C}$ by $\epsilon $ multiplied
by $\left. -\partial _{k}\pi _{0}\left( x^{C}\right) \right/ \partial
_{j}\pi _{0}\left( x^{C}\right) $, resulting in a perturbed value of the
regression function. Subtracting the perturbed regression function from the
original (\ref{md}) will cancel out $g_{d}$ due to the equality of $\pi
\left( x\right) $, giving rise to an equation for $\beta _{d,k}^{C}$ and $%
\beta _{d,j}^{C}$. Note that the multiplier $\left. -\partial _{k}\pi
_{0}\left( x^{C}\right) \right/ \partial _{j}\pi _{0}\left( x^{C}\right) $
is the partial derivative of $x_{j}^{C}$ with respect to $x_{k}^{C}$ if we
view $x_{j}^{C}$ as an implicit function of the other continuous covariates
by equating $\pi _{0}\left( x^{C}\right) $ to a constant. Accordingly,
Assumption NL2.(v) implies two linearly independent equations, ensuring an
exact solution (i.e., identification) of $\beta _{d,k}^{C}$ and $\beta
_{d,j}^{C}$.

It is worth mentioning that our strategy naturally features
over-identification in the sense that if there is a pair of points
satisfying Assumption NL2, then there must be infinitely many pairs of
points for which NL2 holds, because $X^{C}$ is continuously distributed.
Moreover, there will generally be more than one value of $X^{D}$ that
satisfies Assumptions S, as illustrated by the application in Section \ref%
{sec:application}. Such values of $X^{D}$ can be taken as alternative
benchmarks, under which Assumption NL may hold for another set of (pairs of)
points of $X^{C}$. Consequently, the identification can be represented as
the average of solutions over all (pairs of) points of $X^{C}$ satisfying
Assumption NL and over all values of $X^{D}$ satisfying Assumptions S.

According to (\ref{MTEiden}), Theorem \ref{theorem:main} implies the
identification of MTE in the absence of IVs. Specifically, This result
allows practitioners to include all the relevant observed characteristics
into both treatment and outcome equations, without imposing any exclusion or
full independence assumptions. Under Theorem \ref{theorem:main}, the
conventional causal parameters can also be identified without instruments,
provided that the support of $P$ contains 0 and/or 1 (which implies
identifiability of $g_{d}\left( 0\right) $ and/or $g_{d}\left( 1\right) $): 
\begin{eqnarray}
\Delta ^{\text{ATE}}\left( x\right) &=&x^{\prime }\left( \beta _{1}-\beta
_{0}\right) +\left[ g_{1}\left( 1\right) -g_{0}\left( 0\right) \right] ,
\label{CPiden} \\
\Delta ^{\text{TT}}\left( x\right) &=&x^{\prime }\left( \beta _{1}-\beta
_{0}\right) +g_{1}\left( \pi \left( x\right) \right) +\frac{\left( 1-\pi
\left( x\right) \right) g_{0}\left( \pi \left( x\right) \right) -g_{0}\left(
0\right) }{\pi \left( x\right) },  \notag \\
\Delta ^{\text{TUT}}\left( x\right) &=&x^{\prime }\left( \beta _{1}-\beta
_{0}\right) +\frac{g_{1}\left( 1\right) -\pi \left( x\right) g_{1}\left( \pi
\left( x\right) \right) }{1-\pi \left( x\right) }-g_{0}\left( \pi \left(
x\right) \right) ,  \notag \\
\Delta ^{\text{LATE}}\left( x,v_{1},v_{2}\right) &=&x^{\prime }\left( \beta
_{1}-\beta _{0}\right) +\frac{v_{2}g_{1}\left( v_{2}\right)
-v_{1}g_{1}\left( v_{1}\right) +\left( 1-v_{2}\right) g_{0}\left(
v_{2}\right) -\left( 1-v_{1}\right) g_{0}\left( v_{1}\right) }{v_{2}-v_{1}}.
\notag
\end{eqnarray}

\section{Estimation}

\label{sec:estimation}Our identification strategy for the IV-free MTE
implies a separate estimation procedure that works with the partially linear
regression (\ref{md}) by each treatment status. In particular, we recommend
the kernel-weighted pairwise difference estimation method proposed by \cite%
{ahn1993semiparametric} because of its computational simplicity,
well-established asymptotic properties, and, most importantly, its capacity
to effectively leverage the overidentifying information that underlies the
data. Suppose that $\left\{ \left( Y_{i},D_{i},X_{i}\right) :i=1,2,\cdots
,n\right\} $ is a random sample of observations on $\left( Y,D,X\right) $.
In the first step, we estimate the nonparametrically specified propensity
score using the kernel method, that is, 
\begin{equation}
\hat{\pi}\left( x\right) =\frac{\sum_{i=1}^{n}D_{i}\left[ \prod_{l=1}^{\dim
\left( X^{C}\right) }k_{1}\left( \left. \left( X_{il}^{C}-x_{l}^{C}\right)
\right/ h_{1l}\right) \right] 1\left\{ X_{i}^{D}=x^{D}\right\} }{%
\sum_{i=1}^{n}\left[ \prod_{l=1}^{\dim \left( X^{C}\right) }k_{1}\left(
\left. \left( X_{il}^{C}-x_{l}^{C}\right) \right/ h_{1l}\right) \right]
1\left\{ X_{i}^{D}=x^{D}\right\} },  \label{PSfunction}
\end{equation}%
and construct $\hat{P}_{i}=\hat{\pi}\left( X_{i}\right) $ by leaving the $i$%
-th observation out in the estimation, where $h_{1l}$, $l=1,2,\cdots ,\dim
\left( X^{C}\right) $, are bandwidths and $k_{1}$ is a univariate kernel
function. If the dimension of $X^{D}$ is large, then a smoothed kernel for
discrete covariates \citep{racine2004nonparametric} can be applied as a
substitute for the indicator function, to alleviate the potential problem of
inadequate observations in each data cell divided by the support of $X^{D}$.
If the number of continuous covariates is not small either, then the
well-known curse of dimensionality will appear, and a linear-index
specification may thus be practically more relevant when modeling the
propensity score. The index should include a series of interaction terms and
quadratic or even higher-order terms of continuous covariates to supply
sufficient nonlinear variation. The linear-index propensity score can be
estimated by parametric probit/logit or semiparametric methods 
\citep[e.g.,][]{powell1989semiparametric,
ichimura1993semiparametric, klein1993efficient, lewbel2000semiparametric},
depending on the distributional assumption on the error term.

In the second step, we estimate $\beta _{d}$ for each $d$ through a weighted
pairwise difference least squares regression:%
\begin{eqnarray}
\hat{\beta}_{d} &=&\arg \min_{\beta }\sum_{i=1}^{n-1}\sum_{j=i+1}^{n}\hat{%
\omega}_{dij}\left[ \left( Y_{i}-Y_{j}\right) -\left( X_{i}-X_{j}\right)
^{\prime }\beta \right] ^{2}  \notag \\
&=&\left[ \sum_{i=1}^{n-1}\sum_{j=i+1}^{n}\hat{\omega}_{dij}\left(
X_{i}-X_{j}\right) \left( X_{i}-X_{j}\right) ^{\prime }\right] ^{-1}\left[
\sum_{i=1}^{n-1}\sum_{j=i+1}^{n}\hat{\omega}_{dij}\left( X_{i}-X_{j}\right)
\left( Y_{i}-Y_{j}\right) \right] ,  \label{PDLS}
\end{eqnarray}%
where the weights are given by%
\begin{equation*}
\hat{\omega}_{dij}=1\left\{ D_{i}=D_{j}=d\right\} \frac{1}{h_{2}}k_{2}\left( 
\frac{\hat{P}_{i}-\hat{P}_{j}}{h_{2}}\right) ,
\end{equation*}%
with $h_{2}$ and $k_{2}$ being the bandwidth and kernel function,
respectively, which can be different from those in the first step. Given $%
\hat{\beta}_{d}$, the nonlinear function $g_{d}\left( p\right) $ and its
derivative function $g_{d}^{\left( 1\right) }\left( p\right) $ at any $p$ in
the support of $P$ can be estimated by the local linear method, namely,%
\begin{equation*}
\left( 
\begin{array}{c}
\hat{g}_{d}\left( p\right) \\ 
\hat{g}_{d}^{\left( 1\right) }\left( p\right)%
\end{array}%
\right) =\left[ \sum_{i=1}^{n}\hat{w}_{di}\left( p\right) \left( 
\begin{array}{c}
1 \\ 
\hat{P}_{i}-p%
\end{array}%
\right) \left( 
\begin{array}{c}
1 \\ 
\hat{P}_{i}-p%
\end{array}%
\right) ^{\prime }\right] ^{-1}\left[ \sum_{i=1}^{n}\hat{w}_{di}\left(
p\right) \left( 
\begin{array}{c}
1 \\ 
\hat{P}_{i}-p%
\end{array}%
\right) \left( Y_{i}-X_{i}^{\prime }\hat{\beta}_{d}\right) \right] ,
\end{equation*}%
where%
\begin{equation*}
\hat{w}_{di}\left( p\right) =1\left\{ D_{i}=d\right\} \frac{1}{h_{3}}%
k_{3}\left( \frac{\hat{P}_{i}-p}{h_{3}}\right) .
\end{equation*}%
Finally, we plug $\hat{\beta}_{d}$, $\hat{g}_{d}$, $\hat{g}_{d}^{\left(
1\right) }$, and $\hat{\pi}$ into identification equations (\ref{MTEiden})
and (\ref{CPiden}) to estimate the IV-free MTE and other causal parameters,
as follows:%
\begin{eqnarray}
\hat{\Delta}^{\text{MTE}}\left( x,v\right) &=&x^{\prime }\left( \hat{\beta}%
_{1}-\hat{\beta}_{0}\right) +\left[ \hat{g}_{1}\left( v\right) -\hat{g}%
_{0}\left( v\right) \right] +v\hat{g}_{1}^{\left( 1\right) }\left( v\right)
+\left( 1-v\right) \hat{g}_{0}^{\left( 1\right) }\left( v\right) ,
\label{MTEhat} \\
\hat{\Delta}^{\text{ATE}}\left( x\right) &=&x^{\prime }\left( \hat{\beta}%
_{1}-\hat{\beta}_{0}\right) +\left[ \hat{g}_{1}\left( 1\right) -\hat{g}%
_{0}\left( 0\right) \right] ,  \notag \\
\hat{\Delta}^{\text{TT}}\left( x\right) &=&x^{\prime }\left( \hat{\beta}_{1}-%
\hat{\beta}_{0}\right) +\hat{g}_{1}\left( \hat{\pi}\left( x\right) \right) +%
\frac{\left( 1-\hat{\pi}\left( x\right) \right) \hat{g}_{0}\left( \hat{\pi}%
\left( x\right) \right) -\hat{g}_{0}\left( 0\right) }{\hat{\pi}\left(
x\right) },  \notag \\
\hat{\Delta}^{\text{TUT}}\left( x\right) &=&x^{\prime }\left( \hat{\beta}%
_{1}-\hat{\beta}_{0}\right) +\frac{\hat{g}_{1}\left( 1\right) -\hat{\pi}%
\left( x\right) \hat{g}_{1}\left( \hat{\pi}\left( x\right) \right) }{1-\hat{%
\pi}\left( x\right) }-\hat{g}_{0}\left( \hat{\pi}\left( x\right) \right) , 
\notag \\
\hat{\Delta}^{\text{LATE}}\left( x,v_{1},v_{2}\right) &=&x^{\prime }\left( 
\hat{\beta}_{1}-\hat{\beta}_{0}\right) +\frac{v_{2}\hat{g}_{1}\left(
v_{2}\right) -v_{1}\hat{g}_{1}\left( v_{1}\right) +\left( 1-v_{2}\right) 
\hat{g}_{0}\left( v_{2}\right) -\left( 1-v_{1}\right) \hat{g}_{0}\left(
v_{1}\right) }{v_{2}-v_{1}}.  \notag
\end{eqnarray}

\begin{theorem}
\label{theorem:AN1} Suppose the assumptions of Theorem \ref{theorem:main}
hold. Further suppose Assumptions E.1-E.4 given in Appendix \ref{appendix:AN}
hold. Then for any interior point $v$ of the support of $P$ and any $x$
satisfying $0<\pi \left( x\right) <1$, we have%
\begin{eqnarray*}
\sqrt{nh_{3}^{3}}\left[ \hat{\Delta}^{\text{MTE}}\left( x,v\right) -\Delta ^{%
\text{MTE}}\left( x,v\right) \right] &\rightarrow &N\left( 0,\sigma _{\text{%
MTE}}^{2}\left( v\right) \right) , \\
\sqrt{nh_{3}}\left[ \hat{\Delta}^{\text{ATE}}\left( x\right) -\Delta ^{\text{%
ATE}}\left( x\right) -b_{\text{ATE}}h_{3}^{2}\right] &\rightarrow &N\left(
0,\sigma _{\text{ATE}}^{2}\right) , \\
\sqrt{nh_{3}}\left[ \hat{\Delta}^{\text{TT}}\left( x\right) -\Delta ^{\text{%
TT}}\left( x\right) -b_{\text{TT}}\left( \pi \left( x\right) \right)
h_{3}^{2}\right] &\rightarrow &N\left( 0,\sigma _{\text{TT}}^{2}\left( \pi
\left( x\right) \right) \right) , \\
\sqrt{nh_{3}}\left[ \hat{\Delta}^{\text{TUT}}\left( x\right) -\Delta ^{\text{%
TUT}}\left( x\right) -b_{\text{TUT}}\left( \pi \left( x\right) \right)
h_{3}^{2}\right] &\rightarrow &N\left( 0,\sigma _{\text{TUT}}^{2}\left( \pi
\left( x\right) \right) \right) , \\
\sqrt{nh_{3}}\left[ \hat{\Delta}^{\text{LATE}}\left( x,v_{1},v_{2}\right)
-\Delta ^{\text{LATE}}\left( x,v_{1},v_{2}\right) -b_{\text{LATE}}\left(
v_{1},v_{2}\right) h_{3}^{2}\right] &\rightarrow &N\left( 0,\sigma _{\text{%
LATE}}^{2}\left( v_{1},v_{2}\right) \right) ,
\end{eqnarray*}%
where the asymptotic variances $\sigma _{\text{MTE}}^{2},\cdots ,\sigma _{%
\text{LATE}}^{2}$ are defined in (\ref{sigmaMTE})-(\ref{sigmaLATE}),
respectively, and the asymptotic bias terms are defined as 
\begin{eqnarray*}
b_{\text{ATE}} &=&\frac{\kappa _{2}}{2}\left( g_{1}^{\left( 2\right) }\left(
1\right) -g_{0}^{\left( 2\right) }\left( 0\right) \right) , \\
b_{\text{TT}}\left( p\right) &=&\frac{\kappa _{2}}{2}\left( g_{1}^{\left(
2\right) }\left( p\right) +\frac{\left( 1-p\right) g_{0}^{\left( 2\right)
}\left( p\right) -g_{0}^{\left( 2\right) }\left( 0\right) }{p}\right) , \\
b_{\text{TUT}}\left( p\right) &=&\frac{\kappa _{2}}{2}\left( \frac{%
g_{1}^{\left( 2\right) }\left( 1\right) -pg_{1}^{\left( 2\right) }\left(
p\right) }{1-p}-g_{0}^{\left( 2\right) }\left( p\right) \right) , \\
b_{\text{LATE}}\left( v_{1},v_{2}\right) &=&\frac{\kappa _{2}}{2}\left( 
\frac{v_{2}g_{1}^{\left( 2\right) }\left( v_{2}\right) -v_{1}g_{1}^{\left(
2\right) }\left( v_{1}\right) +\left( 1-v_{2}\right) g_{0}^{\left( 2\right)
}\left( v_{2}\right) -\left( 1-v_{1}\right) g_{0}^{\left( 2\right) }\left(
v_{1}\right) }{v_{2}-v_{1}}\right) ,
\end{eqnarray*}%
with $\kappa _{2}=\int k_{3}\left( u\right) u^{2}du$.
\end{theorem}

The semiparametric MTE estimator $\hat{\Delta}^{\text{MTE}}\left( x,v\right) 
$ converges at the same rate $O_{p}\left( \left( nh_{3}^{3}\right)
^{-1/2}\right) $ as the local linear estimator for the derivative of a
regression function \citep[e.g.,][Theorem 2.7]{li2007nonparametric}, while
the estimators of the aggregated causal parameters converge at the same rate 
$O_{p}\left( \left( nh_{3}\right) ^{-1/2}\right) $ as the kernel estimator
for a regression function itself. This is because the estimation of MTE
involves plugging in $\hat{g}_{d}^{\left( 1\right) }$, the nonparametric
derivative estimate, which converges slower than $\hat{g}_{d}$, the
nonparametric regression function estimate. Taking the usually applied
rule-of-thumb bandwidth $h_{3}\sim n^{-1/5}$ for instance, the MTE estimator
will converge at the rate $O_{p}\left( n^{-1/5}\right) $ and the aggregated
causal parameters will converge at the rate $O_{p}\left( n^{-2/5}\right) $.
Since the slope estimate $\hat{\beta}_{d}$ converges at the parametric rate $%
O_{p}\left( n^{-1/2}\right) $ that is always faster than $\hat{g}%
_{d}^{\left( 1\right) }$ and $\hat{g}_{d}$, the observed part $x^{\prime
}\left( \hat{\beta}_{1}-\hat{\beta}_{0}\right) $ is asymptotically
negligible and has no impact on the asymptotic distributions of the
estimators. Therefore, for the causal parameters that $x$ does not enter the
unobserved part such as the MTE, ATE, and LATE, the asymptotic bias and
variance of their estimators will not depend on the observed characteristics.

The kernel-weighted pairwise difference estimator has the advantage of
having a closed-form expression, so we need not solve any formidable
optimization problems. However, it faces the challenging problem of
bandwidth selection as most semiparametric estimation methods.
Alternatively, we can consider imposing a parametric specification on the
unobservable heterogeneity of the MTE such that $E\left[ \left.
U_{d}\right\vert V=v\right] =E\left[ \left. U_{d}\right\vert V=v;\theta _{d}%
\right] $ for finite dimensional $\theta _{d}$ \citep{heckman2005structural}%
, e.g., the polynomial specification such that $E\left[ \left.
U_{d}\right\vert V=v\right] =\sum_{j=0}^{J}\theta _{dj}v^{j}$ or the normal
polynomial specification such that $E\left[ \left. U_{d}\right\vert V=v%
\right] =\sum_{j=0}^{J}\theta _{dj}\Phi ^{-j}\left( v\right) $, for a fixed $%
J$. In the latter, setting $J=1$ will match Heckman's normal sample
selection model. Under the parametric restriction, the second step becomes a
global regression $E\left[ Y\left\vert X,D=d\right. \right] =X^{\prime
}\beta _{d}+g_{d}\left( P;\theta _{d}\right) $ with parameterized
selectivity bias correction term $g_{d}\left( p;\theta _{d}\right) $ for
each $d$; therefore, the tuning of $h_{2}$ and $h_{3}$ is circumvented.
Another advantage of a parametrically specified second step is the lower
computational burden relative to the pairwise difference estimator that is
defined by double summation which requires a squared amount of calculation %
\citep{pan2023penalized}. The simulation in Appendix \ref{appendix:simulate}
shows that MTE estimators with parametric second-step perform well when the
specification is correct or nearly correct, while the semiparametric MTE
estimator performs robustly across different designs.

Notably, the local IV (LIV) estimation procedure can be adapted to our model
as well, though it needs no IV. Unlike the separate estimation procedure,
the adapted LIV approach is based on a whole-sample regression:%
\begin{eqnarray}
E\left[ Y\left\vert X\right. \right] &=&E\left[ Y_{0}+D\left(
Y_{1}-Y_{0}\right) \left\vert X\right. \right]  \notag \\
&=&E\left[ Y_{0}\left\vert X\right. \right] +E\left[ \left.
Y_{1}-Y_{0}\right\vert X,D=1\right] \Pr \left( D=1\left\vert X\right. \right)
\notag \\
&=&\alpha _{0}+X^{\prime }\beta _{0}+PX^{\prime }\left( \beta _{1}-\beta
_{0}\right) +q\left( P\right) ,  \label{aLIV}
\end{eqnarray}%
where%
\begin{equation*}
q\left( p\right) =pE\left[ U_{1}-U_{0}\left\vert V\leq p\right. \right]
=\int_{0}^{p}E\left[ U_{1}-U_{0}\left\vert V=v\right. \right] dv.
\end{equation*}%
Note that since%
\begin{equation*}
q^{\left( 1\right) }\left( p\right) \equiv \frac{dq\left( p\right) }{dp}=E%
\left[ U_{1}-U_{0}\left\vert V=p\right. \right] ,
\end{equation*}%
the MTE is equal to the derivative of the regression function (\ref{aLIV})
with respect to $P$. As a consequence, it would be sufficient to estimate $%
\beta _{d}$ for $d=0,1$ and functions $q$ and $q^{\left( 1\right) }$. Given
the estimated propensity score, we can likewise use the pairwise difference
principle to obtain%
\begin{eqnarray*}
\left( 
\begin{array}{c}
\check{\beta}_{0} \\ 
\check{\delta}%
\end{array}%
\right) &=&\left[ \sum_{i=1}^{n-1}\sum_{j=i+1}^{n}k_{2}\left( \frac{\hat{P}%
_{i}-\hat{P}_{j}}{h_{2}}\right) \left( 
\begin{array}{c}
X_{i}-X_{j} \\ 
\hat{P}_{i}X_{i}-\hat{P}_{j}X_{j}%
\end{array}%
\right) \left( 
\begin{array}{c}
X_{i}-X_{j} \\ 
\hat{P}_{i}X_{i}-\hat{P}_{j}X_{j}%
\end{array}%
\right) ^{\prime }\right] ^{-1} \\
&&\cdot \left[ \sum_{i=1}^{n-1}\sum_{j=i+1}^{n}k_{2}\left( \frac{\hat{P}_{i}-%
\hat{P}_{j}}{h_{2}}\right) \left( 
\begin{array}{c}
X_{i}-X_{j} \\ 
\hat{P}_{i}X_{i}-\hat{P}_{j}X_{j}%
\end{array}%
\right) \left( Y_{i}-Y_{j}\right) \right] ,
\end{eqnarray*}%
where $\check{\delta}$ is an estimator for $\beta _{1}-\beta _{0}$ according
to (\ref{aLIV}). Given $\check{\beta}_{0}$ and $\check{\delta}$, we apply
the local linear method as well, yielding%
\begin{eqnarray*}
\left( 
\begin{array}{c}
\check{r}\left( p\right) \\ 
\check{q}^{\left( 1\right) }\left( p\right)%
\end{array}%
\right) &=&\left[ \sum_{i=1}^{n}k_{3}\left( \frac{\hat{P}_{i}-p}{h_{3}}%
\right) \left( 
\begin{array}{c}
1 \\ 
\hat{P}_{i}-p%
\end{array}%
\right) \left( 
\begin{array}{c}
1 \\ 
\hat{P}_{i}-p%
\end{array}%
\right) ^{\prime }\right] ^{-1} \\
&&\cdot \left[ \sum_{i=1}^{n}k_{3}\left( \frac{\hat{P}_{i}-p}{h_{3}}\right)
\left( 
\begin{array}{c}
1 \\ 
\hat{P}_{i}-p%
\end{array}%
\right) \left( Y_{i}-X_{i}^{\prime }\check{\beta}_{0}-\hat{P}%
_{i}X_{i}^{\prime }\check{\delta}\right) \right] ,
\end{eqnarray*}%
where $\check{r}\left( p\right) $ is an estimator for $\alpha _{0}+q\left(
p\right) $. Then, we construct the instrument-free LIV estimators for the
MTE and other treatment parameters as 
\begin{eqnarray*}
\check{\Delta}^{\text{MTE}}\left( x,v\right) &=&x^{\prime }\check{\delta}+%
\check{q}^{\left( 1\right) }\left( v\right) , \\
\check{\Delta}^{\text{ATE}}\left( x\right) &=&x^{\prime }\check{\delta}+%
\check{r}\left( 1\right) -\check{r}\left( 0\right) , \\
\check{\Delta}^{\text{TT}}\left( x\right) &=&x^{\prime }\check{\delta}+\frac{%
\check{r}\left( \hat{\pi}\left( x\right) \right) -\check{r}\left( 0\right) }{%
\hat{\pi}\left( x\right) }, \\
\check{\Delta}^{\text{TUT}}\left( x\right) &=&x^{\prime }\check{\delta}+%
\frac{\check{r}\left( 1\right) -\check{r}\left( \hat{\pi}\left( x\right)
\right) }{1-\hat{\pi}\left( x\right) }, \\
\check{\Delta}^{\text{LATE}}\left( x,v_{1},v_{2}\right) &=&x^{\prime }\check{%
\delta}+\frac{\check{r}\left( v_{2}\right) -\check{r}\left( v_{1}\right) }{%
v_{2}-v_{1}}.
\end{eqnarray*}%
The asymptotic normality of these adapted LIV estimators can be established
in a way analogous to Theorem \ref{theorem:AN1}. Furthermore, if we accept a
parametric specification on $E\left[ U_{1}-U_{0}\left\vert V=v\right. \right]
$ and thus on $q\left( p\right) $, the adapted LIV approach can be
implemented through a parametric least squares regression.

In summary, the identification based on functional forms can accommodate
most of the frequently-used estimation procedures in IV-based MTE models.
However, since it is infeasible to manipulate the nonlinear variation, our
identification result does not allow for the counterfactual analysis via
estimating the policy-relevant treatment effect (PRTE) which is also an
important application of IV-based MTE models \citep{heckman2005structural}.

\section{Empirical application}

\label{sec:application}In this section, we revisit the heterogeneous
long-term effects of the Head Start program on educational attainment and
labor market outcomes by using the proposed IV-free MTE method. Head Start,
which began in 1965, is one of the largest early child care programs in the
United States. The program is targeted at children from low-income families
and can provide such children with preschool, health, and nutritional
services. Currently, Head Start serves more than a million children, at an
annual cost of 10 billion dollars. As a federally funded large-scale
program, Head Start has encountered concerns about its effectiveness and
thus spawned numerous studies to evaluate its educational and economic
effects on the participants. Early studies focused mainly on short-term
benefits and found that participation in Head Start is associated with
improved test scores and reduced grade repetition at the beginning of
primary schooling. However, such benefits seem to fade out during the upper
primary grades \citep{currie2001early}. \cite{garces2002longer} provided the
first empirical evidence for the longer-term effects of Head Start on high
school completion, college attendance, earnings, and crime. Since then, the
literature has shifted its focus to the medium- and long-term or
intergenerational \citep[e.g.,][]{barr2022breaking} gains of Head Start
enrollment.

However, despite the enormous policy interest, evidence for the long-term
effectiveness of Head Start is not unified, as summarized in Figure 1 in 
\cite{de2020head}. The lack of consistency between these results may be due
to differences in the population or problems related to the empirical
approach \citep{elango2016early}. For example, the LATE obtained by the
family fixed-effects approach 
\citep[e.g.,][]{garces2002longer, deming2009early,
bauer2016long} relies on families that differ from other Head Start families
in size and in other observable dimensions. Moreover, the sibling comparison
design underlying that approach is limited by endogeneity concerns. To
reconcile the divergent evidence, \cite{de2020head} evaluated the
heterogeneous long-term effects of Head Start by using a distributional
treatment effect approach that relies on two weak stochastic dominance
assumptions, instead of restrictive IV assumptions. The authors found
substantial heterogeneity in the returns to Head Start. Specifically, they
found that the program has positive and statistically significant effects on
education and wage income for the lower end of the distribution of
participants.

To produce a complete picture, we assess the causal heterogeneity of Head
Start from another perspective; namely, we examine the effects across
different levels of unobserved resistance to participation in the program,
rather than across the distribution of long-term outcomes. By relating the
treatment effects to participation decision, the MTE is informative about
the nature of selection into treatment and allows the computation of various
causal parameters, such as the ATE, TT, and TUT. Another feature of the MTE
is that its description of the effect heterogeneity is irrelevant to the
specific outcome variable. Instead, the MTE curve depicts the treatment
effects on the unobserved determinants of the treatment. This is another
advantage of the MTE in the case of multiple outcomes of interest, as in
this application, where the interpretation of the heterogeneity is kept
consistent across different outcomes. Finally, in contrast to the
distributional treatment effects partially identified by \cite{de2020head},
our MTE method can achieve point identification and estimation.

We use the data provided by \cite{de2020head}, which are from Round 16 (1994
survey year) of the National Longitudinal Study of Youth 1979 (NLSY79). The
sample is restricted to the 1960--64 cohorts, because the first cohort
eligible for Head Start was born in 1960. In addition, the sample excludes
individuals who participated in any preschool programs other than Head
Start, implying that we estimate the returns to Head Start relative to
informal care. The treatment variable is whether the respondents attended
the Head Start program as a child, and the outcome variables are the
respondents' highest years of education and logarithmic yearly wage incomes
in their early 30s (they were between 30 and 34 years old in 1994). The
covariates are age, gender, race, parental education, and family income in
1978. We refer the reader to the original paper for additional details on
the data, sample, and variables. For our analysis, since we apply
nonparametric estimation in the first step, we recode parental education
into two categories to reduce the number of the data cells or subsamples
split by different values of discrete covariates. Specifically, we redefine
parental education as a binary variable equal to one if at least one parent
went to college or zero if both parents are high school graduates or lower.

We first verify the credibility of Assumptions S and NL for the Head Start
data. Table \ref{table:support} in Appendix \ref{appendix:implement}
lists the support of the nonparametrically
estimated propensity score for all data cells that are split by different
values of discrete covariates. We find that the 34th data cell has a full
support on the unit interval. Therefore, Assumption S must hold if we choose
the corresponding values (32 years old, female, black race, parental
education being college or higher) as the benchmark. We then verify
Assumption NL for the 34th data cell. Note that only one continuous
covariate exists in the data, that is, family income in 1978. Hence, we
invoke Assumption NL1, which requires the propensity score function to have
at least one stationary point. The left panel of Figure \ref%
{Fig:AssumptionNL1} plots the estimated propensity score as a univariate
function of the continuous covariate for the 34th data cell. We find that
the propensity score is highly nonlinear and nonmonotonic in family income
in 1978, possibly owing to the parents' various self-selection on
participation in the program. So Assumption NL1 is fulfilled.

According to our identification result, it is enough to find one benchmark
value of $X^{D}$ for which the support of the propensity score satisfies
Assumption S and the propensity score function satisfies Assumption NL, as
fulfilled above. In another application, however, the chances are
that the propensity score function corresponding to the fully supported data
cell violates Assumption NL or even there is no fully supported data cell.
If this is the case, we need to try values of $X^{D}$ for which the
verification of Assumption S may be not so trivial. As an example, we
consider an additional data cell and retest Assumptions S and NL. We choose
the largest data cell (i.e., the 29th data cell with 175 observations) and
reset the benchmark to be 32 years old, male, white race, and high school or
lower parental education. Now closer examination of Assumption S is
necessary, as the support for this data cell is limited to $\left[ 0,0.194%
\right] $, which is totally separate from the supports for some other data
cells. To this end, we need to find a distinct value of each discrete
variable such that the propensity score has overlapping support with $\left[
0,0.194\right] $. For example, if the value of age is altered from 32 to 30
(or 31, 33, 34) while the other discrete variables remain unaltered, we will
go to the 5th (or 17th, 41st, 53rd) data cell with support $\left[ 0,0.177%
\right] $ (or $\left[ 0,0.242\right] ,\left[ 0,0.154\right] ,\left[ 0,0.204%
\right] $), which is overlapped with $\left[ 0,0.194\right] $ as required.
Altering the values of gender, race, and parental education leads to the
35th, 25th (or 27th), and 30th data cells with supports $\left[ 0,0.526%
\right] $, $\left[ 0,0.771\right] $ (or $\left[ 0,0.623\right] $), and $%
\left[ 0,0.061\right] $, respectively, which all meet the overlapping
condition as well. Therefore, Assumption S holds for the largest data cell.
The validation of Assumption NL1 for this data cell can be demonstrated as
well by the existence of stationary points of the estimated propensity score
function, as illustrated in the right panel of Figure \ref{Fig:AssumptionNL1}%
.

We next turn to the estimation. In the first step, we nonparametrically
regress the treatment variable (Head Start) on all of the covariates to
generate propensity scores for the sample, by employing the kernel
estimation method in (\ref{PSfunction}) with the Gaussian kernel and the
rule-of-thumb bandwidth as given in Table \ref{table:implement}. Figure \ref%
{Fig:CommomSupport} plots the frequency distribution of the estimated
propensity score by treatment status. The figure shows that the propensity
score in our sample follows a bimodal distribution, with the main peak being
at approximately 0.5 for the participants and approximately 0.1 for the
nonparticipants. To reduce the potential impact of outliers, we trim the
observations of the 1\% smallest and 1\% largest propensity scores in the
following steps. This trimming leads to a common support ranging from 0 to
0.6, as indicated by the two dashed vertical lines. Given the estimated
propensity score, we then estimate in sequence the linear coefficients $%
\beta _{1}$ and $\beta _{0}$, MTE, and summary treatment effect measures
including ATE, TT, and TUT. Since the separate estimation procedure exploits
more identifying information behind the data than the LIV procedure in
general \citep{brinch2017beyond}, we focus on the results of the separate
estimation.

Figure \ref{Fig:MTEnormal} shows the estimated MTE under Heckman's normal
specification in which $E\left[ \left. U_{d}\right\vert V=v\right] =\rho
_{d0}+\rho _{dV}\Phi ^{-1}\left( v\right) $ for $d=0,1$, where $\Phi \left(
\cdot \right) $ is the CDF of standard normal distribution, which can be
derived by assuming that $\left( U_{d},U\right) $ follows a bivariate normal
distribution as in Example \ref{example:depend}. The MTE curves, evaluated
at the mean values of $X$, relate the unobserved component $U_{1}-U_{0}$ of
the treatment effect to the unobserved component $V$ of the treatment
choice. A high value of $V$ implies a low probability of treatment; thus, we
interpret $V$ as resistance to participation in Head Start. The MTE curve
for wage income plotted in panel B decreases with resistance, revealing a
pattern of selection on gains, as expected. In other words, based on the
unobserved characteristics, the children who were most likely to enroll in
Head Start benefitted the most from the program in terms of their labor
market outcome. However, when the outcome is educational attainment in panel
A, the positive slope of the MTE curve points to a pattern of reverse
selection on gains. In consequence, the TUT exceeds the ATE, which in turn
exceeds the TT. The same pattern was observed by \cite%
{cornelissen2018benefits} when estimating school-readiness return to a
preschool program in Germany. This phenomenon may be attributed partly to
the fact that parents have their own objectives in deciding childcare
arrangements. Nevertheless, the disagreement between selection patterns for
the two outcomes raises concern about the possible functional form
misspecification of the normal MTE, which is strictly restricted to be
monotonic in the resistance to treatment. Therefore, we consider a
nonparametric specification for $E\left[ \left. U_{d}\right\vert V=v\right] $
and implement a semiparametric separate estimation procedure.

Figure \ref{Fig:MTEsemi} plots the semiparametric MTE curves for education
and labor market outcomes in panels A and B, respectively. Under the
flexible specification, the MTE curves are no longer monotonic, and the
clear pattern of selection disappears. In the case of education outcome, the
curve is initially flat, then becomes an inverted U shape, with a
statistically significant positive effect appearing in the region of the
peak, corresponding to the children with resistance to treatment ranging
from 0.32 to 0.42. The complex feature of this curve is hardly captured by
any parsimonious parametric function. Similar observations are seen in the
case of wage income, in which the MTE curve is nonmonotonic, with a complex
shape, and significantly greater than zero for less than 10\% of the
children who were most likely to attend childcare early. A comparison of the
summary treatment effect measures indicates weak selection on gains for both
outcomes. Table \ref{table1} reports the semiparametric estimates for the
effects of the covariates on potential outcomes and their difference based
on (\ref{PDLS}). Columns 3 and 6 show that girls gain significantly more
returns to Head Start attendance than boys. However, other than gender, no
substantial observable heterogeneity exists in the treatment effects of the
program, though parental education and family income in 1978 have a
significantly positive effect on the respondents' potential education and
potential labor market outcomes in both the treated and untreated states.

\section{Conclusion}

\label{sec:conclusion}We propose a novel method for defining, identifying,
and estimating the MTE in the absence of IVs. Our MTE model allows all the
covariates to be correlated to the structural treatment error. In this
model, we define the MTE based on a reduced-form treatment error that (i) is
uniformly distributed on the unit interval, (ii) is statistically
independent of covariates, and (iii) has several economic meanings. The
independence property facilitates the identification of our defined MTE. We
provide sufficient conditions under which the proposed IV-free MTE can be
point identified based on functional forms. The conditions are standard in a
certain sense. The conditional mean independence assumption is equivalent to
the separability assumption commonly imposed in the MTE literature, and is
implied by and much weaker than the full independence assumption. The
linearity and nonlinearity assumptions are the foundation of identification
based on functional forms, and can make sense in most empirical studies. We
prove the identification by using a construction method. Our identification
strategy allows the adaptation of most of the existing estimation procedures
for the IV-based MTE, such as separate estimation, LIV estimation,
parametric estimation, and semiparametric estimation. For the empirical
application, we evaluate the MTE of the Head Start program on long-term
education and labor market outcomes, in which an IV for Head Start
participation is difficult to acquire. We find significant positive effects
for individuals with\ medium-level or low resistance to treatment, and
substantial heterogeneity exists.

\section*{Acknowledgements}

\addcontentsline{toc}{section}{Acknowledgements}

We wish to thank the Editor and two anonymous reviewers for their valuable
comments and suggestions. Financial supports for this research are provided
through the National Social Science Foundation of China (Grant 24FTJB010) and the National Natural Science Foundation of China (Grants 72173142, 72034006, 72342034, 72173083).

\bibliographystyle{Chicago}
\bibliography{0RefMTE}
\addcontentsline{toc}{section}{References}

\clearpage

\begin{figure}[tbh]
\caption{An informal test for Assumption NL1}
\label{Fig:AssumptionNL1}%
\begin{minipage}[t]{0.49\linewidth}
  \centerline{\includegraphics[width=7.626cm,height=5.55cm]{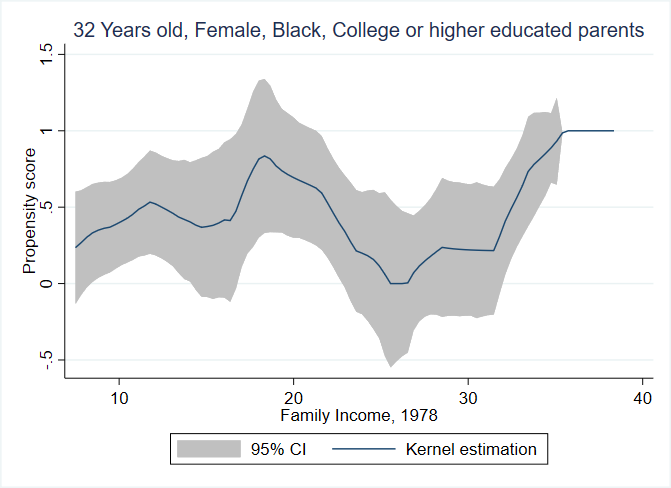}}
\end{minipage}
\begin{minipage}[t]{0.49\linewidth}
  \centerline{\includegraphics[width=7.626cm,height=5.55cm]{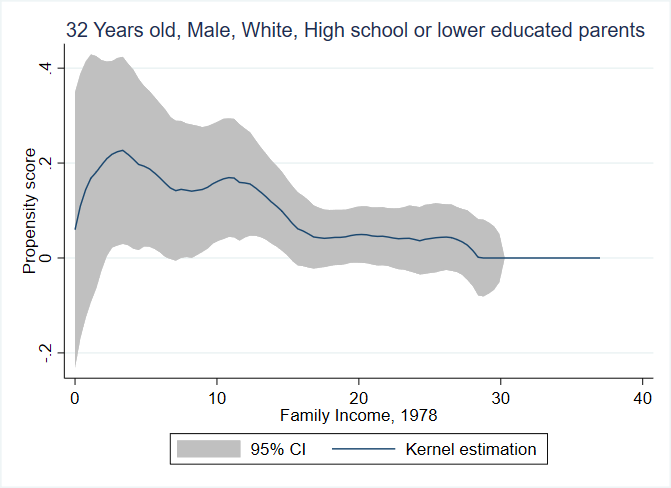}}
\end{minipage}
{\small Notes: The figures plot the nonparametrically estimated propensity
score as a function of the continuous covariate, using in the left panel the
data cell where the estimated propensity score has a full support (i.e., the
34th data cell in Table \ref{table:support}), and using in the right panel
the largest data cell that has 175 observations (i.e., the 29th data cell in
Table \ref{table:support}).}
\end{figure}

\begin{figure}[!htb]
\caption{Common support}
\label{Fig:CommomSupport}%
\centerline{\includegraphics[width=12.76cm,height=9.28cm]{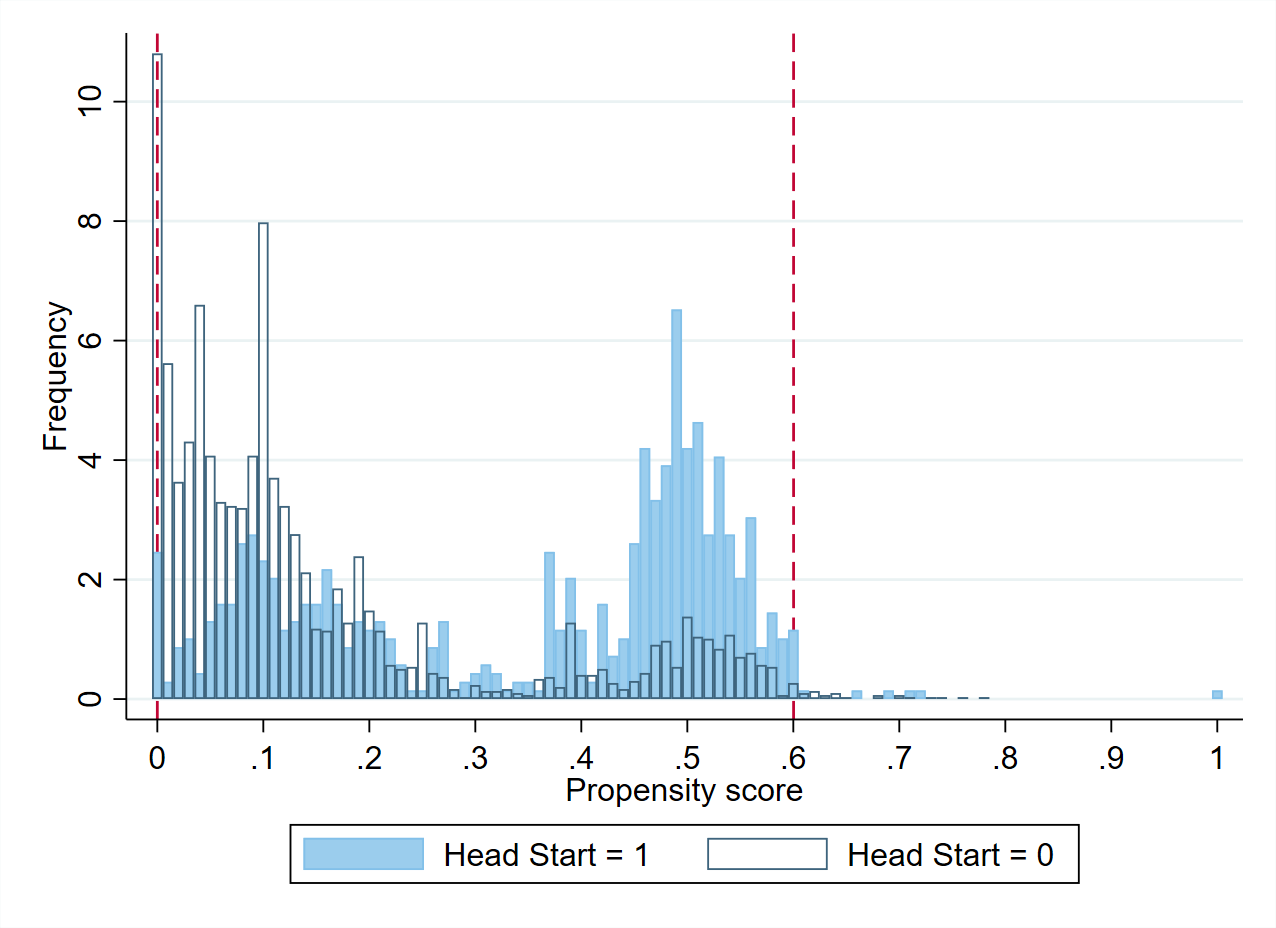}} {\small Notes: The figure plots the frequency distribution of the nonparametrically
estimated propensity score by treatment status. The dashed reference lines
indicate the lower limit (0) and upper limit (0.6) of the propensity score
with common support (based on 1\% trimming on both sides).}
\end{figure}

\begin{figure}[htb]
\caption{MTE curves under normal specification}
\label{Fig:MTEnormal}
\bigskip \setlength{\abovecaptionskip}{-0.2cm} 
\begin{minipage}[t]{0.48\linewidth}
  \centerline{\includegraphics[width=7.626cm,height=5.55cm]{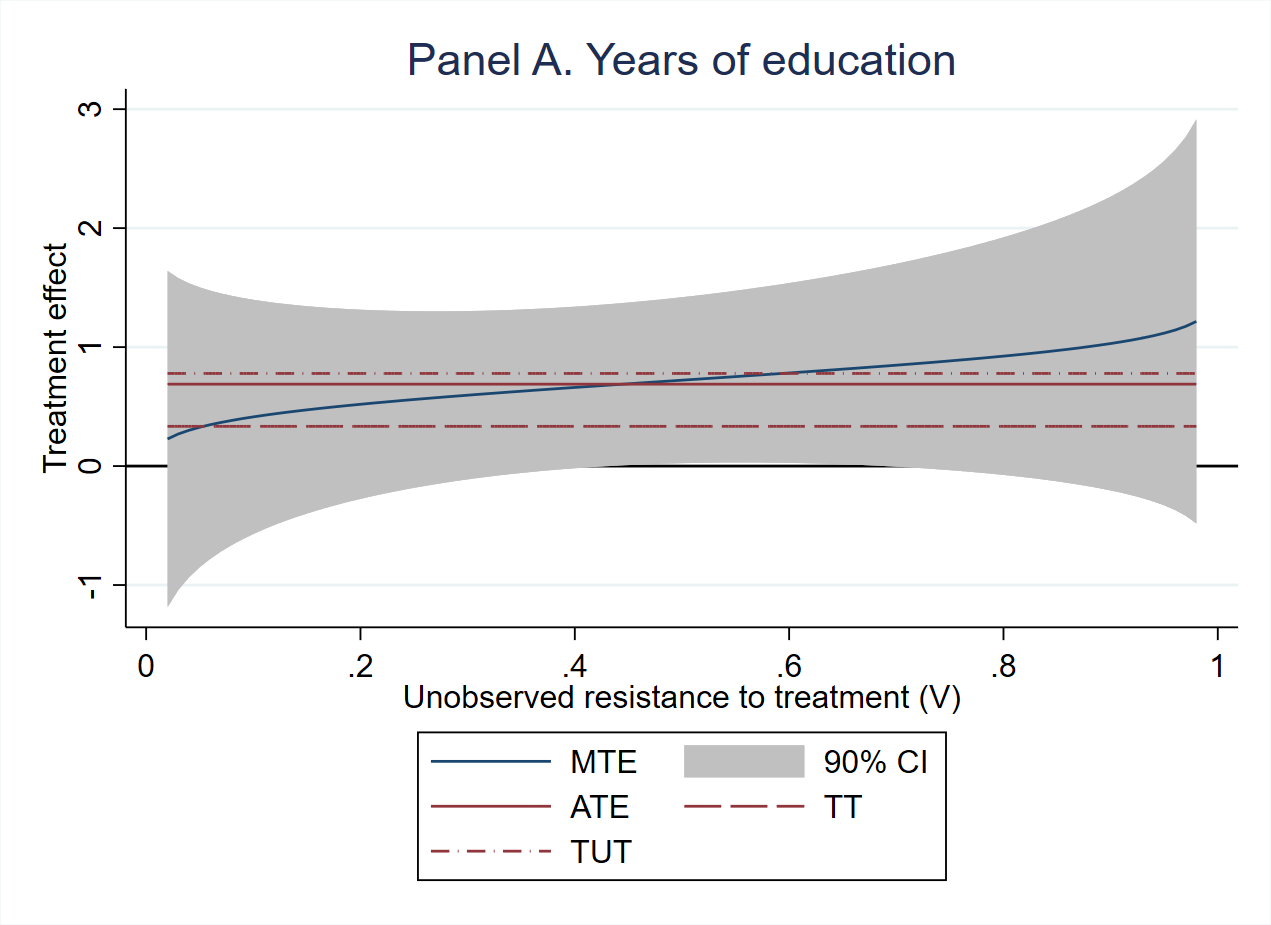}}
\end{minipage}
\begin{minipage}[t]{0.48\linewidth}
  \centerline{\includegraphics[width=7.626cm,height=5.55cm]{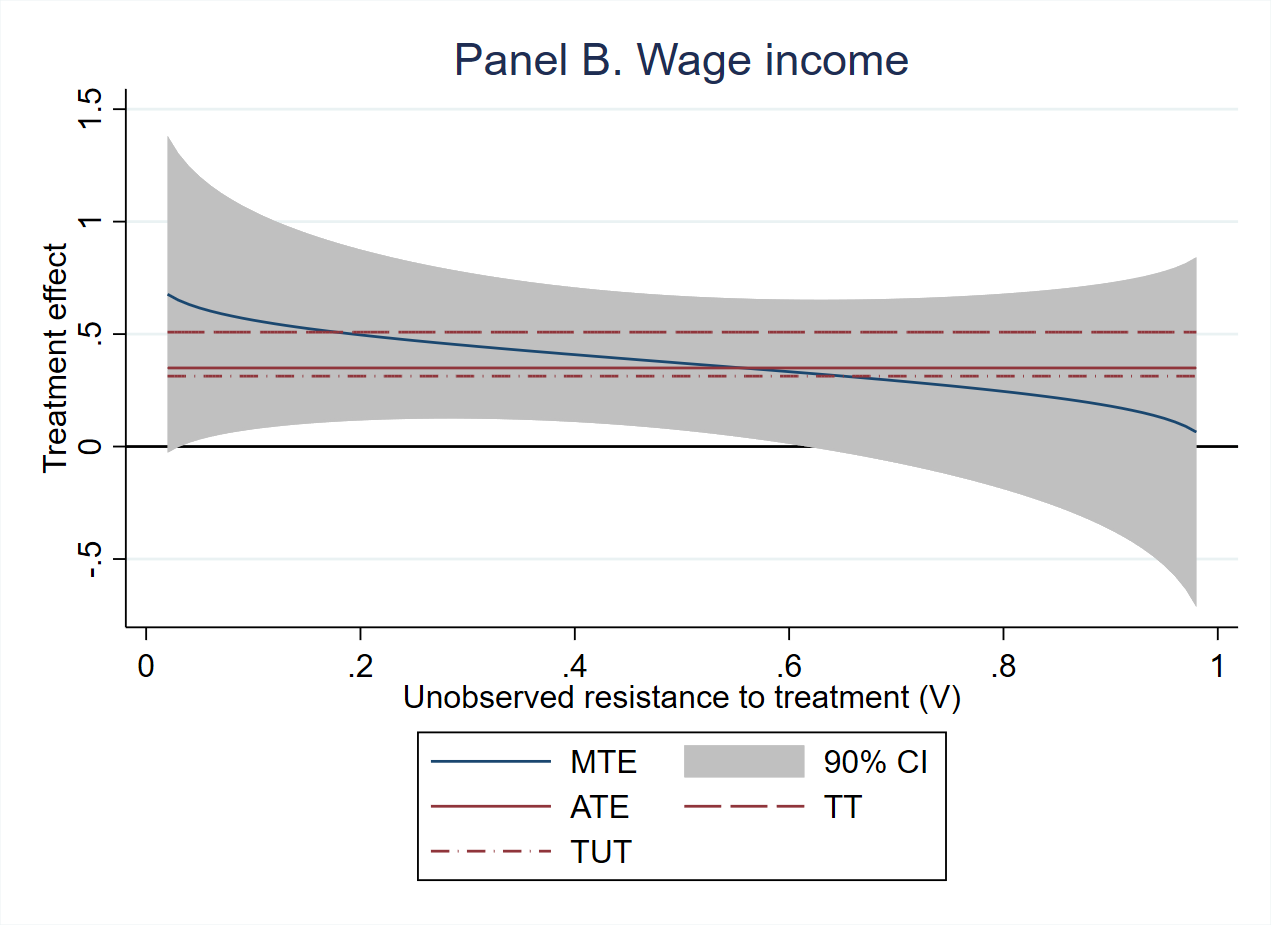}}
\end{minipage}
{\small Notes: The MTE is estimated by the separate procedure based on
parametric normal specification and evaluated at mean values of the
covariates.\ Panels A and B depict the estimated MTE curves for education
outcome and labor market outcome, respectively. The 90\% confidence interval
is based on bootstrapping with 1,000 replications.}
\end{figure}

\begin{figure}[!htb]
\caption{MTE curves under semiparametric specification}
\label{Fig:MTEsemi}
\bigskip \setlength{\abovecaptionskip}{-0.2cm} 
\begin{minipage}[t]{0.48\linewidth}
  \centerline{\includegraphics[width=7.626cm,height=5.55cm]{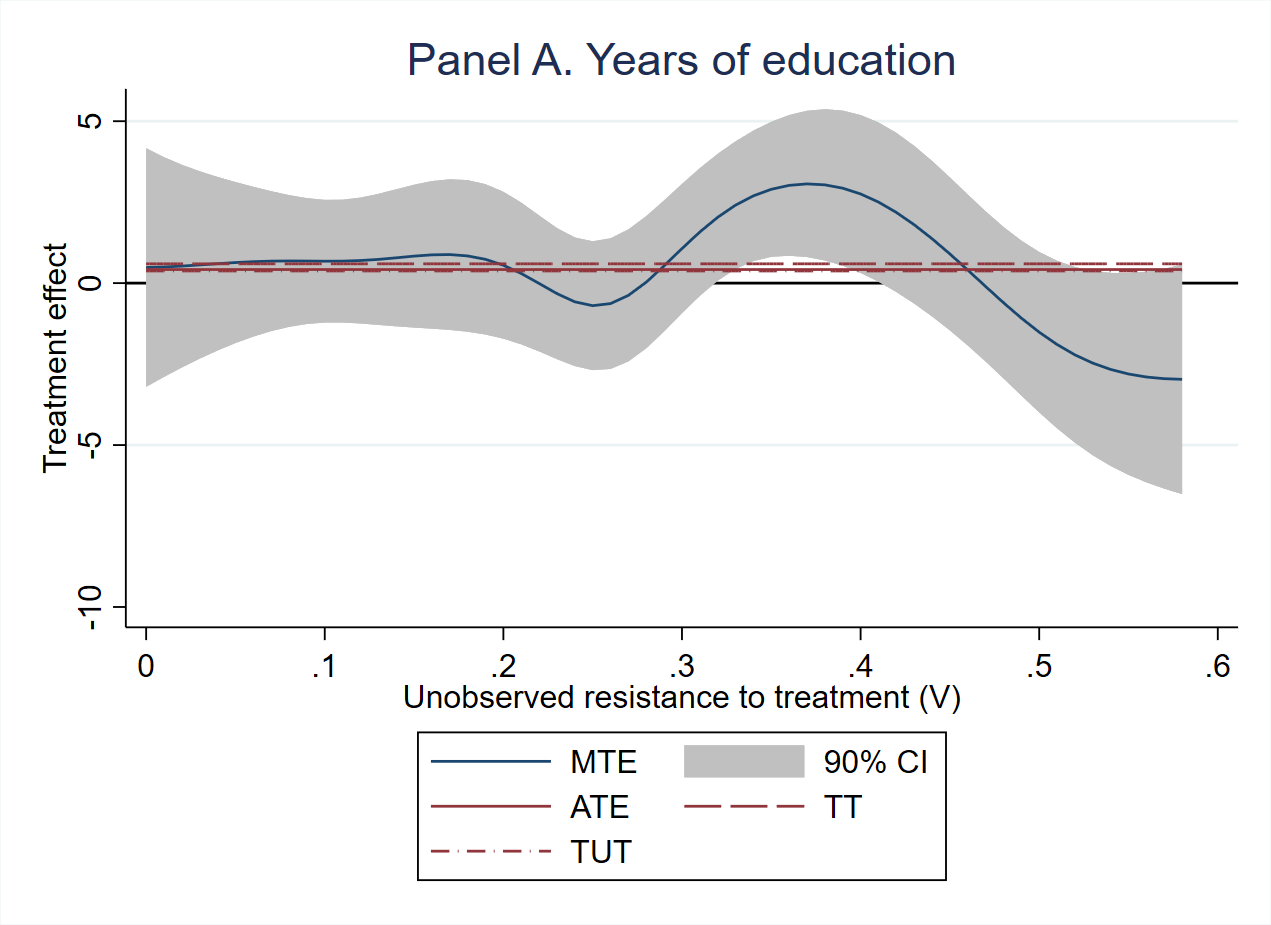}}
\end{minipage}
\begin{minipage}[t]{0.48\linewidth}
  \centerline{\includegraphics[width=7.626cm,height=5.55cm]{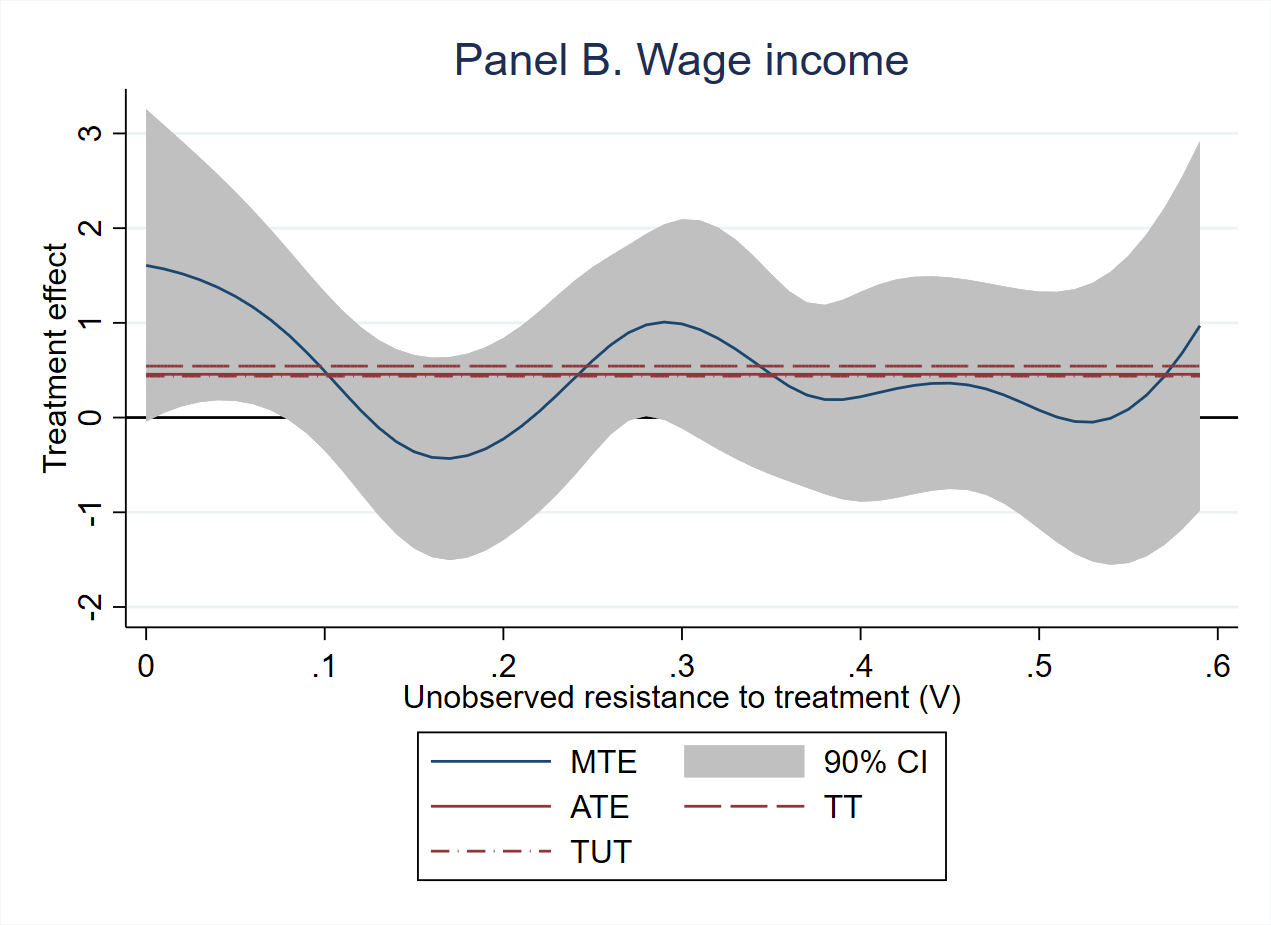}}
\end{minipage}
{\small Notes: The MTE is estimated by the separate procedure based on the
semiparametric specification, with the rule-of-thumb bandwidth, and
evaluated at the mean values of the covariates.\ Panels A and B depict the
estimated MTE curves for education outcome and labor market outcome,
respectively. The 90\% confidence interval is based on bootstrapping with
1,000 replications.}
\end{figure}

\begin{table}[htb]
\caption{Semiparametric estimates of outcome equation coefficients}
\label{table1}
\medskip \renewcommand\arraystretch{1.1} {\small 
\begin{tabular}{lccccccc}
\hline\hline
& \multicolumn{3}{c}{Years of education} &  & \multicolumn{3}{c}{Wage income}
\\ 
& Treated & Untreated & Difference &  & Treated & Untreated & Difference \\ 
& (1) & (2) & (3) &  & (4) & (5) & (6) \\ \hline
Age & 0.000 & 0.000 & 0.000 &  & 0.005 & 0.007 & -0.002 \\ 
& (0.056) & (0.027) & (0.062) &  & (0.030) & (0.014) & (0.033) \\ 
Female & 0.558$^{\ast \ast \ast }$ & 0.254$^{\ast \ast \ast }$ & 0.304$%
^{\ast }$ &  & -0.321$^{\ast \ast \ast }$ & -0.503$^{\ast \ast \ast }$ & 
0.182$^{\ast \ast }$ \\ 
& (0.159) & (0.078) & (0.174) &  & (0.083) & (0.039) & (0.091) \\ 
Black & -0.378 & -0.226 & -0.152 &  & -0.015 & -0.404$^{\ast \ast \ast }$ & 
0.389 \\ 
& (0.518) & (0.227) & (0.559) &  & (0.216) & (0.145) & (0.263) \\ 
Hispanic & -0.642 & -0.403$^{\ast \ast \ast }$ & -0.239 &  & -0.097 & -0.000
& -0.097 \\ 
& (0.408) & (0.145) & (0.430) &  & (0.169) & (0.068) & (0.183) \\ 
Parental education & 1.764$^{\ast \ast \ast }$ & 1.796$^{\ast \ast \ast }$ & 
-0.032 &  & 0.247$^{\ast \ast }$ & 0.197$^{\ast \ast \ast }$ & 0.050 \\ 
& (0.254) & (0.108) & (0.271) &  & (0.096) & (0.050) & (0.107) \\ 
Family income 1978 \ \ \  & 0.053$^{\ast \ast \ast }$ & 0.047$^{\ast \ast
\ast }$ & 0.006 &  & 0.021$^{\ast \ast \ast }$ & 0.014$^{\ast \ast \ast }$ & 
0.007 \\ 
& (0.011) & (0.004) & (0.012) &  & (0.005) & (0.002) & (0.005) \\ 
ATE &  &  & 0.421 &  &  &  & 0.457$^{\ast \ast }$ \\ 
&  &  & (0.492) &  &  &  & (0.228) \\ 
TT &  &  & 0.597 &  &  &  & 0.543 \\ 
&  &  & (1.047) &  &  &  & (0.484) \\ 
TUT &  &  & 0.376 &  &  &  & 0.438$^{\ast }$ \\ 
&  &  & (0.579) &  &  &  & (0.266) \\ \hline
Sample size & \multicolumn{3}{c}{4,554} &  & \multicolumn{3}{c}{3,589} \\ 
\hline\hline
\end{tabular}%
\medskip \newline
{Notes: Columns 1 and 4 display the estimates of coefficients in the treated
state ($\beta _{1}$ in Equation [\ref{MTEiden}]), and columns 2 and 5
display the estimates of coefficients in the untreated state ($\beta _{0}$).
Columns 3 and 6 display the difference in the estimates between the treated
and untreated states ($\beta _{1}-\beta _{0}$), as well as the summary
causal parameters (i.e., ATE, TT, and TUT). Bootstrapped standard errors
from 1,000 replications are reported in parentheses. \newline
$^{\ast }$ Significant at the 10\% level \newline
$^{\ast \ast }$ Significant at the 5\% level \newline
$^{\ast \ast \ast }$ Significant at the 1\% level}}
\end{table}

\clearpage
\setcounter{page}{1}

\begin{appendix}

\begin{center}
{\Large \textbf{Online Appendices} to ``Marginal treatment effects in the
absence of instrumental variables'' }
\end{center}

\renewcommand{\thetable}{A.\arabic{table}} \setcounter{table}{0}
\renewcommand*{\theHtable}{\thetable}
\renewcommand{\thefigure}{A.\arabic{figure}} \setcounter{figure}{0}
\renewcommand*{\theHfigure}{\thefigure}

\section{Additional results of the empirical application}

Based on $\hat{\beta}_{1}$ and $\hat{\beta}_{0}$ reported in Table \ref%
{table1}, and the separate estimates of $E\left[ \left. U_{d}\right\vert V=v%
\right] $ for $d=0,1$ under the semiparametric specification, we can
estimate the marginal structural functions%
\begin{equation*}
E\left[ \left. Y_{d}\right\vert V=v\right] =E\left[ X\right] ^{\prime }\beta
_{d}+E\left[ \left. U_{d}\right\vert V=v\right]
\end{equation*}%
for potential outcomes $Y_{1}$ and $Y_{0}$, which we plot in Figure \ref%
{Fig:MSFsemi}. Panel A sheds light on the significantly positive effect of
Head Start on education for the respondents with medium-level resistance to
treatment, revealing that their gains from the program are driven mainly by
the remarkably low educational attainment when untreated. Panel B leads to
similar results for wage income, where significant gain emerges for low
values of $V$. Moreover, the relatively flat curve of potential wage income
in the treated state implies that early childcare attendance serves as an
equalizer that diminishes the intergroup difference in the labor market
outcome.

\begin{figure}[htb]
\caption{Semiparametric estimates of marginal structural functions}
\label{Fig:MSFsemi}
\bigskip \setlength{\abovecaptionskip}{-0.2cm} 
\begin{minipage}[t]{0.48\linewidth}
  \centerline{\includegraphics[width=7.626cm,height=5.55cm]{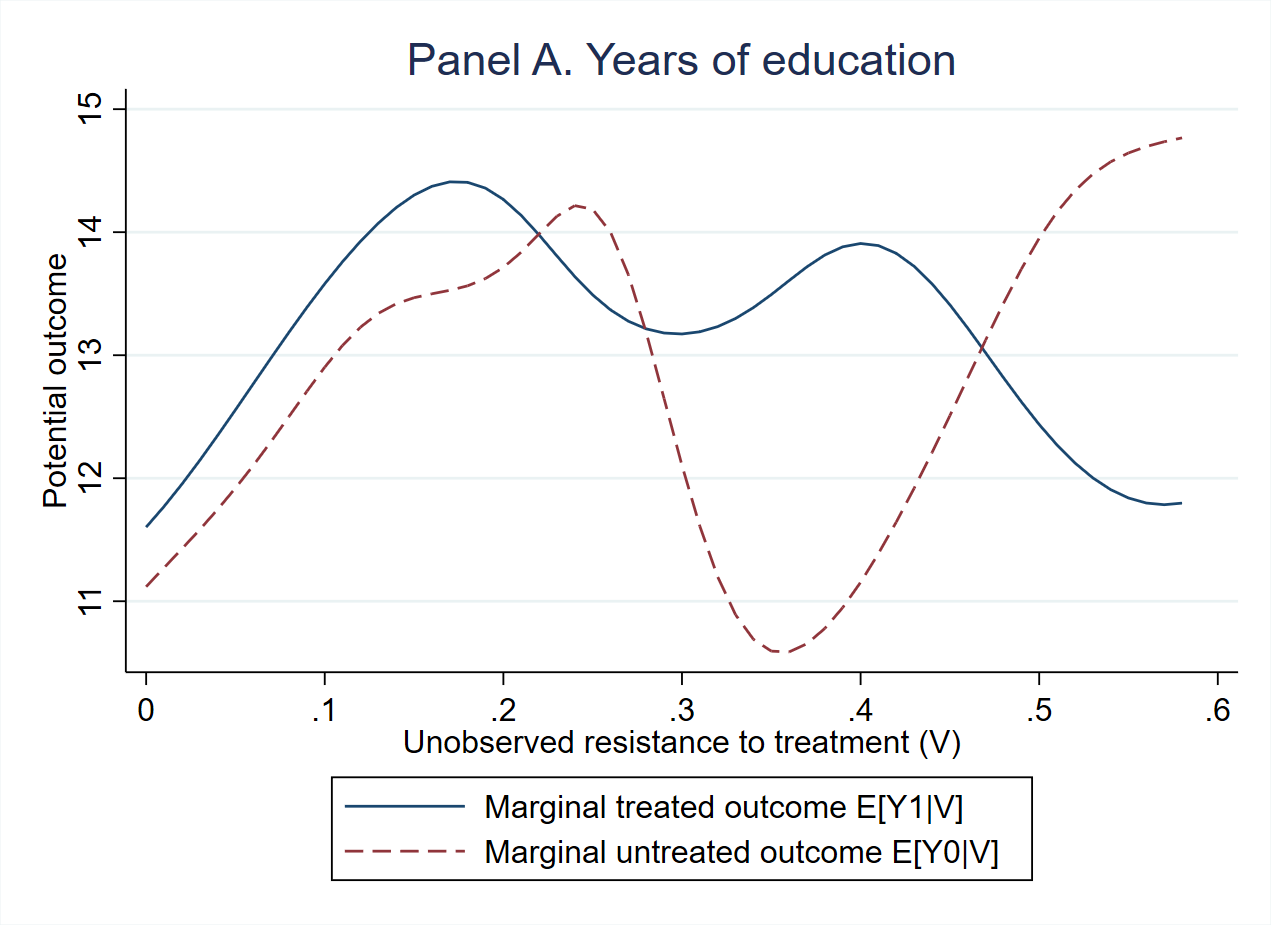}}
\end{minipage}
\begin{minipage}[t]{0.48\linewidth}
  \centerline{\includegraphics[width=7.626cm,height=5.55cm]{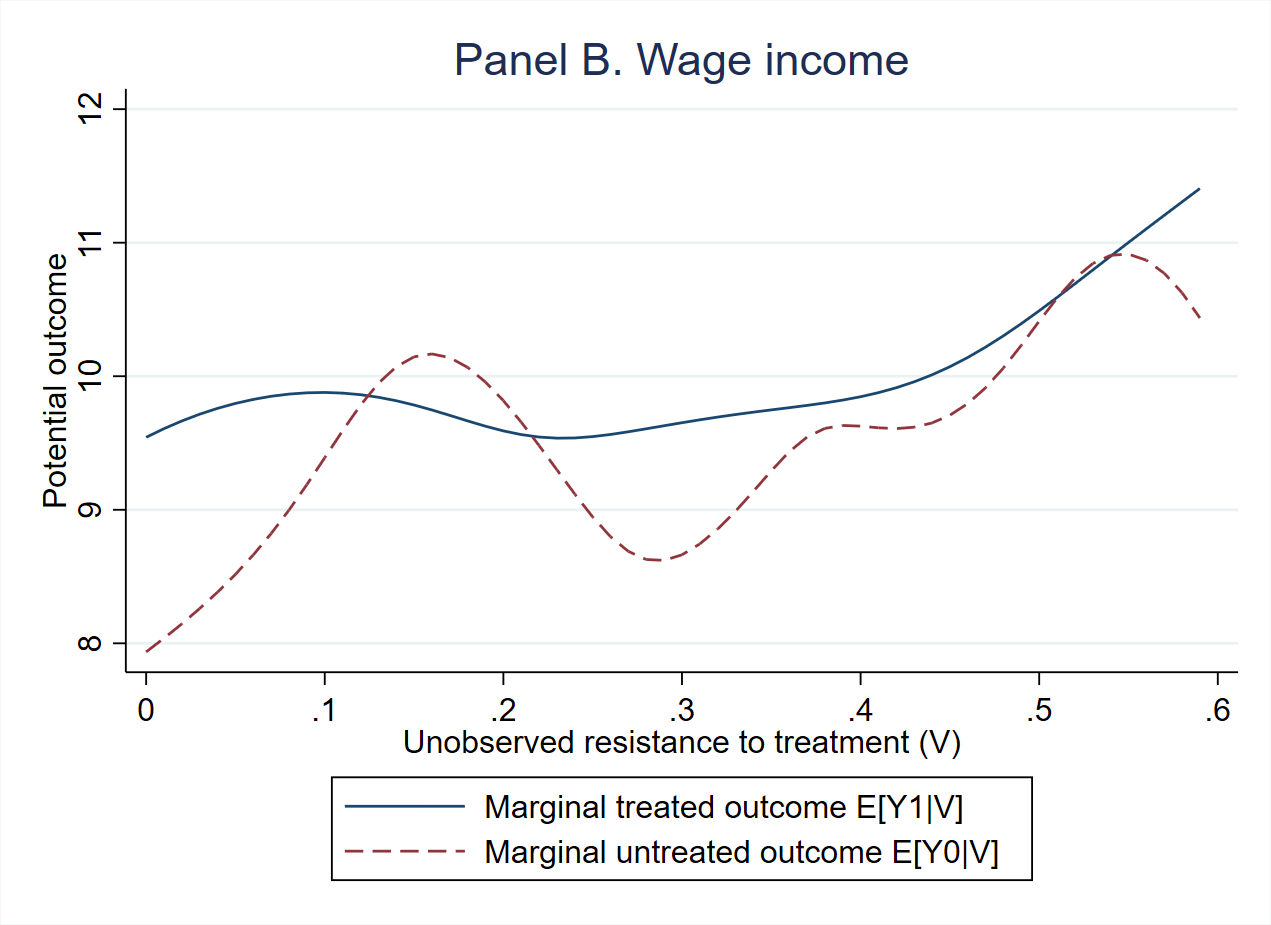}}
\end{minipage}
{\small Notes: The marginal structural function refers to the expected
potential outcome conditional on the unobserved resistance to treatment.
Panels A and B depict the estimated marginal structural function curves for
education outcome and labor market outcome, respectively. The 90\%
confidence interval is based on bootstrapping with 1,000 replications.}
\end{figure}

In addition, we can semiparametrically estimate conditional expectations of
the observed responses given the unobserved determinants of treatment
choice, such as the marginal probability of participation%
\begin{equation*}
E\left[ \left. D\right\vert V=v\right] =\Pr \left( \left. P\geq v\right\vert
V=v\right) =\Pr \left( P\geq v\right) ,
\end{equation*}%
which mirrors the distribution of the propensity score, and the marginal
observed outcome%
\begin{eqnarray*}
E\left[ \left. Y\right\vert V=v\right] &=&E\left[ \left. Y_{0}\right\vert V=v%
\right] +E\left[ \left. D\left( Y_{1}-Y_{0}\right) \right\vert V=v\right] \\
&=&E\left[ \left. Y_{0}\right\vert V=v\right] +E\left[ \left. DX\right\vert
V=v\right] ^{\prime }\left( \beta _{1}-\beta _{0}\right) +E\left[ \left.
D\left( U_{1}-U_{0}\right) \right\vert V=v\right] \\
&=&E\left[ \left. Y_{0}\right\vert V=v\right] +E\left[ 1\left\{ P\geq
v\right\} X\right] ^{\prime }\left( \beta _{1}-\beta _{0}\right) +\Pr \left(
P\geq v\right) \cdot E\left[ \left. U_{1}-U_{0}\right\vert V=v\right] ,
\end{eqnarray*}%
where the last equality follows from the independence of $X$ and $V$ and
from the conditional mean independence of $X$ and $U_{d}$ given $V$.
Plugging in proper estimates of each component in the above equations, we
obtain the estimated marginal response functions and plot them in Figure \ref%
{Fig:MRFsemi}. In particular, the marginal observed outcome curves relate
the individuals with significant positive returns (with medium $V$ in Panel
A and small $V$ in Panel B) to those with low education and low wage income,
thereby bridging our results on the MTE and those of \cite{de2020head} on
the distributional treatment effect.

\begin{figure}[!htb]
\caption{Semiparametric estimates of marginal response functions}
\label{Fig:MRFsemi}
\bigskip \setlength{\abovecaptionskip}{-0.2cm} 
\begin{minipage}[t]{0.48\linewidth}
  \centerline{\includegraphics[width=7.626cm,height=5.55cm]{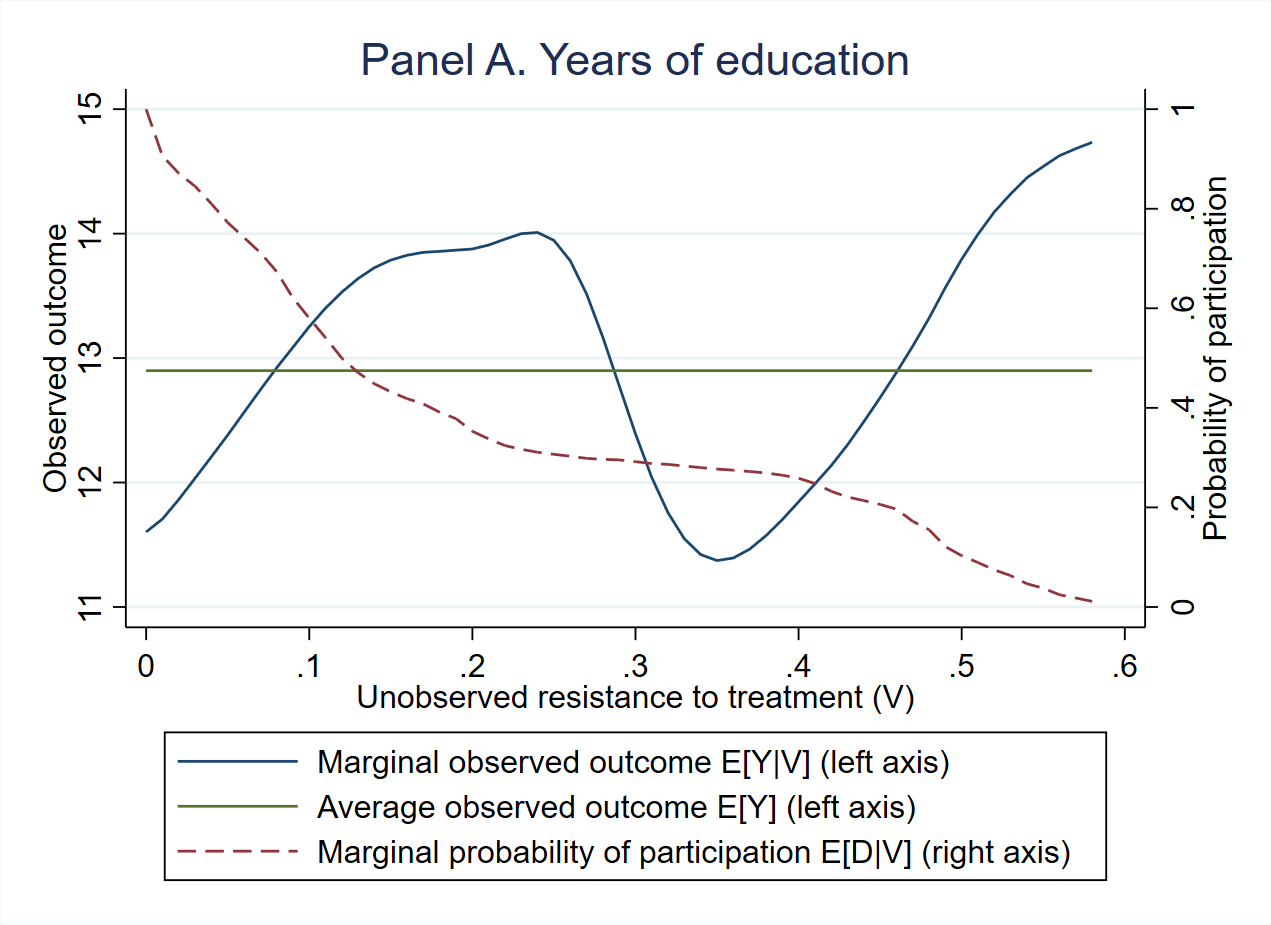}}
\end{minipage}
\begin{minipage}[t]{0.48\linewidth}
  \centerline{\includegraphics[width=7.626cm,height=5.55cm]{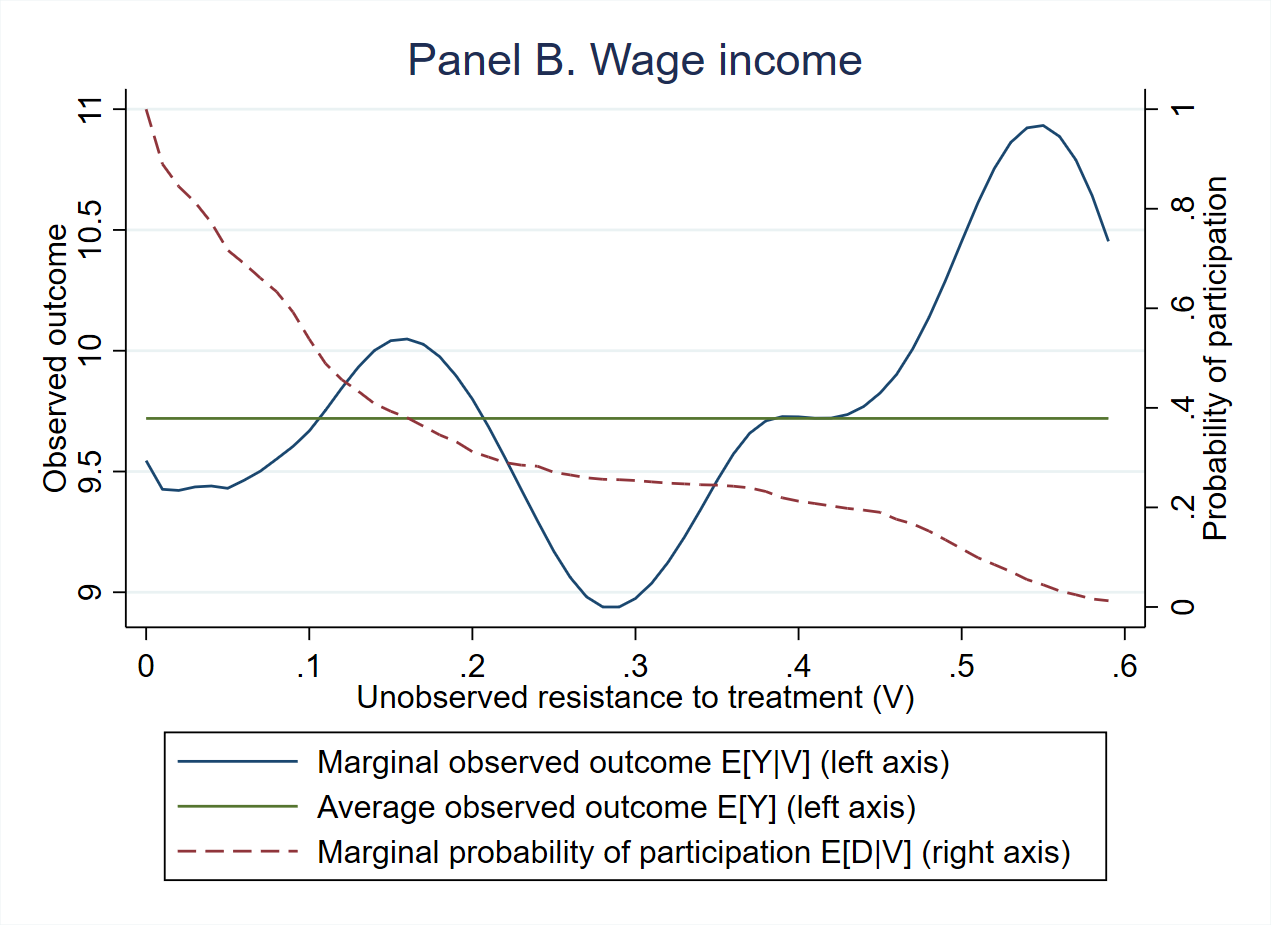}}
\end{minipage}
{\small Notes: This figure plots the (conditional) expectation of the
observed outcome (solid line on the left y-axis), and the conditional
probability of participation (dashed line on the right y-axis). Panels A and
B illustrate the estimated marginal response function curves for education
outcome and labor market outcome, respectively. The 90\% confidence interval
is based on bootstrapping with 1,000 replications.}
\end{figure}

\renewcommand{\theequation}{B.\arabic{equation}} \setcounter{equation}{0} %
\renewcommand*{\theHequation}{\theequation}

\section{Proof of Theorem \protect\ref{theorem:main}}

\label{appendix:proof}Denote $m_{d}\left( x\right) =E\left[ Y\left\vert
X=x,D=d\right. \right] $ and $m_{d0}\left( x^{C}\right) =m_{d}\left(
x^{C},0\right) $, $d=0,1$. Note that $\pi \left( x\right) $ and $m_{d}\left(
x\right) $, and thus $\pi _{0}\left( x^{C}\right) $ and $m_{d0}\left(
x^{C}\right) $, are identified functions because they are conditional
expectations of observed variables. We first consider identifying $\beta
_{d}^{C}$, the coefficients of continuous covariates, from $\pi _{0}\left(
x^{C}\right) $ and $m_{d0}\left( x^{C}\right) $. By equation (\ref{md}), we
have%
\begin{equation}
m_{d0}\left( x^{C}\right) =x^{C\prime }\beta _{d}^{C}+g_{d}\left( \pi
_{0}\left( x^{C}\right) \right) .  \label{IDreg}
\end{equation}%
When $\dim \left( X^{C}\right) =1$, by differentiating both sides of (\ref%
{IDreg}) at $\tilde{x}^{C}$ and invoking Assumption NL1, we obtain%
\begin{equation*}
m_{d0}^{\left( 1\right) }\left( \tilde{x}^{C}\right) =\beta
_{d}^{C}+g_{d}^{\left( 1\right) }\left( \pi _{0}\left( \tilde{x}^{C}\right)
\right) \pi _{0}^{\left( 1\right) }\left( \tilde{x}^{C}\right) =\beta
_{d}^{C},
\end{equation*}%
which identifies $\beta _{d}^{C}$ for $d=0,1$. When $\dim \left(
X^{C}\right) \geq 2$, for $k,j\in \left\{ 1,2,\cdots ,\dim \left(
X^{C}\right) \right\} $ satisfying Assumption NL2, taking the partial
derivatives of $m_{d0}\left( x^{C}\right) $ with respect to $x_{k}^{C}$ and $%
x_{j}^{C}$ yields that%
\begin{eqnarray*}
\partial _{k}m_{d0}\left( x^{C}\right) &=&\beta _{d,k}^{C}+g_{d}^{\left(
1\right) }\left( \pi _{0}\left( x^{C}\right) \right) \partial _{k}\pi
_{0}\left( x^{C}\right) , \\
\partial _{j}m_{d0}\left( x^{C}\right) &=&\beta _{d,j}^{C}+g_{d}^{\left(
1\right) }\left( \pi _{0}\left( x^{C}\right) \right) \partial _{j}\pi
_{0}\left( x^{C}\right) .
\end{eqnarray*}%
It follows from Assumption NL2.(i)-(ii) that%
\begin{equation*}
\frac{\partial _{k}m_{d0}\left( x^{C}\right) -\beta _{d,k}^{C}}{\partial
_{k}\pi _{0}\left( x^{C}\right) }=g_{d}^{\left( 1\right) }\left( \pi
_{0}\left( x^{C}\right) \right) =\frac{\partial _{j}m_{d0}\left(
x^{C}\right) -\beta _{d,j}^{C}}{\partial _{j}\pi _{0}\left( x^{C}\right) },
\end{equation*}%
so that%
\begin{equation}
\partial _{k}m_{d0}\left( x^{C}\right) \partial _{j}\pi _{0}\left(
x^{C}\right) -\partial _{j}m_{d0}\left( x^{C}\right) \partial _{k}\pi
_{0}\left( x^{C}\right) =\partial _{j}\pi _{0}\left( x^{C}\right) \beta
_{d,k}^{C}-\partial _{k}\pi _{0}\left( x^{C}\right) \beta _{d,j}^{C},
\label{11}
\end{equation}%
which is linear in $\beta _{d,k}^{C}$ and $\beta _{d,j}^{C}$. The same
equation is obtained if we evaluate the expression at another point $\tilde{x%
}^{C}$ that satisfies Assumption NL2.(iii)-(iv), which gives%
\begin{equation}
\left( 
\begin{array}{c}
\partial _{k}m_{d0}\left( x^{C}\right) \partial _{j}\pi _{0}\left(
x^{C}\right) -\partial _{j}m_{d0}\left( x^{C}\right) \partial _{k}\pi
_{0}\left( x^{C}\right) \\ 
\partial _{k}m_{d0}\left( \tilde{x}^{C}\right) \partial _{j}\pi _{0}\left( 
\tilde{x}^{C}\right) -\partial _{j}m_{d0}\left( \tilde{x}^{C}\right)
\partial _{k}\pi _{0}\left( \tilde{x}^{C}\right)%
\end{array}%
\right) =\Psi \left( 
\begin{array}{c}
\beta _{d,k}^{C} \\ 
\beta _{d,j}^{C}%
\end{array}%
\right) ,  \label{betaC}
\end{equation}%
where%
\begin{equation*}
\Psi =\left( 
\begin{array}{cc}
\partial _{j}\pi _{0}\left( x^{C}\right) , & -\partial _{k}\pi _{0}\left(
x^{C}\right) \\ 
\partial _{j}\pi _{0}\left( \tilde{x}^{C}\right) , & -\partial _{k}\pi
_{0}\left( \tilde{x}^{C}\right)%
\end{array}%
\right) .
\end{equation*}%
Assumption NL2.(v) ensures that the determinant of $\Psi $ is nonzero, which
shows that $\Psi $ is nonsingular. Therefore, equation (\ref{betaC}) can be
solved for $\beta _{d,k}^{C}$ and $\beta _{d,j}^{C}$ by inverting $\Psi $,
thereby identifying $\beta _{d,k}^{C}$ and $\beta _{d,j}^{C}$. Given
identification of $\beta _{d,k}^{C}$, we can then identify all other
coefficient $\beta _{d,l}^{C}$ in $\beta _{d}^{C}$ by solving (\ref{11})
with the subscript $j$ replaced by $l$, which gives%
\begin{equation*}
\beta _{d,l}^{C}=\frac{\partial _{l}m_{d0}\left( x^{C}\right) \partial
_{k}\pi _{0}\left( x^{C}\right) -\partial _{k}m_{d0}\left( x^{C}\right)
\partial _{l}\pi _{0}\left( x^{C}\right) +\partial _{l}\pi _{0}\left(
x^{C}\right) \beta _{d,k}^{C}}{\partial _{k}\pi _{0}\left( x^{C}\right) }.
\end{equation*}%
Given the identification of $\beta _{d}^{C}$, the function $g_{d}$ is
identified on the support of $\pi _{0}\left( X^{C}\right) $ by%
\begin{equation*}
g_{d}\left( p\right) =E\left[ \left. m_{d0}\left( X^{C}\right) -X^{C\prime
}\beta _{d}^{C}\right\vert \pi _{0}\left( X^{C}\right) =p\right] .
\end{equation*}

Next, we consider identifying $\beta _{d}^{D}$, the coefficients of discrete
covariates. For each $k\in \left\{ 1,2,\cdots ,\dim \left( X^{D}\right)
\right\} $, we have%
\begin{equation*}
m_{d}\left( x^{C},x^{Dk}\right) =x^{C\prime }\beta _{d}^{C}+x_{k}^{D}\beta
_{d,k}^{D}+g_{d}\left( \pi \left( x^{C},x^{Dk}\right) \right)
\end{equation*}%
for any $x^{C}$ in the support of $X^{C}$. By Assumption S, there exists $%
x^{C}\left( k\right) $ in the support of $X^{C}$ such that $\pi \left(
x^{C}\left( k\right) ,x^{Dk}\right) $ is in the support of $\pi _{0}\left(
X^{C}\right) $. It follows from the above identification result that $%
g_{d}\left( \pi \left( x^{C}\left( k\right) ,x^{Dk}\right) \right) $ is
identified. Consequently, $\beta _{d,k}^{D}$ is identified by%
\begin{equation*}
\beta _{d,k}^{D}=\frac{m_{d}\left( x^{C}\left( k\right) ,x^{Dk}\right)
-x^{C}\left( k\right) ^{\prime }\beta _{d}^{C}-g_{d}\left( \pi \left(
x^{C}\left( k\right) ,x^{Dk}\right) \right) }{x_{k}^{D}}.
\end{equation*}%
This argument holds for each $k\in \left\{ 1,2,\cdots ,\dim \left(
X^{D}\right) \right\} $, thereby identifying $\beta _{d}^{D}$. Finally,
given the identification of $\beta _{d}=\left( \beta _{d}^{C},\beta
_{d}^{D}\right) $, it follows from (\ref{md}) that the function $g_{d}$ is
identified on the support of $P=\pi \left( X\right) $ by%
\begin{equation*}
g_{d}\left( p\right) =E\left[ \left. m_{d}\left( X\right) -X^{\prime }\beta
_{d}\right\vert \pi \left( X\right) =p\right] ,
\end{equation*}%
for $d=0,1$.

\renewcommand{\thetheorem}{C.\arabic{theorem}} \setcounter{theorem}{0} %
\renewcommand*{\theHtheorem}{\thetheorem} \renewcommand{\theequation}{C.%
\arabic{equation}} \setcounter{equation}{0} \renewcommand*{\theHequation}{%
\theequation}

\section{Further discussion on the nonlinearity assumption}

\label{appendix:NL}In light of the proof of Theorem \ref{theorem:main}, our
identification strategy may work even when no continuous covariate is
available in the data. In that case, we put a discretely distributed
covariate in $X^{C}$.

\begin{theorem}
\label{theorem:A1} Theorem \ref{theorem:main} holds if Assumption NL1 is
replaced by that there exist two different constants $x^{C}$, $\tilde{x}^{C}$%
\ in the support of $X^{C}$\ such that $\pi _{0}\left( x^{C}\right) =\pi
_{0}\left( \tilde{x}^{C}\right) $.
\end{theorem}

\begin{proof}
When $\dim \left( X^{C}\right) =1$, it follows from (\ref{IDreg})
and the assumption of Theorem \ref{theorem:A1} that%
\begin{equation*}
m_{d0}\left( x^{C}\right) -x^{C}\beta _{d}^{C}=g_{d}\left( \pi _{0}\left(
x^{C}\right) \right) =g_{d}\left( \pi _{0}\left( \tilde{x}^{C}\right)
\right) =m_{d0}\left( \tilde{x}^{C}\right) -\tilde{x}^{C}\beta _{d}^{C}.
\end{equation*}%
Hence, $\beta _{d}^{C}$ is identified as%
\begin{equation*}
\beta _{d}^{C}=\frac{m_{d0}\left( x^{C}\right) -m_{d0}\left( \tilde{x}%
^{C}\right) }{x^{C}-\tilde{x}^{C}}.
\end{equation*}%
The remaining part of the proof is the same as the proof of Theorem \ref%
{theorem:main}.
\end{proof}

The assumption of Theorem \ref{theorem:A1} in place of NL1 requires the
univariate function $\pi _{0}$ to be not one-to-one, but imposes no other
smoothness or continuity restriction on $\pi _{0}$. Therefore, it doesn't
require $X^{C}$ to be continuously valued. In the extreme case of $X^{C}$
containing only one binary covariate, this assumption will be equivalent to
the full irrelevance of $X^{C}$ to the treatment probability, which is a
condition suggested by \cite{chamberlain1986asymptotic} for identification
of the sample selection model. The following two theorems show that
Assumption NL1 and the assumption of Theorem \ref{theorem:A1} can be
extended to the case of multiple continuous covariates.

\begin{theorem}
\label{theorem:A2} Theorem \ref{theorem:main} holds if Assumption NL is
replaced by that (i) there are two points $x^{C}$ and $\tilde{x}^{C}$ in the
support of $X^{C}$ and an index $k$ such that $x_{k}^{C}\neq \tilde{x}%
_{k}^{C}$, $x_{-k}^{C}=\tilde{x}_{-k}^{C}$, and $\pi _{0}\left( x^{C}\right)
=\pi _{0}\left( \tilde{x}^{C}\right) $, and (ii) there exists $\breve{x}^{C}$
in the support of $X^{C} $ such that functions $\pi _{0}$ and $m_{d0}$ are
differentiable at $\breve{x}^{C}$, with $\partial _{k}\pi _{0}\left( \breve{x%
}^{C}\right) \neq 0 $.
\end{theorem}

\begin{proof}
It follows from equation (\ref{IDreg}) and $\pi _{0}\left( x^{C}\right) =\pi
_{0}\left( \tilde{x}^{C}\right) $ that%
\begin{equation*}
m_{d0}\left( x^{C}\right) -x_{k}^{C}\beta _{d,k}^{C}-x_{-k}^{C\prime }\beta
_{d,-k}^{C}=m_{d0}\left( \tilde{x}^{C}\right) -\tilde{x}_{k}^{C}\beta
_{d,k}^{C}-\tilde{x}_{-k}^{C^{\prime }}\beta _{d,-k}^{C}.
\end{equation*}%
Hence, $\beta _{d,k}^{C}$ is identified by condition (i) as%
\begin{equation*}
\beta _{d,k}^{C}=\frac{m_{d0}\left( x^{C}\right) -m_{d0}\left( \tilde{x}%
^{C}\right) }{x_{k}^{C}-\tilde{x}_{k}^{C}}.
\end{equation*}%
For any index $j\neq k$, taking the partial derivatives of $m_{d0}\left(
x^{C}\right) $ with respect to $x_{k}^{C}$ and $x_{j}^{C}$ at $\breve{x}^{C}$
yields that%
\begin{eqnarray}
\partial _{k}m_{d0}\left( \breve{x}^{C}\right) &=&\beta
_{d,k}^{C}+g_{d}^{\left( 1\right) }\left( \pi _{0}\left( \breve{x}%
^{C}\right) \right) \partial _{k}\pi _{0}\left( \breve{x}^{C}\right) ,
\label{partial_k} \\
\partial _{j}m_{d0}\left( \breve{x}^{C}\right) &=&\beta
_{d,j}^{C}+g_{d}^{\left( 1\right) }\left( \pi _{0}\left( \breve{x}%
^{C}\right) \right) \partial _{j}\pi _{0}\left( \breve{x}^{C}\right) .
\label{partial_j}
\end{eqnarray}%
Since $\partial _{k}\pi _{0}\left( \breve{x}^{C}\right) \neq 0$ by condition
(ii), we can find $g_{d}^{\left( 1\right) }\left( \pi _{0}\left( \breve{x}%
^{C}\right) \right) $ from (\ref{partial_k}) to be%
\begin{equation*}
g_{d}^{\left( 1\right) }\left( \pi _{0}\left( \breve{x}^{C}\right) \right) =%
\frac{\partial _{k}m_{d0}\left( \breve{x}^{C}\right) -\beta _{d,k}^{C}}{%
\partial _{k}\pi _{0}\left( \breve{x}^{C}\right) },
\end{equation*}%
and substitute it into (\ref{partial_j}) to identify $\beta _{d,j}^{C}$ as%
\begin{equation*}
\beta _{d,j}^{C}=\partial _{j}m_{d0}\left( \breve{x}^{C}\right) -\frac{%
\partial _{j}\pi _{0}\left( \breve{x}^{C}\right) }{\partial _{k}\pi
_{0}\left( \breve{x}^{C}\right) }\left[ \partial _{k}m_{d0}\left( \breve{x}%
^{C}\right) -\beta _{d,k}^{C}\right] .
\end{equation*}%
Therefore, $\beta _{d}^{C}$ is identified. The remaining part of the proof
is the same as the proof of Theorem \ref{theorem:main}.
\end{proof}

\begin{theorem}
\label{theorem:A3} Theorem \ref{theorem:main} holds if Assumption NL is
replaced by that there exist two points $x^{C}$ and $\tilde{x}^{C}$ in the
support of $X^{C}$ and an index $k$ such that functions $\pi _{0}$ and $%
m_{d0}$ are differentiable at $x^{C}$ and $\tilde{x}^{C}$, with $\partial
_{k}\pi _{0}\left( x^{C}\right) =0 $ and $\partial _{k}\pi _{0}\left( \tilde{%
x}^{C}\right) \neq 0$.
\end{theorem}

\begin{proof}
The identification of $\beta _{d,k}^{C}$ is analogous to Theorem %
\ref{theorem:main}. Then, the identification of all the other coefficients $%
\beta _{d,j}^{C}$ in $\beta _{d}^{C}$ and the remaining components of the
model is the same as in Theorem \ref{theorem:A2} and Theorem \ref%
{theorem:main}, respectively.
\end{proof}

Alternatively, the identification of coefficients of continuous covariates
can be attained by exploiting the local irrelevance of the function $%
g_{d}\left( p\right) $.

\begin{theorem}
\label{theorem:A4} Theorem \ref{theorem:main} holds if Assumption NL is
replaced by the existence of a vector $x^{C}$ in the support of $X^{C}$ such
that $g_{d}^{\left( 1\right) }\left( \pi _{0}\left( x^{C}\right) \right) =0$
for $d=0,1$, and that functions $\pi _{0}$ and $m_{d0}$ are differentiable
at $x^{C}$.
\end{theorem}

\begin{proof}
For each $k\in \left\{ 1,2,\cdots ,\dim \left( X^{C}\right) \right\} $,
taking the partial derivative of both sides of (\ref{IDreg}) with respect to 
$x_{k}^{C}$ gives%
\begin{equation*}
\partial _{k}m_{d0}\left( x^{C}\right) =\beta _{d,k}^{C}+g_{d}^{\left(
1\right) }\left( \pi _{0}\left( x^{C}\right) \right) \partial _{k}\pi
_{0}\left( x^{C}\right) =\beta _{d,k}^{C},
\end{equation*}%
which immediately identifies all $\beta _{d,k}^{C}$ in $\beta _{d}^{C}$. The
remaining part of the proof is the same as that of Theorem \ref{theorem:main}.
\end{proof}

The above discussion further illustrates the over-identification feature of
our strategy. In fact, the strength of identification of our method is
determined by the number of nonlinearity points in the data satisfying
Assumption NL or Theorems \ref{theorem:A1}-\ref{theorem:A4}, given different
benchmark values of $X^{D}$.

\renewcommand{\thetheorem}{D.\arabic{theorem}} \setcounter{theorem}{0} %
\renewcommand*{\theHtheorem}{\thetheorem} \renewcommand{\theequation}{D.%
\arabic{equation}} \setcounter{equation}{0} \renewcommand*{\theHequation}{%
\theequation}

\section{Limited valued outcomes}

\label{appendix:LDV}To accommodate limited valued outcomes, we first
generalize the linear additive model to a linear latent index model.

\medskip

\textbf{Assumption D.L} (Linear index). \textit{Assume that }$Y_{d}=l\left(
X^{\prime }\gamma _{d},V_{d}\right) $\textit{, }$d=0,1$\textit{, holds for
some known link function }$l$\textit{\ and unknown coefficients }$\gamma
_{d} $\textit{, where }$V_{d}$\textit{\ is unobserved error term of }$Y_{d} $%
\textit{.}

\medskip

The generalized linear model applies to a variety of frequently encountered
limited dependent variables, such as $l\left( t,v\right) =1\left\{ t\geq
v\right\} $ for binary $Y\in \left\{ 0,1\right\} $, $l\left( t,v\right)
=\max \left( t-v,0\right) $ for censored (at zero) $Y\geq 0$, $l\left(
t,v\right) =\exp \left( t+v\right) \left/ \left[ 1+\exp \left( t+v\right) %
\right] \right. $ for proportion-valued $Y\in \left( 0,1\right) $, $l\left(
t,v\right) =\sum_{j=1}^{s}1\left\{ t+v\geq c_{j}\right\} $ for ordered $Y\in
\left\{ 0,1,\cdots ,s\right\} $ with known thresholds $c_{1}<c_{2}<\cdots
<c_{s}$, and so forth. Although setting $l\left( t,v\right) =t+v$ reduces to
the linear additive model, the following discussion doesn't include Theorem %
\ref{theorem:main} as a special case because $\gamma _{d}$ here can be only
identified up to scale for a general link function $l$. Any changes in the
scaling of $\gamma _{d}$ can be freely absorbed into $l$, and a scale
normalization is needed for identification of $\gamma _{d}$ and $l$.

\medskip

\textbf{Assumption D.N} (Normalization). \textit{Decompose }$X^{\prime
}\gamma _{d}=X^{C\prime }\gamma _{d}^{C}+X^{D\prime }\gamma _{d}^{D}$\textit{%
, where }$X^{C}$\textit{\ and }$X^{D}$\textit{\ consist of covariates that
are continuously and discretely distributed, respectively. Assume that }$%
\gamma _{d,1}^{C}$\textit{, the first element of }$\gamma _{d}^{C}$\textit{,
equal to 1.}

\medskip

Assumption D.N imposes the convenient normalization that the first
continuous covariate has a unit coefficient. This scaling of $\gamma _{d}$
is arbitrary and is innocuous because our focus is on the identification of
MTE, which is a difference in the conditional expectations of $l\left(
X^{\prime }\gamma _{d},V_{d}\right) $, rather than on separate
identification of $\gamma _{d}$ and $l$. Since in the generalized linear
model the potential outcomes are not necessarily additive in the error term,
the mean independence assumption needs to be strengthened to a stricter
distributional independence assumption.

\medskip

\textbf{Assumption D.CI} (Conditional Independence). \textit{Assume that }$%
\left. V_{d}\perp \!\!\!\!\perp X\right\vert V$\textit{\ for }$d=0,1$\textit{%
, namely, }$V_{d}$\textit{\ is independent of }$X$\textit{\ conditional on }$%
V$\textit{, where }$V$\textit{\ is the reduced-form treatment error in
equation (\ref{normalized}).}

\medskip

Recalling that $V\perp \!\!\!\!\perp X$ by definition, Assumption D.CI is
equivalent to the full independence $\left( V_{d},V\right) \perp
\!\!\!\!\perp X$ for $d=0,1$, which implies that both the marginal
distribution of $V_{d}$ and the copula of $V_{d}$ and $V$ are independent of 
$X$. Under Assumptions D.L and D.CI, we have 
\begin{equation*}
\Delta ^{\text{MTE}}\left( x,v\right) =E\left[ \left. l\left( x^{\prime
}\gamma _{1},V_{1}\right) -l\left( x^{\prime }\gamma _{0},V_{0}\right)
\right\vert V=v\right] .
\end{equation*}%
The additive nonseparability of the observables and unobservables
constitutes the primary difficulty in identifying MTE in the limited outcome
case. Meanwhile, Assumptions D.L and D.CI lead to a double index form of the
observable outcome regression functions for each treatment status, that is, 
\begin{equation}
E\left[ \left. Y\cdot 1\left\{ D=d\right\} \right\vert X=x\right]
=r_{d}\left( x^{\prime }\gamma _{d},\pi \left( x\right) \right)  \label{A1}
\end{equation}%
for $d=0,1$, where \label{A2} 
\begin{eqnarray}
r_{0}\left( t,p\right) &=&\int_{p}^{1}E\left[ \left. l\left( t,V_{0}\right)
\right\vert V=v\right] dv, \\
r_{1}\left( t,p\right) &=&\int_{0}^{p}E\left[ \left. l\left( t,V_{1}\right)
\right\vert V=v\right] dv.
\end{eqnarray}%
Provided that $\pi \left( x\right) $ is a nonlinear index, $\gamma _{d}$ and
thus $r_{d}$ can be identified based on functional forms, which ensures
identification of MTE since 
\begin{equation}
\Delta ^{\text{MTE}}\left( x,v\right) =\partial _{2}r_{1}\left( x^{\prime
}\gamma _{1},v\right) +\partial _{2}r_{0}\left( x^{\prime }\gamma
_{0},v\right) ,  \label{AMTE}
\end{equation}%
where $\partial _{2}r_{d}$ is the partial derivative of $r_{d}$ with respect
to its second argument. Like in the case of unlimited outcomes, the key
powers of this identification strategy are supplied by the nonlinearity of $%
\pi \left( x\right) $, as specified by the following assumption.

\medskip

\textbf{Assumption D.NL} (Non-Linearity). \textit{Assume that the functions }%
$\pi _{0}$\textit{\ and }$r_{d}$\textit{\ are differentiable and denote
their partial derivatives with respect to the }$k$\textit{-th argument as }$%
\partial _{k}\pi _{0}$\textit{\ and }$\partial _{k}r_{d}$\textit{, where }$%
\pi _{0}\left( x^{C}\right) =\pi \left( x^{C},0\right) $\textit{\ and }$%
r_{d} $\textit{, }$d=0,1$\textit{, are defined in (\ref{A1}) and (\ref{A2}).
Assume that there exist two vectors }$x^{C}$\textit{\ and }$\tilde{x}^{C}$%
\textit{\ on the support of }$X^{C}$\textit{\ and two elements }$k$\textit{\
and }$j$\textit{\ of the set }$\left\{ 1,2,\cdots ,\dim \left( X^{C}\right)
\right\} $\textit{\ such that }(i)\textit{\ }$\partial _{1}r_{d}\left(
x^{C\prime }\gamma _{d}^{C},\pi _{0}\left( x^{C}\right) \right) \neq 0$%
\textit{, }(ii)\textit{\ }$\partial _{1}r_{d}\left( \tilde{x}^{C\prime
}\gamma _{d}^{C},\pi _{0}\left( \tilde{x}^{C}\right) \right) \neq 0$\textit{%
, }(iii)\textit{\ }$\partial _{k}\pi _{0}\left( x^{C}\right) \neq \partial
_{1}\pi _{0}\left( x^{C}\right) \gamma _{d,k}^{C}$\textit{, }(iv)\textit{\ }$%
\partial _{j}\pi _{0}\left( x^{C}\right) \neq \partial _{1}\pi _{0}\left(
x^{C}\right) \gamma _{d,j}^{C}$\textit{, }(v)\textit{\ }$\partial _{l}\pi
_{0}\left( \tilde{x}^{C}\right) \neq \partial _{1}\pi _{0}\left( \tilde{x}%
^{C}\right) \gamma _{d,l}^{C}$\textit{\ for }$l=2,\cdots ,\dim \left(
X^{C}\right) $\textit{, and}%
\begin{eqnarray*}
&\text{(vi) }&\partial _{k}\pi _{0}\left( x^{C}\right) \partial _{j}\pi
_{0}\left( \tilde{x}^{C}\right) -\partial _{k}\pi _{0}\left( \tilde{x}%
^{C}\right) \partial _{j}\pi _{0}\left( x^{C}\right) \\
&&\neq \left[ \partial _{1}\pi _{0}\left( x^{C}\right) \partial _{j}\pi
_{0}\left( \tilde{x}^{C}\right) -\partial _{1}\pi _{0}\left( \tilde{x}%
^{C}\right) \partial _{j}\pi _{0}\left( x^{C}\right) \right] \gamma
_{d,k}^{C} \\
&&-\left[ \partial _{1}\pi _{0}\left( x^{C}\right) \partial _{k}\pi
_{0}\left( \tilde{x}^{C}\right) -\partial _{1}\pi _{0}\left( \tilde{x}%
^{C}\right) \partial _{k}\pi _{0}\left( x^{C}\right) \right] \gamma
_{d,j}^{C}.
\end{eqnarray*}

\medskip

Assumption D.NL.(i)-(ii) require $r_{d}$ to depend on the linear index, and
(iii)-(vi) essentially require some nonlinear variation in $\pi _{0}$ under
the scale normalization. As in Assumption NL, it is difficult to construct
examples other than $\pi _{0}\left( x^{C}\right) =f\left( x^{C\prime }\gamma
\right) $ that violates Assumption D.NL. Finally, a support assumption and
an invertibility assumption are imposed for technical reasons.

\medskip

\textbf{Assumption D.S} (Support). \textit{For each }$k\in \left\{
1,2,\cdots ,\dim \left( X^{D}\right) \right\} $\textit{, assume for some }$%
x_{k}^{D}\neq 0$\textit{\ in the support of }$X_{k}^{D}$\textit{\ that there
exists }$x^{C}\left( k\right) $\textit{\ in the support of }$X^{C}$\textit{\
such that }$\pi \left( x^{C}\left( k\right) ,x^{Dk}\right) $\textit{\ is in
the support of }$\pi _{0}\left( X^{C}\right) $\textit{\ and that }$%
x^{C}\left( k\right) ^{\prime }\gamma _{d}^{C}+x_{k}^{D}\gamma _{d,k}^{D}$%
\textit{\ is in the support of }$X^{C\prime }\gamma _{d}^{C}$\textit{\ for }$%
d=0,1$\textit{.}

\medskip

\textbf{Assumption D.I} (Invertibility). \textit{Assume that the function }$%
r_{d}$\textit{\ is invertible on its first argument for }$d=0,1$\textit{.}

\medskip

For the case of a binary outcome, we have $r_{0}\left( t,p\right)
=\int_{p}^{1}F_{\left. V_{0}\right\vert V}\left( \left. t\right\vert
v\right) dv$ and $r_{1}\left( t,p\right) =\int_{0}^{p}F_{\left.
V_{1}\right\vert V}\left( \left. t\right\vert v\right) dv$, hence a
sufficient condition for Assumption D.I to hold is that $V_{d}$ is
continuously distributed with support $%
\mathbb{R}
$, conditional on $V$, for $d=0,1$. Interestingly, the same continuous
distribution condition also suffices for Assumption D.I in the censored case
with $l\left( t,v\right) =\max \left( t-v,0\right) $, where $\partial
_{1}r_{0}\left( t,p\right) =\int_{p}^{1}F_{\left. V_{0}\right\vert V}\left(
\left. t\right\vert v\right) dv$ and $\partial _{1}r_{1}\left( t,p\right)
=\int_{0}^{p}F_{\left. V_{1}\right\vert V}\left( \left. t\right\vert
v\right) dv$. Under the imposed assumptions, the following identification
theorem for the generalized linear model follows immediately from 
\citet[Theorems
3.1-3.2]{escanciano2016identification}.

\begin{theorem}
\label{theorem:A} If Assumptions D.L, D.NL, D.CI, D.S, D.N, and D.I hold,
then $\gamma _{d}$ and $r_{d}\left( t,p\right) $ at all points in the
support of $X^{\prime }\gamma _{d}$\ and $P$ are identified for $d=0,1$.
\end{theorem}

Having established identification of $\gamma _{d}$ and $r_{d}$, the MTE for
limited valued outcomes can be identified by (\ref{AMTE}) and estimated by
the semiparametric least squares method of \cite%
{escanciano2016identification}. Alternatively, the LIV estimation procedure
may be adapted according to%
\begin{equation*}
E\left[ \left. Y\right\vert X=x\right] =s\left( x^{\prime }\gamma
_{0},x^{\prime }\gamma _{1},\pi \left( x\right) \right)
\end{equation*}%
and%
\begin{eqnarray*}
\Delta ^{\text{MTE}}\left( x,v\right) &=&\frac{\partial s\left( x^{\prime
}\gamma _{0},x^{\prime }\gamma _{1},v\right) }{\partial v}, \\
\Delta ^{\text{ATE}}\left( x\right) &=&s\left( x^{\prime }\gamma
_{0},x^{\prime }\gamma _{1},1\right) -s\left( x^{\prime }\gamma
_{0},x^{\prime }\gamma _{1},0\right) , \\
\Delta ^{\text{TT}}\left( x\right) &=&\frac{s\left( x^{\prime }\gamma
_{0},x^{\prime }\gamma _{1},\pi \left( x\right) \right) -s\left( x^{\prime
}\gamma _{0},x^{\prime }\gamma _{1},0\right) }{\pi \left( x\right) }, \\
\Delta ^{\text{TUT}}\left( x\right) &=&\frac{s\left( x^{\prime }\gamma
_{0},x^{\prime }\gamma _{1},1\right) -s\left( x^{\prime }\gamma
_{0},x^{\prime }\gamma _{1},\pi \left( x\right) \right) }{1-\pi \left(
x\right) }, \\
\Delta ^{\text{LATE}}\left( x,v_{1},v_{2}\right) &=&\frac{s\left( x^{\prime
}\gamma _{0},x^{\prime }\gamma _{1},v_{2}\right) -s\left( x^{\prime }\gamma
_{0},x^{\prime }\gamma _{1},v_{1}\right) }{v_{2}-v_{1}},
\end{eqnarray*}%
where%
\begin{equation*}
s\left( t_{0},t_{1},p\right) =r_{0}\left( t_{0},p\right) +r_{1}\left(
t_{1},p\right) .
\end{equation*}

\renewcommand{\theequation}{E.\arabic{equation}} \setcounter{equation}{0} %
\renewcommand*{\theHequation}{\theequation}

\renewcommand{\thelemma}{E.\arabic{lemma}} \setcounter{theorem}{0} %
\renewcommand*{\theHlemma}{\thelemma}

\section{Sufficient conditions and proofs for Theorem \protect\ref%
{theorem:AN1}}

\label{appendix:AN}

\subsection{Sufficient conditions}

We first give the assumptions that have been mentioned when presenting
Theorem \ref{theorem:AN1}. We denote $f_{X}\left( \cdot \right) $ and $%
f_{P}\left( \cdot \right) $ as the probability density function (PDF) of $%
X=\left( X^{C},X^{D}\right) $ and $P=\pi \left( X\right) $, respectively,
and $f_{P}^{\left( 1\right) }\left( p\right) =\left. \partial f_{P}\left(
p\right) \right/ \partial p$ as the derivative function of $f_{P}\left(
\cdot \right) $ if it is differentiable. For $d=0,1$, we define $%
V_{d}=Y-X^{\prime }\beta _{d}-g_{d}\left( P\right) $, which satisfies $E%
\left[ \left. V_{d}\right\vert X,D=d\right] =0$ because by definition $%
g_{d}\left( p\right) =E\left[ \left. Y-X^{\prime }\beta _{d}\right\vert \pi
\left( X\right) =p,D=d\right] $. And we denote $\sigma _{V_{d}}^{2}\left(
p\right) =Var\left( \left. V_{d}\right\vert P=p,D=d\right) =E\left[ \left.
V_{d}^{2}\right\vert P=p,D=d\right] $. Note that%
\begin{eqnarray*}
\sigma _{V_{d}}^{2}\left( p\right) &=&E\left[ \left. \left(
U_{d}-g_{d}\left( P\right) \right) ^{2}\right\vert P=p,D=d\right] \\
&=&E\left[ \left. U_{d}^{2}\right\vert P=p,D=d\right] -g_{d}\left( p\right)
\\
&=&\sigma _{U_{d}}^{2}\left( p\right) -g_{d}\left( p\right) .
\end{eqnarray*}%
In addition, we denote%
\begin{equation*}
\delta _{d}\left( p\right) =E\left[ \left. 1\left\{ D=d\right\} \right\vert
P=p\right] =\left\{ 
\begin{array}{cc}
p, & \text{if }d=1, \\ 
1-p, & \text{if }d=0,%
\end{array}%
\right.
\end{equation*}%
and%
\begin{equation*}
\delta _{d}^{\left( 1\right) }\left( p\right) =\frac{\partial \delta
_{d}\left( p\right) }{\partial p}=\left\{ 
\begin{array}{cc}
1, & \text{if }d=1, \\ 
-1, & \text{if }d=0.%
\end{array}%
\right.
\end{equation*}

\medskip

\textbf{Assumption E.1} (First step). For an integer $s\geq 2$, we assume
that:

(i) $\pi \left( x^{C},x^{D}\right) $ and $f_{X}\left( x^{C},x^{D}\right) $
are $s$-times differentiable with respect to $x^{C}$ for all $x^{D}$, and
the derivative functions are all Lipschitz continuous and bounded;

(ii) $\inf_{x^{C}}f_{X}\left( x^{C},x^{D}\right) \geq c$ for a positive
constant $c$ for all $x^{D}$;

(iii) $k_{1}\left( \cdot \right) $ is a bounded $s$-th order kernel function
satisfying that $\int k_{1}\left( u\right) du=1$, $\int u^{t}k_{1}\left(
u\right) du=0$ for $t=1,\cdots ,s-1$, and $\int u^{s}k_{1}\left( u\right)
du\neq 0$, and the functions $K_{1l}\left( u\right) =\left\vert u\right\vert
^{l}k_{1}\left( u\right) $ are Lipschitz continuous for all $0\leq l\leq s+1$%
.

(iv) $h_{1l}\rightarrow 0$ for $l=1,2,\cdots ,\dim \left( X^{C}\right) $,
and $\left. nh_{11}h_{12}\cdots h_{1,\dim \left( X^{C}\right) }\right/ \ln
n\rightarrow \infty $, as $n\rightarrow \infty $.

\medskip

\textbf{Assumption E.2} (Second step). For $d=0,1$, we assume that:

(i) $\left\Vert \hat{\beta}_{d}-\beta _{d}\right\Vert =O_{p}\left(
n^{-1/2}\right) $;

(ii) $g_{d}\left( \cdot \right) $, $f_{P}\left( \cdot \right) $, and $\sigma
_{V_{d}}^{2}\left( \cdot \right) $ are twice continuously differentiable
over the support of $P$, and $\inf_{p\in \left[ 0,1\right] }f_{P}\left(
p\right) \geq c$ for a positive constant $c$.

(iii) $k_{3}\left( \cdot \right) $ is a bounded, symmetric, compactly
supported, and three times continuously differentiable kernel function;

(iv) $h_{3}\rightarrow 0$, $nh_{3}^{3}\rightarrow \infty $, and $%
nh_{3}^{7}\rightarrow 0$, as $n\rightarrow \infty $.

\medskip

\textbf{Assumption E.3} (Rate). We assume that as $n\rightarrow \infty $,
(i) $nh_{3}\eta _{1}^{2}\rightarrow 0$, (ii) $\left. n\eta _{1}^{6}\right/
h_{3}^{7}\rightarrow 0$, (iii) $nh_{3}^{3}\left\vert h_{1}\right\vert
\rightarrow \infty $, and (iv) $n^{3/2}h_{3}^{5}\left\vert h_{1}\right\vert
^{2}\rightarrow \infty $, where%
\begin{eqnarray*}
\eta _{1} &=&tr\left( h_{1}^{s}\right) +\sqrt{\frac{\ln n}{n\left\vert
h_{1}\right\vert }}, \\
tr\left( h_{1}^{s}\right) &=&\sum_{l=1}^{\dim \left( X^{C}\right)
}h_{1l}^{s}, \\
\left\vert h_{1}\right\vert &=&h_{11}h_{12}\cdots h_{1,\dim \left(
X^{C}\right) }.
\end{eqnarray*}

\medskip

\textbf{Assumption E.4} (Moment). We assume that $E\left[ \left\Vert
X\right\Vert ^{2}\right] <\infty $ and $E\left[ V_{d}^{2}\right] <\infty $
for $d=0,1$.

\medskip

Assumption E.1.(iii) employs the higher-order kernel device to reduce the
bias and improve the convergence rate of the first-step estimation (\ref%
{PSfunction}). Since the first step is a multivariate kernel regression, the
usual symmetric (second-order) kernel will necessarily lead to a rate of
convergence that dominates the subsequent local linear and MTE estimation,
which is apparently undesirable. By using the higher-order kernel, Lemma \ref%
{lemma:firststep} below implies that%
\begin{equation*}
\sup_{x}\left\vert \hat{\pi}\left( x\right) -\pi \left( x\right) \right\vert
=O_{p}\left( \eta _{1}\right) .
\end{equation*}%
Therefore, Assumption E.3.(i) ensures the asymptotic negligibility of the
first-step estimation error relative to the subsequent local linear
estimation. If $h_{1l}$ follows a rate $O\left( n^{-1\left/ \left( 2s+\dim
\left( X^{C}\right) \right) \right. }\right) $ that optimally balances the
squared bias and the variance, the first-step rate of convergence will
become $O_{p}\left( \left( \ln n\right) ^{1/2}n^{-s\left/ \left( 2s+\dim
\left( X^{C}\right) \right) \right. }\right) $. Further suppose that $h_{3}$
follows a rule-of-thumb rate $O\left( n^{-1/5}\right) $, then Assumption
E.3.(i) will reduce to $s>2\dim \left( X^{C}\right) $. Note that under the
supposed rates of $h_{1l}$ and $h_{3}$, Assumption E.3.(ii), (iii), and (iv)
will reduce to $s>2\dim \left( X^{C}\right) $, $s>\left( 3/4\right) \dim
\left( X^{C}\right) $, and $s>\left( 3/2\right) \dim \left( X^{C}\right) $,
respectively, which are all implied by Assumption E.3.(i).

Assumption E.2.(i) imposes a weak rate condition on the slope estimation in
the second step. The $\sqrt{n}$-consistency of the weighted pairwise
difference estimation (\ref{PDLS}) has been established by 
\citet[Corollary
3.1]{ahn1993semiparametric}. Actually, most semiparametric estimators for
sample selection models satisfy Assumption E.2.(i) as well, such as the
Robinson-type partialling-out estimator \cite[p.307]{pagan1999nonparametric}%
, the symmetry-based pairwise difference estimator %
\citep{chen2010semiparametric}, and the integrated least squares estimator %
\citep{pan2022semiparametric}, among others.

\subsection{Proof of Theorem \protect\ref{theorem:AN1}}

We first develop the asymptotic normality of the local linear estimator $%
\left( \hat{g}_{d}\left( p\right) ,\hat{g}_{d}^{\left( 1\right) }\left(
p\right) \right) $ for $d=0,1$. To this end, we define the infeasible
estimator as%
\begin{equation*}
\left( 
\begin{array}{c}
\tilde{g}_{d}\left( p\right) \\ 
\tilde{g}_{d}^{\left( 1\right) }\left( p\right)%
\end{array}%
\right) =\left[ \sum_{i=1}^{n}w_{di}\left( p\right) \left( 
\begin{array}{c}
1 \\ 
P_{i}-p%
\end{array}%
\right) \left( 
\begin{array}{c}
1 \\ 
P_{i}-p%
\end{array}%
\right) ^{\prime }\right] ^{-1}\left[ \sum_{i=1}^{n}w_{di}\left( p\right)
\left( 
\begin{array}{c}
1 \\ 
P_{i}-p%
\end{array}%
\right) \left( Y_{i}-X_{i}^{\prime }\beta _{d}\right) \right] ,
\end{equation*}%
where%
\begin{equation*}
w_{di}\left( p\right) =1\left\{ D_{i}=d\right\} \frac{1}{h_{3}}k_{3}\left( 
\frac{P_{i}-p}{h_{3}}\right) .
\end{equation*}%
We will analyze $\left( \tilde{g}_{d}\left( p\right) ,\tilde{g}_{d}^{\left(
1\right) }\left( p\right) \right) $ and $\left( \hat{g}_{d}\left( p\right) ,%
\hat{g}_{d}^{\left( 1\right) }\left( p\right) \right) -\left( \tilde{g}%
_{d}\left( p\right) ,\tilde{g}_{d}^{\left( 1\right) }\left( p\right) \right) 
$ in order.

The asymptotic normality of $\left( \tilde{g}_{d}\left( p\right) ,\tilde{g}%
_{d}^{\left( 1\right) }\left( p\right) \right) $ can be established by
standard arguments for kernel regression 
\citep[e.g.,][Subsection
2.7.2]{li2007nonparametric}, except for the presence of the treatment status
indicator $1\left\{ D_{i}=d\right\} $ in our case. For completeness, we
provide a proof for the asymptotic normality of $\left( \tilde{g}_{d}\left(
p\right) ,\tilde{g}_{d}^{\left( 1\right) }\left( p\right) \right) $ in Lemma %
\ref{lemma:secondstep}. In the proof, we obtain an asymptotically linear
representation of $\left( \tilde{g}_{d}\left( p\right) ,\tilde{g}%
_{d}^{\left( 1\right) }\left( p\right) \right) $ in (\ref{AsymptoticLinear}):%
\begin{equation}
J_{n}\left[ \left( 
\begin{array}{c}
\tilde{g}_{d}\left( p\right) \\ 
\tilde{g}_{d}^{\left( 1\right) }\left( p\right)%
\end{array}%
\right) -\left( 
\begin{array}{c}
g_{d}\left( p\right) \\ 
g_{d}^{\left( 1\right) }\left( p\right)%
\end{array}%
\right) -\left( 
\begin{array}{c}
\frac{\kappa _{2}}{2}g_{d}^{\left( 2\right) }\left( p\right) h_{3}^{2} \\ 
0%
\end{array}%
\right) \right] =\Gamma _{d}J_{n}\frac{1}{n}\sum_{i=1}^{n}w_{di}\left(
p\right) \left( 
\begin{array}{c}
1 \\ 
\frac{P_{i}-p}{h_{3}^{2}}%
\end{array}%
\right) V_{di}+o_{p}\left( 1\right) ,  \label{gtilde}
\end{equation}%
where%
\begin{eqnarray*}
\Gamma _{d} &=&\left( 
\begin{array}{cc}
1\left/ \left[ \delta _{d}\left( p\right) f_{P}\left( p\right) \right]
\right. , & 0 \\ 
0, & 1\left/ \left[ \kappa _{2}\delta _{d}\left( p\right) f_{P}\left(
p\right) \right] \right.%
\end{array}%
\right) , \\
J_{n} &=&\left( 
\begin{array}{cc}
\sqrt{nh_{3}} & 0 \\ 
0 & \sqrt{nh_{3}^{3}}%
\end{array}%
\right) .
\end{eqnarray*}%
We next consider $J_{n}\left[ \left( \hat{g}_{d}\left( p\right) ,\hat{g}%
_{d}^{\left( 1\right) }\left( p\right) \right) -\left( \tilde{g}_{d}\left(
p\right) ,\tilde{g}_{d}^{\left( 1\right) }\left( p\right) \right) \right] $.

For $r=0,1,2$, we denote%
\begin{eqnarray*}
\hat{B}_{r} &=&\frac{1}{n}\sum_{i}^{n}\hat{w}_{di}\left( p\right) \left( 
\hat{P}_{i}-p\right) ^{r}, \\
B_{r} &=&\frac{1}{n}\sum_{i=1}^{n}w_{di}\left( p\right) \left(
P_{i}-p\right) ^{r},
\end{eqnarray*}%
and%
\begin{equation*}
\hat{Q}=\left( 
\begin{array}{cc}
\hat{B}_{0}, & \hat{B}_{1} \\ 
\hat{B}_{1}, & \hat{B}_{2}%
\end{array}%
\right) ,\text{ }Q=\left( 
\begin{array}{cc}
B_{0}, & B_{1} \\ 
B_{1}, & B_{2}%
\end{array}%
\right) ,
\end{equation*}%
then%
\begin{eqnarray*}
\left( 
\begin{array}{c}
\hat{g}_{d}\left( p\right) \\ 
\hat{g}_{d}^{\left( 1\right) }\left( p\right)%
\end{array}%
\right) &=&\hat{Q}^{-1}\frac{1}{n}\sum_{i}\hat{w}_{di}\left( p\right) \left( 
\begin{array}{c}
1 \\ 
\hat{P}_{i}-p%
\end{array}%
\right) \left( Y_{i}-X_{i}^{\prime }\hat{\beta}_{d}\right) \\
&=&\frac{1}{\left\vert \hat{Q}\right\vert }\frac{1}{n}\sum_{i}\hat{w}%
_{di}\left( p\right) \left( 
\begin{array}{c}
\hat{B}_{2}-\left( \hat{P}_{i}-p\right) \hat{B}_{1} \\ 
-\hat{B}_{1}+\left( \hat{P}_{i}-p\right) \hat{B}_{0}%
\end{array}%
\right) \left( Y_{i}-X_{i}^{\prime }\hat{\beta}_{d}\right) , \\
\left( 
\begin{array}{c}
\tilde{g}_{d}\left( p\right) \\ 
\tilde{g}_{d}^{\left( 1\right) }\left( p\right)%
\end{array}
\right) &=&\frac{1}{\left\vert Q\right\vert }\frac{1}{n}\sum_{i}w_{di}\left(
p\right) \left( 
\begin{array}{c}
B_{2}-\left( P_{i}-p\right) B_{1} \\ 
-B_{1}+\left( P_{i}-p\right) B_{0}%
\end{array}%
\right) \left( Y_{i}-X_{i}^{\prime }\beta _{d}\right) ,
\end{eqnarray*}%
where%
\begin{equation*}
\left\vert \hat{Q}\right\vert =\hat{B}_{0}\hat{B}_{2}-\hat{B}_{1}^{2},\text{ 
}\left\vert Q\right\vert =B_{0}B_{2}-B_{1}^{2}.
\end{equation*}%
For $r=0,1$, Further denote%
\begin{eqnarray*}
\hat{L}_{r} &=&\frac{1}{n}\sum_{i}\hat{w}_{di}\left( p\right) \left[ \hat{B}%
_{r+1}-\left( \hat{P}_{i}-p\right) \hat{B}_{r}\right] \left(
Y_{i}-X_{i}^{\prime }\hat{\beta}_{d}\right) , \\
L_{r} &=&\frac{1}{n}\sum_{i}w_{di}\left( p\right) \left[ B_{r+1}-\left(
P_{i}-p\right) B_{r}\right] \left( Y_{i}-X_{i}^{\prime }\beta _{d}\right) ,
\end{eqnarray*}%
then we have%
\begin{equation*}
\left( 
\begin{array}{c}
\hat{g}_{d}\left( p\right) \\ 
\hat{g}_{d}^{\left( 1\right) }\left( p\right)%
\end{array}%
\right) =\frac{1}{\left\vert \hat{Q}\right\vert }\left( 
\begin{array}{c}
\hat{L}_{1} \\ 
-\hat{L}_{0}%
\end{array}
\right) ,\text{ }\left( 
\begin{array}{c}
\tilde{g}_{d}\left( p\right) \\ 
\tilde{g}_{d}^{\left( 1\right) }\left( p\right)%
\end{array}%
\right) =\frac{1}{\left\vert Q\right\vert }\left( 
\begin{array}{c}
L_{1} \\ 
-L_{0}%
\end{array}
\right) .
\end{equation*}

It follows from Lemma \ref{lemma:secondstep} that%
\begin{eqnarray*}
\frac{L_{1}}{\left\vert Q\right\vert } &=&g_{d}\left( p\right) +O_{p}\left(
h_{3}^{2}+\frac{1}{\sqrt{nh_{3}}}\right) =g_{d}\left( p\right) +o_{p}\left(
h_{3}\right) , \\
\frac{L_{0}}{\left\vert Q\right\vert } &=&-g_{d}^{\left( 1\right) }\left(
p\right) +O_{p}\left( \frac{1}{\sqrt{nh_{3}^{3}}}\right) =-g_{d}^{\left(
1\right) }\left( p\right) +o_{p}\left( 1\right) .
\end{eqnarray*}%
Therefore,%
\begin{eqnarray}
\hat{g}_{d}\left( p\right) -\tilde{g}_{d}\left( p\right) &=&\frac{\hat{L}_{1}%
}{\left\vert \hat{Q}\right\vert }-\frac{L_{1}}{\left\vert Q\right\vert }=%
\left[ \left( \hat{L}_{1}-L_{1}\right) -\left( \left\vert \hat{Q}\right\vert
-\left\vert Q\right\vert \right) \frac{L_{1}}{\left\vert Q\right\vert }%
\right] \frac{1}{\left\vert \hat{Q}\right\vert }  \notag \\
&=&\frac{\left( \hat{L}_{1}-L_{1}\right) -\left( \left\vert \hat{Q}%
\right\vert -\left\vert Q\right\vert \right) g_{d}\left( p\right)
+o_{p}\left( h_{3}\left( \left\vert \hat{Q}\right\vert -\left\vert
Q\right\vert \right) \right) }{\left\vert \hat{Q}\right\vert },
\label{gdhat-tilde} \\
\hat{g}_{d}^{\left( 1\right) }\left( p\right) -\tilde{g}_{d}^{\left(
1\right) }\left( p\right) &=&-\frac{\left( \hat{L}_{0}-L_{0}\right) +\left(
\left\vert \hat{Q}\right\vert -\left\vert Q\right\vert \right) g_{d}^{\left(
1\right) }\left( p\right) +o_{p}\left( \left\vert \hat{Q}\right\vert
-\left\vert Q\right\vert \right) }{\left\vert \hat{Q}\right\vert }.
\label{gd1hat-tilde}
\end{eqnarray}%
Denote%
\begin{equation}
V_{di}\left( p\right) =Y_{i}-X_{i}^{\prime }\beta _{d}-g_{d}\left( p\right)
=V_{di}+g_{d}\left( P_{i}\right) -g_{d}\left( p\right) ,  \label{Vdip}
\end{equation}%
which satisfies $E\left[ \left. V_{di}\left( p\right) \right\vert
X_{i},D_{i}=d\right] =g_{d}\left( P_{i}\right) -g_{d}\left( p\right) $ by
definition, and denote%
\begin{eqnarray*}
\hat{C}_{r} &=&\frac{1}{n}\sum_{i}^{n}\hat{w}_{di}\left( p\right) \left( 
\hat{P}_{i}-p\right) ^{r}V_{di}\left( p\right) , \\
C_{r} &=&\frac{1}{n}\sum_{i=1}^{n}w_{di}\left( p\right) \left(
P_{i}-p\right) ^{r}V_{di}\left( p\right) , \\
\hat{G}_{r} &=&\frac{1}{n}\sum_{i}^{n}\hat{w}_{di}\left( p\right) \left( 
\hat{P}_{i}-p\right) ^{r}X_{i},
\end{eqnarray*}%
for $r=0,1$. By several calculations,%
\begin{eqnarray*}
\left( \hat{L}_{1}-L_{1}\right) -\left( \left\vert \hat{Q}\right\vert
-\left\vert Q\right\vert \right) g_{d}\left( p\right) &=&\left( \hat{C}_{0}%
\hat{B}_{2}-C_{0}B_{2}\right) -\left( \hat{C}_{1}\hat{B}_{1}-C_{1}B_{1}%
\right) \\
&&-\left( \hat{\beta}_{d}-\beta _{d}\right) ^{\prime }\left( \hat{B}_{2}\hat{%
G}_{0}-\hat{B}_{1}\hat{G}_{1}\right) , \\
\hat{L}_{0}-L_{0} &=&\left( \hat{C}_{0}\hat{B}_{1}-C_{0}B_{1}\right) -\left( 
\hat{C}_{1}\hat{B}_{0}-C_{1}B_{0}\right) \\
&&-\left( \hat{\beta}_{d}-\beta _{d}\right) ^{\prime }\left( \hat{B}_{1}\hat{%
G}_{0}-\hat{B}_{0}\hat{G}_{1}\right) , \\
\left\vert \hat{Q}\right\vert -\left\vert Q\right\vert &=&\left( \hat{B}_{0}%
\hat{B}_{2}-B_{0}B_{2}\right) -\left( \hat{B}_{1}^{2}-B_{1}^{2}\right) .
\end{eqnarray*}

It follows from Lemmas \ref{lemma:kernel}-\ref{lemma:Grhat} that%
\begin{eqnarray*}
\hat{C}_{0}\hat{B}_{2}-C_{0}B_{2} &=&\left( \hat{C}_{0}-C_{0}\right) \left( 
\hat{B}_{2}-B_{2}\right) +C_{0}\left( \hat{B}_{2}-B_{2}\right) +\left( \hat{C%
}_{0}-C_{0}\right) B_{2} \\
&=&O_{p}\left( \frac{1}{\sqrt{nh_{3}}}\right) o_{p}\left( \sqrt{\frac{h_{3}}{%
n}}\right) +o_{p}\left( h_{3}\right) o_{p}\left( \sqrt{\frac{h_{3}}{n}}%
\right) \\
&&+\left( \hat{C}_{0}-C_{0}\right) \left[ \kappa _{2}\delta _{d}\left(
p\right) f_{P}\left( p\right) h_{3}^{2}+o_{p}\left( h_{3}^{3}\right) \right]
\\
&=&\frac{\kappa _{2}\delta _{d}\left( p\right) f_{P}\left( p\right)
g_{d}^{\left( 1\right) }\left( p\right) }{n}\sum_{i=1}^{n}\delta _{d}\left(
P_{i}\right) k_{3}^{\left( 1\right) }\left( \frac{P_{i}-p}{h_{3}}\right)
\left( P_{i}-p\right) \left( D_{i}-P_{i}\right) \\
&&+o_{p}\left( \sqrt{\frac{h_{3}^{3}}{n}}\right) ,
\end{eqnarray*}%
\begin{eqnarray*}
\hat{C}_{1}\hat{B}_{1}-C_{1}B_{1} &=&\left( \hat{C}_{1}-C_{1}\right) \left( 
\hat{B}_{1}-B_{1}\right) +\left( \hat{C}_{1}-C_{1}\right) B_{1}+C_{1}\left( 
\hat{B}_{1}-B_{1}\right) \\
&=&O_{p}\left( \sqrt{\frac{h_{3}}{n}}\right) o_{p}\left( \frac{1}{\sqrt{%
nh_{3}}}\right) +O_{p}\left( \sqrt{\frac{h_{3}}{n}}\right) O_{p}\left(
h_{3}^{2}\right) +O_{p}\left( h_{3}^{2}\right) o_{p}\left( \frac{1}{\sqrt{%
nh_{3}}}\right) \\
&=&o_{p}\left( \sqrt{\frac{h_{3}^{3}}{n}}\right) ,
\end{eqnarray*}%
\begin{eqnarray*}
\left( \hat{\beta}_{d}-\beta _{d}\right) ^{\prime }\left( \hat{B}_{2}\hat{G}%
_{0}-\hat{B}_{1}\hat{G}_{1}\right) &\leq &\left\Vert \hat{\beta}_{d}-\beta
_{d}\right\Vert \left[ 
\begin{array}{c}
\left( \left\vert \hat{B}_{2}-B_{2}\right\vert +\left\vert B_{2}\right\vert
\right) \left\Vert \hat{G}_{0}\right\Vert \\ 
+\left( \left\vert \hat{B}_{1}-B_{1}\right\vert +\left\vert B_{1}\right\vert
\right) \left\Vert \hat{G}_{1}\right\Vert%
\end{array}%
\right] \\
&=&O_{p}\left( \frac{1}{\sqrt{n}}\right) \left[ 
\begin{array}{c}
\left( o_{p}\left( \sqrt{\frac{h_{3}}{n}}\right) +O_{p}\left(
h_{3}^{2}\right) \right) O_{p}\left( 1\right) \\ 
+\left( o_{p}\left( \frac{1}{\sqrt{nh_{3}}}\right) +O_{p}\left(
h_{3}^{2}\right) \right) O_{p}\left( h_{3}\right)%
\end{array}%
\right] \\
&=&o_{p}\left( \sqrt{\frac{h_{3}^{3}}{n}}\right) ,
\end{eqnarray*}%
\begin{eqnarray*}
\hat{C}_{0}\hat{B}_{1}-C_{0}B_{1} &=&\left( \hat{C}_{0}-C_{0}\right) \left( 
\hat{B}_{1}-B_{1}\right) +C_{0}\left( \hat{B}_{1}-B_{1}\right) +\left( \hat{C%
}_{0}-C_{0}\right) B_{1} \\
&=&O_{p}\left( \frac{1}{\sqrt{nh_{3}}}\right) o_{p}\left( \frac{1}{\sqrt{%
nh_{3}}}\right) +o_{p}\left( h_{3}\right) o_{p}\left( \frac{1}{\sqrt{nh_{3}}}%
\right) +O_{p}\left( \frac{1}{\sqrt{nh_{3}}}\right) O_{p}\left(
h_{3}^{2}\right) \\
&=&o_{p}\left( \sqrt{\frac{h_{3}}{n}}\right) ,
\end{eqnarray*}%
\begin{eqnarray*}
\hat{C}_{1}\hat{B}_{0}-C_{1}B_{0} &=&\left( \hat{C}_{1}-C_{1}\right) \left( 
\hat{B}_{0}-B_{0}\right) +C_{1}\left( \hat{B}_{0}-B_{0}\right) +\left( \hat{C%
}_{1}-C_{1}\right) B_{0} \\
&=&O_{p}\left( \sqrt{\frac{h_{3}}{n}}\right) o_{p}\left( \frac{1}{\sqrt{%
nh_{3}^{3}}}\right) +O_{p}\left( h_{3}^{2}\right) o_{p}\left( \frac{1}{\sqrt{%
nh_{3}^{3}}}\right) \\
&&+\left( \hat{C}_{1}-C_{1}\right) \left[ \delta _{d}\left( p\right)
f_{P}\left( p\right) +o_{p}\left( h_{3}\right) \right] \\
&=&\frac{\delta _{d}\left( p\right) f_{P}\left( p\right) g_{d}^{\left(
1\right) }\left( p\right) }{n}\sum_{i=1}^{n}\delta _{d}\left( P_{i}\right) %
\left[ k_{3}\left( \frac{P_{i}-p}{h_{3}}\right) +k_{3}^{\left( 1\right)
}\left( \frac{P_{i}-p}{h_{3}}\right) \frac{P_{i}-p}{h_{3}}\right] \\
&&\cdot \frac{P_{i}-p}{h_{3}}\left( D_{i}-P_{i}\right) +o_{p}\left( \sqrt{%
\frac{h_{3}}{n}}\right) ,
\end{eqnarray*}%
\begin{eqnarray*}
\left( \hat{\beta}_{d}-\beta _{d}\right) ^{\prime }\left( \hat{B}_{1}\hat{G}%
_{0}-\hat{B}_{0}\hat{G}_{1}\right) &\leq &\left\Vert \hat{\beta}_{d}-\beta
_{d}\right\Vert \left[ 
\begin{array}{c}
\left( \left\vert \hat{B}_{1}-B_{1}\right\vert +\left\vert B_{1}\right\vert
\right) \left\Vert \hat{G}_{0}\right\Vert \\ 
+\left( \left\vert \hat{B}_{0}-B_{0}\right\vert +\left\vert B_{0}\right\vert
\right) \left\Vert \hat{G}_{1}\right\Vert%
\end{array}%
\right] \\
&=&O_{p}\left( \frac{1}{\sqrt{n}}\right) \left[ 
\begin{array}{c}
\left( o_{p}\left( \frac{1}{\sqrt{nh_{3}}}\right) +O_{p}\left(
h_{3}^{2}\right) \right) O_{p}\left( 1\right) \\ 
+\left( o_{p}\left( \frac{1}{\sqrt{nh_{3}^{3}}}\right) +O_{p}\left( 1\right)
\right) O_{p}\left( h_{3}\right)%
\end{array}%
\right] \\
&=&o_{p}\left( \sqrt{\frac{h_{3}}{n}}\right) ,
\end{eqnarray*}%
\begin{eqnarray*}
\hat{B}_{0}\hat{B}_{2}-B_{0}B_{2} &=&\left( \hat{B}_{0}-B_{0}\right) \left( 
\hat{B}_{2}-B_{2}\right) +B_{0}\left( \hat{B}_{2}-B_{2}\right) +\left( \hat{B%
}_{0}-B_{0}\right) B_{2} \\
&=&o_{p}\left( \frac{1}{\sqrt{nh_{3}^{3}}}\right) o_{p}\left( \sqrt{\frac{%
h_{3}}{n}}\right) +O_{p}\left( 1\right) o_{p}\left( \sqrt{\frac{h_{3}}{n}}%
\right) +o_{p}\left( \frac{1}{\sqrt{nh_{3}^{3}}}\right) O_{p}\left(
h_{3}^{2}\right) \\
&=&o_{p}\left( \sqrt{\frac{h_{3}}{n}}\right) ,
\end{eqnarray*}%
and%
\begin{eqnarray*}
\hat{B}_{1}^{2}-B_{1}^{2} &=&\left( \hat{B}_{1}-B_{1}\right) \left[ \left( 
\hat{B}_{1}-B_{1}\right) +2B_{1}\right] \\
&=&o_{p}\left( \frac{1}{\sqrt{nh_{3}}}\right) \left[ o_{p}\left( \frac{1}{%
\sqrt{nh_{3}}}\right) +O_{p}\left( h_{3}^{2}\right) \right] \\
&=&o_{p}\left( \sqrt{\frac{h_{3}}{n}}\right) .
\end{eqnarray*}%
Therefore,%
\begin{eqnarray}
\left( \hat{L}_{1}-L_{1}\right) -\left( \left\vert \hat{Q}\right\vert
-\left\vert Q\right\vert \right) g_{d}\left( p\right) &=&\frac{\kappa
_{2}\delta _{d}\left( p\right) f_{P}\left( p\right) g_{d}^{\left( 1\right)
}\left( p\right) }{n}\sum_{i=1}^{n}\delta _{d}\left( P_{i}\right)
k_{3}^{\left( 1\right) }\left( \frac{P_{i}-p}{h_{3}}\right)  \notag \\
&&\cdot \left( P_{i}-p\right) \left( D_{i}-P_{i}\right) +o_{p}\left( \sqrt{%
\frac{h_{3}^{3}}{n}}\right) ,  \label{L1hat}
\end{eqnarray}%
\begin{eqnarray}
\hat{L}_{0}-L_{0} &=&-\frac{\delta _{d}\left( p\right) f_{P}\left( p\right)
g_{d}^{\left( 1\right) }\left( p\right) }{n}\sum_{i=1}^{n}\delta _{d}\left(
P_{i}\right) \left[ k_{3}\left( \frac{P_{i}-p}{h_{3}}\right) +k_{3}^{\left(
1\right) }\left( \frac{P_{i}-p}{h_{3}}\right) \frac{P_{i}-p}{h_{3}}\right] 
\notag \\
&&\cdot \frac{P_{i}-p}{h_{3}}\left( D_{i}-P_{i}\right) +o_{p}\left( \sqrt{%
\frac{h_{3}}{n}}\right) ,  \label{L0hat}
\end{eqnarray}%
and%
\begin{equation}
\left\vert \hat{Q}\right\vert -\left\vert Q\right\vert =o_{p}\left( \sqrt{%
\frac{h_{3}}{n}}\right) .  \label{detQhat}
\end{equation}%
In addition, it follows from Lemma \ref{lemma:kernel} that%
\begin{equation*}
\left\vert Q\right\vert =\kappa _{2}\delta _{d}^{2}\left( p\right)
f_{P}^{2}\left( p\right) h_{3}^{2}+o_{p}\left( h_{3}^{2}\right) .
\end{equation*}%
Hence,%
\begin{equation*}
\frac{\left\vert \hat{Q}\right\vert }{\left\vert Q\right\vert }-1=\frac{%
\left\vert \hat{Q}\right\vert -\left\vert Q\right\vert }{\left\vert
Q\right\vert }=\frac{o_{p}\left( 1\left/ \sqrt{nh_{3}^{3}}\right. \right) }{%
\kappa _{2}\delta _{d}^{2}\left( p\right) f_{P}^{2}\left( p\right)
+o_{p}\left( 1\right) }=o_{p}\left( \frac{1}{\sqrt{nh_{3}^{3}}}\right)
=o_{p}\left( 1\right) ,
\end{equation*}%
and%
\begin{equation}
\frac{1}{\left\vert \hat{Q}\right\vert }=\frac{1}{\left\vert Q\right\vert }%
\left( 1+o_{p}\left( 1\right) \right) =\frac{1}{\kappa _{2}\delta
_{d}^{2}\left( p\right) f_{P}^{2}\left( p\right) h_{3}^{2}}\left(
1+o_{p}\left( 1\right) \right) .  \label{1/detQhat}
\end{equation}

Substituting (\ref{L1hat})-(\ref{1/detQhat}) into (\ref{gdhat-tilde})-(\ref%
{gd1hat-tilde}) yields that%
\begin{eqnarray*}
\hat{g}_{d}\left( p\right) -\tilde{g}_{d}\left( p\right) &=&\left( \frac{%
g_{d}^{\left( 1\right) }\left( p\right) }{\delta _{d}\left( p\right)
f_{P}\left( p\right) }\right) \frac{1}{n}\sum_{i=1}^{n}\delta _{d}\left(
P_{i}\right) \frac{P_{i}-p}{h_{3}^{2}}k_{3}^{\left( 1\right) }\left( \frac{%
P_{i}-p}{h_{3}}\right) \left( D_{i}-P_{i}\right) \\
&&+o_{p}\left( \sqrt{\frac{1}{nh_{3}}}\right) , \\
\hat{g}_{d}^{\left( 1\right) }\left( p\right) -\tilde{g}_{d}^{\left(
1\right) }\left( p\right) &=&\left( \frac{g_{d}^{\left( 1\right) }\left(
p\right) }{\kappa _{2}\delta _{d}\left( p\right) f_{P}\left( p\right) }%
\right) \frac{1}{n}\sum_{i=1}^{n}\delta _{d}\left( P_{i}\right) \left[ \frac{%
1}{h_{3}}k_{3}\left( \frac{P_{i}-p}{h_{3}}\right) \right. \\
&&\left. +\frac{P_{i}-p}{h_{3}^{2}}k_{3}^{\left( 1\right) }\left( \frac{%
P_{i}-p}{h_{3}}\right) \right] \frac{P_{i}-p}{h_{3}^{2}}\left(
D_{i}-P_{i}\right) +o_{p}\left( \sqrt{\frac{1}{nh_{3}^{3}}}\right) .
\end{eqnarray*}%
Further combining with (\ref{gtilde}), we obtain%
\begin{equation}
J_{n}\left[ \left( 
\begin{array}{c}
\hat{g}_{d}\left( p\right) \\ 
\hat{g}_{d}^{\left( 1\right) }\left( p\right)%
\end{array}
\right) -\left( 
\begin{array}{c}
g_{d}\left( p\right) \\ 
g_{d}^{\left( 1\right) }\left( p\right)%
\end{array}
\right)-\left( 
\begin{array}{c}
\frac{\kappa _{2}}{2}g_{d}^{\left( 2\right) }\left( p\right) h_{3}^{2} \\ 
0%
\end{array}
\right) \right] =\Gamma _{d}A_{dn}\left( p\right) +o_{p}\left( 1\right) ,
\label{asymplinear}
\end{equation}%
where%
\begin{eqnarray*}
A_{dn}\left( p\right) &=&J_{n}\frac{1}{n}\sum_{i=1}^{n}\left( 
\begin{array}{c}
1 \\ 
\frac{P_{i}-p}{h_{3}^{2}}%
\end{array}%
\right) \left[ w_{di}\left( p\right) V_{di}+\left( 
\begin{array}{c}
\Lambda _{di}\left( p\right) \\ 
\Lambda _{di}\left( p\right) +\Upsilon _{di}\left( p\right)%
\end{array}%
\right) \left( D_{i}-P_{i}\right) \right] , \\
\Lambda _{di}\left( p\right) &=&g_{d}^{\left( 1\right) }\left( p\right)
\delta _{d}\left( P_{i}\right) \frac{P_{i}-p}{h_{3}^{2}}k_{3}^{\left(
1\right) }\left( \frac{P_{i}-p}{h_{3}}\right) , \\
\Upsilon _{di}\left( p\right) &=&g_{d}^{\left( 1\right) }\left( p\right)
\delta _{d}\left( P_{i}\right) \frac{1}{h_{3}}k_{3}\left( \frac{P_{i}-p}{%
h_{3}}\right) .
\end{eqnarray*}

Since%
\begin{eqnarray*}
E\left[ \left. w_{di}\left( p\right) V_{di}\right\vert P_{i}\right] &=&\frac{%
1}{h_{3}}k_{3}\left( \frac{P_{i}-p}{h_{3}}\right) \delta _{d}\left(
P_{i}\right) E\left[ \left. V_{di}\right\vert P_{i},D_{i}=d\right] =0, \\
E\left[ \left. D_{i}-P_{i}\right\vert P_{i}\right] &=&0, \\
E\left[ \left. w_{di}\left( p\right) V_{di}\left( D_{i}-P_{i}\right)
\right\vert P_{i}\right] &=&\left( d-P_{i}\right) E\left[ \left.
w_{di}\left( p\right) V_{di}\right\vert P_{i}\right] =0,
\end{eqnarray*}%
we have{\small 
\begin{eqnarray*}
E\left[ A_{dn}\left( p\right) \right] &=&0, \\
Var\left( A_{dn\left( 1\right) }\left( p\right) \right) &=&h_{3}\left( E 
\left[ w_{di}^{2}\left( p\right) V_{di}^{2}\right] +E\left[ \Lambda
_{di}^{2}\left( p\right) \left( D_{i}-P_{i}\right) ^{2}\right] \right) \\
&=&\kappa \delta _{d}\left( p\right) f_{P}\left( p\right) \sigma
_{V_{d}}^{2}\left( p\right) +\kappa _{22}^{\left( 1\right) }p\left(
1-p\right) \delta _{d}^{2}\left( p\right) f_{P}\left( p\right) \left(
g_{d}^{\left( 1\right) }\left( p\right) \right) ^{2}+o\left( 1\right) , \\
Var\left( A_{dn\left( 2\right) }\left( p\right) \right) &=&h_{3}^{3}\left( E 
\left[ \frac{\left( P_{i}-p\right) ^{2}}{h_{3}^{4}}w_{di}^{2}\left( p\right)
V_{di}^{2}\right] +E\left[ \frac{\left( P_{i}-p\right) ^{2}}{h_{3}^{4}}%
\left( \Lambda _{di}\left( p\right) +\Upsilon _{di}\left( p\right) \right)
^{2}\left( D_{i}-P_{i}\right) ^{2}\right] \right) \\
&=&\kappa _{22}\delta _{d}\left( p\right) f_{P}\left( p\right) \sigma
_{V_{d}}^{2}\left( p\right) +\left( \kappa _{24}^{\left( 1\right) }-2\kappa
_{22}\right) p\left( 1-p\right) \delta _{d}^{2}\left( p\right) f_{P}\left(
p\right) \left( g_{d}^{\left( 1\right) }\left( p\right) \right) ^{2}+o\left(
1\right) , \\
E\left[ A_{dn\left( 1\right) }\left( p\right) A_{dn\left( 2\right) }\left(
p\right) \right] &=&h_{3}^{2}\left( E\left[ \frac{P_{i}-p}{h_{3}^{2}}%
w_{di}^{2}\left( p\right) V_{di}^{2}\right] +E\left[ \frac{P_{i}-p}{h_{3}^{2}%
}\left( \Lambda _{di}^{2}\left( p\right) +\Lambda _{di}\left( p\right)
\Upsilon _{di}\left( p\right) \right) \left( D_{i}-P_{i}\right) ^{2}\right]
\right) \\
&=&O\left( h_{3}\right) =o\left( 1\right) ,
\end{eqnarray*}%
}where%
\begin{eqnarray*}
\kappa _{22}^{\left( 1\right) } &=&\int \left( k_{3}^{\left( 1\right)
}\left( u\right) \right) ^{2}u^{2}du, \\
\kappa _{24}^{\left( 1\right) } &=&\int \left( k_{3}^{\left( 1\right)
}\left( u\right) \right) ^{2}u^{4}du.
\end{eqnarray*}%
It then follows from Liapunov's central limit theorem that%
\begin{equation*}
A_{dn}\left( p\right) \rightarrow N\left( 0,\text{diag}\left( 
\begin{array}{c}
\kappa \delta _{d}\left( p\right) f_{P}\left( p\right) \sigma
_{V_{d}}^{2}\left( p\right) +\kappa _{22}^{\left( 1\right) }p\left(
1-p\right) \delta _{d}^{2}\left( p\right) f_{P}\left( p\right) \left(
g_{d}^{\left( 1\right) }\left( p\right) \right) ^{2} \\ 
\kappa _{22}\delta _{d}\left( p\right) f_{P}\left( p\right) \sigma
_{V_{d}}^{2}\left( p\right) +\left( \kappa _{24}^{\left( 1\right) }-2\kappa
_{22}\right) p\left( 1-p\right) \delta _{d}^{2}\left( p\right) f_{P}\left(
p\right) \left( g_{d}^{\left( 1\right) }\left( p\right) \right) ^{2}%
\end{array}%
\right) \right) ,
\end{equation*}%
and thus from (\ref{asymplinear}) that%
\begin{equation*}
J_{n}\left[ \left( 
\begin{array}{c}
\hat{g}_{d}\left( p\right) \\ 
\hat{g}_{d}^{\left( 1\right) }\left( p\right)%
\end{array}%
\right) -\left( 
\begin{array}{c}
g_{d}\left( p\right) \\ 
g_{d}^{\left( 1\right) }\left( p\right)%
\end{array}%
\right) -\left( 
\begin{array}{c}
\frac{\kappa _{2}}{2}g_{d}^{\left( 2\right) }\left( p\right) h_{3}^{2} \\ 
0%
\end{array}%
\right) \right] \rightarrow N\left( 0,\Sigma _{d}\left( p\right) \right) ,
\end{equation*}%
where%
\begin{equation*}
\Sigma _{d}\left( p\right) =\left( 
\begin{array}{cc}
\frac{\kappa \sigma _{V_{d}}^{2}\left( p\right) +\kappa _{22}^{\left(
1\right) }p\left( 1-p\right) \delta _{d}\left( p\right) \left( g_{d}^{\left(
1\right) }\left( p\right) \right) ^{2}}{\delta _{d}\left( p\right)
f_{P}\left( p\right) } & 0 \\ 
0 & \frac{\kappa _{22}\sigma _{V_{d}}^{2}\left( p\right) +\left( \kappa
_{24}^{\left( 1\right) }-2\kappa _{22}\right) p\left( 1-p\right) \delta
_{d}\left( p\right) \left( g_{d}^{\left( 1\right) }\left( p\right) \right)
^{2}}{\kappa _{2}^{2}\delta _{d}\left( p\right) f_{P}\left( p\right) }%
\end{array}%
\right) .
\end{equation*}

Moreover, we can deduce the joint asymptotic normality of $\left( \hat{g}%
_{1}\left( p\right) ,\hat{g}_{1}^{\left( 1\right) }\left( p\right) \right) $
and $\left( \hat{g}_{0}\left( p\right) ,\hat{g}_{0}^{\left( 1\right) }\left(
p\right) \right) $ from the asymptotically linear representation (\ref%
{asymplinear}). Due to%
\begin{eqnarray*}
E\left[ A_{1n\left( 1\right) }\left( p\right) A_{0n\left( 1\right) }\left(
p\right) \right] &=&h_{3}E\left[ \Lambda _{1i}\left( p\right) \Lambda
_{0i}\left( p\right) \left( D_{i}-P_{i}\right) ^{2}\right] =\kappa
_{22}^{\left( 1\right) }g_{1}^{\left( 1\right) }\left( p\right)
g_{0}^{\left( 1\right) }\left( p\right) p^{2}\left( 1-p\right)
^{2}f_{P}\left( p\right) +o\left( 1\right) , \\
E\left[ A_{1n\left( 1\right) }\left( p\right) A_{0n\left( 2\right) }\left(
p\right) \right] &=&h_{3}^{2}E\left[ \frac{P_{i}-p}{h_{3}^{2}}\Lambda
_{1i}\left( p\right) \left( \Lambda _{0i}\left( p\right) +\Upsilon
_{0i}\left( p\right) \right) \left( D_{i}-P_{i}\right) ^{2}\right] =O\left(
h_{3}\right) =o\left( 1\right) , \\
E\left[ A_{1n\left( 2\right) }\left( p\right) A_{0n\left( 1\right) }\left(
p\right) \right] &=&h_{3}^{2}E\left[ \frac{P_{i}-p}{h_{3}^{2}}\left( \Lambda
_{1i}\left( p\right) +\Upsilon _{1i}\left( p\right) \right) \Lambda
_{0i}\left( p\right) \left( D_{i}-P_{i}\right) ^{2}\right] =O\left(
h_{3}\right) =o\left( 1\right) , \\
E\left[ A_{1n\left( 2\right) }\left( p\right) A_{0n\left( 2\right) }\left(
p\right) \right] &=&h_{3}^{3}E\left[ \frac{\left( P_{i}-p\right) ^{2}}{%
h_{3}^{4}}\left( \Lambda _{1i}\left( p\right) +\Upsilon _{1i}\left( p\right)
\right) \left( \Lambda _{0i}\left( p\right) +\Upsilon _{0i}\left( p\right)
\right) \left( D_{i}-P_{i}\right) ^{2}\right] \\
&=&\left( \kappa _{24}^{\left( 1\right) }-2\kappa _{22}\right) g_{1}^{\left(
1\right) }\left( p\right) g_{0}^{\left( 1\right) }\left( p\right)
p^{2}\left( 1-p\right) ^{2}f_{P}\left( p\right) +o\left( 1\right) ,
\end{eqnarray*}%
it follows from Liapunov's central limit theorem that%
\begin{equation}
\left( 
\begin{array}{cc}
J_{n} & 0 \\ 
0 & J_{n}%
\end{array}%
\right) \left[ \left( 
\begin{array}{c}
\hat{g}_{1}\left( p\right) \\ 
\hat{g}_{1}^{\left( 1\right) }\left( p\right) \\ 
\hat{g}_{0}\left( p\right) \\ 
\hat{g}_{0}^{\left( 1\right) }\left( p\right)%
\end{array}%
\right) -\left( 
\begin{array}{c}
g_{1}\left( p\right) \\ 
g_{1}^{\left( 1\right) }\left( p\right) \\ 
g_{0}\left( p\right) \\ 
g_{0}^{\left( 1\right) }\left( p\right)%
\end{array}%
\right) -\left( 
\begin{array}{c}
\frac{\kappa _{2}}{2}g_{1}^{\left( 2\right) }\left( p\right) h_{3}^{2} \\ 
0 \\ 
\frac{\kappa _{2}}{2}g_{0}^{\left( 2\right) }\left( p\right) h_{3}^{2} \\ 
0%
\end{array}%
\right) \right] \rightarrow N\left( 0,\Sigma \left( p\right) \right) ,
\label{jointAN}
\end{equation}%
where%
\begin{equation*}
\Sigma \left( p\right) =\left( 
\begin{array}{cc}
\Sigma _{1}\left( p\right) & \Sigma _{10}\left( p\right) \\ 
\Sigma _{10}\left( p\right) & \Sigma _{0}\left( p\right)%
\end{array}%
\right) ,
\end{equation*}%
and%
\begin{equation*}
\Sigma _{10}\left( p\right) =\left( 
\begin{array}{cc}
\displaystyle\frac{\kappa _{22}^{\left( 1\right) }p\left( 1-p\right)
g_{1}^{\left( 1\right) }\left( p\right) g_{0}^{\left( 1\right) }\left(
p\right) }{f_{P}\left( p\right) } & 0 \\ 
0 & \displaystyle\frac{\left( \kappa _{24}^{\left( 1\right) }-2\kappa
_{22}\right) p\left( 1-p\right) g_{1}^{\left( 1\right) }\left( p\right)
g_{0}^{\left( 1\right) }\left( p\right) }{\kappa _{2}^{2}f_{P}\left(
p\right) }%
\end{array}%
\right) .
\end{equation*}

Now the asymptotic normality of the estimators for MTE and relevant causal
parameters can be readily deduced from (\ref{jointAN}). For the MTE
estimator, we have%
\begin{eqnarray*}
&&\sqrt{nh_{3}^{3}}\left[ \hat{\Delta}^{\text{MTE}}\left( x,v\right) -\Delta
^{\text{MTE}}\left( x,v\right) \right] \\
&=&x^{\prime }\left[ \sqrt{nh_{3}^{3}}\left( \hat{\beta}_{1}-\beta
_{1}\right) -\sqrt{nh_{3}^{3}}\left( \hat{\beta}_{0}-\beta _{0}\right) %
\right] \\
&&+\left[ \sqrt{nh_{3}^{3}}\left( \hat{g}_{1}\left( v\right) -g_{1}\left(
v\right) \right) -\sqrt{nh_{3}^{3}}\left( \hat{g}_{0}\left( v\right)
-g_{0}\left( v\right) \right) \right] \\
&&+v\sqrt{nh_{3}^{3}}\left( \hat{g}_{1}^{\left( 1\right) }\left( v\right)
-g_{1}^{\left( 1\right) }\left( v\right) \right) +\left( 1-v\right) \sqrt{%
nh_{3}^{3}}\left( \hat{g}_{0}^{\left( 1\right) }\left( v\right)
-g_{0}^{\left( 1\right) }\left( v\right) \right) \\
&=&O_{p}\left( \sqrt{h_{3}^{3}}\right) +O\left( \sqrt{nh_{3}^{7}}\right)
+O_{p}\left( h_{3}\right) \\
&&+v\sqrt{nh_{3}^{3}}\left( \hat{g}_{1}^{\left( 1\right) }\left( v\right)
-g_{1}^{\left( 1\right) }\left( v\right) \right) +\left( 1-v\right) \sqrt{%
nh_{3}^{3}}\left( \hat{g}_{0}^{\left( 1\right) }\left( v\right)
-g_{0}^{\left( 1\right) }\left( v\right) \right) \\
&\rightarrow &N\left( 0,\sigma _{\text{MTE}}^{2}\left( v\right) \right) ,
\end{eqnarray*}%
where%
\begin{equation}
\sigma _{\text{MTE}}^{2}\left( v\right) =\frac{\kappa _{22}\left[ v\sigma
_{V_{1}}^{2}\left( v\right) +\left( 1-v\right) \sigma _{V_{0}}^{2}\left(
v\right) \right] +\left( \kappa _{24}^{\left( 1\right) }-2\kappa
_{22}\right) v\left( 1-v\right) \left[ vg_{1}^{\left( 1\right) }\left(
v\right) +\left( 1-v\right) g_{0}^{\left( 1\right) }\left( v\right) \right]
^{2}}{\kappa _{2}^{2}f_{P}\left( v\right) }.  \label{sigmaMTE}
\end{equation}%
For the relevant causal parameter estimators, we have%
\begin{eqnarray*}
&&\sqrt{nh_{3}}\left[ \hat{\Delta}^{\text{ATE}}\left( x\right) -\Delta ^{%
\text{ATE}}\left( x\right) -\frac{\kappa _{2}}{2}\left( g_{1}^{\left(
2\right) }\left( 1\right) -g_{0}^{\left( 2\right) }\left( 0\right) \right)
h_{3}^{2}\right] \\
&=&x^{\prime }\left[ \sqrt{nh_{3}}\left( \hat{\beta}_{1}-\beta _{1}\right) -%
\sqrt{nh_{3}}\left( \hat{\beta}_{0}-\beta _{0}\right) \right] \\
&&+\sqrt{nh_{3}}\left[ \hat{g}_{1}\left( 1\right) -g_{1}\left( 1\right) -%
\frac{\kappa _{2}}{2}g_{1}^{\left( 2\right) }\left( 1\right) h_{3}^{2}\right]
-\sqrt{nh_{3}}\left[ \hat{g}_{0}\left( 0\right) -g_{0}\left( 0\right) -\frac{%
\kappa _{2}}{2}g_{0}^{\left( 2\right) }\left( 0\right) h_{3}^{2}\right] \\
&=&O_{p}\left( \sqrt{h_{3}}\right) +\sqrt{nh_{3}}\left[ \hat{g}_{1}\left(
1\right) -g_{1}\left( 1\right) -\frac{\kappa _{2}}{2}g_{1}^{\left( 2\right)
}\left( 1\right) h_{3}^{2}\right] -\sqrt{nh_{3}}\left[ \hat{g}_{0}\left(
0\right) -g_{0}\left( 0\right) -\frac{\kappa _{2}}{2}g_{0}^{\left( 2\right)
}\left( 0\right) h_{3}^{2}\right] \\
&\rightarrow &N\left( 0,\sigma _{\text{ATE}}^{2}\right) ,
\end{eqnarray*}%
\begin{eqnarray*}
&&\sqrt{nh_{3}}\left[ \hat{\Delta}^{\text{TT}}\left( x\right) -\Delta ^{%
\text{TT}}\left( x\right) -\frac{\kappa _{2}}{2}\left( g_{1}^{\left(
2\right) }\left( \pi \left( x\right) \right) +\frac{\left( 1-\pi \left(
x\right) \right) g_{0}^{\left( 2\right) }\left( \pi \left( x\right) \right)
-g_{0}^{\left( 2\right) }\left( 0\right) }{\pi \left( x\right) }\right)
h_{3}^{2}\right] \\
&=&x^{\prime }\left[ \sqrt{nh_{3}}\left( \hat{\beta}_{1}-\beta _{1}\right) -%
\sqrt{nh_{3}}\left( \hat{\beta}_{0}-\beta _{0}\right) \right] +\sqrt{nh_{3}}%
\left[ \hat{g}_{1}\left( \hat{\pi}\left( x\right) \right) -g_{1}\left( \pi
\left( x\right) \right) \right] \\
&&+\sqrt{nh_{3}}\left[ \frac{\left( 1-\hat{\pi}\left( x\right) \right) \hat{g%
}_{0}\left( \hat{\pi}\left( x\right) \right) -\hat{g}_{0}\left( 0\right) }{%
\hat{\pi}\left( x\right) }-\frac{\left( 1-\pi \left( x\right) \right)
g_{0}\left( \pi \left( x\right) \right) -g_{0}\left( 0\right) }{\pi \left(
x\right) }\right] \\
&&-\sqrt{nh_{3}}\frac{\kappa _{2}}{2}\left[ g_{1}^{\left( 2\right) }\left(
\pi \left( x\right) \right) +\frac{1-\pi \left( x\right) }{\pi \left(
x\right) }g_{0}^{\left( 2\right) }\left( \pi \left( x\right) \right) -\frac{1%
}{\pi \left( x\right) }g_{0}^{\left( 2\right) }\left( 0\right) \right]
h_{3}^{2} \\
&=&O_{p}\left( \sqrt{h_{3}}\right) +O_{p}\left( \sqrt{nh_{3}}\eta
_{1}\right) +\sqrt{nh_{3}}\left[ \hat{g}_{1}\left( \pi \left( x\right)
\right) -g_{1}\left( \pi \left( x\right) \right) -\frac{\kappa _{2}}{2}%
g_{1}^{\left( 2\right) }\left( \pi \left( x\right) \right) h_{3}^{2}\right]
\\
&&+O_{p}\left( \sqrt{nh_{3}}\eta _{1}\right) +\frac{1-\pi \left( x\right) }{%
\pi \left( x\right) }\sqrt{nh_{3}}\left[ \hat{g}_{0}\left( \pi \left(
x\right) \right) -g_{0}\left( \pi \left( x\right) \right) -\frac{\kappa _{2}%
}{2}g_{0}^{\left( 2\right) }\left( \pi \left( x\right) \right) h_{3}^{2}%
\right] \\
&&+O_{p}\left( \sqrt{nh_{3}}\eta _{1}\right) -\frac{1}{\pi \left( x\right) }%
\sqrt{nh_{3}}\left[ \hat{g}_{0}\left( 0\right) -g_{0}\left( 0\right) -\frac{%
\kappa _{2}}{2}g_{0}^{\left( 2\right) }\left( 0\right) h_{3}^{2}\right] \\
&\rightarrow &N\left( 0,\sigma _{\text{TT}}^{2}\left( \pi \left( x\right)
\right) \right) ,
\end{eqnarray*}%
\begin{eqnarray*}
&&\sqrt{nh_{3}}\left[ \hat{\Delta}^{\text{TUT}}\left( x\right) -\Delta ^{%
\text{TUT}}\left( x\right) -\frac{\kappa _{2}}{2}\left( \frac{g_{1}^{\left(
2\right) }\left( 1\right) -\pi \left( x\right) g_{1}^{\left( 2\right)
}\left( \pi \left( x\right) \right) }{1-\pi \left( x\right) }-g_{0}^{\left(
2\right) }\left( \pi \left( x\right) \right) \right) h_{3}^{2}\right] \\
&=&-\sqrt{nh_{3}}\left[ \hat{g}_{0}\left( \pi \left( x\right) \right)
-g_{0}\left( \pi \left( x\right) \right) -\frac{\kappa _{2}}{2}g_{0}^{\left(
2\right) }\left( \pi \left( x\right) \right) h_{3}^{2}\right] \\
&&-\frac{\pi \left( x\right) }{1-\pi \left( x\right) }\sqrt{nh_{3}}\left[ 
\hat{g}_{1}\left( \pi \left( x\right) \right) -g_{1}\left( \pi \left(
x\right) \right) -\frac{\kappa _{2}}{2}g_{1}^{\left( 2\right) }\left( \pi
\left( x\right) \right) h_{3}^{2}\right] \\
&&+\frac{1}{1-\pi \left( x\right) }\sqrt{nh_{3}}\left[ \hat{g}_{1}\left(
1\right) -g_{1}\left( 1\right) -\frac{\kappa _{2}}{2}g_{1}^{\left( 2\right)
}\left( 1\right) h_{3}^{2}\right] +o_{p}\left( 1\right) \\
&\rightarrow &N\left( 0,\sigma _{\text{TUT}}^{2}\left( \pi \left( x\right)
\right) \right) ,
\end{eqnarray*}%
\begin{equation*}
\sqrt{nh_{3}}\left[ 
\begin{array}{c}
\hat{\Delta}^{\text{LATE}}\left( x,v_{1},v_{2}\right) -\Delta ^{\text{LATE}%
}\left( x,v_{1},v_{2}\right) \\ 
-\frac{\kappa _{2}}{2}\left( \frac{v_{2}g_{1}^{\left( 2\right) }\left(
v_{2}\right) -v_{1}g_{1}^{\left( 2\right) }\left( v_{1}\right) +\left(
1-v_{2}\right) g_{0}^{\left( 2\right) }\left( v_{2}\right) -\left(
1-v_{1}\right) g_{0}^{\left( 2\right) }\left( v_{1}\right) }{v_{2}-v_{1}}%
\right) h_{3}^{2}%
\end{array}%
\right] \rightarrow N\left( 0,\sigma _{\text{LATE}}^{2}\left(
v_{1},v_{2}\right) \right) ,
\end{equation*}%
where%
\begin{equation}
\sigma _{\text{ATE}}^{2}=\frac{\kappa \sigma _{V_{1}}^{2}\left( 1\right) }{%
f_{P}\left( 1\right) }+\frac{\kappa \sigma _{V_{0}}^{2}\left( 0\right) }{%
f_{P}\left( 0\right) },  \label{sigmaATE}
\end{equation}%
\begin{eqnarray}
\sigma _{\text{TT}}^{2}\left( p\right) &=&\frac{\kappa \sigma
_{V_{1}}^{2}\left( p\right) +\kappa _{22}^{\left( 1\right) }p^{2}\left(
1-p\right) \left( g_{1}^{\left( 1\right) }\left( p\right) \right) ^{2}}{%
pf_{P}\left( p\right) }  \notag \\
&&+\left( \frac{1-p}{p}\right) ^{2}\frac{\kappa \sigma _{V_{0}}^{2}\left(
p\right) +\kappa _{22}^{\left( 1\right) }p\left( 1-p\right) ^{2}\left(
g_{0}^{\left( 1\right) }\left( p\right) \right) ^{2}}{\left( 1-p\right)
f_{P}\left( p\right) }  \notag \\
&&+\frac{2\kappa _{22}^{\left( 1\right) }\left( 1-p\right) ^{2}g_{1}^{\left(
1\right) }\left( p\right) g_{0}^{\left( 1\right) }\left( p\right) }{%
f_{P}\left( p\right) }+\frac{\kappa \sigma _{V_{0}}^{2}\left( 0\right) }{%
p^{2}f_{P}\left( 0\right) },  \label{sigmaTT}
\end{eqnarray}%
\begin{eqnarray}
\sigma _{\text{TUT}}^{2}\left( p\right) &=&\frac{\kappa \sigma
_{V_{0}}^{2}\left( p\right) +\kappa _{22}^{\left( 1\right) }p\left(
1-p\right) ^{2}\left( g_{0}^{\left( 1\right) }\left( p\right) \right) ^{2}}{%
\left( 1-p\right) f_{P}\left( p\right) }  \notag \\
&&+\left( \frac{p}{1-p}\right) ^{2}\frac{\kappa \sigma _{V_{1}}^{2}\left(
p\right) +\kappa _{22}^{\left( 1\right) }p^{2}\left( 1-p\right) \left(
g_{1}^{\left( 1\right) }\left( p\right) \right) ^{2}}{pf_{P}\left( p\right) }
\notag \\
&&+\frac{2\kappa _{22}^{\left( 1\right) }p^{2}g_{1}^{\left( 1\right) }\left(
p\right) g_{0}^{\left( 1\right) }\left( p\right) }{f_{P}\left( p\right) }+%
\frac{\kappa \sigma _{V_{1}}^{2}\left( 1\right) }{\left( 1-p\right)
^{2}f_{P}\left( 1\right) },  \label{sigmaTUT}
\end{eqnarray}%
\begin{eqnarray}
\sigma _{\text{LATE}}^{2}\left( v_{1},v_{2}\right) &=&\frac{\kappa }{\left(
v_{2}-v_{1}\right) ^{2}f_{P}\left( v_{1}\right) }\left[ v_{1}\sigma
_{V_{1}}^{2}\left( v_{1}\right) +\left( 1-v_{1}\right) \sigma
_{V_{0}}^{2}\left( v_{1}\right) \right]  \notag \\
&&+\frac{\kappa }{\left( v_{2}-v_{1}\right) ^{2}f_{P}\left( v_{2}\right) }%
\left[ v_{2}\sigma _{V_{1}}^{2}\left( v_{2}\right) +\left( 1-v_{2}\right)
\sigma _{V_{0}}^{2}\left( v_{2}\right) \right]  \notag \\
&&+\frac{\kappa _{22}^{\left( 1\right) }v_{1}\left( 1-v_{1}\right) }{\left(
v_{2}-v_{1}\right) ^{2}f_{P}\left( v_{1}\right) }\left[ v_{1}g_{1}^{\left(
1\right) }\left( v_{1}\right) +\left( 1-v_{1}\right) g_{0}^{\left( 1\right)
}\left( v_{1}\right) \right] ^{2}  \notag \\
&&+\frac{\kappa _{22}^{\left( 1\right) }v_{2}\left( 1-v_{2}\right) }{\left(
v_{2}-v_{1}\right) ^{2}f_{P}\left( v_{2}\right) }\left[ v_{2}g_{1}^{\left(
1\right) }\left( v_{2}\right) +\left( 1-v_{2}\right) g_{0}^{\left( 1\right)
}\left( v_{2}\right) \right] ^{2},  \label{sigmaLATE}
\end{eqnarray}%
which completes the proof.

\subsection{Auxiliary lemmas}

\begin{lemma}
\label{lemma:firststep} Under Assumption E.1, we have%
\begin{eqnarray*}
\sup_{x}\left\vert \hat{\pi}\left( x\right) -\pi \left( x\right) \right\vert
&=&O\left( \sum_{l=1}^{\dim \left( X^{C}\right) }h_{1l}^{s}+\sqrt{\frac{\ln n%
}{nh_{11}h_{12}\cdots h_{1,\dim \left( X^{C}\right) }}}\right) \text{ almost
surely,} \\
\sup_{x}\left\vert \hat{f}_{X}\left( x\right) -f_{X}\left( x\right)
\right\vert &=&O\left( \sum_{l=1}^{\dim \left( X^{C}\right) }h_{1l}^{s}+%
\sqrt{\frac{\ln n}{nh_{11}h_{12}\cdots h_{1,\dim \left( X^{C}\right) }}}%
\right) \text{ almost surely,}
\end{eqnarray*}%
where $\hat{\pi}\left( x\right) $ is defined in (\ref{PSfunction}) and $\hat{%
f}_{X}\left( x\right) $ is the normalized denominator of $\hat{\pi}\left(
x\right) $, namely,%
\begin{equation}
\hat{f}_{X}\left( x\right) =\frac{1}{nh_{11}h_{12}\cdots h_{1,\dim \left(
X^{C}\right) }}\sum_{i=1}^{n}\left[ \prod_{l=1}^{\dim \left( X^{C}\right)
}k_{1}\left( \frac{X_{il}^{C}-x_{l}^{C}}{h_{1l}}\right) \right] 1\left\{
X_{i}^{D}=x^{D}\right\} .  \label{fXhat}
\end{equation}
\end{lemma}

\begin{proof}
This proof follows straightforwardly from combining Subsection
1.11 and Theorems 1.4 and 2.6 of \cite{li2007nonparametric}, thus is omitted here.
\end{proof}

\begin{lemma}
\label{lemma:secondstep} Under Assumption E.2.(ii)-(iv), the infeasible
local linear estimators $\left( \tilde{g}_{d}\left( p\right) ,\tilde{g}%
_{d}^{\left( 1\right) }\left( p\right) \right) $, $d=0,1$, are
asymptotically normally distributed for any interior point $p$ of the
support of $P$:%
\begin{equation*}
\left( 
\begin{array}{cc}
\sqrt{nh_{3}} & 0 \\ 
0 & \sqrt{nh_{3}^{3}}%
\end{array}%
\right) \left[ \left( 
\begin{array}{c}
\tilde{g}_{d}\left( p\right) \\ 
\tilde{g}_{d}^{\left( 1\right) }\left( p\right)%
\end{array}
\right) -\left( 
\begin{array}{c}
g_{d}\left( p\right) \\ 
g_{d}^{\left( 1\right) }\left( p\right)%
\end{array}
\right)-\left( 
\begin{array}{c}
\frac{\kappa _{2}}{2}g_{d}^{\left( 2\right) }\left( p\right) h_{3}^{2} \\ 
0%
\end{array}
\right) \right] \rightarrow N\left( 0,\Xi _{d}\left( p\right) \right) ,
\end{equation*}%
where%
\begin{equation*}
\Xi _{d}\left( p\right) =\left( 
\begin{array}{cc}
\left. \kappa \sigma _{V_{d}}^{2}\left( p\right) \right/ \left[ \delta
_{d}\left( p\right) f_{P}\left( p\right) \right] & 0 \\ 
0 & \left. \kappa _{22}\sigma _{V_{d}}^{2}\left( p\right) \right/ \left[
\kappa _{2}^{2}\delta _{d}\left( p\right) f_{P}\left( p\right) \right]%
\end{array}%
\right) ,
\end{equation*}%
and%
\begin{eqnarray*}
\kappa &=&\int k_{3}^{2}\left( u\right) du, \\
\kappa _{2} &=&\int k_{3}\left( u\right) u^{2}du, \\
\kappa _{22} &=&\int k_{3}^{2}\left( u\right) u^{2}du.
\end{eqnarray*}
\end{lemma}

\begin{proof}
Denote%
\begin{equation*}
Q_{d}=\frac{1}{n}\sum_{i=1}^{n}w_{di}\left( p\right) \left( 
\begin{array}{c}
1 \\ 
P_{i}-p%
\end{array}%
\right) \left( 
\begin{array}{c}
1 \\ 
P_{i}-p%
\end{array}%
\right) ^{\prime }
\end{equation*}%
as the normalized denominator of $\left( \tilde{g}_{d}\left( p\right) ,%
\tilde{g}_{d}^{\left( 1\right) }\left( p\right) \right) $. By $Y-X^{\prime
}\beta _{d}=g_{d}\left( P\right) +V_{d}$, we have%
\begin{eqnarray*}
\left( 
\begin{array}{c}
\tilde{g}_{d}\left( p\right) \\ 
\tilde{g}_{d}^{\left( 1\right) }\left( p\right)%
\end{array}%
\right) &=&Q_{d}^{-1}\frac{1}{n}\sum_{i=1}^{n}w_{di}\left( p\right) \left( 
\begin{array}{c}
1 \\ 
P_{i}-p%
\end{array}%
\right) \left[ g_{d}\left( P_{i}\right) +V_{di}\right] \\
&=&\left( 
\begin{array}{c}
g_{d}\left( p\right) \\ 
g_{d}^{\left( 1\right) }\left( p\right)%
\end{array}%
\right) +Q_{d}^{-1}\frac{1}{n}\sum_{i=1}^{n}w_{di}\left( p\right) \left( 
\begin{array}{c}
1 \\ 
P_{i}-p%
\end{array}%
\right) \left[ \frac{1}{2}g_{d}^{\left( 2\right) }\left( p\right) \left(
P_{i}-p\right) ^{2}+R_{di}+V_{di}\right] ,
\end{eqnarray*}%
where the second equality follows from a Taylor expansion of $g_{d}\left(
P_{i}\right) $ around $P_{i}=p$, and $R_{di}=\frac{1}{6}g_{d}^{\left(
3\right) }\left( P_{i}^{\ast }\right) \left( P_{i}-p\right) ^{3}$ with $%
P_{i}^{\ast }$ an intermediate value. Note that the three times
differentiability of $g_{d}\left( \cdot \right) $ required in the expansion
is stronger than Assumption E.2.(ii) and is imposed only for notational
simplicity. Actually, the twice continuous differentiability of $g_{d}\left(
\cdot \right) $ imposed by Assumption E.2.(ii) suffices for the same
conclusion to hold.

By Lemma \ref{lemma:kernel},%
\begin{equation*}
Q_{d}=\left( 
\begin{array}{cc}
\frac{1}{n}\sum_{i=1}^{n}w_{di}\left( p\right) , & \frac{1}{n}%
\sum_{i=1}^{n}w_{di}\left( p\right) \left( P_{i}-p\right) \\ 
\frac{1}{n}\sum_{i=1}^{n}w_{di}\left( p\right) \left( P_{i}-p\right) , & 
\frac{1}{n}\sum_{i=1}^{n}w_{di}\left( p\right) \left( P_{i}-p\right) ^{2}%
\end{array}%
\right) =\left( 
\begin{array}{cc}
O_{p}\left( 1\right) & O_{p}\left( h_{3}^{2}\right) \\ 
O_{p}\left( h_{3}^{2}\right) & O_{p}\left( h_{3}^{2}\right)%
\end{array}%
\right) ,
\end{equation*}%
which is an asymptotically singular matrix. Hence, we further normalize $%
Q_{d}$ by a matrix,%
\begin{equation*}
G=\left( 
\begin{array}{cc}
1, & 0 \\ 
0, & h_{3}^{-2}%
\end{array}%
\right) ,
\end{equation*}%
such that%
\begin{eqnarray*}
GQ_{d} &=&\left( 
\begin{array}{cc}
\frac{1}{n}\sum_{i=1}^{n}w_{di}\left( p\right) , & \frac{1}{n}%
\sum_{i=1}^{n}w_{di}\left( p\right) \left( P_{i}-p\right) \\ 
\frac{1}{nh_{3}^{2}}\sum_{i=1}^{n}w_{di}\left( p\right) \left(
P_{i}-p\right) , & \frac{1}{nh_{3}^{2}}\sum_{i=1}^{n}w_{di}\left( p\right)
\left( P_{i}-p\right) ^{2}%
\end{array}%
\right) \\
&\overset{p}{\longrightarrow }&T_{d}\equiv \left( 
\begin{array}{cc}
\delta _{d}\left( p\right) f_{P}\left( p\right) , & 0 \\ 
\kappa _{2}\left[ \delta _{d}^{\left( 1\right) }\left( p\right) f_{P}\left(
p\right) +\delta _{d}\left( p\right) f_{P}^{\left( 1\right) }\left( p\right) %
\right] , & \kappa _{2}\delta _{d}\left( p\right) f_{P}\left( p\right)%
\end{array}%
\right) ,
\end{eqnarray*}%
Since $T_{d}$ is invertible, we have%
\begin{equation*}
\left( GQ_{d}\right) ^{-1}\overset{p}{\longrightarrow }T_{d}^{-1}=\left( 
\begin{array}{cc}
1, & 0 \\ 
-\left[ \left. \delta _{d}^{\left( 1\right) }\left( p\right) \right/ \delta
_{d}\left( p\right) +\left. f_{P}^{\left( 1\right) }\left( p\right) \right/
f_{P}\left( p\right) \right] , & 1\left/ \kappa _{2}\right.%
\end{array}%
\right) \frac{1}{\delta _{d}\left( p\right) f_{P}\left( p\right) }.
\end{equation*}%
Therefore,%
\begin{eqnarray*}
&&J_{n}\left[ \left(
\begin{array}{c}
\tilde{g}_{d}\left( p\right)  \\ 
\tilde{g}_{d}^{\left( 1\right) }\left( p\right)
\end{array}
\right) -\left(
\begin{array}{c}
g_{d}\left( p\right)  \\ 
g_{d}^{\left( 1\right) }\left( p\right)
\end{array}
\right) \right] \\
&=&J_{n}\left( GQ_{d}\right) ^{-1}G\frac{1}{n}\sum_{i=1}^{n}w_{di}\left(
p\right) \left( 
\begin{array}{c}
1 \\ 
P_{i}-p%
\end{array}%
\right) \left[ \frac{1}{2}g_{d}^{\left( 2\right) }\left( p\right) \left(
P_{i}-p\right) ^{2}+R_{di}+V_{di}\right] \\
&=&J_{n}T_{d}^{-1}\frac{1}{n}\sum_{i=1}^{n}w_{di}\left( p\right) \left( 
\begin{array}{c}
1 \\ 
\frac{P_{i}-p}{h_{3}^{2}}%
\end{array}%
\right) \left[ \frac{1}{2}g_{d}^{\left( 2\right) }\left( p\right) \left(
P_{i}-p\right) ^{2}+R_{di}+V_{di}\right] +o_{p}\left( 1\right) ,
\end{eqnarray*}%
where%
\begin{equation*}
J_{n}=\left( 
\begin{array}{cc}
\sqrt{nh_{3}} & 0 \\ 
0 & \sqrt{nh_{3}^{3}}%
\end{array}%
\right) .
\end{equation*}

Note that for%
\begin{equation*}
\Gamma _{d}=diag\left( T_{d}^{-1}\right) =\left( 
\begin{array}{cc}
1\left/ \left[ \delta _{d}\left( p\right) f_{P}\left( p\right) \right]
\right. , & 0 \\ 
0, & 1\left/ \left[ \kappa _{2}\delta _{d}\left( p\right) f_{P}\left(
p\right) \right] \right.%
\end{array}%
\right) ,
\end{equation*}%
we have%
\begin{eqnarray*}
&&J_{n}T_{d}^{-1}\frac{1}{n}\sum_{i=1}^{n}w_{di}\left( p\right) \left( 
\begin{array}{c}
1 \\ 
\frac{P_{i}-p}{h_{3}^{2}}%
\end{array}%
\right) \left[ \frac{1}{2}g_{d}^{\left( 2\right) }\left( p\right) \left(
P_{i}-p\right) ^{2}+R_{di}+V_{di}\right] \\
&=&J_{n}\Gamma _{d}\frac{1}{n}\sum_{i=1}^{n}w_{di}\left( p\right) \left( 
\begin{array}{c}
1 \\ 
\frac{P_{i}-p}{h_{3}^{2}}%
\end{array}%
\right) \left[ \frac{1}{2}g_{d}^{\left( 2\right) }\left( p\right) \left(
P_{i}-p\right) ^{2}+R_{di}+V_{di}\right] +o_{p}\left( 1\right) ,
\end{eqnarray*}%
because%
\begin{eqnarray*}
\sqrt{nh_{3}^{3}}T_{d\left( 2,1\right) }^{-1}\frac{1}{n}\sum_{i=1}^{n}w_{di}%
\left( p\right) \frac{1}{2}g_{d}^{\left( 2\right) }\left( p\right) \left(
P_{i}-p\right) ^{2} &=&O\left( \sqrt{nh_{3}^{7}}+h_{3}^{3}\right) =o\left(
1\right) , \\
\sqrt{nh_{3}^{3}}T_{d\left( 2,1\right) }^{-1}\frac{1}{n}\sum_{i=1}^{n}w_{di}%
\left( p\right) R_{di} &=&O\left( \sqrt{nh_{3}^{9}}+h_{3}^{4}\right)
=o\left( 1\right) , \\
\sqrt{nh_{3}^{3}}T_{d\left( 2,1\right) }^{-1}\frac{1}{n}\sum_{i=1}^{n}w_{di}%
\left( p\right) V_{di} &=&O\left( h_{3}\right) =o\left( 1\right) ,
\end{eqnarray*}%
where%
\begin{equation*}
T_{d\left( 2,1\right) }^{-1}=-\frac{\left. \delta _{d}^{\left( 1\right)
}\left( p\right) \right/ \delta _{d}\left( p\right) +\left. f_{P}^{\left(
1\right) }\left( p\right) \right/ f_{P}\left( p\right) }{\delta _{d}\left(
p\right) f_{P}\left( p\right) }.
\end{equation*}%
Therefore,%
\begin{eqnarray}
&&J_{n}\left[ \left(
\begin{array}{c}
\tilde{g}_{d}\left( p\right)  \\ 
\tilde{g}_{d}^{\left( 1\right) }\left( p\right)
\end{array}
\right) -\left(
\begin{array}{c}
g_{d}\left( p\right)  \\ 
g_{d}^{\left( 1\right) }\left( p\right)
\end{array}
\right) \right]   \notag \\
&=&J_{n}\Gamma _{d}\frac{1}{n}\sum_{i=1}^{n}w_{di}\left( p\right) \left( 
\begin{array}{c}
1 \\ 
\frac{P_{i}-p}{h_{3}^{2}}%
\end{array}%
\right) \left[ \frac{1}{2}g_{d}^{\left( 2\right) }\left( p\right) \left(
P_{i}-p\right) ^{2}+R_{di}+V_{di}\right] +o_{p}\left( 1\right)  \notag \\
&=&\Gamma _{d}\left[ A_{1d}\left( p\right) +A_{2d}\left( p\right)
+A_{3d}\left( p\right) \right] +o_{p}\left( 1\right) ,  \label{gdtilde}
\end{eqnarray}%
where%
\begin{eqnarray*}
A_{1d}\left( p\right) &=&\left(
\begin{array}{c}
A_{1d\left( 1\right) }\left(p\right)  \\ 
A_{1d\left( 2\right) }\left( p\right)
\end{array}
\right) \equiv J_{n}\frac{1%
}{n}\sum_{i=1}^{n}w_{di}\left( p\right) \left( 
\begin{array}{c}
1 \\ 
\frac{P_{i}-p}{h_{3}^{2}}%
\end{array}%
\right) \frac{1}{2}g_{d}^{\left( 2\right) }\left( p\right) \left(
P_{i}-p\right) ^{2}, \\
A_{2d}\left( p\right) &=&\left(
\begin{array}{c}
A_{2d\left( 1\right) }\left(p\right)  \\ 
A_{2d\left( 2\right) }\left( p\right)
\end{array}
\right) \equiv J_{n}\frac{1%
}{n}\sum_{i=1}^{n}w_{di}\left( p\right) \left( 
\begin{array}{c}
1 \\ 
\frac{P_{i}-p}{h_{3}^{2}}%
\end{array}%
\right) R_{di} \\
&=&J_{n}\frac{1}{6n}\sum_{i=1}^{n}w_{di}\left( p\right) \left( 
\begin{array}{c}
1 \\ 
\frac{P_{i}-p}{h_{3}^{2}}%
\end{array}%
\right) g_{d}^{\left( 3\right) }\left( P_{i}^{\ast }\right) \left(
P_{i}-p\right) ^{3}, \\
A_{3d}\left( p\right) &=&\left(
\begin{array}{c}
A_{3d\left( 1\right) }\left(p\right)  \\ 
A_{3d\left( 2\right) }\left( p\right)
\end{array}
\right) \equiv J_{n}\frac{1%
}{n}\sum_{i=1}^{n}w_{di}\left( p\right) \left( 
\begin{array}{c}
1 \\ 
\frac{P_{i}-p}{h_{3}^{2}}%
\end{array}%
\right) V_{di}.
\end{eqnarray*}

We first consider $A_{1d}\left( p\right) $. Since%
\begin{eqnarray*}
E\left[ A_{1d\left( 1\right) }\left( p\right) \right] &=&\frac{\sqrt{nh_{3}}%
}{2}g_{d}^{\left( 2\right) }\left( p\right) E\left[ w_{di}\left( p\right)
\left( P_{i}-p\right) ^{2}\right] =\frac{\kappa _{2}}{2}g_{d}^{\left(
2\right) }\left( p\right) \delta _{d}\left( p\right) f_{P}\left( p\right) 
\sqrt{nh_{3}^{5}}+o\left( 1\right) , \\
E\left[ A_{1d\left( 2\right) }\left( p\right) \right] &=&\frac{\sqrt{%
nh_{3}^{3}}}{2}g_{d}^{\left( 2\right) }\left( p\right) E\left[ w_{di}\left(
p\right) \frac{P_{i}-p}{h_{3}^{2}}\left( P_{i}-p\right) ^{2}\right] =O\left( 
\sqrt{nh_{3}^{7}}\right) =o\left( 1\right) ,
\end{eqnarray*}%
and%
\begin{eqnarray*}
Var\left( A_{1d\left( 1\right) }\left( p\right) \right) &\leq &\frac{h_{3}}{4%
}\left[ g_{d}^{\left( 2\right) }\left( p\right) \right] ^{2}E\left[
w_{di}^{2}\left( p\right) \left( P_{i}-p\right) ^{4}\right] =O\left(
h_{3}^{2}\right) =o\left( 1\right) , \\
Var\left( A_{1d\left( 2\right) }\left( p\right) \right) &\leq &\frac{%
h_{3}^{3}}{4}\left[ g_{d}^{\left( 2\right) }\left( p\right) \right] ^{2}E%
\left[ w_{di}^{2}\left( p\right) \frac{\left( P_{i}-p\right) ^{6}}{h_{3}^{4}}%
\right] =O\left( h_{3}^{4}\right) =o\left( 1\right) ,
\end{eqnarray*}%
it follows from Markov's inequality that%
\begin{eqnarray*}
A_{1d\left( 1\right) }\left( p\right) &=&\frac{\kappa _{2}}{2}g_{d}^{\left(
2\right) }\left( p\right) \delta _{d}\left( p\right) f_{P}\left( p\right) 
\sqrt{nh_{3}^{5}}+o_{p}\left( 1\right) , \\
A_{1d\left( 2\right) }\left( p\right) &=&o_{p}\left( 1\right) ,
\end{eqnarray*}%
namely,%
\begin{equation}
A_{1d}\left( p\right) =J_{n}\left(
\begin{array}{c}
\frac{\kappa _{2}}{2}g_{d}^{\left( 2\right) }\left( p\right) \delta _{d}\left( p\right)
f_{P}\left( p\right) h_{3}^{2}  \\ 
0
\end{array}
\right) +o_{p}\left( 1\right) .
\label{A1d}
\end{equation}%
Similarly, we can show that%
\begin{equation}
A_{2d}\left( p\right) =o_{p}\left( 1\right) .  \label{A2d}
\end{equation}

We next consider $A_{3d}\left( p\right) $. By the definition of $V_{di}$, we
have $E\left[ A_{3d}\left( p\right) \right] =0$. For the variance-covariance
term, we have%
\begin{gather*}
Var\left( A_{3d\left( 1\right) }\left( p\right) \right) =h_{3}E\left[
w_{di}^{2}\left( p\right) V_{di}^{2}\right] =\kappa \sigma
_{V_{d}}^{2}\left( p\right) \delta _{d}\left( p\right) f_{P}\left( p\right)
+O\left( h_{3}\right) , \\
Var\left( A_{3d\left( 2\right) }\left( p\right) \right) =h_{3}^{3}E\left[
w_{di}^{2}\left( p\right) \frac{\left( P_{i}-p\right) ^{2}}{h_{3}^{4}}%
V_{di}^{2}\right] =\kappa _{22}\sigma _{V_{d}}^{2}\left( p\right) \delta
_{d}\left( p\right) f_{P}\left( p\right) +O\left( h_{3}\right) , \\
Cov\left( A_{3d\left( 1\right) }\left( p\right) ,A_{3d\left( 2\right)
}\left( p\right) \right) =h_{3}^{2}E\left[ w_{di}^{2}\left( p\right) \frac{%
P_{i}-p}{h_{3}^{2}}V_{di}^{2}\right] =O\left( h_{3}\right) .
\end{gather*}%
It then follows from Liapunov's central limit theorem that%
\begin{equation}
A_{3d}\left( p\right) \rightarrow N\left( 0,\Omega _{d}\right) ,  \label{A3d}
\end{equation}%
where%
\begin{equation*}
\Omega _{d}=\left( 
\begin{array}{cc}
\kappa \sigma _{V_{d}}^{2}\left( p\right) \delta _{d}\left( p\right)
f_{P}\left( p\right) & 0 \\ 
0 & \kappa _{22}\sigma _{V_{d}}^{2}\left( p\right) \delta _{d}\left(
p\right) f_{P}\left( p\right)%
\end{array}%
\right) .
\end{equation*}

Substituting (\ref{A1d}) and (\ref{A2d}) into (\ref{gdtilde}) yields that%
\begin{equation}
J_{n}\left[ \left(
\begin{array}{c}
\tilde{g}_{d}\left( p\right)  \\ 
\tilde{g}_{d}^{\left( 1\right) }\left( p\right)
\end{array}
\right) -\left(
\begin{array}{c}
g_{d}\left( p\right)  \\ 
g_{d}^{\left( 1\right) }\left( p\right)
\end{array}
\right) \right] =J_{n}\left(
\begin{array}{c}
\frac{\kappa _{2}}{2}g_{d}^{\left( 2\right)
}\left( p\right) h_{3}^{2}  \\ 
0
\end{array}
\right) +\Gamma _{d}A_{3d}\left( p\right)
+o_{p}\left( 1\right) ,  \label{AsymptoticLinear}
\end{equation}%
and the conclusion follows from (\ref{A3d}).
\end{proof}

\begin{lemma}
\label{lemma:kernel} Under Assumption E.2, we have%
\begin{eqnarray*}
\frac{1}{n}\sum_{i=1}^{n}w_{di}\left( p\right) &=&\delta _{d}\left( p\right)
f_{P}\left( p\right) +o_{p}\left( h_{3}\right) , \\
\frac{1}{nh_{3}^{2}}\sum_{i=1}^{n}w_{di}\left( p\right) \left(
P_{i}-p\right) &=&\kappa _{2}\left[ \delta _{d}^{\left( 1\right) }\left(
p\right) f_{P}\left( p\right) +\delta _{d}\left( p\right) f_{P}^{\left(
1\right) }\left( p\right) \right] +O_{p}\left( \frac{1}{\sqrt{nh_{3}^{3}}}%
\right) , \\
\frac{1}{nh_{3}^{2}}\sum_{i=1}^{n}w_{di}\left( p\right) \left(
P_{i}-p\right) ^{2} &=&\kappa _{2}\delta _{d}\left( p\right) f_{P}\left(
p\right) +o_{p}\left( h_{3}\right) .
\end{eqnarray*}
\end{lemma}

\begin{proof}
Since the data are i.i.d., we have%
\begin{eqnarray*}
E\left[ \frac{1}{n}\sum_{i=1}^{n}w_{di}\left( p\right) \right] &=&E\left[
1\left\{ D_{i}=d\right\} \frac{1}{h_{3}}k_{3}\left( \frac{P_{i}-p}{h_{3}}%
\right) \right] \\
&=&E\left[ \delta _{d}\left( P_{i}\right) \frac{1}{h_{3}}k_{3}\left( \frac{%
P_{i}-p}{h_{3}}\right) \right] \\
&=&\int \delta _{d}\left( v\right) \frac{1}{h_{3}}k_{3}\left( \frac{v-p}{%
h_{3}}\right) f_{P}\left( v\right) dv \\
&=&\int \delta _{d}\left( p+uh_{3}\right) k_{3}\left( u\right) f_{P}\left(
p+uh_{3}\right) du \\
&=&\delta _{d}\left( p\right) f_{P}\left( p\right) +O\left( h_{3}^{2}\right)
,
\end{eqnarray*}%
where the second equality follows from the law of iterated expectation. On
the other hand,%
\begin{eqnarray*}
Var\left( \frac{1}{n}\sum_{i=1}^{n}w_{di}\left( p\right) \right) &=&\frac{1}{%
nh_{3}^{2}}Var\left( 1\left\{ D_{i}=d\right\} k_{3}\left( \frac{P_{i}-p}{%
h_{3}}\right) \right) \\
&\leq &\frac{1}{nh_{3}^{2}}E\left[ 1\left\{ D_{i}=d\right\} k_{3}^{2}\left( 
\frac{P_{i}-p}{h_{3}}\right) \right] \\
&=&O\left( \frac{1}{nh_{3}}\right) .
\end{eqnarray*}%
Therefore,%
\begin{equation*}
\left( E\left[ \left( \frac{1}{n}\sum_{i=1}^{n}w_{di}\left( p\right) -\delta
_{d}\left( p\right) f_{P}\left( p\right) \right) ^{2}\right] \right)
^{1/2}=O\left( h_{3}^{2}+\frac{1}{\sqrt{nh_{3}}}\right) =o\left(
h_{3}\right) ,
\end{equation*}%
where the last equality is by Assumption E.2.(iv). The first assertion of
the lemma then follows from Markov's inequality.

For the second assertion of the lemma, we have%
\begin{eqnarray*}
E\left[ \frac{1}{nh_{3}^{2}}\sum_{i=1}^{n}w_{di}\left( p\right) \left(
P_{i}-p\right) \right] &=&E\left[ \delta _{d}\left( P_{i}\right) \frac{1}{%
h_{3}^{3}}k_{3}\left( \frac{P_{i}-p}{h_{3}}\right) \left( P_{i}-p\right) %
\right] \\
&=&\int \delta _{d}\left( v\right) \frac{1}{h_{3}^{3}}k_{3}\left( \frac{v-p}{%
h_{3}}\right) \left( v-p\right) f_{P}\left( v\right) dv \\
&=&\frac{1}{h_{3}}\int \delta _{d}\left( p+uh_{3}\right) f_{P}\left(
p+uh_{3}\right) k_{3}\left( u\right) udu \\
&=&\kappa _{2}\left[ \delta _{d}^{\left( 1\right) }\left( p\right)
f_{P}\left( p\right) +\delta _{d}\left( p\right) f_{P}^{\left( 1\right)
}\left( p\right) \right] +O\left( h_{3}^{2}\right) ,
\end{eqnarray*}%
and%
\begin{eqnarray*}
Var\left( \frac{1}{nh_{3}^{2}}\sum_{i=1}^{n}w_{di}\left( p\right) \left(
P_{i}-p\right) \right) &\leq &\frac{1}{nh_{3}^{6}}E\left[ 1\left\{
D_{i}=d\right\} k_{3}^{2}\left( \frac{P_{i}-p}{h_{3}}\right) \left(
P_{i}-p\right) ^{2}\right] \\
&=&O\left( \frac{1}{nh_{3}^{3}}\right) .
\end{eqnarray*}%
Therefore,%
\begin{eqnarray*}
\left( E\left[ \left( \frac{1}{nh_{3}^{2}}\sum_{i=1}^{n}w_{di}\left(
p\right) \left( P_{i}-p\right) -\kappa _{2}\left[ \delta _{d}^{\left(
1\right) }\left( p\right) f_{P}\left( p\right) +\delta _{d}\left( p\right)
f_{P}^{\left( 1\right) }\left( p\right) \right] \right) ^{2}\right] \right)
^{1/2} &=&O\left( h_{3}^{2}+\frac{1}{\sqrt{nh_{3}^{3}}}\right) \\
&=&O\left( \frac{1}{\sqrt{nh_{3}^{3}}}\right) ,
\end{eqnarray*}%
and the second assertion of the lemma follows. Analogously, the third
assertion of the lemma can be readily shown.
\end{proof}

\begin{lemma}
\label{lemma:kernelV} Under Assumption E.2, we have%
\begin{eqnarray*}
\frac{1}{n}\sum_{i=1}^{n}w_{di}\left( p\right) V_{di}\left( p\right)
&=&o_{p}\left( h_{3}\right) , \\
\frac{1}{nh_{3}^{2}}\sum_{i=1}^{n}w_{di}\left( p\right) \left(
P_{i}-p\right) V_{di}\left( p\right) &=&\kappa _{2}\delta _{d}\left(
p\right) f_{P}\left( p\right) g_{d}^{\left( 1\right) }\left( p\right)
+O_{p}\left( \frac{1}{\sqrt{nh_{3}^{3}}}\right) ,
\end{eqnarray*}%
where $V_{di}\left( p\right) $ is defined in (\ref{Vdip}).
\end{lemma}

\begin{proof}
Note that $E\left[ \left. V_{di}\left( p\right) \right\vert
P_{i},D_{i}=d\right] =g_{d}\left( P_{i}\right) -g_{d}\left( p\right) $
because $E\left[ \left. V_{di}\right\vert P_{i},D_{i}=d\right] =0$ by
definition. Hence,%
\begin{eqnarray*}
E\left[ \frac{1}{n}\sum_{i=1}^{n}w_{di}\left( p\right) V_{di}\left( p\right) %
\right]  &=&E\left[ \delta _{d}\left( P_{i}\right) \left( g_{d}\left(
P_{i}\right) -g_{d}\left( p\right) \right) \frac{1}{h_{3}}k_{3}\left( \frac{%
P_{i}-p}{h_{3}}\right) \right]  \\
&=&O\left( h_{3}^{2}\right) , \\
E\left[ \frac{1}{nh_{3}^{2}}\sum_{i=1}^{n}w_{di}\left( p\right) \left(
P_{i}-p\right) V_{di}\left( p\right) \right]  &=&E\left[ \delta _{d}\left(
P_{i}\right) \left( g_{d}\left( P_{i}\right) -g_{d}\left( p\right) \right) 
\frac{1}{h_{3}^{3}}k_{3}\left( \frac{P_{i}-p}{h_{3}}\right) \left(
P_{i}-p\right) \right]  \\
&=&\kappa _{2}\delta _{d}\left( p\right) f_{P}\left( p\right) g_{d}^{\left(
1\right) }\left( p\right) +O\left( h_{3}^{2}\right) .
\end{eqnarray*}%
For the variances, the same arguments as in Lemma \ref{lemma:kernel} show
that%
\begin{eqnarray*}
Var\left( \frac{1}{n}\sum_{i=1}^{n}w_{di}\left( p\right) V_{di}\left(
p\right) \right)  &=&O\left( \frac{1}{nh_{3}}\right) , \\
Var\left( \frac{1}{nh_{3}^{2}}\sum_{i=1}^{n}w_{di}\left( p\right) \left(
P_{i}-p\right) V_{di}\left( p\right) \right)  &=&O\left( \frac{1}{nh_{3}^{3}}%
\right) .
\end{eqnarray*}%
The conclusion then follows from Markov's inequality.
\end{proof}

\begin{lemma}
\label{lemma:Brhat-Br} Under Assumptions E.1, E.2, and E.3, we have%
\begin{eqnarray*}
\frac{1}{n}\sum_{i=1}^{n}\hat{w}_{di}\left( p\right) -\frac{1}{n}%
\sum_{i=1}^{n}w_{di}\left( p\right) &=&o_{p}\left( \frac{1}{\sqrt{nh_{3}^{3}}%
}\right) , \\
\frac{1}{n}\sum_{i=1}^{n}\hat{w}_{di}\left( p\right) \left( \hat{P}%
_{i}-p\right) -\frac{1}{n}\sum_{i=1}^{n}w_{di}\left( p\right) \left(
P_{i}-p\right) &=&o_{p}\left( \frac{1}{\sqrt{nh_{3}}}\right) , \\
\frac{1}{n}\sum_{i=1}^{n}\hat{w}_{di}\left( p\right) \left( \hat{P}%
_{i}-p\right) ^{2}-\frac{1}{n}\sum_{i=1}^{n}w_{di}\left( p\right) \left(
P_{i}-p\right) ^{2} &=&o_{p}\left( \sqrt{\frac{h_{3}}{n}}\right) .
\end{eqnarray*}
\end{lemma}

\begin{proof}
By a Taylor expansion,%
\begin{eqnarray*}
&&\frac{1}{n}\sum_{i=1}^{n}\hat{w}_{di}\left( p\right) -\frac{1}{n}%
\sum_{i=1}^{n}w_{di}\left( p\right) \\
&=&\frac{1}{n}\sum_{i=1}^{n}1\left\{ D_{i}=d\right\} \frac{1}{h_{3}^{2}}%
k_{3}^{\left( 1\right) }\left( \frac{P_{i}-p}{h_{3}}\right) \left( \hat{P}%
_{i}-P_{i}\right) \\
&&+\frac{1}{2n}\sum_{i=1}^{n}1\left\{ D_{i}=d\right\} \frac{1}{h_{3}^{3}}%
k_{3}^{\left( 2\right) }\left( \frac{P_{i}-p}{h_{3}}\right) \left( \hat{P}%
_{i}-P_{i}\right) ^{2} \\
&&+\frac{1}{6n}\sum_{i=1}^{n}1\left\{ D_{i}=d\right\} \frac{1}{h_{3}^{4}}%
k_{3}^{\left( 3\right) }\left( \frac{\hat{P}_{i}^{\ast }-p}{h_{3}}\right)
\left( \hat{P}_{i}-P_{i}\right) ^{3},
\end{eqnarray*}%
where $\hat{P}_{i}^{\ast }$ is an intermediate value between $\hat{P}_{i}$
and $P_{i}$. For any i.i.d. nonnegative random variables $\left\{
Z_{ni}:1\leq i\leq n\right\} $ with $0<EZ_{ni}<\infty $, we have $\left.
\left( 1/n\right) \sum_{i=1}^{n}Z_{ni}\right/ EZ_{ni}=O_{p}\left( 1\right) $
by Markov's inequality. Thus, it follows from Lemma \ref{lemma:firststep}
and Assumption E.3.(i)-(ii) that%
\begin{eqnarray*}
\left\vert \frac{1}{n}\sum_{i=1}^{n}1\left\{ D_{i}=d\right\} \frac{1}{%
h_{3}^{2}}k_{3}^{\left( 1\right) }\left( \frac{P_{i}-p}{h_{3}}\right) \left( 
\hat{P}_{i}-P_{i}\right) \right\vert &\leq &\sup_{i}\left\vert \hat{P}%
_{i}-P_{i}\right\vert \frac{1}{n}\sum_{i=1}^{n}\frac{1}{h_{3}^{2}}\left\vert
k_{3}^{\left( 1\right) }\left( \frac{P_{i}-p}{h_{3}}\right) \right\vert \\
&=&O_{p}\left( \eta _{1}\right) O_{p}\left( \frac{1}{h_{3}^{2}}E\left\vert
k_{3}^{\left( 1\right) }\left( \frac{P_{i}-p}{h_{3}}\right) \right\vert
\right) \\
&=&O_{p}\left( \frac{\eta _{1}}{h_{3}}\right) =o_{p}\left( \frac{1}{\sqrt{%
nh_{3}^{3}}}\right) ,
\end{eqnarray*}%
\begin{eqnarray*}
\left\vert \frac{1}{n}\sum_{i=1}^{n}1\left\{ D_{i}=d\right\} \frac{1}{%
h_{3}^{3}}k_{3}^{\left( 2\right) }\left( \frac{P_{i}-p}{h_{3}}\right) \left( 
\hat{P}_{i}-P_{i}\right) ^{2}\right\vert &\leq &O_{p}\left( \eta
_{1}^{2}\right) O_{p}\left( \frac{1}{h_{3}^{3}}E\left\vert k_{3}^{\left(
2\right) }\left( \frac{P_{i}-p}{h_{3}}\right) \right\vert \right) \\
&=&O_{p}\left( \frac{\eta _{1}^{2}}{h_{3}^{2}}\right) =o_{p}\left( \frac{1}{%
nh_{3}^{3}}\right) =o_{p}\left( \frac{1}{\sqrt{nh_{3}^{3}}}\right) ,
\end{eqnarray*}%
and%
\begin{equation*}
\left\vert \frac{1}{n}\sum_{i=1}^{n}1\left\{ D_{i}=d\right\} \frac{1}{%
h_{3}^{4}}k_{3}^{\left( 3\right) }\left( \frac{\hat{P}_{i}^{\ast }-p}{h_{3}}%
\right) \left( \hat{P}_{i}-P_{i}\right) ^{3}\right\vert \leq O_{p}\left( 
\frac{\eta _{1}^{3}}{h_{3}^{4}}\right) =o_{p}\left( \frac{1}{h_{3}^{4}}\sqrt{%
\frac{h_{3}^{7}}{n}}\right) =o_{p}\left( \frac{1}{\sqrt{nh_{3}^{3}}}\right) .
\end{equation*}%
Therefore,%
\begin{equation*}
\frac{1}{n}\sum_{i=1}^{n}\hat{w}_{di}\left( p\right) -\frac{1}{n}%
\sum_{i=1}^{n}w_{di}\left( p\right) =o_{p}\left( \frac{1}{\sqrt{nh_{3}^{3}}}%
\right) .
\end{equation*}

For the second claim of the lemma, we decompose the left-hand side into
three terms:%
\begin{eqnarray}
&&\frac{1}{n}\sum_{i=1}^{n}\hat{w}_{di}\left( p\right) \left( \hat{P}%
_{i}-p\right) -\frac{1}{n}\sum_{i=1}^{n}w_{di}\left( p\right) \left(
P_{i}-p\right)  \notag \\
&=&\frac{1}{n}\sum_{i=1}^{n}w_{di}\left( p\right) \left( \hat{P}%
_{i}-P_{i}\right) +\frac{1}{n}\sum_{i=1}^{n}\left[ \hat{w}_{di}\left(
p\right) -w_{di}\left( p\right) \right] \left( P_{i}-p\right)  \notag \\
&&+\frac{1}{n}\sum_{i=1}^{n}\left[ \hat{w}_{di}\left( p\right) -w_{di}\left(
p\right) \right] \left( \hat{P}_{i}-P_{i}\right) .  \label{B1hat_decompose}
\end{eqnarray}%
For the first term, we have%
\begin{equation}
\left\vert \frac{1}{n}\sum_{i=1}^{n}w_{di}\left( p\right) \left( \hat{P}%
_{i}-P_{i}\right) \right\vert \leq O_{p}\left( \eta _{1}\right) O_{p}\left(
E\left\vert w_{di}\left( p\right) \right\vert \right) =O_{p}\left( \eta
_{1}\right) O_{p}\left( 1\right) =o_{p}\left( \frac{1}{\sqrt{nh_{3}}}\right)
.  \label{B1hat1}
\end{equation}%
For the second term, a Taylor expansion yields that%
\begin{eqnarray*}
&&\frac{1}{n}\sum_{i=1}^{n}\left[ \hat{w}_{di}\left( p\right) -w_{di}\left(
p\right) \right] \left( P_{i}-p\right) \\
&=&\frac{1}{n}\sum_{i=1}^{n}1\left\{ D_{i}=d\right\} \frac{1}{h_{3}^{2}}%
k_{3}^{\left( 1\right) }\left( \frac{P_{i}-p}{h_{3}}\right) \left(
P_{i}-p\right) \left( \hat{P}_{i}-P_{i}\right) \\
&&+\frac{1}{2n}\sum_{i=1}^{n}1\left\{ D_{i}=d\right\} \frac{1}{h_{3}^{3}}%
k_{3}^{\left( 2\right) }\left( \frac{P_{i}-p}{h_{3}}\right) \left(
P_{i}-p\right) \left( \hat{P}_{i}-P_{i}\right) ^{2} \\
&&+\frac{1}{6n}\sum_{i=1}^{n}1\left\{ D_{i}=d\right\} \frac{1}{h_{3}^{4}}%
k_{3}^{\left( 3\right) }\left( \frac{\hat{P}_{i}^{\ast }-p}{h_{3}}\right)
\left( P_{i}-p\right) \left( \hat{P}_{i}-P_{i}\right) ^{3}.
\end{eqnarray*}%
It follows that%
\begin{eqnarray*}
&&\left\vert \frac{1}{n}\sum_{i=1}^{n}1\left\{ D_{i}=d\right\} \frac{1}{%
h_{3}^{2}}k_{3}^{\left( 1\right) }\left( \frac{P_{i}-p}{h_{3}}\right) \left(
P_{i}-p\right) \left( \hat{P}_{i}-P_{i}\right) \right\vert \\
&\leq &\sup_{i}\left\vert \hat{P}_{i}-P_{i}\right\vert \frac{1}{n}%
\sum_{i=1}^{n}\frac{1}{h_{3}^{2}}\left\vert k_{3}^{\left( 1\right) }\left( 
\frac{P_{i}-p}{h_{3}}\right) \right\vert \left\vert P_{i}-p\right\vert \\
&=&O_{p}\left( \eta _{1}\right) O_{p}\left( \frac{1}{h_{3}}E\left[
\left\vert k_{3}^{\left( 1\right) }\left( \frac{P_{i}-p}{h_{3}}\right)
\right\vert \left\vert \frac{P_{i}-p}{h_{3}}\right\vert \right] \right) \\
&=&O_{p}\left( \eta _{1}\right) O_{p}\left( 1\right) =o_{p}\left( \frac{1}{%
\sqrt{nh_{3}}}\right) ,
\end{eqnarray*}%
\begin{eqnarray*}
&&\left\vert \frac{1}{n}\sum_{i=1}^{n}1\left\{ D_{i}=d\right\} \frac{1}{%
h_{3}^{3}}k_{3}^{\left( 2\right) }\left( \frac{P_{i}-p}{h_{3}}\right) \left(
P_{i}-p\right) \left( \hat{P}_{i}-P_{i}\right) ^{2}\right\vert \\
&\leq &O_{p}\left( \eta _{1}^{2}\right) O_{p}\left( \frac{1}{h_{3}^{2}}%
E\left\vert k_{3}^{\left( 2\right) }\left( \frac{P_{i}-p}{h_{3}}\right)
\right\vert \left\vert \frac{P_{i}-p}{h_{3}}\right\vert \right) \\
&=&O_{p}\left( \frac{\eta _{1}^{2}}{h_{3}}\right) =o_{p}\left( \frac{1}{%
nh_{3}^{2}}\right) =o_{p}\left( \frac{1}{\sqrt{nh_{3}}}\right) ,
\end{eqnarray*}%
and%
\begin{equation*}
\left\vert \frac{1}{n}\sum_{i=1}^{n}1\left\{ D_{i}=d\right\} \frac{1}{%
h_{3}^{4}}k_{3}^{\left( 3\right) }\left( \frac{\hat{P}_{i}^{\ast }-p}{h_{3}}%
\right) \left( P_{i}-p\right) \left( \hat{P}_{i}-P_{i}\right)
^{3}\right\vert \leq O_{p}\left( \frac{\eta _{1}^{3}}{h_{3}^{4}}\right)
=o_{p}\left( \frac{1}{h_{3}^{4}}\sqrt{\frac{h_{3}^{7}}{n}}\right)
=o_{p}\left( \frac{1}{\sqrt{nh_{3}}}\right) .
\end{equation*}%
Therefore,%
\begin{equation}
\frac{1}{n}\sum_{i=1}^{n}\left[ \hat{w}_{di}\left( p\right) -w_{di}\left(
p\right) \right] \left( P_{i}-p\right) =o_{p}\left( \frac{1}{\sqrt{nh_{3}}}%
\right) .  \label{B1hat2}
\end{equation}%
For the third term of (\ref{B1hat_decompose}), we have%
\begin{equation*}
\left\vert \frac{1}{n}\sum_{i=1}^{n}\left[ \hat{w}_{di}\left( p\right)
-w_{di}\left( p\right) \right] \left( \hat{P}_{i}-P_{i}\right) \right\vert
\leq \sup_{i}\left\vert \hat{P}_{i}-P_{i}\right\vert \frac{1}{n}%
\sum_{i=1}^{n}\left\vert \hat{w}_{di}\left( p\right) -w_{di}\left( p\right)
\right\vert .
\end{equation*}%
As for the first claim of the lemma, we can show that%
\begin{equation}
\frac{1}{n}\sum_{i=1}^{n}\left\vert \hat{w}_{di}\left( p\right)
-w_{di}\left( p\right) \right\vert =o_{p}\left( \frac{1}{\sqrt{nh_{3}^{3}}}%
\right) =o_{p}\left( 1\right) .  \label{absB0hat}
\end{equation}%
It then follows that%
\begin{equation}
\left\vert \frac{1}{n}\sum_{i=1}^{n}\left[ \hat{w}_{di}\left( p\right)
-w_{di}\left( p\right) \right] \left( \hat{P}_{i}-P_{i}\right) \right\vert
\leq O_{p}\left( \eta _{1}\right) o_{p}\left( 1\right) =o_{p}\left( \frac{1}{%
\sqrt{nh_{3}}}\right) .  \label{B1hat3}
\end{equation}%
Substituting (\ref{B1hat1}), (\ref{B1hat2}), and (\ref{B1hat3}) into (\ref%
{B1hat_decompose}), we obtain%
\begin{equation*}
\frac{1}{n}\sum_{i=1}^{n}\hat{w}_{di}\left( p\right) \left( \hat{P}%
_{i}-p\right) -\frac{1}{n}\sum_{i=1}^{n}w_{di}\left( p\right) \left(
P_{i}-p\right) =o_{p}\left( \frac{1}{\sqrt{nh_{3}}}\right) .
\end{equation*}

For the third claim of the lemma, we decompose the left-hand side into five
terms:%
\begin{eqnarray}
&&\frac{1}{n}\sum_{i=1}^{n}\hat{w}_{di}\left( p\right) \left( \hat{P}%
_{i}-p\right) ^{2}-\frac{1}{n}\sum_{i=1}^{n}w_{di}\left( p\right) \left(
P_{i}-p\right) ^{2}  \notag \\
&=&\frac{1}{n}\sum_{i=1}^{n}\left[ \hat{w}_{di}\left( p\right) -w_{di}\left(
p\right) \right] \left( \hat{P}_{i}-P_{i}\right) ^{2}+\frac{2}{n}%
\sum_{i=1}^{n}\left[ \hat{w}_{di}\left( p\right) -w_{di}\left( p\right) %
\right] \left( \hat{P}_{i}-P_{i}\right) \left( P_{i}-p\right)  \notag \\
&&+\frac{1}{n}\sum_{i=1}^{n}w_{di}\left( p\right) \left( \hat{P}%
_{i}-P_{i}\right) ^{2}+\frac{2}{n}\sum_{i=1}^{n}w_{di}\left( p\right) \left( 
\hat{P}_{i}-P_{i}\right) \left( P_{i}-p\right)  \notag \\
&&+\frac{1}{n}\sum_{i=1}^{n}\left[ \hat{w}_{di}\left( p\right) -w_{di}\left(
p\right) \right] \left( P_{i}-p\right) ^{2}.  \label{B2hat_decompose}
\end{eqnarray}%
For the first term, it follows from (\ref{absB0hat}) that%
\begin{eqnarray}
\left\vert \frac{1}{n}\sum_{i=1}^{n}\left[ \hat{w}_{di}\left( p\right)
-w_{di}\left( p\right) \right] \left( \hat{P}_{i}-P_{i}\right)
^{2}\right\vert &\leq &\sup_{i}\left\vert \hat{P}_{i}-P_{i}\right\vert ^{2}%
\frac{1}{n}\sum_{i=1}^{n}\left\vert \hat{w}_{di}\left( p\right)
-w_{di}\left( p\right) \right\vert  \notag \\
&=&O_{p}\left( \eta _{1}^{2}\right) o_{p}\left( 1\right) =o_{p}\left( \frac{1%
}{nh_{3}}\right)  \notag \\
&=&o_{p}\left( \frac{1}{\sqrt{nh_{3}^{3}}}\sqrt{\frac{h_{3}}{n}}\right)
=o_{p}\left( \sqrt{\frac{h_{3}}{n}}\right) .  \label{B2hat1}
\end{eqnarray}%
For the second term, a Taylor expansion yields that%
\begin{eqnarray*}
&&\frac{1}{n}\sum_{i=1}^{n}\left[ \hat{w}_{di}\left( p\right) -w_{di}\left(
p\right) \right] \left( \hat{P}_{i}-P_{i}\right) \left( P_{i}-p\right) \\
&=&\frac{1}{n}\sum_{i=1}^{n}1\left\{ D_{i}=d\right\} \frac{1}{h_{3}^{2}}%
k_{3}^{\left( 1\right) }\left( \frac{P_{i}-p}{h_{3}}\right) \left(
P_{i}-p\right) \left( \hat{P}_{i}-P_{i}\right) ^{2} \\
&&+\frac{1}{2n}\sum_{i=1}^{n}1\left\{ D_{i}=d\right\} \frac{1}{h_{3}^{3}}%
k_{3}^{\left( 2\right) }\left( \frac{\hat{P}_{i}^{\ast }-p}{h_{3}}\right)
\left( P_{i}-p\right) \left( \hat{P}_{i}-P_{i}\right) ^{3}.
\end{eqnarray*}%
Analogously, it follows that%
\begin{eqnarray*}
\left\vert \frac{1}{n}\sum_{i=1}^{n}1\left\{ D_{i}=d\right\} \frac{1}{%
h_{3}^{2}}k_{3}^{\left( 1\right) }\left( \frac{P_{i}-p}{h_{3}}\right) \left(
P_{i}-p\right) \left( \hat{P}_{i}-P_{i}\right) ^{2}\right\vert &\leq
&O_{p}\left( \eta _{1}^{2}\right) =o_{p}\left( \frac{1}{nh_{3}}\right)
=o_{p}\left( \sqrt{\frac{h_{3}}{n}}\right) , \\
\left\vert \frac{1}{n}\sum_{i=1}^{n}1\left\{ D_{i}=d\right\} \frac{1}{%
h_{3}^{3}}k_{3}^{\left( 2\right) }\left( \frac{\hat{P}_{i}^{\ast }-p}{h_{3}}%
\right) \left( P_{i}-p\right) \left( \hat{P}_{i}-P_{i}\right)
^{3}\right\vert &\leq &O_{p}\left( \frac{\eta _{1}^{3}}{h_{3}^{3}}\right)
=o_{p}\left( \frac{\sqrt{h_{3}^{7}}}{h_{3}^{3}\sqrt{n}}\right) =o_{p}\left( 
\sqrt{\frac{h_{3}}{n}}\right) .
\end{eqnarray*}%
Therefore,%
\begin{equation}
\frac{1}{n}\sum_{i=1}^{n}\left[ \hat{w}_{di}\left( p\right) -w_{di}\left(
p\right) \right] \left( \hat{P}_{i}-P_{i}\right) \left( P_{i}-p\right)
=o_{p}\left( \sqrt{\frac{h_{3}}{n}}\right) .  \label{B2hat2}
\end{equation}%
For the third and fourth terms of (\ref{B2hat_decompose}), we have%
\begin{eqnarray}
\left\vert \frac{1}{n}\sum_{i=1}^{n}w_{di}\left( p\right) \left( \hat{P}%
_{i}-P_{i}\right) ^{2}\right\vert &\leq &O_{p}\left( \eta _{1}^{2}\right)
O_{p}\left( E\left\vert w_{di}\left( p\right) \right\vert \right)  \notag \\
&=&O_{p}\left( \eta _{1}^{2}\right) O_{p}\left( 1\right) =o_{p}\left( \sqrt{%
\frac{h_{3}}{n}}\right) ,  \label{B2hat3} \\
\left\vert \frac{1}{n}\sum_{i=1}^{n}w_{di}\left( p\right) \left( \hat{P}%
_{i}-P_{i}\right) \left( P_{i}-p\right) \right\vert &\leq &O_{p}\left( \eta
_{1}\right) O_{p}\left( E\left\vert w_{di}\left( p\right) \left(
P_{i}-p\right) \right\vert \right)  \notag \\
&=&O_{p}\left( \eta _{1}h_{3}\right) =o_{p}\left( \sqrt{\frac{h_{3}}{n}}%
\right) .  \label{B2hat4}
\end{eqnarray}%
For the last term of (\ref{B2hat_decompose}), we can show as well by a
Taylor expansion that%
\begin{equation}
\frac{1}{n}\sum_{i=1}^{n}\left[ \hat{w}_{di}\left( p\right) -w_{di}\left(
p\right) \right] \left( P_{i}-p\right) ^{2}=o_{p}\left( \sqrt{\frac{h_{3}}{n}%
}\right) .  \label{B2hat5}
\end{equation}%
The conclusion follows from combining (\ref{B2hat_decompose})-(\ref{B2hat5}).
\end{proof}

\begin{lemma}
\label{lemma:Crhat-Cr} Under Assumptions E.1, E.2, E.3, and E.4, we have%
\begin{eqnarray*}
&&\frac{1}{n}\sum_{i}^{n}\hat{w}_{di}\left( p\right) V_{di}\left( p\right) -%
\frac{1}{n}\sum_{i}^{n}w_{di}\left( p\right) V_{di}\left( p\right) \\
&=&\frac{g_{d}^{\left( 1\right) }\left( p\right) }{n}\sum_{i=1}^{n}\delta
_{d}\left( P_{i}\right) \frac{1}{h_{3}}k_{3}^{\left( 1\right) }\left( \frac{%
P_{i}-p}{h_{3}}\right) \frac{P_{i}-p}{h_{3}}\left( D_{i}-P_{i}\right)
+o_{p}\left( \frac{1}{\sqrt{nh_{3}}}\right) \\
&=&O_{p}\left( \frac{1}{\sqrt{nh_{3}}}\right) ,
\end{eqnarray*}%
and%
\begin{eqnarray*}
&&\frac{1}{n}\sum_{i}^{n}\hat{w}_{di}\left( p\right) \left( \hat{P}%
_{i}-p\right) V_{di}\left( p\right) -\frac{1}{n}\sum_{i}^{n}w_{di}\left(
p\right) \left( P_{i}-p\right) V_{di}\left( p\right) \\
&=&\frac{g_{d}^{\left( 1\right) }\left( p\right) }{n}\sum_{i=1}^{n}\delta
_{d}\left( P_{i}\right) \left[ k_{3}\left( \frac{P_{i}-p}{h_{3}}\right)
+k_{3}^{\left( 1\right) }\left( \frac{P_{i}-p}{h_{3}}\right) \frac{P_{i}-p}{%
h_{3}}\right] \frac{P_{i}-p}{h_{3}}\left( D_{i}-P_{i}\right) +o_{p}\left( 
\sqrt{\frac{h_{3}}{n}}\right) \\
&=&O_{p}\left( \sqrt{\frac{h_{3}}{n}}\right) ,
\end{eqnarray*}%
where $V_{di}\left( p\right) $ is defined in (\ref{Vdip}).
\end{lemma}

\begin{proof}
By a Taylor expansion,%
\begin{eqnarray}
&&\frac{1}{n}\sum_{i=1}^{n}\hat{w}_{di}\left( p\right) V_{di}\left( p\right)
-\frac{1}{n}\sum_{i=1}^{n}w_{di}\left( p\right) V_{di}\left( p\right)  \notag
\\
&=&\frac{1}{n}\sum_{i=1}^{n}1\left\{ D_{i}=d\right\} \frac{1}{h_{3}^{2}}%
k_{3}^{\left( 1\right) }\left( \frac{P_{i}-p}{h_{3}}\right) \left( \hat{P}%
_{i}-P_{i}\right) V_{di}\left( p\right)  \notag \\
&&+\frac{1}{2n}\sum_{i=1}^{n}1\left\{ D_{i}=d\right\} \frac{1}{h_{3}^{3}}%
k_{3}^{\left( 2\right) }\left( \frac{P_{i}-p}{h_{3}}\right) \left( \hat{P}%
_{i}-P_{i}\right) ^{2}V_{di}\left( p\right)  \notag \\
&&+\frac{1}{6n}\sum_{i=1}^{n}1\left\{ D_{i}=d\right\} \frac{1}{h_{3}^{4}}%
k_{3}^{\left( 3\right) }\left( \frac{\hat{P}_{i}^{\ast }-p}{h_{3}}\right)
\left( \hat{P}_{i}-P_{i}\right) ^{3}V_{di}\left( p\right)  \notag \\
&\triangleq &\left( \text{I}\right) +\left( \text{II}\right) +\left( \text{%
III}\right) ,  \label{C0hat_decompose}
\end{eqnarray}%
where $\hat{P}_{i}^{\ast }$ is an intermediate value between $\hat{P}_{i}$
and $P_{i}$. Note that $\left( \text{I}\right) $ and $\left( \text{II}%
\right) $ can be represented as U-statistics of order two and three,
respectively, plus higher order terms. Specifically, by denoting $%
W_{i}=\left( Y_{i},D_{i},X_{i}\right) $ as the $i$-th observation, and%
\begin{eqnarray*}
K_{1h}\left( X_{j},X_{i}\right) &=&\prod_{l=1}^{\dim \left( X^{C}\right) }%
\frac{1}{\left\vert h_{1}\right\vert }k_{1}\left( \frac{X_{jl}^{C}-X_{il}^{C}%
}{h_{1l}}\right) 1\left\{ X_{j}^{D}=X_{i}^{D}\right\} , \\
m_{1}\left( W_{i},W_{j}\right) &=&1\left\{ D_{i}=d\right\} \frac{1}{h_{3}^{2}%
}k_{3}^{\left( 1\right) }\left( \frac{P_{i}-p}{h_{3}}\right) V_{di}\left(
p\right) \left( D_{j}-P_{i}\right) K_{1h}\left( X_{j},X_{i}\right) , \\
m_{2}\left( W_{i},W_{j},W_{l}\right) &=&1\left\{ D_{i}=d\right\} \frac{1}{%
h_{3}^{3}}k_{3}^{\left( 2\right) }\left( \frac{P_{i}-p}{h_{3}}\right)
V_{di}\left( p\right) \left( D_{j}-P_{i}\right) K_{1h}\left(
X_{j},X_{i}\right) \left( D_{l}-P_{i}\right) K_{1h}\left( X_{l},X_{i}\right)
, \\
M_{1} &=&\frac{1}{n\left( n-1\right) }\sum_{i}\sum_{j\neq i}\frac{%
m_{1}\left( W_{i},W_{j}\right) }{f_{X}\left( X_{i}\right) }, \\
M_{2} &=&\frac{1}{n\left( n-1\right) \left( n-2\right) }\sum_{i}\sum_{j\neq
i}\sum_{l\neq j,i}\frac{m_{2}\left( W_{i},W_{j},W_{l}\right) }{%
f_{X}^{2}\left( X_{i}\right) },
\end{eqnarray*}%
we have%
\begin{eqnarray}
\left( \text{I}\right) &=&\frac{1}{n\left( n-1\right) }\sum_{i}\sum_{j\neq i}%
\frac{m_{1}\left( W_{i},W_{j}\right) }{\hat{f}_{X}\left( X_{i}\right) }%
=M_{1}\cdot \left( 1+O_{p}\left( \eta _{1}\right) \right) ,  \label{(I)} \\
\left( \text{II}\right) &=&\frac{1}{n\left( n-1\right) ^{2}}%
\sum_{i}\sum_{j\neq i}\sum_{l\neq i}\frac{m_{2}\left(
W_{i},W_{j},W_{l}\right) }{\hat{f}_{X}^{2}\left( X_{i}\right) }=M_{2}\cdot
\left( 1+O_{p}\left( \eta _{1}\right) \right) ,  \label{(II)}
\end{eqnarray}%
where $\hat{f}_{X}\left( x\right) $ is defined in (\ref{fXhat}) and the last
equalities follow from Lemma \ref{lemma:firststep}.

To analyze the U-statistics $M_{1}$ and $M_{2}$, we need to calculate the
conditional expectations of their kernels:%
\begin{eqnarray*}
E\left[ \left. \frac{m_{1}\left( W_{i},W_{j}\right) }{f_{X}\left(
X_{i}\right) }\right\vert W_{i}\right]  &=&O\left( tr\left( h_{1}^{s}\right)
\right) \cdot 1\left\{ D_{i}=d\right\} \frac{1}{h_{3}^{2}}k_{3}^{\left(
1\right) }\left( \frac{P_{i}-p}{h_{3}}\right) V_{di}\left( p\right) , \\
E\left[ \left. \frac{m_{1}\left( W_{i},W_{j}\right) }{f_{X}\left(
X_{i}\right) }\right\vert W_{j}\right]  &=&\delta _{d}\left( P_{j}\right) 
\frac{g_{d}\left( P_{j}\right) -g_{d}\left( p\right) }{h_{3}^{2}}%
k_{3}^{\left( 1\right) }\left( \frac{P_{j}-p}{h_{3}}\right) \left(
D_{j}-P_{j}\right) \left[ 1+O\left( tr\left( h_{1}^{s}\right) \right) \right]
, \\
E\left[ \left. \frac{m_{2}\left( W_{i},W_{j},W_{l}\right) }{f_{X}^{2}\left(
X_{i}\right) }\right\vert W_{i}\right]  &=&O\left( tr\left(
h_{1}^{2s}\right) \right) \cdot 1\left\{ D_{i}=d\right\} \frac{1}{h_{3}^{3}}%
k_{3}^{\left( 2\right) }\left( \frac{P_{i}-p}{h_{3}}\right) V_{di}\left(
p\right) , \\
E\left[ \left. \frac{m_{2}\left( W_{i},W_{j},W_{l}\right) }{f_{X}^{2}\left(
X_{i}\right) }\right\vert W_{j}\right]  &=&O\left( tr\left( h_{1}^{s}\right)
\right) \cdot \delta _{d}\left( P_{j}\right) \frac{g_{d}\left( P_{j}\right)
-g_{d}\left( p\right) }{h_{3}^{2}}k_{3}^{\left( 1\right) }\left( \frac{%
P_{j}-p}{h_{3}}\right) \left( D_{j}-P_{j}\right)  \\
&&\cdot \left[ 1+O\left( tr\left( h_{1}^{s}\right) \right) \right] , \\
E\left[ \left. \frac{m_{2}\left( W_{i},W_{j},W_{l}\right) }{f_{X}^{2}\left(
X_{i}\right) }\right\vert W_{l}\right]  &=&O\left( tr\left( h_{1}^{s}\right)
\right) \cdot \delta _{d}\left( P_{l}\right) \frac{g_{d}\left( P_{l}\right)
-g_{d}\left( p\right) }{h_{3}^{2}}k_{3}^{\left( 1\right) }\left( \frac{%
P_{l}-p}{h_{3}}\right) \left( D_{l}-P_{l}\right)  \\
&&\cdot \left[ 1+O\left( tr\left( h_{1}^{s}\right) \right) \right] ,
\end{eqnarray*}%
where we use $E\left[ \left. V_{di}\left( p\right) \right\vert X_{i},D_{i}=d%
\right] =g_{d}\left( P_{i}\right) -g_{d}\left( p\right) $, and hence,%
\begin{eqnarray*}
E\left[ M_{1}\right]  &=&E\left[ E\left[ \left. \frac{m_{1}\left(
W_{i},W_{j}\right) }{f_{X}\left( X_{i}\right) }\right\vert W_{j}\right] %
\right] =0, \\
E\left[ M_{2}\right]  &=&E\left[ E\left[ \left. \frac{m_{2}\left(
W_{i},W_{j},W_{l}\right) }{f_{X}^{2}\left( X_{i}\right) }\right\vert W_{j}%
\right] \right] =0.
\end{eqnarray*}%
Since by Assumptions E.1.(ii) and E.3.(iii)-(iv),%
\begin{eqnarray*}
E\left[ \frac{m_{1}^{2}\left( W_{i},W_{j}\right) }{f_{X}^{2}\left(
X_{i}\right) }\right]  &=&O\left( \frac{1}{h_{3}^{3}\left\vert
h_{1}\right\vert }\right) =o\left( n\right) , \\
E\left[ \frac{m_{2}^{2}\left( W_{i},W_{j},W_{l}\right) }{f_{X}^{4}\left(
X_{i}\right) }\right]  &=&O\left( \frac{1}{h_{3}^{5}\left\vert
h_{1}\right\vert ^{2}}\right) =o\left( n^{3/2}\right) ,
\end{eqnarray*}%
it follows from the projection method for U-statistics 
\citep[e.g.,][Lemma
3.1]{powell1989semiparametric} that%
\begin{equation*}
M_{1}=\frac{1}{n}\sum_{i=1}^{n}E\left[ \left. \frac{m_{1}\left(
W_{i},W_{j}\right) }{f_{X}\left( X_{i}\right) }\right\vert W_{i}\right] +%
\frac{1}{n}\sum_{j=1}^{n}E\left[ \left. \frac{m_{1}\left( W_{i},W_{j}\right) 
}{f_{X}\left( X_{i}\right) }\right\vert W_{j}\right] +o_{p}\left( \frac{1}{%
\sqrt{n}}\right) ,
\end{equation*}%
and from the variances of U-statistics \citep[e.g.,][Corollary 3.2]{shao2003}
that%
\begin{eqnarray*}
Var\left( M_{2}\right)  &\leq &\frac{3}{n}\left[ Var\left( E\left[ \left. 
\frac{m_{2}\left( W_{i},W_{j},W_{l}\right) }{f_{X}^{2}\left( X_{i}\right) }%
\right\vert W_{i}\right] \right) +2Var\left( E\left[ \left. \frac{%
m_{2}\left( W_{i},W_{j},W_{l}\right) }{f_{X}^{2}\left( X_{i}\right) }%
\right\vert W_{j}\right] \right) \right] +O\left( \frac{1}{n^{2}}\right)  \\
&=&O\left( \frac{tr\left( h_{1}^{4s}\right) }{nh_{3}^{5}}+\frac{tr\left(
h_{1}^{2s}\right) }{nh_{3}}+\frac{1}{n^{2}}\right) =O\left( \frac{1}{%
n^{3}h_{3}^{7}}+\frac{1}{n^{2}h_{3}^{2}}+\frac{1}{n^{2}}\right) =O\left( 
\frac{1}{n^{3}h_{3}^{7}}+\frac{1}{n^{2}h_{3}^{2}}\right) ,
\end{eqnarray*}%
where we use $tr\left( h_{1}^{s}\right) =O\left( \eta _{1}\right) =o\left(
1\left/ \sqrt{nh_{3}}\right. \right) $ by Assumption E.3.(i). Since%
\begin{equation*}
\frac{1}{n}\sum_{i=1}^{n}E\left[ \left. \frac{m_{1}\left( W_{i},W_{j}\right) 
}{f_{X}\left( X_{i}\right) }\right\vert W_{i}\right] =O\left( tr\left(
h_{1}^{s}\right) \right) \cdot \frac{1}{n}\sum_{i=1}^{n}1\left\{
D_{i}=d\right\} \frac{1}{h_{3}^{2}}k_{3}^{\left( 1\right) }\left( \frac{%
P_{i}-p}{h_{3}}\right) V_{di}\left( p\right) ,
\end{equation*}%
and%
\begin{eqnarray*}
E\left[ \frac{1}{n}\sum_{i=1}^{n}1\left\{ D_{i}=d\right\} \frac{1}{h_{3}^{2}}%
k_{3}^{\left( 1\right) }\left( \frac{P_{i}-p}{h_{3}}\right) V_{di}\left(
p\right) \right]  &=&E\left[ \delta _{d}\left( P_{i}\right) \frac{%
g_{d}\left( P_{i}\right) -g_{d}\left( p\right) }{h_{3}^{2}}k_{3}^{\left(
1\right) }\left( \frac{P_{i}-p}{h_{3}}\right) \right]  \\
&=&O\left( 1\right) , \\
Var\left( \frac{1}{n}\sum_{i=1}^{n}1\left\{ D_{i}=d\right\} \frac{1}{%
h_{3}^{2}}k_{3}^{\left( 1\right) }\left( \frac{P_{i}-p}{h_{3}}\right)
V_{di}\left( p\right) \right)  &\leq &\frac{1}{n}E\left[ \frac{1}{h_{3}^{4}}%
\left\vert k_{3}^{\left( 1\right) }\left( \frac{P_{i}-p}{h_{3}}\right)
\right\vert ^{2}V_{di}^{2}\left( p\right) \right]  \\
&=&O\left( \frac{1}{nh_{3}^{3}}\right) =o\left( 1\right) ,
\end{eqnarray*}%
it follows from Markov's inequality that%
\begin{equation*}
\frac{1}{n}\sum_{i=1}^{n}E\left[ \left. \frac{m_{1}\left( W_{i},W_{j}\right) 
}{f_{X}\left( X_{i}\right) }\right\vert W_{i}\right] =O\left( tr\left(
h_{1}^{s}\right) \right) \cdot O_{p}\left( 1\right) =O_{p}\left( tr\left(
h_{1}^{s}\right) \right) =o_{p}\left( \frac{1}{\sqrt{nh_{3}}}\right) ,
\end{equation*}%
Hence,%
\begin{eqnarray}
M_{1} &=&\frac{1}{n}\sum_{j=1}^{n}\delta _{d}\left( P_{j}\right) \frac{%
g_{d}\left( P_{j}\right) -g_{d}\left( p\right) }{h_{3}^{2}}k_{3}^{\left(
1\right) }\left( \frac{P_{j}-p}{h_{3}}\right) \left( D_{j}-P_{j}\right)
+o_{p}\left( \frac{1}{\sqrt{nh_{3}}}\right)   \notag \\
&=&\frac{g_{d}^{\left( 1\right) }\left( p\right) }{n}\sum_{j=1}^{n}\delta
_{d}\left( P_{j}\right) \frac{P_{j}-p}{h_{3}^{2}}k_{3}^{\left( 1\right)
}\left( \frac{P_{j}-p}{h_{3}}\right) \left( D_{j}-P_{j}\right) +o_{p}\left( 
\frac{1}{\sqrt{nh_{3}}}\right) .  \label{M1}
\end{eqnarray}%
On the other hand,%
\begin{equation*}
E\left[ M_{2}^{2}\right] =Var\left( M_{2}\right) =O\left( \frac{1}{%
n^{3}h_{3}^{7}}+\frac{1}{n^{2}h_{3}^{2}}\right) =o\left( \frac{1}{nh_{3}}%
\right) ,
\end{equation*}%
and thus by Markov's inequality,%
\begin{equation}
M_{2}=o_{p}\left( \frac{1}{\sqrt{nh_{3}}}\right) .  \label{M2}
\end{equation}%
Substituting (\ref{M1}) and (\ref{M2}) into (\ref{(I)}) and (\ref{(II)}),
respectively, we obtain%
\begin{equation}
\left( \text{I}\right) +\left( \text{II}\right) =\frac{g_{d}^{\left(
1\right) }\left( p\right) }{n}\sum_{i=1}^{n}\delta _{d}\left( P_{i}\right) 
\frac{1}{h_{3}}k_{3}^{\left( 1\right) }\left( \frac{P_{i}-p}{h_{3}}\right) 
\frac{P_{i}-p}{h_{3}}\left( D_{i}-P_{i}\right) +o_{p}\left( \frac{1}{\sqrt{%
nh_{3}}}\right) .  \label{(I)+(II)}
\end{equation}%
For the term $\left( \text{III}\right) $, it follows from Lemma \ref%
{lemma:firststep} that%
\begin{equation}
\left( \text{III}\right) \leq \frac{1}{6n}\sum_{i=1}^{n}\frac{1}{h_{3}^{4}}%
\left\vert k_{3}^{\left( 3\right) }\left( \frac{\hat{P}_{i}^{\ast }-p}{h_{3}}%
\right) \right\vert \left\vert \hat{P}_{i}-P_{i}\right\vert ^{3}=O_{p}\left( 
\frac{\eta _{1}^{3}}{h_{3}^{4}}\right) =o_{p}\left( \frac{1}{\sqrt{nh_{3}}}%
\right) ,  \label{(III)}
\end{equation}%
where the last equality follows from Assumption E.3.(ii). The first claim of
the lemma follows from substituting (\ref{(I)+(II)}) and (\ref{(III)}) into (%
\ref{C0hat_decompose}).

For the second claim, we decompose the left-hand side into three terms:%
\begin{eqnarray}
&&\frac{1}{n}\sum_{i=1}^{n}\hat{w}_{di}\left( p\right) \left( \hat{P}%
_{i}-p\right) V_{di}\left( p\right) -\frac{1}{n}\sum_{i=1}^{n}w_{di}\left(
p\right) \left( P_{i}-p\right) V_{di}\left( p\right)   \notag \\
&=&\frac{1}{n}\sum_{i=1}^{n}w_{di}\left( p\right) \left( \hat{P}%
_{i}-P_{i}\right) V_{di}\left( p\right) +\frac{1}{n}\sum_{i=1}^{n}\left[ 
\hat{w}_{di}\left( p\right) -w_{di}\left( p\right) \right] \left(
P_{i}-p\right) V_{di}\left( p\right)   \notag \\
&&+\frac{1}{n}\sum_{i=1}^{n}\left[ \hat{w}_{di}\left( p\right) -w_{di}\left(
p\right) \right] \left( \hat{P}_{i}-P_{i}\right) V_{di}\left( p\right) .
\label{C1hat_decompose}
\end{eqnarray}%
Denote%
\begin{eqnarray*}
m_{3}\left( W_{i},W_{j}\right)  &=&\frac{1}{h_{3}}w_{di}\left( p\right)
V_{di}\left( p\right) \left( D_{j}-P_{i}\right) K_{1h}\left(
X_{j},X_{i}\right) , \\
M_{3} &=&\frac{1}{n\left( n-1\right) }\sum_{i}\sum_{j\neq i}\frac{%
m_{3}\left( W_{i},W_{j}\right) }{f_{X}\left( X_{i}\right) },
\end{eqnarray*}%
then it follows from Lemma \ref{lemma:firststep} that%
\begin{equation*}
\frac{1}{n}\sum_{i=1}^{n}w_{di}\left( p\right) \left( \hat{P}%
_{i}-P_{i}\right) V_{di}\left( p\right) =\frac{h_{3}}{n\left( n-1\right) }%
\sum_{i}\sum_{j\neq i}\frac{m_{3}\left( W_{i},W_{j}\right) }{\hat{f}%
_{X}\left( X_{i}\right) }=h_{3}M_{3}\cdot \left( 1+O_{p}\left( \eta
_{1}\right) \right) .
\end{equation*}%
Since%
\begin{eqnarray*}
E\left[ \frac{m_{3}^{2}\left( W_{i},W_{j}\right) }{f_{X}^{2}\left(
X_{i}\right) }\right]  &=&O\left( \frac{1}{h_{3}^{3}\left\vert
h_{1}\right\vert }\right) =o\left( n\right) , \\
E\left[ \left. \frac{m_{3}\left( W_{i},W_{j}\right) }{f_{X}\left(
X_{i}\right) }\right\vert W_{i}\right]  &=&O\left( tr\left( h_{1}^{s}\right)
\right) \cdot \frac{1}{h_{3}}w_{di}\left( p\right) V_{di}\left( p\right) , \\
E\left[ \left. \frac{m_{3}\left( W_{i},W_{j}\right) }{f_{X}\left(
X_{i}\right) }\right\vert W_{j}\right]  &=&\delta _{d}\left( P_{j}\right) %
\left[ g_{d}\left( P_{j}\right) -g_{d}\left( p\right) \right] \frac{1}{%
h_{3}^{2}}k_{3}\left( \frac{P_{j}-p}{h_{3}}\right) \left( D_{j}-P_{j}\right) %
\left[ 1+O\left( tr\left( h_{1}^{s}\right) \right) \right] , \\
E\left[ \frac{m_{3}\left( W_{i},W_{j}\right) }{f_{X}\left( X_{i}\right) }%
\right]  &=&0,
\end{eqnarray*}%
it follows from the projection of U-statistics that%
\begin{eqnarray*}
M_{3} &=&\frac{1}{n}\sum_{i=1}^{n}E\left[ \left. \frac{m_{3}\left(
W_{i},W_{j}\right) }{f_{X}\left( X_{i}\right) }\right\vert W_{i}\right] +%
\frac{1}{n}\sum_{j=1}^{n}E\left[ \left. \frac{m_{3}\left( W_{i},W_{j}\right) 
}{f_{X}\left( X_{i}\right) }\right\vert W_{j}\right] +o_{p}\left( \frac{1}{%
\sqrt{n}}\right)  \\
&=&\frac{1}{n}\sum_{j=1}^{n}\delta _{d}\left( P_{j}\right) \frac{g_{d}\left(
P_{j}\right) -g_{d}\left( p\right) }{h_{3}^{2}}k_{3}\left( \frac{P_{j}-p}{%
h_{3}}\right) \left( D_{j}-P_{j}\right) +o_{p}\left( \frac{1}{\sqrt{nh_{3}}}%
\right)  \\
&=&\frac{g_{d}^{\left( 1\right) }\left( p\right) }{n}\sum_{j=1}^{n}\delta
_{d}\left( P_{j}\right) \frac{P_{j}-p}{h_{3}^{2}}k_{3}\left( \frac{P_{j}-p}{%
h_{3}}\right) \left( D_{j}-P_{j}\right) +o_{p}\left( \frac{1}{\sqrt{nh_{3}}}%
\right) .
\end{eqnarray*}%
Therefore,%
\begin{eqnarray}
\frac{1}{n}\sum_{i=1}^{n}w_{di}\left( p\right) \left( \hat{P}%
_{i}-P_{i}\right) V_{di}\left( p\right)  &=&\frac{g_{d}^{\left( 1\right)
}\left( p\right) }{n}\sum_{i=1}^{n}\delta _{d}\left( P_{i}\right)
k_{3}\left( \frac{P_{i}-p}{h_{3}}\right) \frac{P_{i}-p}{h_{3}}\left(
D_{i}-P_{i}\right)   \notag \\
&&+o_{p}\left( \sqrt{\frac{h_{3}}{n}}\right) .  \label{decompose1}
\end{eqnarray}

For the second term in the decomposition (\ref{B1hat_decompose}), a Taylor
expansion gives%
\begin{eqnarray*}
&&\frac{1}{n}\sum_{i}\left[ \hat{w}_{di}\left( p\right) -w_{di}\left(
p\right) \right] \left( P_{i}-p\right) V_{di}\left( p\right)  \\
&=&\frac{1}{n}\sum_{i=1}^{n}1\left\{ D_{i}=d\right\} V_{di}\left( p\right) 
\frac{1}{h_{3}^{2}}k_{3}^{\left( 1\right) }\left( \frac{P_{i}-p}{h_{3}}%
\right) \left( P_{i}-p\right) \left( \hat{P}_{i}-P_{i}\right)  \\
&&+\frac{1}{2n}\sum_{i=1}^{n}1\left\{ D_{i}=d\right\} V_{di}\left( p\right) 
\frac{1}{h_{3}^{3}}k_{3}^{\left( 2\right) }\left( \frac{P_{i}-p}{h_{3}}%
\right) \left( P_{i}-p\right) \left( \hat{P}_{i}-P_{i}\right) ^{2} \\
&&+\frac{1}{6n}\sum_{i=1}^{n}1\left\{ D_{i}=d\right\} V_{di}\left( p\right) 
\frac{1}{h_{3}^{4}}k_{3}^{\left( 3\right) }\left( \frac{\hat{P}_{i}^{\ast }-p%
}{h_{3}}\right) \left( P_{i}-p\right) \left( \hat{P}_{i}-P_{i}\right) ^{3} \\
&\triangleq &\left( \text{IV}\right) +\left( \text{V}\right) +\left( \text{VI%
}\right) .
\end{eqnarray*}%
Denote%
\begin{eqnarray*}
m_{4}\left( W_{i},W_{j}\right)  &=&1\left\{ D_{i}=d\right\} V_{di}\left(
p\right) \frac{1}{h_{3}^{3}}k_{3}^{\left( 1\right) }\left( \frac{P_{i}-p}{%
h_{3}}\right) \left( P_{i}-p\right) \left( D_{j}-P_{i}\right) K_{1h}\left(
X_{j},X_{i}\right) , \\
M_{4} &=&\frac{1}{n\left( n-1\right) }\sum_{i}\sum_{j\neq i}\frac{%
m_{4}\left( W_{i},W_{j}\right) }{f_{X}\left( X_{i}\right) },
\end{eqnarray*}%
then we have%
\begin{equation*}
\left( \text{IV}\right) =\frac{h_{3}}{n\left( n-1\right) }%
\sum_{i}\sum_{j\neq i}\frac{m_{4}\left( W_{i},W_{j}\right) }{\hat{f}%
_{X}\left( X_{i}\right) }=h_{3}M_{4}\cdot \left( 1+O_{p}\left( \eta
_{1}\right) \right) .
\end{equation*}%
Analogously, it follows from the projection of U-statistics that%
\begin{eqnarray*}
M_{4} &=&\frac{1}{n}\sum_{i=1}^{n}\delta _{d}\left( P_{i}\right) \frac{%
g_{d}\left( P_{i}\right) -g_{d}\left( p\right) }{h_{3}^{3}}k_{3}^{\left(
1\right) }\left( \frac{P_{i}-p}{h_{3}}\right) \left( P_{i}-p\right) \left(
D_{i}-P_{i}\right) +o_{p}\left( \frac{1}{\sqrt{nh_{3}}}\right)  \\
&=&\frac{g_{d}^{\left( 1\right) }\left( p\right) }{n}\sum_{i=1}^{n}\delta
_{d}\left( P_{i}\right) \frac{1}{h_{3}}k_{3}^{\left( 1\right) }\left( \frac{%
P_{i}-p}{h_{3}}\right) \left( \frac{P_{i}-p}{h_{3}}\right) ^{2}\left(
D_{i}-P_{i}\right) +o_{p}\left( \frac{1}{\sqrt{nh_{3}}}\right) ,
\end{eqnarray*}%
and hence%
\begin{equation*}
\left( \text{IV}\right) =\frac{g_{d}^{\left( 1\right) }\left( p\right) }{n}%
\sum_{i=1}^{n}\delta _{d}\left( P_{i}\right) k_{3}^{\left( 1\right) }\left( 
\frac{P_{i}-p}{h_{3}}\right) \left( \frac{P_{i}-p}{h_{3}}\right) ^{2}\left(
D_{i}-P_{i}\right) +o_{p}\left( \sqrt{\frac{h_{3}}{n}}\right) .
\end{equation*}%
Moreover, it follows from Lemma \ref{lemma:firststep} that%
\begin{eqnarray*}
\left\vert \left( \text{V}\right) \right\vert  &\leq &O_{p}\left( \eta
_{1}^{2}\right) \cdot O_{p}\left( E\left\vert V_{di}\left( p\right) \frac{1}{%
h_{3}^{3}}k_{3}^{\left( 2\right) }\left( \frac{P_{i}-p}{h_{3}}\right) \left(
P_{i}-p\right) \right\vert \right) =O_{p}\left( \eta _{1}^{2}\right)
=o_{p}\left( \sqrt{\frac{h_{3}}{n}}\right) , \\
\left\vert \left( \text{VI}\right) \right\vert  &\leq &O_{p}\left( \eta
_{1}^{3}\right) \cdot O_{p}\left( E\left\vert V_{di}\left( p\right) \frac{1}{%
h_{3}^{4}}k_{3}^{\left( 3\right) }\left( \frac{\hat{P}_{i}^{\ast }-p}{h_{3}}%
\right) \left( P_{i}-p\right) \right\vert \right) =O_{p}\left( \frac{\eta
_{1}^{3}}{h_{3}^{3}}\right) =o_{p}\left( \sqrt{\frac{h_{3}}{n}}\right) .
\end{eqnarray*}%
Therefore, we have%
\begin{eqnarray}
&&\frac{1}{n}\sum_{i}\left[ \hat{w}_{di}\left( p\right) -w_{di}\left(
p\right) \right] \left( P_{i}-p\right) V_{di}\left( p\right)   \notag \\
&=&\frac{g_{d}^{\left( 1\right) }\left( p\right) }{n}\sum_{i=1}^{n}\delta
_{d}\left( P_{i}\right) k_{3}^{\left( 1\right) }\left( \frac{P_{i}-p}{h_{3}}%
\right) \left( \frac{P_{i}-p}{h_{3}}\right) ^{2}\left( D_{i}-P_{i}\right)
+o_{p}\left( \frac{1}{\sqrt{nh_{3}}}\right) .  \label{decompose2}
\end{eqnarray}

For the third term in the decomposition (\ref{B1hat_decompose}), a Taylor
expansion gives%
\begin{eqnarray}
&&\frac{1}{n}\sum_{i}\left[ \hat{w}_{di}\left( p\right) -w_{di}\left(
p\right) \right] \left( \hat{P}_{i}-P_{i}\right) V_{di}\left( p\right)  
\notag \\
&=&\frac{1}{n}\sum_{i=1}^{n}1\left\{ D_{i}=d\right\} V_{di}\left( p\right) 
\frac{1}{h_{3}^{2}}k_{3}^{\left( 1\right) }\left( \frac{P_{i}-p}{h_{3}}%
\right) \left( \hat{P}_{i}-P_{i}\right) ^{2}  \notag \\
&&+\frac{1}{2n}\sum_{i=1}^{n}1\left\{ D_{i}=d\right\} V_{di}\left( p\right) 
\frac{1}{h_{3}^{3}}k_{3}^{\left( 2\right) }\left( \frac{\hat{P}_{i}^{\ast }-p%
}{h_{3}}\right) \left( \hat{P}_{i}-P_{i}\right) ^{3}  \notag \\
&=&O_{p}\left( \eta _{1}^{2}\right) +O_{p}\left( \frac{\eta _{1}^{3}}{%
h_{3}^{3}}\right) =o_{p}\left( \sqrt{\frac{h_{3}}{n}}\right) .
\label{decompose3}
\end{eqnarray}%
The second claim of the lemma then follows from inserting (\ref{decompose1}%
), (\ref{decompose2}), and (\ref{decompose3}) into the decomposition (\ref%
{C1hat_decompose}).
\end{proof}

\begin{lemma}
\label{lemma:Grhat} Under Assumptions E.1, E.2, E.3, and E.4, we have%
\begin{eqnarray*}
\frac{1}{n}\sum_{i}^{n}\hat{w}_{di}\left( p\right) X_{i} &=&O_{p}\left(
1\right) , \\
\frac{1}{n}\sum_{i}^{n}\hat{w}_{di}\left( p\right) \left( \hat{P}%
_{i}-p\right) X_{i} &=&O_{p}\left( h_{3}\right) .
\end{eqnarray*}
\end{lemma}

\begin{proof}
Note that%
\begin{equation}
\frac{1}{n}\sum_{i}^{n}\hat{w}_{di}\left( p\right) X_{i}=\frac{1}{n}%
\sum_{i}^{n}\left[ \hat{w}_{di}\left( p\right) -w_{di}\left( p\right) \right]
X_{i}+\frac{1}{n}\sum_{i}^{n}w_{di}\left( p\right) X_{i}.  \label{00000}
\end{equation}%
Similar to the first claim of Lemma \ref{lemma:Brhat-Br}, we can show that%
\begin{equation}
\frac{1}{n}\sum_{i}^{n}\left[ \hat{w}_{di}\left( p\right) -w_{di}\left(
p\right) \right] X_{i}=O_{p}\left( \frac{1}{\sqrt{nh_{3}^{3}}}\right)
=o_{p}\left( 1\right) .  \label{11111}
\end{equation}%
For the second term, we have%
\begin{equation*}
E\left[ \frac{1}{n}\sum_{i}^{n}w_{di}\left( p\right) X_{i}\right] =O\left(
1\right)
\end{equation*}%
and%
\begin{equation*}
Var\left( \frac{1}{n}\sum_{i}^{n}w_{di}\left( p\right) X_{i}\right) =O\left( 
\frac{1}{nh_{3}}\right) =o\left( 1\right) ,
\end{equation*}%
hence,%
\begin{equation}
\frac{1}{n}\sum_{i}^{n}w_{di}\left( p\right) X_{i}=O_{p}\left( 1\right) .
\label{22222}
\end{equation}%
The first claim of the lemma follows from substituting (\ref{11111}) and (%
\ref{22222}) into (\ref{00000}).

For the second claim, it follows from the compactness of the support of $%
k_{3}\left( \cdot \right) $ imposed in Assumption E.2.(iii) that%
\begin{equation*}
\left\Vert \frac{1}{n}\sum_{i}^{n}\hat{w}_{di}\left( p\right) \left( \hat{P}%
_{i}-p\right) X_{i}\right\Vert \leq \frac{1}{n}\sum_{i}^{n}\hat{w}%
_{di}\left( p\right) \left\vert \hat{P}_{i}-p\right\vert \left\Vert
X_{i}\right\Vert =O\left( h_{3}\right) \cdot \frac{1}{n}\sum_{i}^{n}\hat{w}%
_{di}\left( p\right) \left\Vert X_{i}\right\Vert .
\end{equation*}%
As above, we can show that%
\begin{eqnarray*}
\frac{1}{n}\sum_{i}^{n}\hat{w}_{di}\left( p\right) \left\Vert
X_{i}\right\Vert &=&\frac{1}{n}\sum_{i}^{n}\left[ \hat{w}_{di}\left(
p\right) -w_{di}\left( p\right) \right] \left\Vert X_{i}\right\Vert +\frac{1%
}{n}\sum_{i}^{n}w_{di}\left( p\right) \left\Vert X_{i}\right\Vert \\
&=&O_{p}\left( \frac{1}{\sqrt{nh_{3}^{3}}}\right) +O_{p}\left( 1\right)
=O_{p}\left( 1\right) .
\end{eqnarray*}%
Therefore,%
\begin{equation*}
\frac{1}{n}\sum_{i}^{n}\hat{w}_{di}\left( p\right) \left( \hat{P}%
_{i}-p\right) X_{i}=O_{p}\left( h_{3}\right) .
\end{equation*}
\end{proof}

\section{Implementation guidance}

\label{appendix:implement}

Table \ref{table:implement} summarizes our estimation procedure and
preliminary tests for the assumptions. In general, we recommend that the
data cell we choose for testing should have at least 50 observations. For
example, there are 40 among all 60 data cells in Table \ref{table:support}
satisfying this criterion. If we examine the supports of the 40 data cells
one by one using the verification method as in our application, we will find
that Assumption S holds for all of them. Moreover, we plot the conditional
propensity score functions for all the 40 data cells in Figure \ref{Fig:R1},
which verifies that Assumption NL1 (nonmonotonicity) holds for all of them
as well. In other words, we can verify the validity of Assumptions S and NL1
in our application regardless of which data cell we choose. This shows that
Assumptions S and NL1 are quite weak and hardly limit the applicability of
our identification strategy in practice. On the other hand, if the number of
discrete covariates is not small and/or all data cells have less than 50
observations, we should consider to use a semiparametric or parametric
approach that imposes a linear index structure on the treatment function. In
this case, we should remember to include nonlinear terms of continuous
covariates in the linear index to ensure the identification.

\renewcommand{\theequation}{G.\arabic{equation}} \setcounter{equation}{0} %
\renewcommand*{\theHequation}{\theequation}

\section{Simulation}

\label{appendix:simulate}

We examine the finite sample property of the MTE estimation based on our
non-IV identification strategy by a Monte Carlo simulation. The treatment
variable data are generated according to the structural treatment equation (%
\ref{treatment}):%
\begin{equation*}
D_{i}=1\left\{ \mu \left( X_{i}\right) \geq U_{i}\right\} ,\text{ }%
i=1,2,\cdots ,n,
\end{equation*}%
where the function $\mu \left( \cdot \right) $ and the distribution of error 
$U_{i}$ will vary across different designs. The vector $X_{i}$ consists of
one continuous covariate and five discrete covariates, namely, $X_{i}=\left(
X_{i}^{C},X_{1i}^{D},X_{2i}^{D},X_{3i}^{D},X_{4i}^{D},X_{5i}^{D}\right)
^{\prime }$, where $X_{i}^{C}$ follows a standard normal distribution and $%
X_{ji}^{D}$ ($j=1,\cdots ,5$) are independent $\left\{ 0,1\right\} $-valued
binary random variables with identical distribution that $\Pr \left(
X_{ji}^{D}=1\right) =0.5$. The potential and observed outcomes are generated
by%
\begin{eqnarray*}
Y_{1i} &=&\beta _{1}^{C}X_{i}^{C}+\sum_{j=1}^{5}\left( X_{ji}^{D}-0.5\right)
+U_{1i}, \\
Y_{0i} &=&\beta _{0}^{C}X_{i}^{C}+\sum_{j=1}^{5}\frac{1}{6-j}\left(
X_{ji}^{D}-0.5\right) +U_{0i}, \\
Y_{i} &=&D_{i}Y_{1i}+\left( 1-D_{i}\right) Y_{0i},\text{ }i=1,2,\cdots ,n,
\end{eqnarray*}%
where $U_{1i}$ is generated by $U_{1i}=-0.75U_{i}+\sqrt{1-0.75^{2}}\varrho
_{i}$ with a standard normal random variable $\varrho _{i}$ that is
independent of $U_{i}$ and $X_{i}$, and $U_{0i}$ follows an independent
standard normal distribution as well. The sample size $n$ is set 4000.

We consider 16 different designs. In the benchmark design (Design 1), the
structural treatment error $U_{i}$ follows a standard normal distribution
that is independent of $X_{i}$, and the treatment function $\mu \left( \cdot
\right) $ is quadratic in the continuous covairate and linear in the
discrete covariates. Specifically, we set%
\begin{equation}
\mu \left( X_{i}\right) =X_{i}^{C}+\frac{1}{2}\left[ \left( X_{i}^{C}\right)
^{2}-1\right] +\sum_{j=1}^{5}\frac{1}{j}\left( X_{ji}^{D}-0.5\right) .
\label{treat_fn}
\end{equation}%
Note that $E\left[ \mu \left( X_{i}\right) \right] =0$ and the probability
of treatment is about $0.5$. In addition, we impose the exclusion
restriction $\beta _{1}^{C}=\beta _{0}^{C}=0$ in Design 1. Other designs
include cases of no exclusion restriction, linear or cubic treatment
function $\mu \left( \cdot \right) $, nonnormal treatment error $U_{i}$,
nonzero correlation between $U_{i}$ and $X_{i}^{C}$ as in Example \ref%
{example:depend}, heteroskedastic $U_{i}$ or $U_{di}$, nonnormal forms of
MTE, and nonlinear outcome equations. Table \ref{table:designs} collects
detailed descriptions of the designs. For each design, the simulation
replicates 1000 times.

In every replication $b=1,\cdots ,1000$, we compute different estimators for
the MTE evaluated at $x=\bar{X}$, the sample mean of $X_{i}$, and $%
v=0.01,0.02,\cdots ,0.99$. The first estimator is fully parametric, assuming
jointly normally distributed errors and quadratic $\mu \left( \cdot \right) $%
: $\mu \left( X_{i}\right) =\gamma _{0}+\gamma _{1}X_{i}^{C}+\gamma
_{2}\left( X_{i}^{C}\right) ^{2}+\sum_{j=1}^{5}\gamma _{2+j}X_{ji}^{D}$. The
normal assumption leads to the widely-used Heckman's two-step estimation.
Nevertheless, the parametric estimator has a high risk of misspecification.
More flexible estimators comprise of a nonparametric first step (\ref%
{PSfunction}) and a parametric second step. In the first step, we use the
second-order ($s=2$) Gaussian kernel as $k_{1}\left( \cdot \right) $ because
higher-order kernels are likely to result in unstable estimates in practice %
\citep{pan2022semiparametric}, and use the rule-of-thumb bandwidth as given
in Table \ref{table:implement}. In the second step, we consider three
parametric specifications on the unobservable heterogeneity of MTE, namely,
a normal specification that $E\left[ \left. U_{di}\right\vert V_{i}=v\right]
=\rho _{d0}+\rho _{dV}\Phi ^{-1}\left( v\right) $, a linear specification
that $E\left[ \left. U_{di}\right\vert V_{i}=v\right] =\theta _{d0}+\theta
_{d1}v$, and a quadratic specification that $E\left[ \left.
U_{di}\right\vert V_{i}=v\right] =\theta _{d0}+\theta _{d1}v+\theta
_{d2}v^{2}$, where $V_{i}=F_{\left. U\right\vert X}\left( \left.
U_{i}\right\vert X_{i}\right) $ is the reduced-form treatment error. The
most flexible estimator under Assumptions L and NL is the semiparametric one
that comprises of a nonparametric first step and a partially linear second
step, leading to the kernel-weighted pairwise difference estimator (\ref%
{PDLS})-(\ref{MTEhat}). We use the Gaussian kernel and rule-of-thumb
bandwidths as well. Theorem \ref{theorem:AN1} has established the asymptotic
normality of this semiparametric MTE estimator. Recall that all the above
estimators impose no exclusion restriction and base the identification on
different functional forms between the treatment and outcome equations. As a
comparison, we also consider the conventional IV-based MTE estimator under
the semiparametric specification, whose implementation differs from the
semiparametric non-IV MTE estimator only in that it assumes the exclusion
restriction $\beta _{1}^{C}=\beta _{0}^{C}=0$. For each MTE estimator, we
calculate the average over replications, that is, $\left. \sum_{b=1}^{1000}%
\hat{\Delta}_{b}^{\text{MTE}}\left( \bar{X},v\right) \right/ 1000$. And we
depict the averaged MTE curves, along with the true MTE curve, for each
design in Figures \ref{Fig:Design1-4}-\ref{Fig:Design13-16}.

In the benchmark design, the parametric and normal MTE estimators are
correctly specified. Therefore, we can see from the upper-left panel of
Figure \ref{Fig:Design1-4} that the parametric MTE performs perfectly and
the normal MTE comes in second, and both of them are superior to the
semiparametric one which exploits the least information of the model. Since
the exclusion restriction holds in this design, the IV-based MTE is also
consistent and performs well. However, by comparing the two semiparametric
estimators, we find little improvement of the IV-based MTE over the non-IV
MTE, implying that the exclusion restriction provides not much extra
identifying power when the identification based on functional forms holds.
The misspecified linear and quadratic MTE estimators perform satisfactorily
as well, mainly because the true MTE curve is close to linear. When the
exclusion restriction is violated as in Design 2, the upper-right panel of
Figure \ref{Fig:Design1-4} shows that the IV-based MTE estimator
deteriorates drastically due to misspecification. In Design 3, Assumption NL
is violated and the identification based on functional forms fails to hold.
The bottom-left panel of Figure \ref{Fig:Design1-4} confirms that in this
case, the normal, linear, quadratic, and semiparametric non-IV MTE
estimators are underidentified, while the IV-based MTE provides much
improvement over its non-IV counterpart. Somewhat surprisingly, however, the
parametric MTE still performs perfectly, suggesting a potential success of
the classical parametric identification based on functional forms, even when
the semiparametric identification fails. The bottom-right panel of Figure %
\ref{Fig:Design1-4} shows that estimation based on functional forms is more
robust than that based on IV to the violation of identifying conditions.

Besides the curves, we also summarize the simulation results of the MTE
estimators by the bias $\left. \sum_{b=1}^{1000}\left[ \hat{\Delta}_{b}^{%
\text{MTE}}\left( \bar{X},v\right) -\Delta ^{\text{MTE}}\left( \bar{X}%
,v\right) \right] \right/ 1000$, the root mean squared error (RMSE) $\sqrt{%
\left. \sum_{b=1}^{1000}\left[ \hat{\Delta}_{b}^{\text{MTE}}\left( \bar{X}%
,v\right) -\Delta ^{\text{MTE}}\left( \bar{X},v\right) \right] ^{2}\right/
1000}$, and the coverage probability of 95\% confidence interval $\left.
\sum_{b=1}^{1000}1\left\{ \left\vert \hat{\Delta}_{b}^{\text{MTE}}\left( 
\bar{X},v\right) -\Delta ^{\text{MTE}}\left( \bar{X},v\right) \right\vert
\leq 1.96\cdot SE\left( \hat{\Delta}^{\text{MTE}}\left( \bar{X},v\right)
\right) \right\} \right/ 1000$, where $SE\left( \hat{\Delta}^{\text{MTE}%
}\left( \bar{X},v\right) \right) $ is approximated by the simulated standard
deviation of the estimator. Table \ref{table:Design1-4} reports the summary
statistics of $\hat{\Delta}_{b}^{\text{MTE}}\left( \bar{X},0.5\right) $ for
Designs 1-4, along with those of the ATE estimator and of $\hat{\beta}%
_{1}^{C}-\hat{\beta}_{0}^{C}$ that is the estimated partial effect of the
continuous covariate on treatment effects. In Design 1, the coverage
probabilities of all the estimators are close to the true size 95\%, partly
demonstrating the validity of a simple $t$-test when the identification is
attained. In Design 2, the IV-based estimator is heavily biased in both
estimation and inference, while in Design 3, the non-IV estimators are
biased except the fully parametric one, which is consistent with the finding
of Figure \ref{Fig:Design1-4}. The summary statistics in Design 4 show that
estimation based on functional forms is much more robust than that based on
IV in terms of inference as well.

In the remaining designs, we relax the exclusion restriction and no longer
report the results for the IV-based estimator. Design 5 considers a case
that allows for nonzero correlation between the continuous covariate and
structural treatment error but Assumption CMI holds, as in Example \ref%
{example:depend}. The upper-left panel of Figure \ref{Fig:Design5-8} shows
that the parametric MTE curve is still almost completely overlapped with the
true MTE curve, while the other four MTE curves are more biased than in the
uncorrelated case (Designs 1-2). Although biased, the semiparametric MTE
estimator is fairly robust in terms of inference as shown by Table \ref%
{table:Design5-8}. Design 6 is inspired by Example \ref{example:hetero},
where nonlinearity comes from the multiplicative heteroscedasticity in the
treatment error. In this design, the quadratic treatment functional form of
the parametric MTE is wrongly specified. Although the misspecification does
not lead to large bias in estimation, it does lead to apparent invalidation
of inference revealed by the severely biased coverage probability of the
parametric estimator. In comparison, the normal estimator that allows for a
nonparametric treatment equation performs satisfactorily in terms of the
coverage probability. Designs 7 and 8 specify normal polynomial functional
forms on the unobservable heterogeneity of MTE, under which the considered
parametric second steps are all misspecified. Therefore, the semiparametric
estimator that allows for nonparametric functional form of the unobservable
heterogeneity of MTE has the best overall performance in these two designs,
as illustrated by the lower panels of Figure \ref{Fig:Design5-8} and the
right half part of Table \ref{table:Design5-8}.

Figures \ref{Fig:Design9-12} and \ref{Fig:Design13-16} depict the MTE curves
in additional designs. Designs 9 and 10 consider nonnormally distributed
treatment error $U_{i}$, under which the parametric MTE estimator is
misspecified and expected to be inconsistent. However, the upper-left panel
of Figure \ref{Fig:Design9-12} shows that the parametric MTE still works in
Design 9. This is mainly because the MTE estimator depends only on the
propensity score (namely, the conditional expectation of treatment) but not
on regression coefficients of the treatment equation. The estimation of the
propensity score is actually a prediction problem. Some forms of
misspecification may have a less serious impact on the prediction than on
the regression. In Design 9, the $t\left( 3\right) $ distribution of the
treatment error is symmetric about zero and resembles the normal
distribution, thus exerting only a small influence on the parametric MTE. In
comparison, the $t\left( 2\right) $ distributed treatment error in Design 10
is more fat-tailed and therefore leads to more biased parametric MTE. The
nonnormally distributed treatment error also results in nonnormal
unobservable heterogeneity of MTE. As in Designs 7 and 8, the semiparametric
MTE performs the best overall in such case. Design 11 considers a cubic
treament function $\mu \left( \cdot \right) $, under which the parametric
MTE is notably biased due to severe misspecification. The last panel of
Figure \ref{Fig:Design9-12} for Design 12 shows that the heteroskedasticity
in outcome errors makes little difference to the performance of the MTE
estimators by comparing with the benchmark design.

Figure \ref{Fig:Design13-16} investigates the impact of nonlinearity in the
outcome equation, and the results are also as expected. That is, if we
correctly include the nonlinear terms during implementation, the MTE
estimator still performs well as shown by the dotted lines that stand for
the correctly specified semiparametric estimator. However, if we omit the
nonlinear terms in the outcome equation, the MTE estimator may be biased,
and the magnitude of bias depends on the specific forms of nonlinearity.
When the nonlinear term is an interaction of two discrete covariates as in
Design 13, the omission of it has quite limited influence on the MTE
estimators. Nevertheless, the omitted variable problem will become obvious
if the interaction term pertains to the continuous covariate as in Design
14. The lower panels of Figure \ref{Fig:Design13-16} further show that
omitting a quadratic term or an exponential term in the outcome equation
will induce more substantial bias in the estimation of MTE.

In summary, the fully parametric estimator and the estimators that specify
nonparametric first step and parametric second step are the first choice
when we are confident that the specification is correct or at least nearly
correct. Otherwise, the semiparametric estimator is perferred because its
performance is the most robust across different designs.

\renewcommand{\thetable}{F.\arabic{table}} \setcounter{table}{0}
\renewcommand*{\theHtable}{\thetable}
\renewcommand{\thefigure}{F.\arabic{figure}} \setcounter{figure}{0}
\renewcommand*{\theHfigure}{\thefigure}

\clearpage

\begin{sidewaystable}[htb]
\caption{An overview of the testing and estimation procedures}
\label{table:implement}
\medskip \renewcommand\arraystretch{1.2}
\resizebox{\linewidth}{!}{
\begin{tabular}{lrL{27cm}}
\hline\hline
\multicolumn{3}{l}{\textbf{Step 0}: Test Assumptions S and NL} \\ \cline{2-3}
& \textbf{Step 0.1}: & List the data cells for all possible values of
discrete covariates and the corresponding conditional supports of the
propensity score, as in Table \ref{table:support}. \\ 
& \textbf{Step 0.2}: & Find a data cell that satisfies Assumption S and has
at least 50 observations. \\ 
&  & \textit{Note}: A method of quick validation of Assumption S is to check
whether the present data cell has full support or whether its support is
overlapped with all other data cells' supports. \\ 
& \textbf{Step 0.3}: & Test Assumption NL for the data cell selected by Step
0.2. \\ 
&  & \textit{If there is only one continuous covariate}, then plot the
propensity score function for the selected data cell as in Figure \ref%
{Fig:AssumptionNL1}, and visually inspect whether the function is
nonmonotonic. \\ 
&  & \textit{If there are at least two continuous covariates}, then test the
nonlinearity of the propensity score function by methods from the literature
on single-index models \citep[e.g.,][]{fan1996consistent,
stute2005nonparametric, xia2009model, maistre2019nonparametric}. \\ 
& \textbf{Step 0.4}: & If Assumption NL holds, go to Step 1. Otherwise,
return to Step 0.2-0.3 until finding one data cell for which Assumptions S
and NL both hold. \\ \hline
\multicolumn{3}{l}{\textbf{Step 1}: Estimate the treatment equation $%
D=1\left\{ \mu \left( X\right) \geq U\right\} =1\left\{ \pi \left( X\right)
\geq V\right\} $} \\ \cline{2-3}
& \textbf{Case 1.1}: & Nonparametric (kernel) estimation. Suppose that both
the functional form of $\mu \left( \cdot \right) $ and the distribution of
the structural treatment error $U$ are nonparametrically specified. Then,
estimate the propensity score $\pi \left( X\right) $ by the kernel estimator
(\ref{PSfunction}). \\ 
&  & \textit{Note}: We recommend to use the Gaussian kernel (i.e., standard
normal density) as $k_{1}\left( \cdot \right) $ and the rule-of-thumb
bandwidth $h_{1}=sd\left( X^{C}\right) \cdot n^{-1\left/ \left( 4+\dim
\left( X^{C}\right) \right) \right. }$, where $sd\left( X^{C}\right) $ is
the sample standard deviation of the continuous covariate. \\ 
& \textbf{Case 1.2}: & Semiparametric (single-index) estimation. Suppose
that $\mu \left( \cdot \right) $ is parametrically specified as a linear
index and the distribution of $U$ is nonparametrically specified. Then,
estimate $\pi \left( X\right) $ by methods from the literature on
single-index models \citep[e.g.,][]{powell1989semiparametric,
ichimura1993semiparametric, klein1993efficient, lewbel2000semiparametric}.
\\ 
& \textbf{Case 1.3}: & Parametric (Probit/Logit) estimation. Suppose that $%
\mu \left( \cdot \right) $ is parametrically specified and $U$ is assumed to
follow a normal/logistic distribution. Then, perform a Probit/Logit
regression on the treatment equation and predict $\pi \left( X\right) $. \\ 
\hline
\multicolumn{3}{l}{\textbf{Step 2}: Estimate the outcome equation (\ref{md})
and MTE (\ref{MTEiden})} \\ \cline{2-3}
& \textbf{Case 2.1}: & Semiparametric (partially linear) estimation. Suppose
that the functional form of $g_{d}\left( \cdot \right) $ is
nonparametrically specified.\ Then, estimate $\beta _{d}$ by the
kernel-weighted pairwise difference estimator (\ref{PDLS}) and $g_{d}\left(
\cdot \right) $ by the local linear estimator. \\ 
&  & \textit{Note}: We recommend to use the Gaussian kernel and
rule-of-thumb bandwidths $h_{2}=h_{3}=sd\left( \hat{P}\right) \cdot
n^{-1\left/ 5\right. }$ when estimating $\beta _{d}$ and $g_{d}\left( \cdot
\right) $. \\ 
& \textbf{Case 2.2}: & Parametric estimation. Suppose that $g_{d}\left(
\cdot \right) $ is parametrically specified as $g_{d}\left( \cdot ;\theta
_{d}\right) $ for a finite-dimensional parameter $\theta _{d}$. We consider
three commonly used specifications. \\ 
&  & \textbf{Case 2.2.1}: Polynomial specification. Suppose that $%
g_{d}\left( \cdot \right) $ is a polynomial function. Then, regress $Y$ on $%
X $, $\hat{P}$, and the powers of $\hat{P}$ for the treated and untreated
groups to estimate $\beta _{d}$ and $g_{d}\left( \cdot \right) $. \\ 
&  & \textbf{Case 2.2.2}: Normal specification. Suppose that $\left(
U,U_{d}\right) $ follows a bivariate normal distribution, which implies $%
g_{d}\left( p\right) =\theta _{1d}\left. \phi \left( \Phi ^{-1}\left(
p\right) \right) \right/ \left( p-1+d\right) $. Therefore, regress $Y$ on $X$
and $\left. \phi \left( \Phi ^{-1}\left( \hat{P}\right) \right) \right/
\left( \hat{P}-1+d\right) $ for each group $d=0,1$ to estimate $\beta _{d}$
and $g_{d}\left( \cdot \right) $. \\ 
&  & \textbf{Case 2.2.3}: Normal polynomial specification, which includes
the normal specification as a special case when the order of polynomial is
one. If the order is set two, then add a term $\Phi ^{-1}\left( \hat{P}%
\right) \left. \phi \left( \Phi ^{-1}\left( \hat{P}\right) \right) \right/
\left( \hat{P}-1+d\right) $ to the above regression. If the order is set
three, further add $\left( \Phi ^{-2}\left( \hat{P}\right) +2\right) \left.
\phi \left( \Phi ^{-1}\left( \hat{P}\right) \right) \right/ \left( \hat{P}%
-1+d\right) $. If the order is set four, further add $\left( \Phi
^{-3}\left( \hat{P}\right) +3\Phi ^{-1}\left( \hat{P}\right) \right) \left.
\phi \left( \Phi ^{-1}\left( \hat{P}\right) \right) \right/ \left( \hat{P}%
-1+d\right) $, and so on. \\ 
& \textbf{Step 2.3}: & Estimate the MTE. After obtaining the estimates of $%
\beta _{d}$ and $g_{d}\left( \cdot \right) $, plug them into the equation (%
\ref{MTEhat}) to construct the estimator for the MTE. \\ \hline\hline
\end{tabular}}
\end{sidewaystable}

\clearpage

\begin{table}[htb]
\caption{Support of the nonparametrically estimated propensity score in each
data cell}
\label{table:support}
\medskip \renewcommand\arraystretch{1.1} {\small 
\begin{tabular}{ccccccc}
\hline\hline
Data cell & Age & Gender & Race & Parental education & Observations & 
Support of $\hat{\pi} _{0}\left( X^{C}\right) $ \\ \hline
1 & 30 & Male & Hispanic & High school or lower & 70 & [0, 0.208] \\ 
2 & 30 & Male & Hispanic & College or higher & 10 & [0, 0.288] \\ 
3 & 30 & Male & Black & High school or lower & 93 & [0, 0.588] \\ 
4 & 30 & Male & Black & College or higher & 21 & [0, 0.813] \\ 
5 & 30 & Male & White & High school or lower & 147 & [0, 0.177] \\ 
6 & 30 & Male & White & College or higher & 78 & [0, 0.208] \\ 
7 & 30 & Female & Hispanic & High school or lower & 65 & [0, 0.434] \\ 
8 & 30 & Female & Hispanic & College or higher & 11 & 0 \\ 
9 & 30 & Female & Black & High school or lower & 100 & [0, 0.671] \\ 
10 & 30 & Female & Black & College or higher & 20 & [0, 0.641] \\ 
11 & 30 & Female & White & High school or lower & 125 & [0, 0.759] \\ 
12 & 30 & Female & White & College or higher & 73 & [0, 0.179] \\ 
13 & 31 & Male & Hispanic & High school or lower & 87 & [0, 0.374] \\ 
14 & 31 & Male & Hispanic & College or higher & 19 & [0, 0.442] \\ 
15 & 31 & Male & Black & High school or lower & 112 & [0.030, 0.548] \\ 
16 & 31 & Male & Black & College or higher & 24 & [0, 0.531] \\ 
17 & 31 & Male & White & High school or lower & 148 & [0, 0.242] \\ 
18 & 31 & Male & White & College or higher & 92 & [0, 0.125] \\ 
19 & 31 & Female & Hispanic & High school or lower & 97 & [0, 0.737] \\ 
20 & 31 & Female & Hispanic & College or higher & 14 & [0, 0.788] \\ 
21 & 31 & Female & Black & High school or lower & 121 & [0.249, 0.817] \\ 
22 & 31 & Female & Black & College or higher & 24 & [0, 0.639] \\ 
23 & 31 & Female & White & High school or lower & 149 & [0, 0.178] \\ 
24 & 31 & Female & White & College or higher & 75 & [0, 0.221] \\ 
25 & 32 & Male & Hispanic & High school or lower & 65 & [0, 0.771] \\ 
26 & 32 & Male & Hispanic & College or higher & 13 & 0 \\ 
27 & 32 & Male & Black & High school or lower & 133 & [0, 0.623] \\ 
28 & 32 & Male & Black & College or higher & 23 & [0, 0.949] \\ 
29 & 32 & Male & White & High school or lower & 175 & [0, 0.194] \\ 
30 & 32 & Male & White & College or higher & 90 & [0, 0.061] \\ \hline\hline
\end{tabular}
}
\end{table}

\renewcommand{\thetable}{F.\arabic{table} (cont.)} \setcounter{table}{1}

\begin{table}[htb]
\caption{Support of the nonparametrically estimated propensity score in each
data cell}
\medskip \renewcommand\arraystretch{1.1} {\small 
\begin{tabular}{ccccccc}
\hline\hline
Data cell & Age & Gender & Race & Parental education & Observations & 
Support of $\hat{\pi} _{0}\left( X^{C}\right) $ \\ \hline
31 & 32 & Female & Hispanic & High school or lower & 94 & [0, 0.642] \\ 
32 & 32 & Female & Hispanic & College or higher & 17 & [0, 0.351] \\ 
33 & 32 & Female & Black & High school or lower & 130 & [0.021, 0.999] \\ 
34 & 32 & Female & Black & College or higher & 34 & [0, 1] \\ 
35 & 32 & Female & White & High school or lower & 155 & [0, 0.526] \\ 
36 & 32 & Female & White & College or higher & 81 & [0, 0.118] \\ 
37 & 33 & Male & Hispanic & High school or lower & 72 & [0, 0.567] \\ 
38 & 33 & Male & Hispanic & College or higher & 15 & [0, 0.743] \\ 
39 & 33 & Male & Black & High school or lower & 123 & [0, 0.818] \\ 
40 & 33 & Male & Black & College or higher & 25 & [0, 0.596] \\ 
41 & 33 & Male & White & High school or lower & 141 & [0, 0.154] \\ 
42 & 33 & Male & White & College or higher & 95 & [0, 0.133] \\ 
43 & 33 & Female & Hispanic & High school or lower & 83 & [0, 0.481] \\ 
44 & 33 & Female & Hispanic & College or higher & 10 & [0, 0.447] \\ 
45 & 33 & Female & Black & High school or lower & 125 & [0.258, 0.628] \\ 
46 & 33 & Female & Black & College or higher & 27 & [0, 0.925] \\ 
47 & 33 & Female & White & High school or lower & 152 & [0, 0.140] \\ 
48 & 33 & Female & White & College or higher & 93 & [0, 0.350] \\ 
49 & 34 & Male & Hispanic & High school or lower & 85 & [0, 0.864] \\ 
50 & 34 & Male & Hispanic & College or higher & 16 & [0, 0.673] \\ 
51 & 34 & Male & Black & High school or lower & 101 & [0.176, 1] \\ 
52 & 34 & Male & Black & College or higher & 27 & [0, 0.790] \\ 
53 & 34 & Male & White & High school or lower & 162 & [0, 0.204] \\ 
54 & 34 & Male & White & College or higher & 85 & [0, 0.057] \\ 
55 & 34 & Female & Hispanic & High school or lower & 73 & [0, 0.443] \\ 
56 & 34 & Female & Hispanic & College or higher & 11 & 0 \\ 
57 & 34 & Female & Black & High school or lower & 109 & [0.271, 0.586] \\ 
58 & 34 & Female & Black & College or higher & 24 & [0, 0.856] \\ 
59 & 34 & Female & White & High school or lower & 143 & [0, 0.5] \\ 
60 & 34 & Female & White & College or higher & 76 & [0, 0.055] \\ 
\hline\hline
\end{tabular}
}
\end{table}

\begin{figure}[tbh]
\caption{Test of Assumption NL1 for all data cells with size larger than 50}
\label{Fig:R1}
\bigskip 
\begin{minipage}[t]{0.245\linewidth}
  \centerline{\includegraphics[width=3.813cm,height=2.773cm]{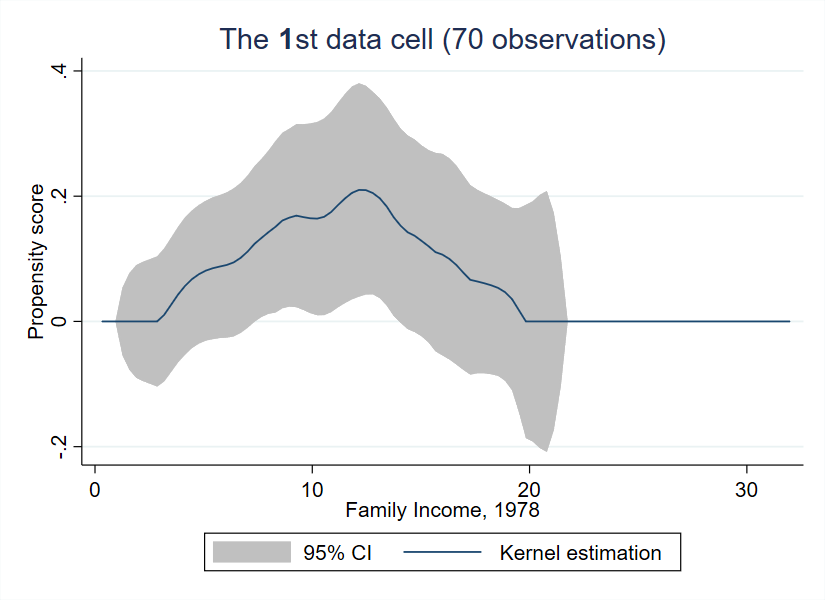}}
\end{minipage}
\begin{minipage}[t]{0.245\linewidth}
  \centerline{\includegraphics[width=3.813cm,height=2.773cm]{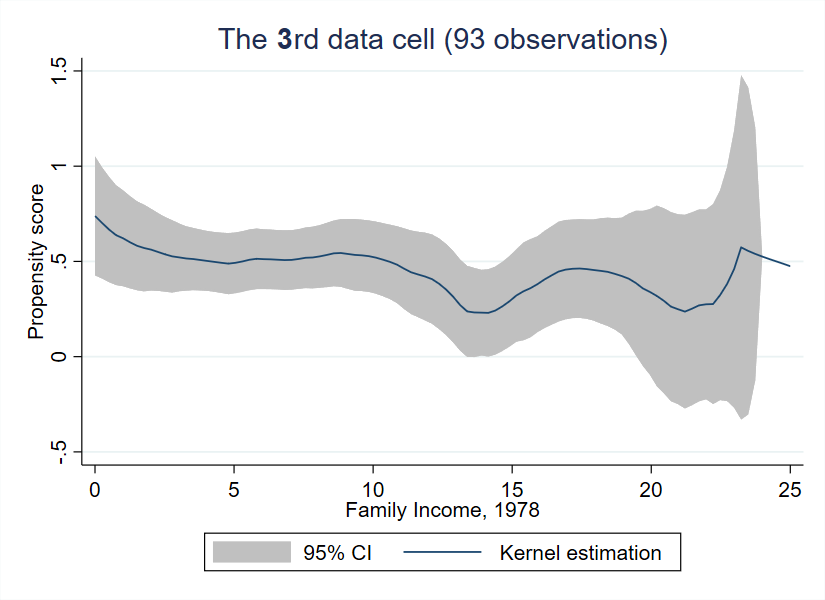}}
\end{minipage}
\begin{minipage}[t]{0.245\linewidth}
  \centerline{\includegraphics[width=3.813cm,height=2.773cm]{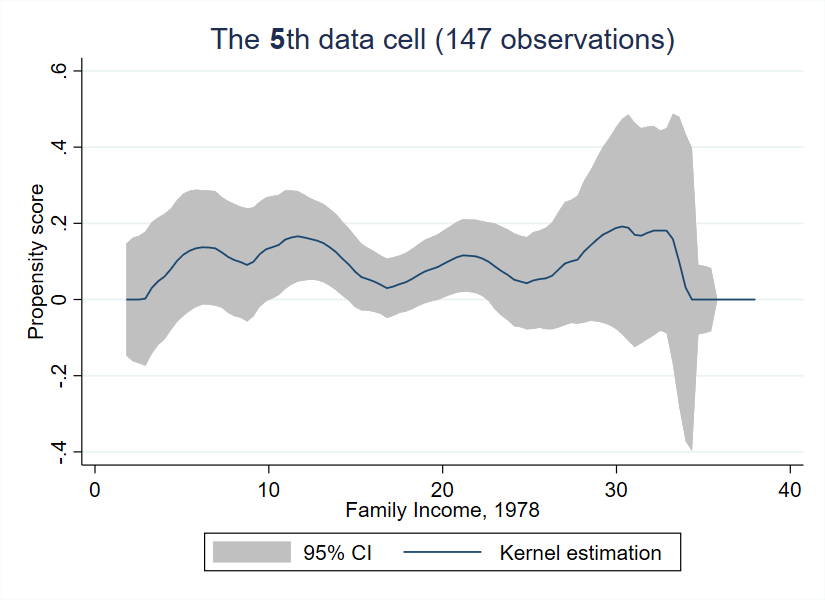}}
\end{minipage}
\begin{minipage}[t]{0.245\linewidth}
  \centerline{\includegraphics[width=3.813cm,height=2.773cm]{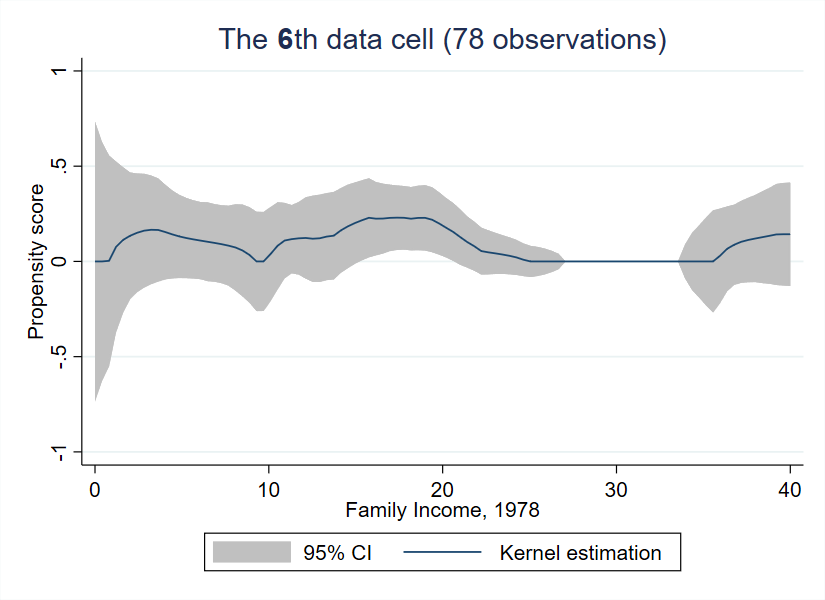}}
\end{minipage}
\begin{minipage}[t]{0.245\linewidth}
  \centerline{\includegraphics[width=3.813cm,height=2.773cm]{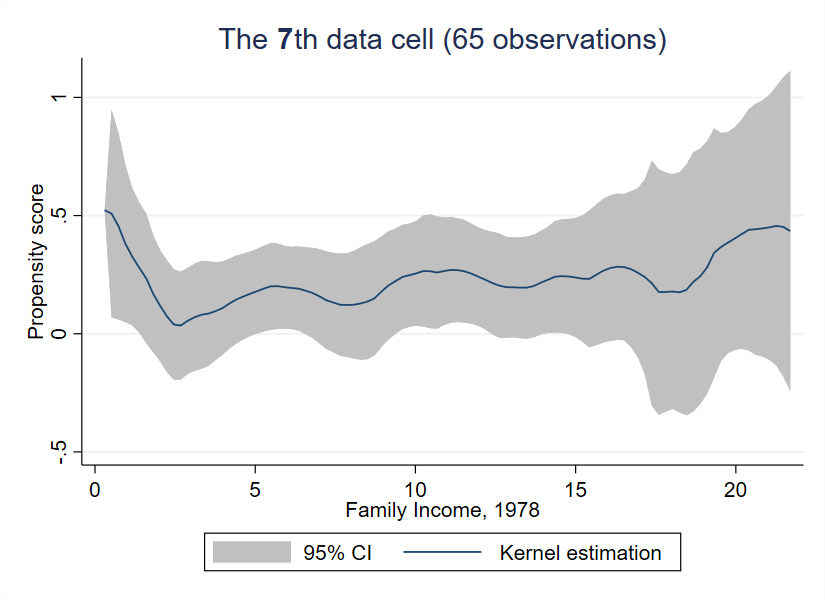}}
\end{minipage}
\begin{minipage}[t]{0.245\linewidth}
  \centerline{\includegraphics[width=3.813cm,height=2.773cm]{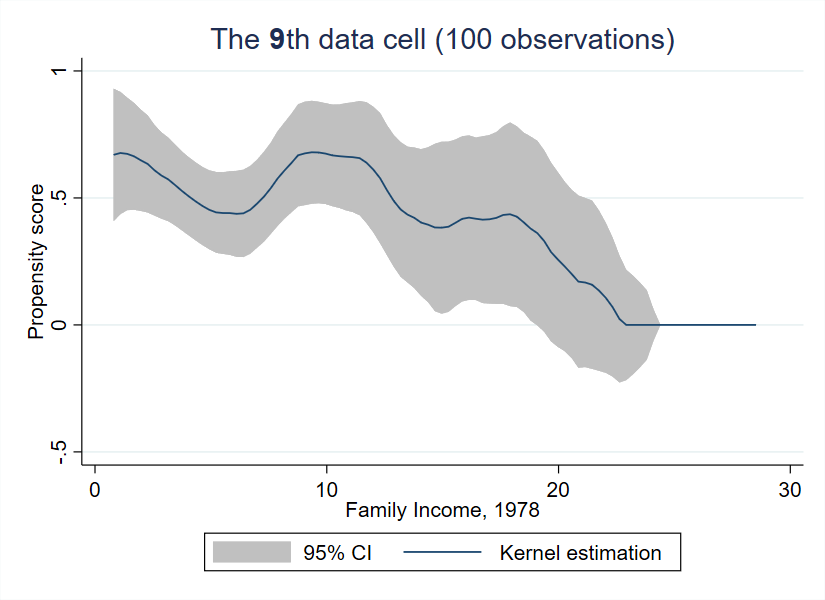}}
\end{minipage}
\begin{minipage}[t]{0.245\linewidth}
  \centerline{\includegraphics[width=3.813cm,height=2.773cm]{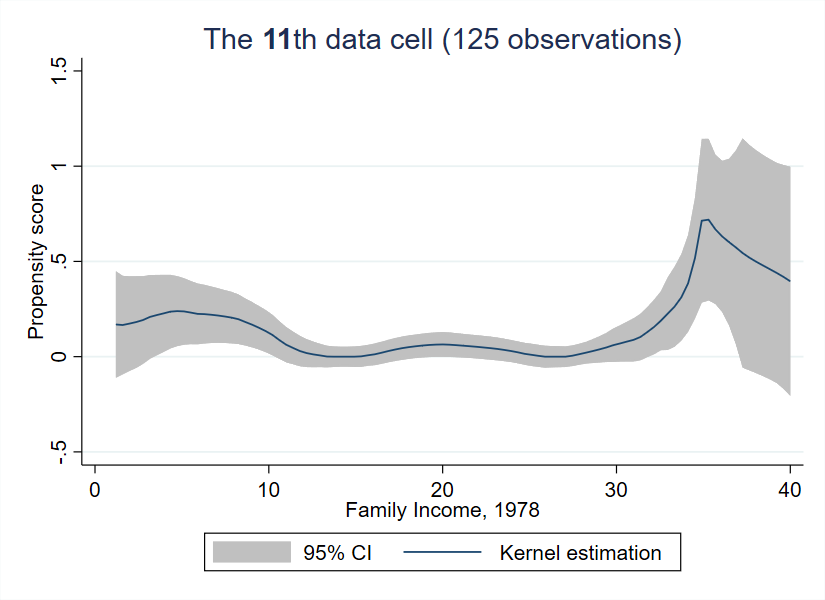}}
\end{minipage}
\begin{minipage}[t]{0.245\linewidth}
  \centerline{\includegraphics[width=3.813cm,height=2.773cm]{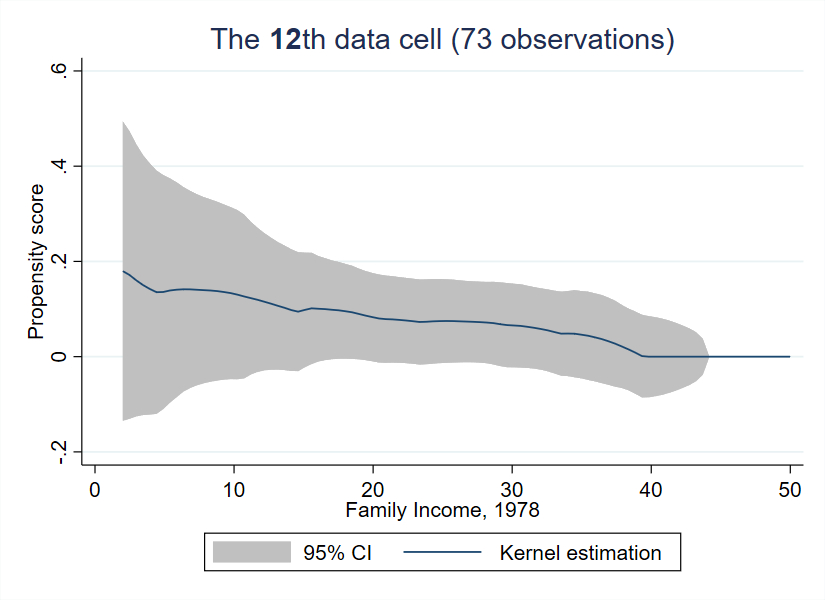}}
\end{minipage}
\begin{minipage}[t]{0.245\linewidth}
  \centerline{\includegraphics[width=3.813cm,height=2.773cm]{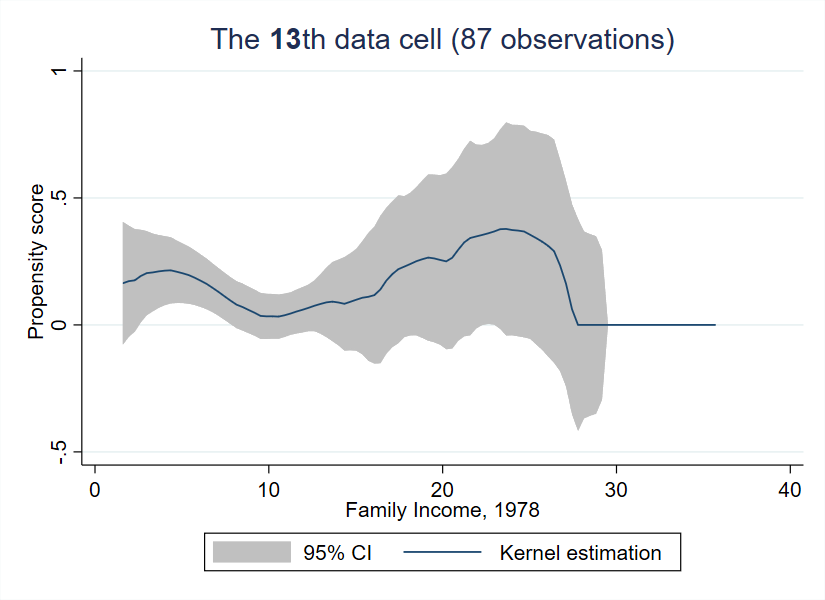}}
\end{minipage}
\begin{minipage}[t]{0.245\linewidth}
  \centerline{\includegraphics[width=3.813cm,height=2.773cm]{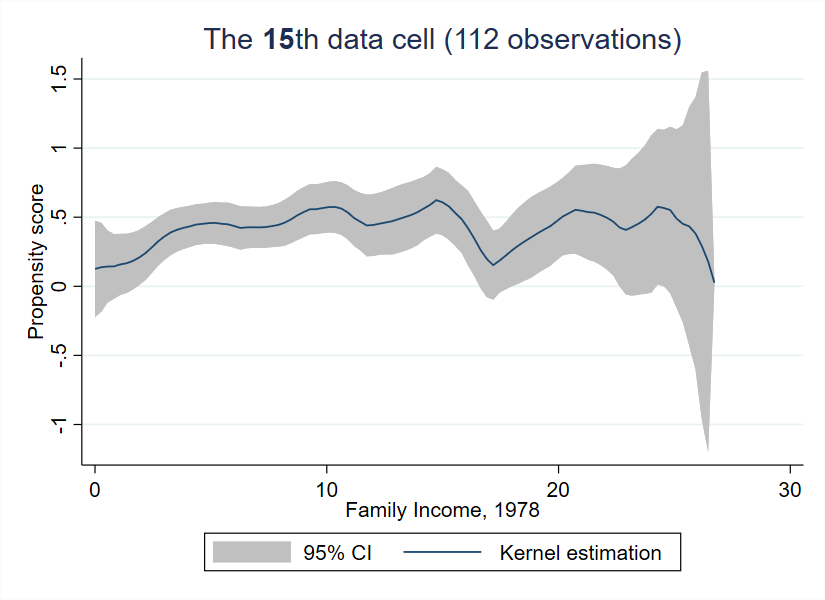}}
\end{minipage}
\begin{minipage}[t]{0.245\linewidth}
  \centerline{\includegraphics[width=3.813cm,height=2.773cm]{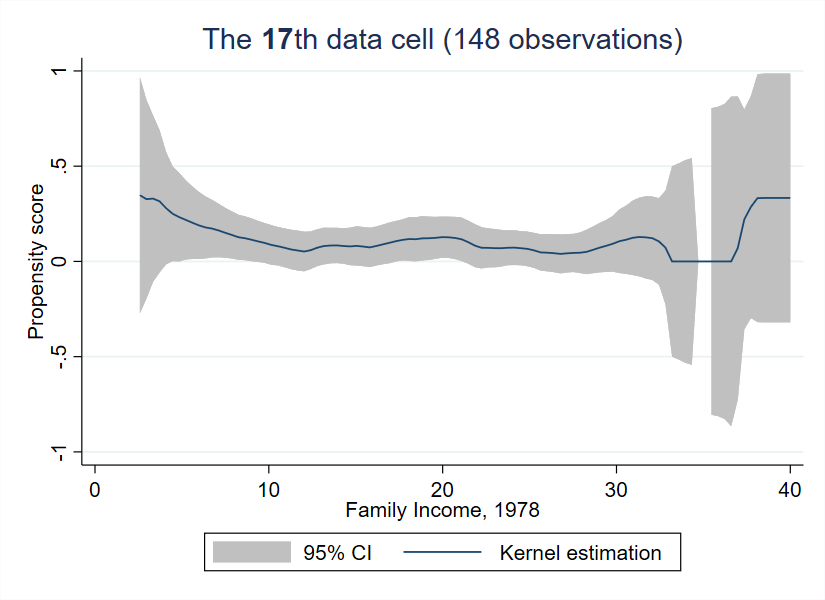}}
\end{minipage}
\begin{minipage}[t]{0.245\linewidth}
  \centerline{\includegraphics[width=3.813cm,height=2.773cm]{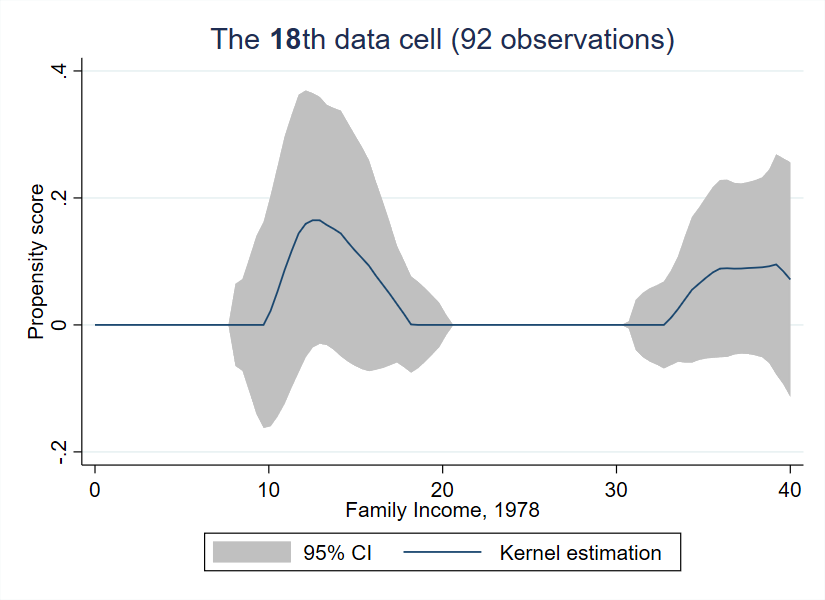}}
\end{minipage}
\begin{minipage}[t]{0.245\linewidth}
  \centerline{\includegraphics[width=3.813cm,height=2.773cm]{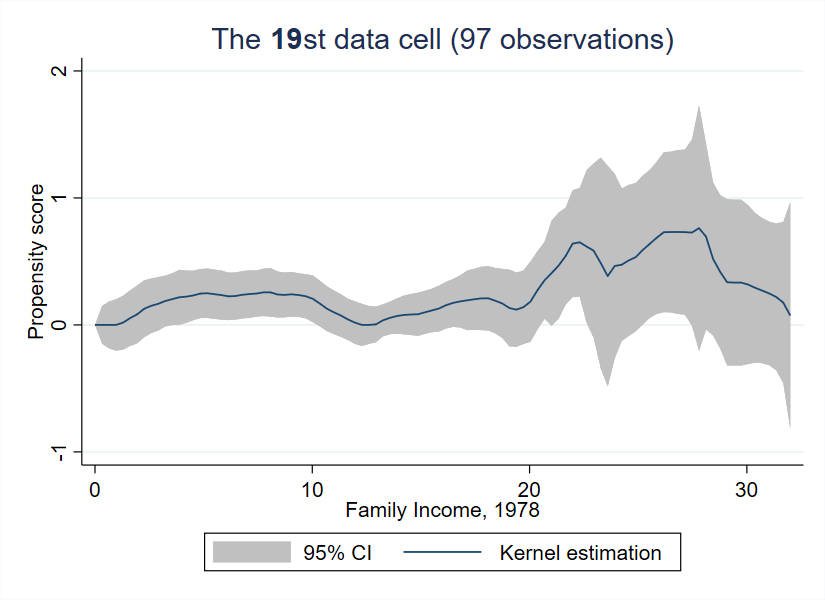}}
\end{minipage}
\begin{minipage}[t]{0.245\linewidth}
  \centerline{\includegraphics[width=3.813cm,height=2.773cm]{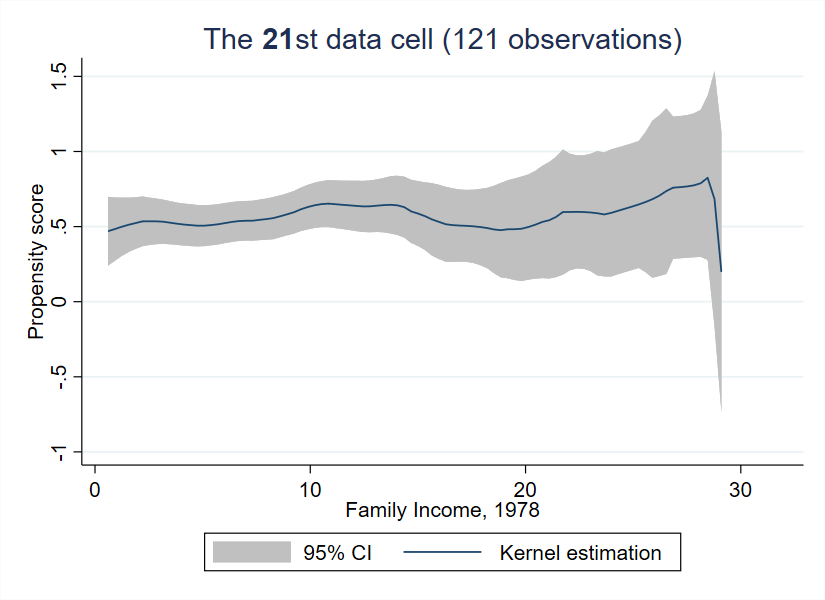}}
\end{minipage}
\begin{minipage}[t]{0.245\linewidth}
  \centerline{\includegraphics[width=3.813cm,height=2.773cm]{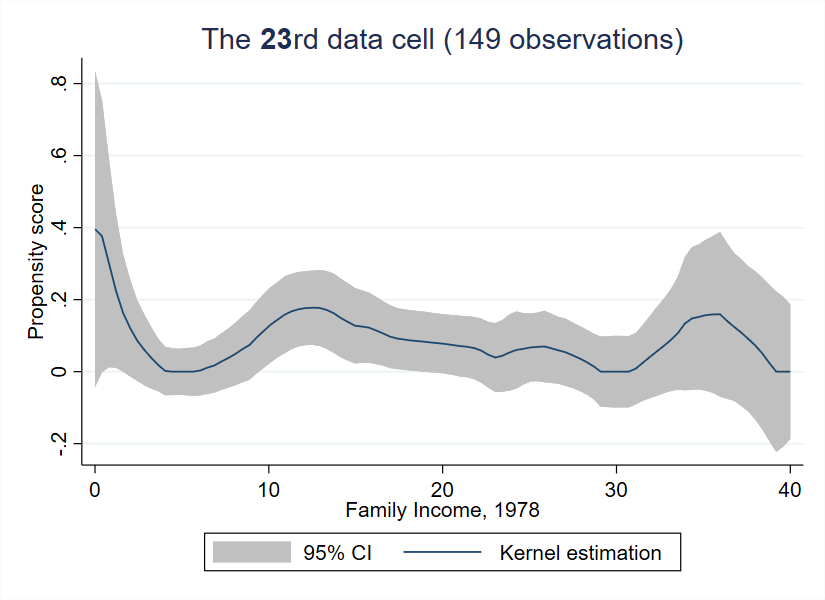}}
\end{minipage}
\begin{minipage}[t]{0.245\linewidth}
  \centerline{\includegraphics[width=3.813cm,height=2.773cm]{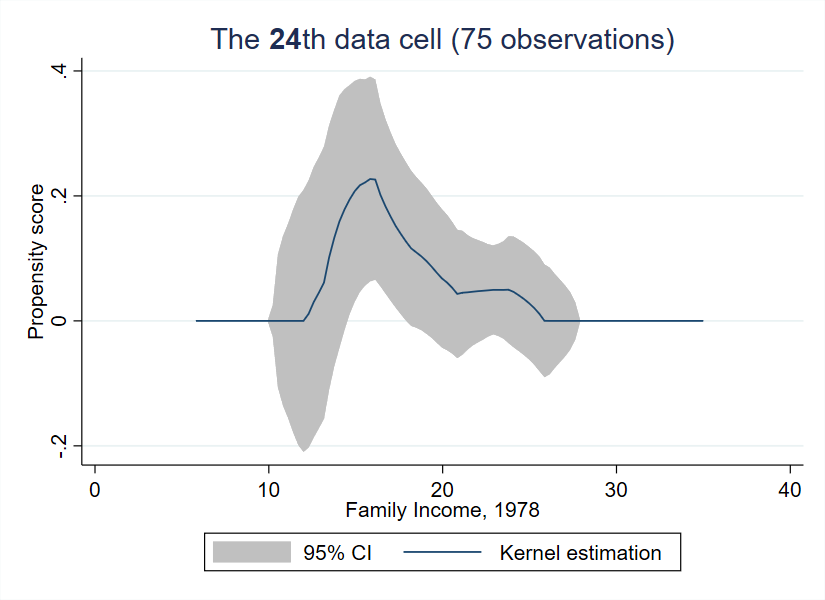}}
\end{minipage}
\begin{minipage}[t]{0.245\linewidth}
  \centerline{\includegraphics[width=3.813cm,height=2.773cm]{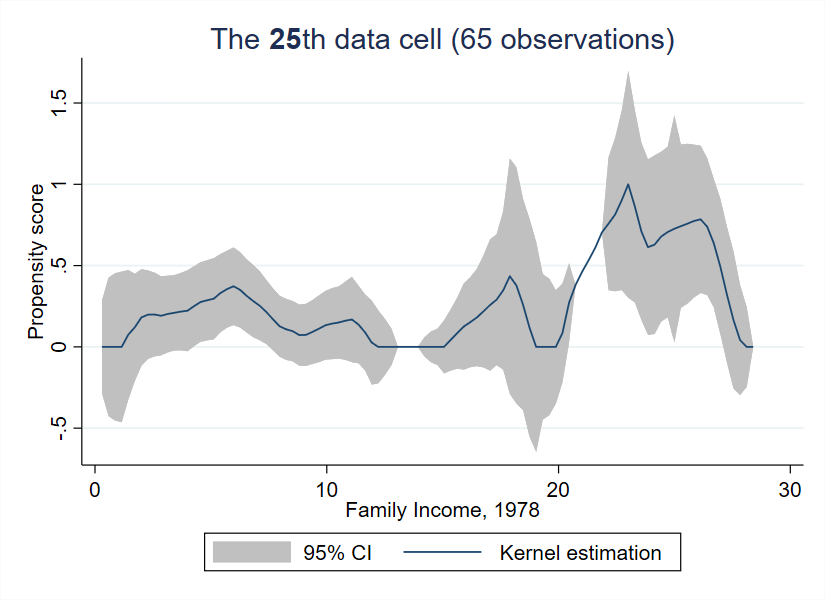}}
\end{minipage}
\begin{minipage}[t]{0.245\linewidth}
  \centerline{\includegraphics[width=3.813cm,height=2.773cm]{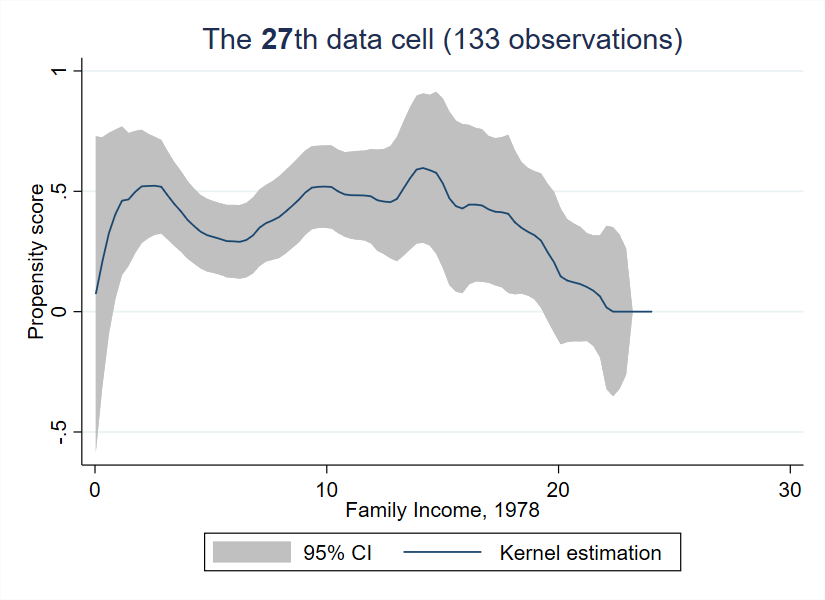}}
\end{minipage}
\begin{minipage}[t]{0.245\linewidth}
  \centerline{\includegraphics[width=3.813cm,height=2.773cm]{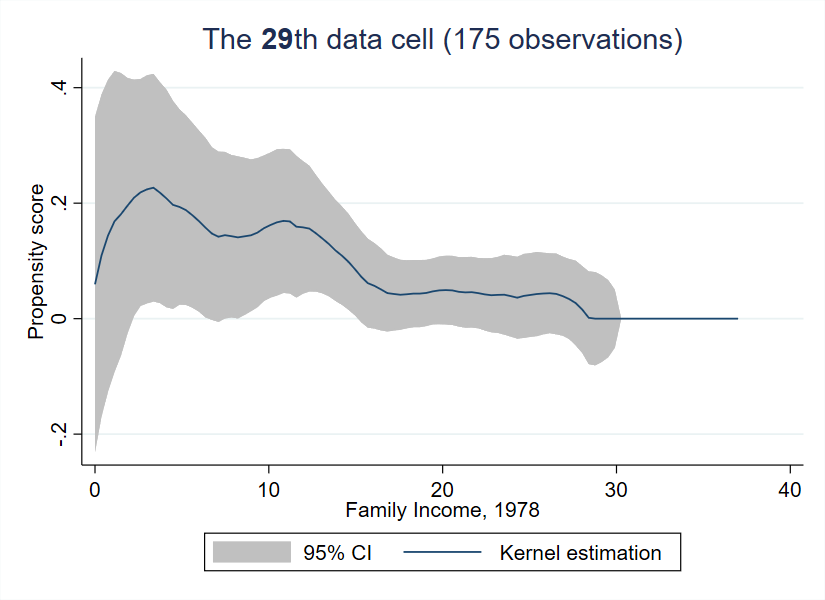}}
\end{minipage}
\begin{minipage}[t]{0.245\linewidth}
  \centerline{\includegraphics[width=3.813cm,height=2.773cm]{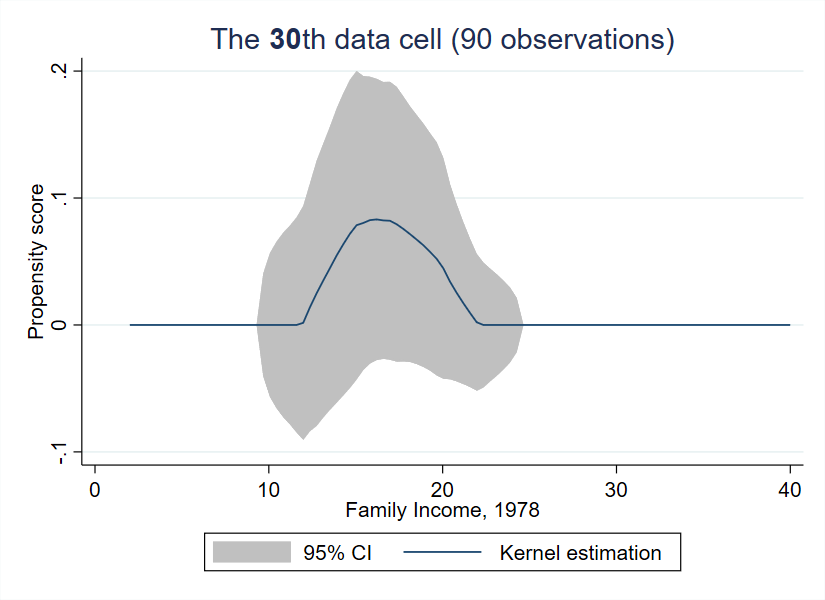}}
\end{minipage}
\begin{minipage}[t]{0.245\linewidth}
  \centerline{\includegraphics[width=3.813cm,height=2.773cm]{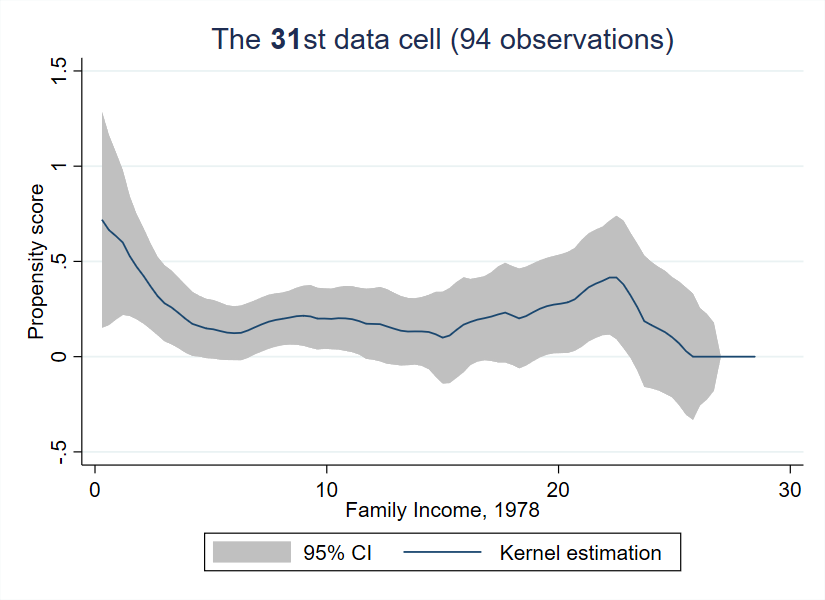}}
\end{minipage}
\begin{minipage}[t]{0.245\linewidth}
  \centerline{\includegraphics[width=3.813cm,height=2.773cm]{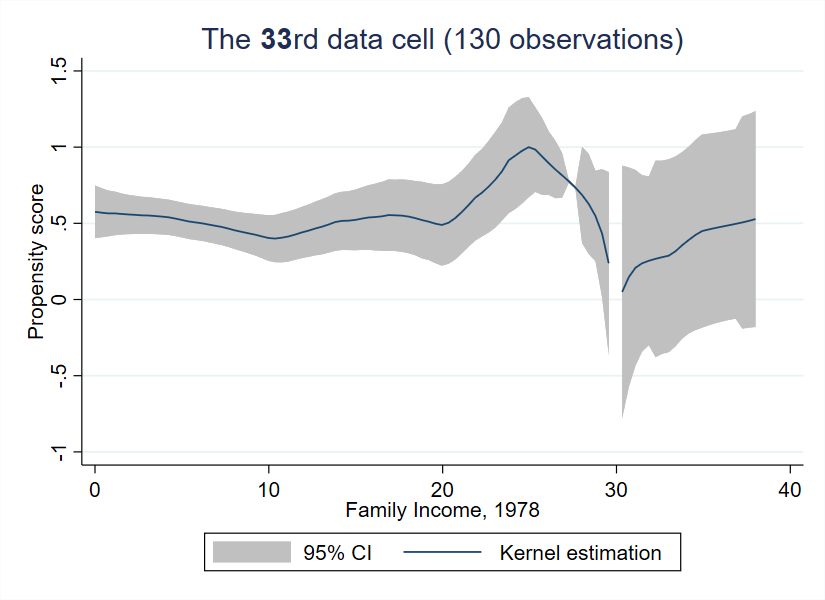}}
\end{minipage}
\begin{minipage}[t]{0.245\linewidth}
  \centerline{\includegraphics[width=3.813cm,height=2.773cm]{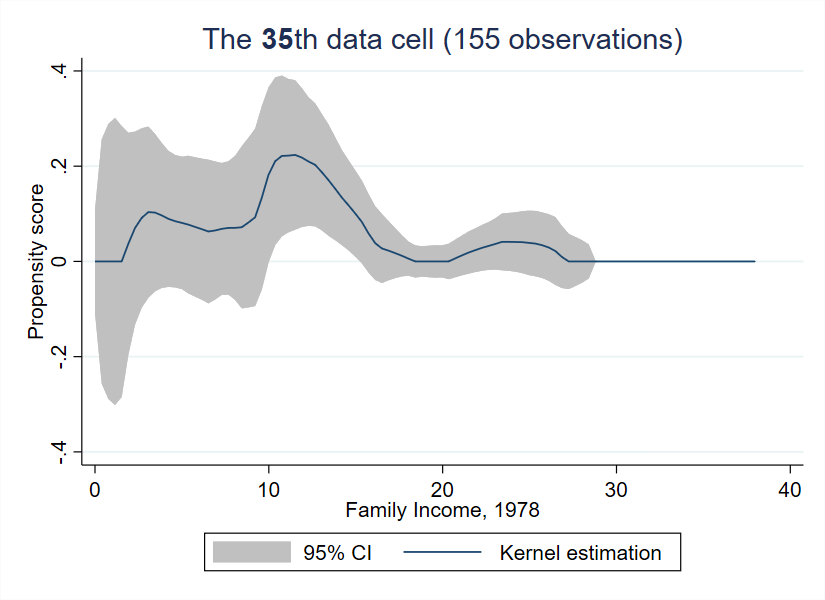}}
\end{minipage}
\begin{minipage}[t]{0.245\linewidth}
  \centerline{\includegraphics[width=3.813cm,height=2.773cm]{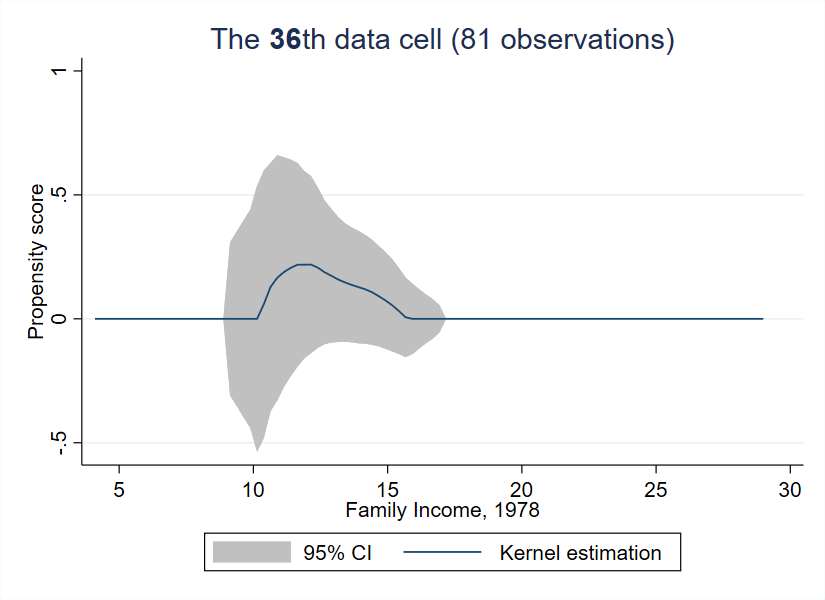}}
\end{minipage}
\begin{minipage}[t]{0.245\linewidth}
  \centerline{\includegraphics[width=3.813cm,height=2.773cm]{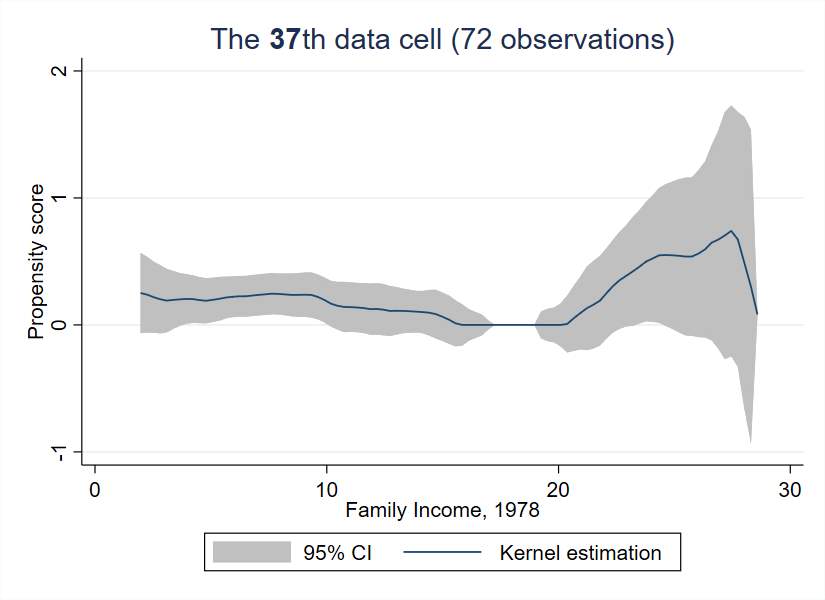}}
\end{minipage}
\begin{minipage}[t]{0.245\linewidth}
  \centerline{\includegraphics[width=3.813cm,height=2.773cm]{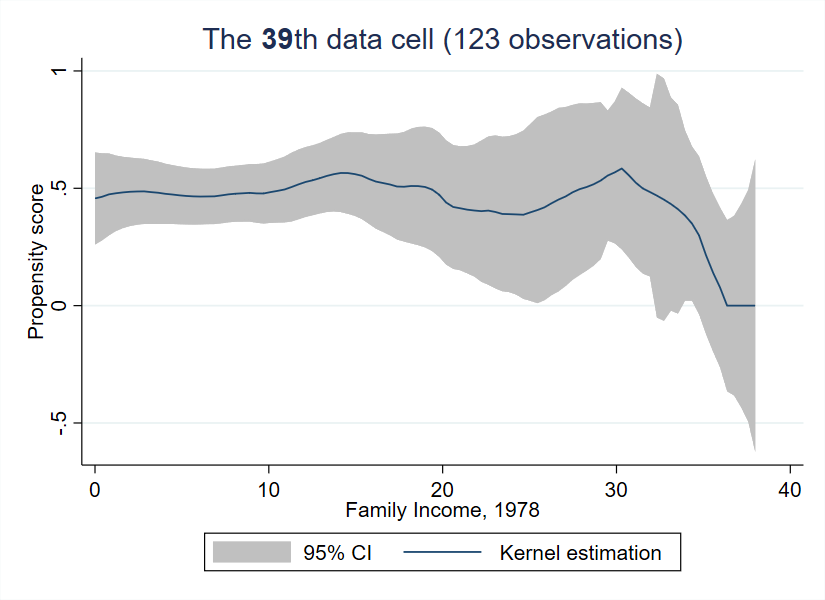}}
\end{minipage}
\begin{minipage}[t]{0.245\linewidth}
  \centerline{\includegraphics[width=3.813cm,height=2.773cm]{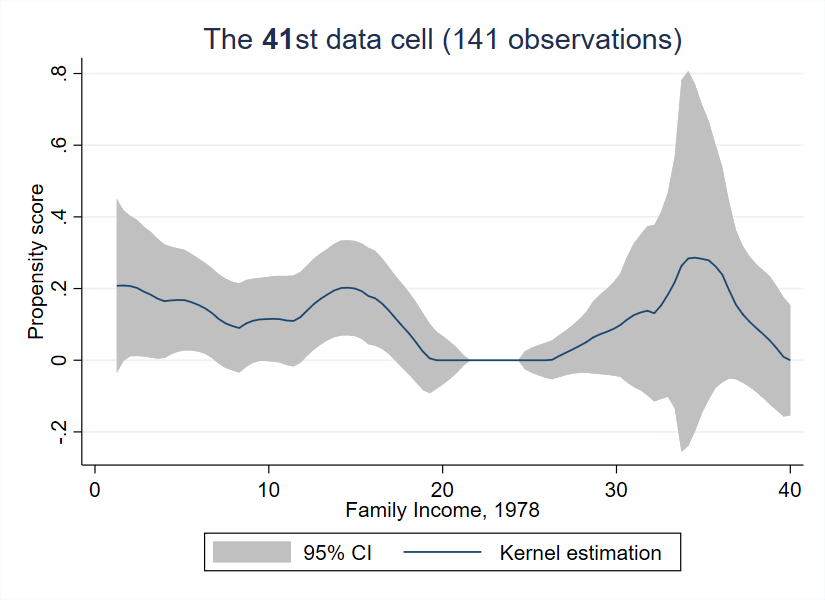}}
\end{minipage}
\begin{minipage}[t]{0.245\linewidth}
  \centerline{\includegraphics[width=3.813cm,height=2.773cm]{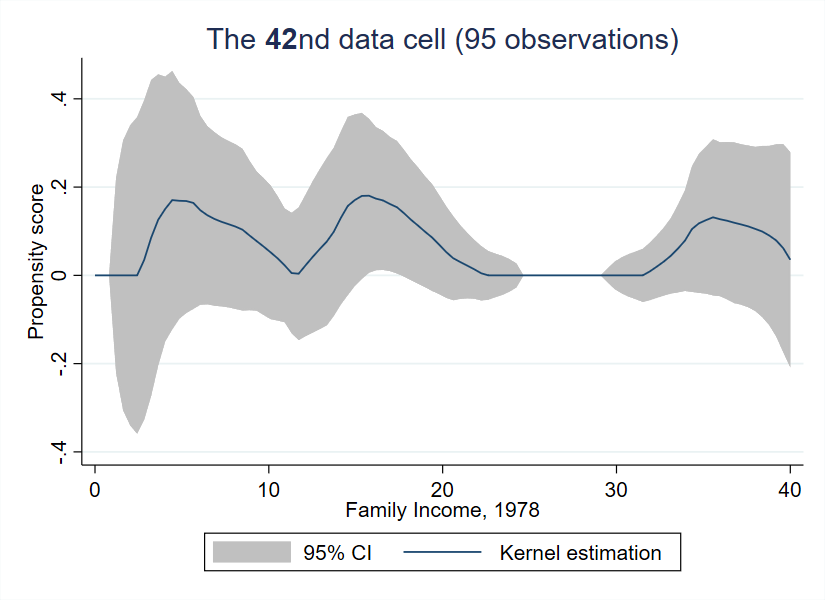}}
\end{minipage}
\end{figure}

\renewcommand{\thefigure}{F.\arabic{figure} (cont.)} \setcounter{figure}{0} 
\renewcommand*{\theHfigure}{\thefigure}

\begin{figure}[tbh]
\caption{Test of Assumption NL1 for all data cells with size larger than 50}
\bigskip 
\begin{minipage}[t]{0.245\linewidth}
  \centerline{\includegraphics[width=3.813cm,height=2.773cm]{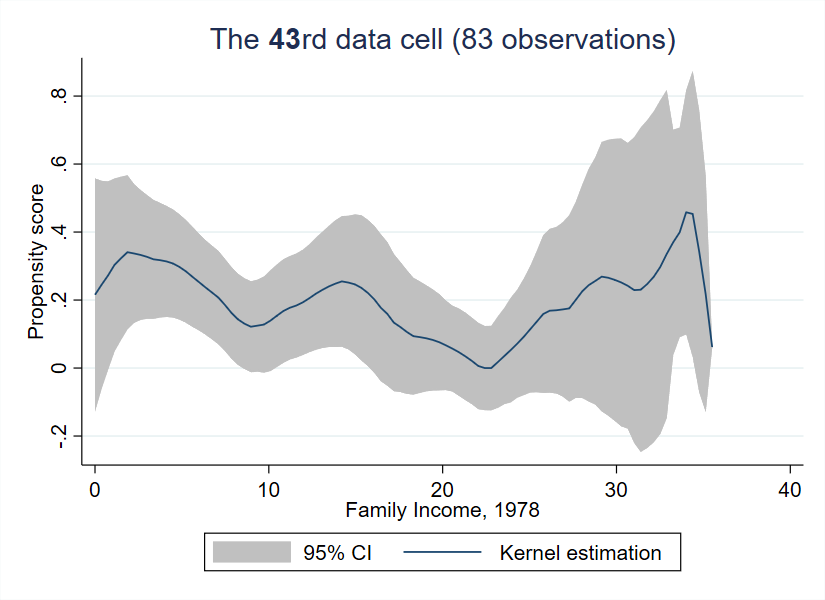}}
\end{minipage}
\begin{minipage}[t]{0.245\linewidth}
  \centerline{\includegraphics[width=3.813cm,height=2.773cm]{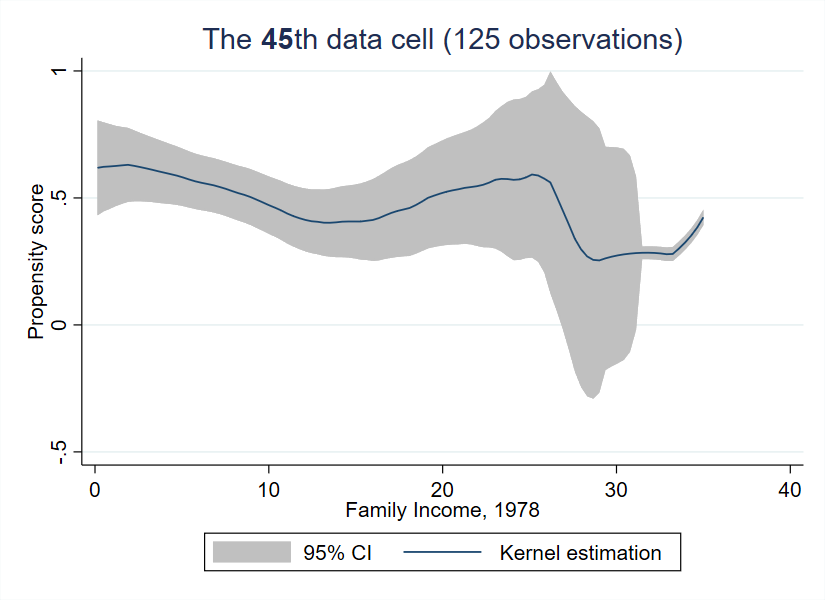}}
\end{minipage}
\begin{minipage}[t]{0.245\linewidth}
  \centerline{\includegraphics[width=3.813cm,height=2.773cm]{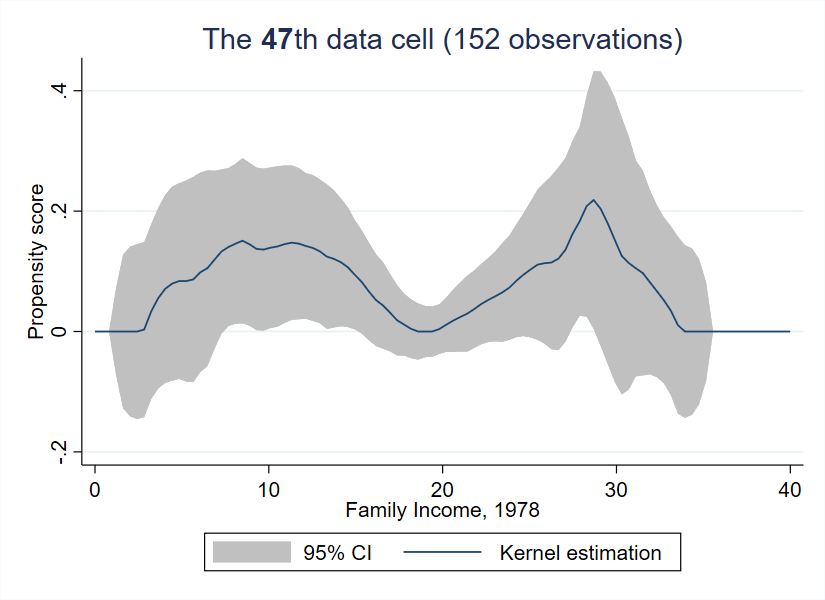}}
\end{minipage}
\begin{minipage}[t]{0.245\linewidth}
  \centerline{\includegraphics[width=3.813cm,height=2.773cm]{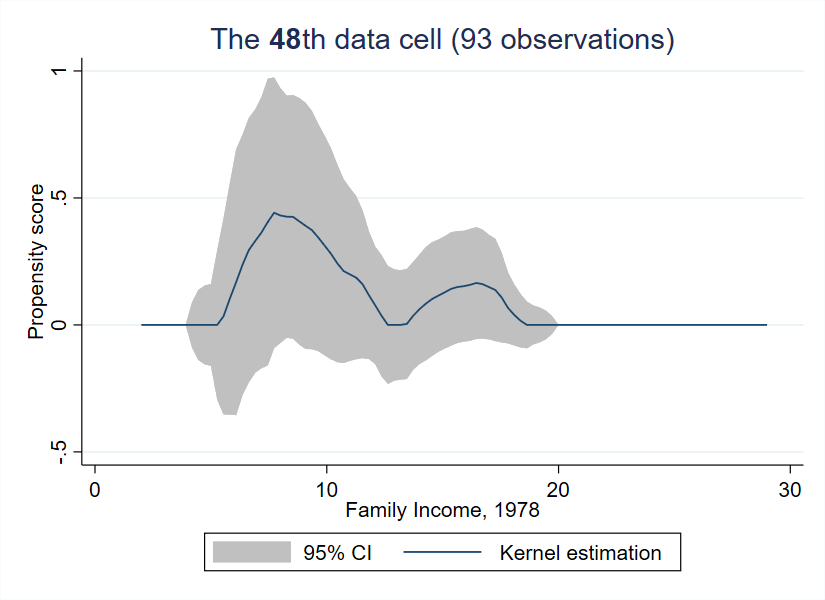}}
\end{minipage}
\begin{minipage}[t]{0.245\linewidth}
  \centerline{\includegraphics[width=3.813cm,height=2.773cm]{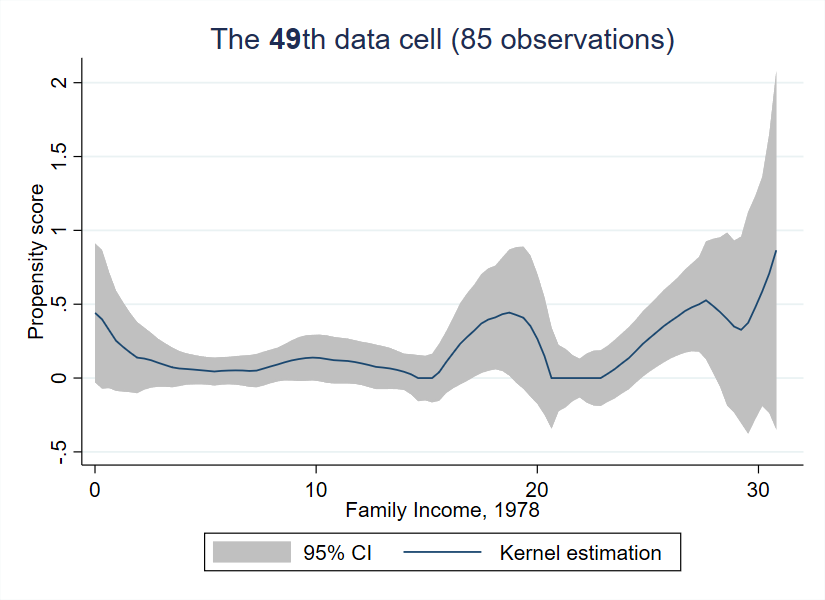}}
\end{minipage}
\begin{minipage}[t]{0.245\linewidth}
  \centerline{\includegraphics[width=3.813cm,height=2.773cm]{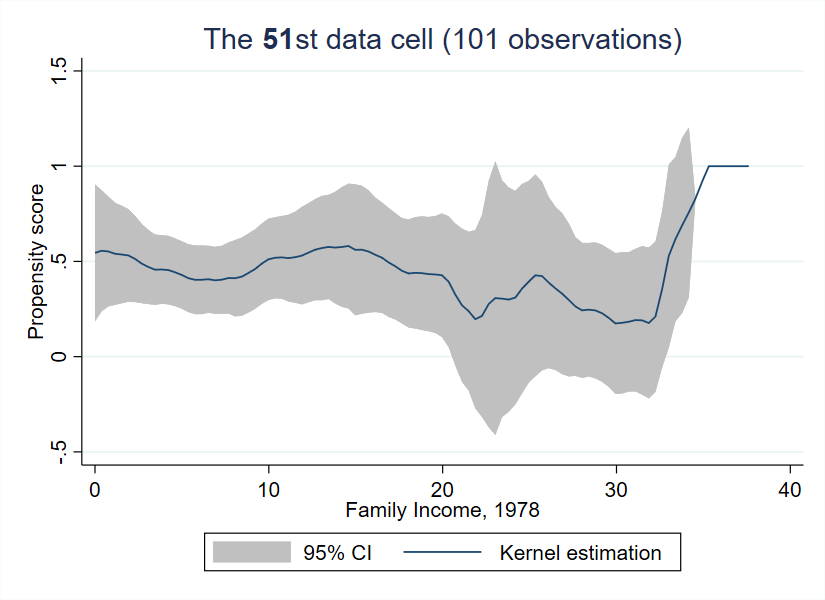}}
\end{minipage}
\begin{minipage}[t]{0.245\linewidth}
  \centerline{\includegraphics[width=3.813cm,height=2.773cm]{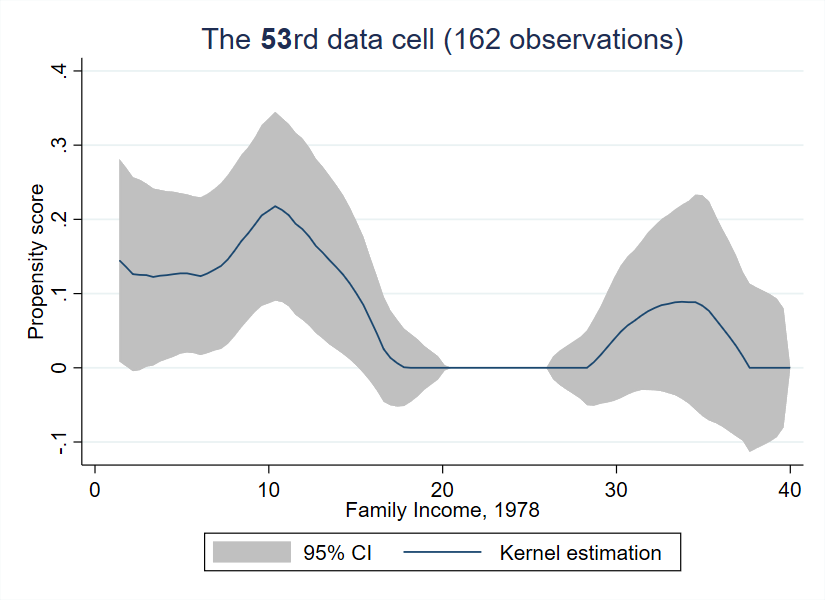}}
\end{minipage}
\begin{minipage}[t]{0.245\linewidth}
  \centerline{\includegraphics[width=3.813cm,height=2.773cm]{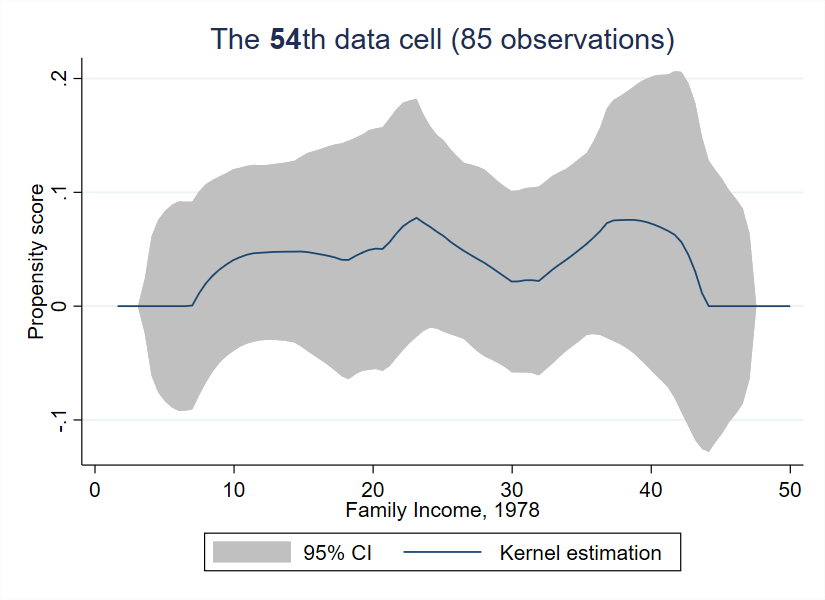}}
\end{minipage}
\begin{minipage}[t]{0.245\linewidth}
  \centerline{\includegraphics[width=3.813cm,height=2.773cm]{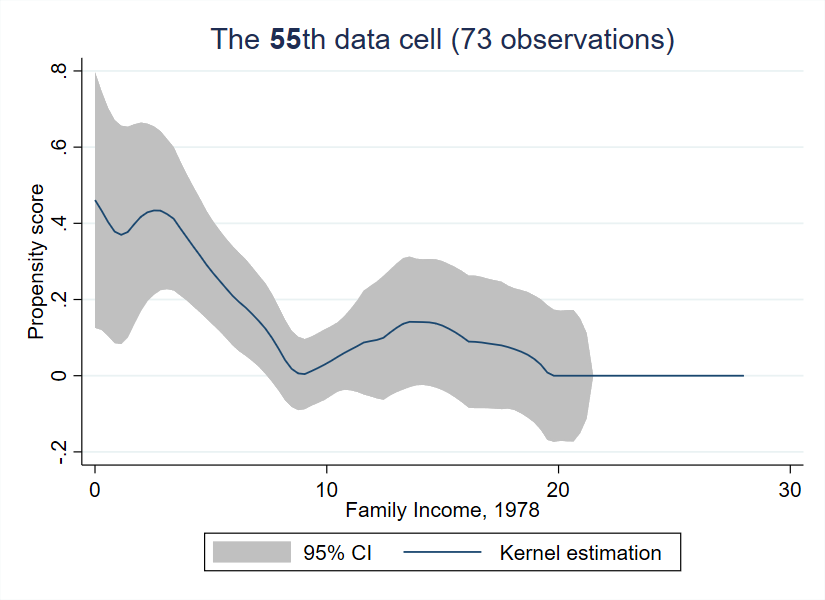}}
\end{minipage}
\begin{minipage}[t]{0.245\linewidth}
  \centerline{\includegraphics[width=3.813cm,height=2.773cm]{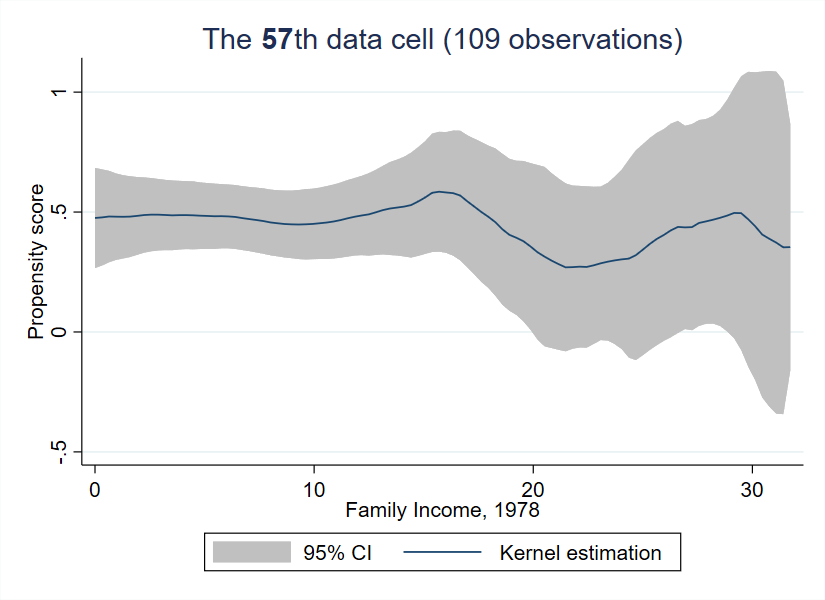}}
\end{minipage}
\begin{minipage}[t]{0.245\linewidth}
  \centerline{\includegraphics[width=3.813cm,height=2.773cm]{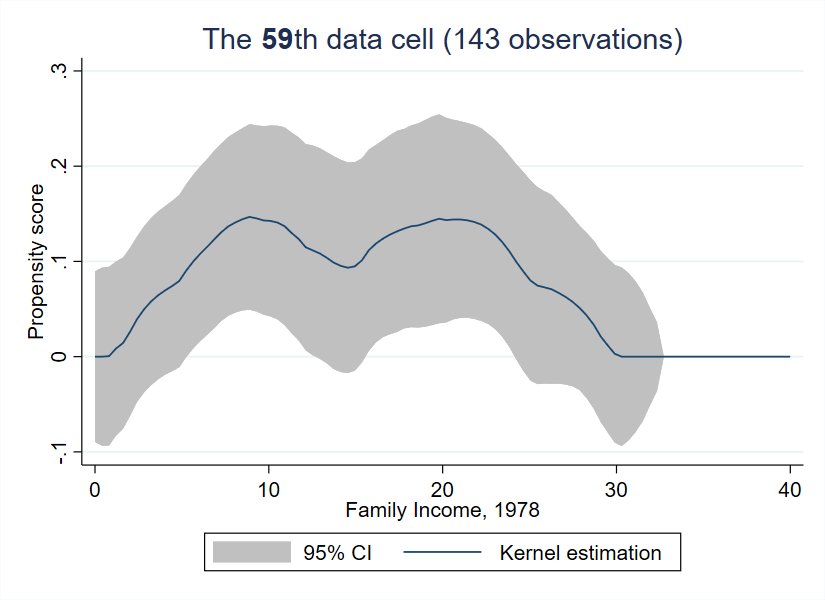}}
\end{minipage}
\begin{minipage}[t]{0.245\linewidth}
  \centerline{\includegraphics[width=3.813cm,height=2.773cm]{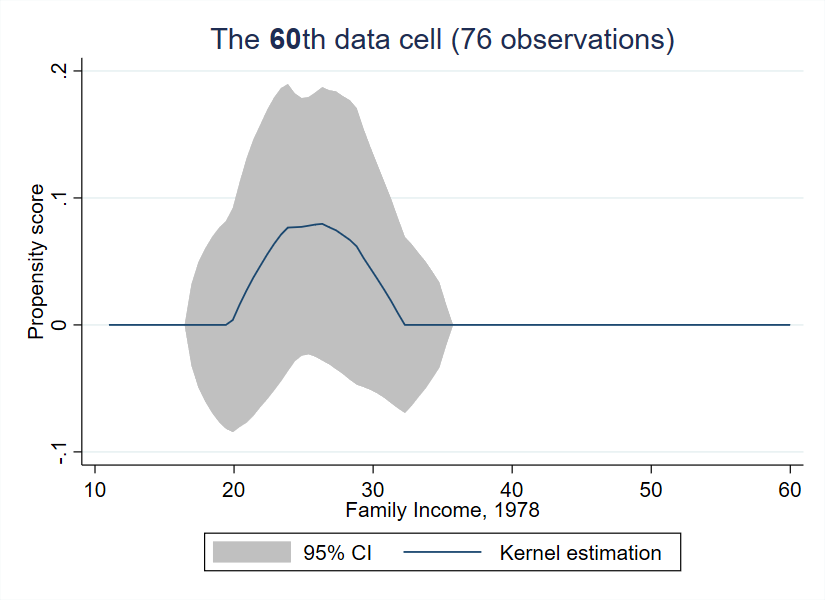}}
\end{minipage}
\end{figure}

\renewcommand{\thetable}{G.\arabic{table}} \setcounter{table}{0}
\renewcommand*{\theHtable}{\thetable} 
\renewcommand{\thefigure}{G.\arabic{figure}} \setcounter{figure}{0} 
\renewcommand*{\theHfigure}{\thefigure}

\clearpage

\begin{sidewaystable}[htb]
\caption{Simulation designs}
\label{table:designs}
\medskip \renewcommand\arraystretch{1.4}
\resizebox{\linewidth}{!}{
\begin{tabular}{clC{11.8cm}cc}
\hline\hline
Design & Exclusion restriction & Other difference from Design 1 & Consistent
estimators & Inconsistent estimators \\ \hline
1 & Yes: $\beta _{1}^{C}=\beta _{0}^{C}=0$ & - & P, N, S, SIV & L, Q \\ 
2 & No: $\beta _{1}^{C}=2,\beta _{0}^{C}=1$ & - & P, N, S & L, Q, SIV \\ 
3 & Yes: $\beta _{1}^{C}=\beta _{0}^{C}=0$ & Linear treatment function: $\mu
\left( X_{i}\right) =X_{i}^{C}+\sum_{j=1}^{5}\frac{1}{j}\left(
X_{ji}^{D}-0.5\right) $. & SIV & P, N, L, Q, S \\ 
4 & No: $\beta _{1}^{C}=2,\beta _{0}^{C}=1$ & Linear treatment function: $%
\mu \left( X_{i}\right) =X_{i}^{C}+\sum_{j=1}^{5}\frac{1}{j}\left(
X_{ji}^{D}-0.5\right) $. & - & P, N, L, Q, S, SIV \\ 
5 & No: $\beta _{1}^{C}=2,\beta _{0}^{C}=1$ & $X_{i}^{C}=\sqrt{1-0.75^{2}}%
U_{i}+0.75\varrho _{i}$ s.t. $Cov\left( U_{i},X_{i}^{C}\right) \neq 0$. & P,
N, S & L, Q \\ 
6 & No: $\beta _{1}^{C}=2,\beta _{0}^{C}=1$ & (i) Linear treatment function: 
$\mu \left( X_{i}\right) =X_{i}^{C}+\sum_{j=1}^{5}\frac{1}{j}\left(
X_{ji}^{D}-0.5\right) $, and (ii) Heteroskedastic $U_{i}$: $U_{i}=\exp
\left( X_{i}^{C}+\frac{1}{2}\left[ \left( X_{i}^{C}\right) ^{2}-1\right]
\right) \epsilon _{i}$, where $\epsilon _{i}\sim N\left( 0,1\right) $ and is
independent of $X_{i}^{C}$. & N, S & P, L, Q \\ 
7 & No: $\beta _{1}^{C}=2,\beta _{0}^{C}=1$ & Normal quadratic MTE: $%
U_{1i}=-\left( \frac{3}{8}U_{i}^{2}+\frac{3}{4}U_{i}-\frac{3}{8}\right) +%
\sqrt{1-\frac{27}{32}}\varrho _{i}$ s.t. $E\left[ \left. U_{1i}\right\vert
V_{i}=v\right] =\frac{3}{8}-\frac{3}{4}\Phi ^{-1}\left( v\right) -\frac{3}{8}%
\Phi ^{-2}\left( v\right) $, where $V_{i}=\Phi \left( U_{i}\right) $. & S & 
P, N, L, Q \\ 
8 & No: $\beta _{1}^{C}=2,\beta _{0}^{C}=1$ & Normal cubic MTE: $U_{1i}=-%
\frac{1}{4}U_{i}^{3}+\frac{1}{4}\varrho _{i}$ s.t. $E\left[ \left.
U_{1i}\right\vert V_{i}=v\right] =-\frac{1}{4}\Phi ^{-3}\left( v\right) $. & 
S & P, N, L, Q \\ 
9 & No: $\beta _{1}^{C}=2,\beta _{0}^{C}=1$ & $U_{i}\sim t\left( 3\right) $
and is standardized to have unit variance. & S & P, N, L, Q \\ 
10 & No: $\beta _{1}^{C}=2,\beta _{0}^{C}=1$ & $U_{i}\sim t\left( 2\right) $.
& S & P, N, L, Q \\ 
11 & No: $\beta _{1}^{C}=2,\beta _{0}^{C}=1$ & $\mu \left( X_{i}\right)
=X_{i}^{C}+\frac{1}{2}\left[ \left( X_{i}^{C}\right) ^{2}-1\right] +\frac{1}{%
3}\left( X_{i}^{C}\right) ^{3}+\sum_{j=1}^{5}\frac{1}{j}\left(
X_{ji}^{D}-0.5\right) $. & N, S & P, L, Q \\ 
12 & No: $\beta _{1}^{C}=2,\beta _{0}^{C}=1$ & Heteroskedastic $U_{di}$: $%
U_{1i}=-0.75U_{i}+\sqrt{1-0.75^{2}}\varrho _{i}$ with $\varrho _{i}\sim
N\left( 0,\exp \left( 2X_{i}^{C}-2\right) \right) $, and $U_{0i}\sim N\left(
0,\left( X_{i}^{C}\right) ^{2}\right) $. & P, N, S & L, Q \\ 
13 & No: $\beta _{1}^{C}=2,\beta _{0}^{C}=1$ & $%
Y_{1i}=2X_{i}^{C}+\sum_{j=1}^{5}\left( X_{ji}^{D}-0.5\right) +\left(
X_{1i}^{D}-0.5\right) \left( X_{2i}^{D}-0.5\right) +U_{1i}$. & Correctly
specified S & P, N, L, Q, S \\ 
14 & No: $\beta _{1}^{C}=2,\beta _{0}^{C}=1$ & $%
Y_{1i}=2X_{i}^{C}+\sum_{j=1}^{5}\left( X_{ji}^{D}-0.5\right)
+X_{i}^{C}\left( X_{1i}^{D}-0.5\right) +U_{1i}$. & Correctly specified S & 
P, N, L, Q, S \\ 
15 & No: $\beta _{1}^{C}=2,\beta _{0}^{C}=1$ & $%
Y_{1i}=2X_{i}^{C}+\sum_{j=1}^{5}\left( X_{ji}^{D}-0.5\right) +\left[ \left(
X_{i}^{C}\right) ^{2}-1\right] +U_{1i}$. & Correctly specified S & P, N, L,
Q, S \\ 
16 & No: $\beta _{1}^{C}=2,\beta _{0}^{C}=1$ & $%
Y_{1i}=2X_{i}^{C}+\sum_{j=1}^{5}\left( X_{ji}^{D}-0.5\right) +0.5\left[ \exp
\left( X_{i}^{C}\right) -\exp \left( 0.5\right) \right] +U_{1i}$. & 
Correctly specified S & P, N, L, Q, S \\ \hline\hline
\end{tabular}
}\medskip \newline
{\footnotesize Notes: P, N, L, Q, S, and SIV stand for parametric, normal, linear,
quadratic, semiparametric, and semiparametric IV estimators for the MTE,
respectively. Please refer to Table \ref{table:estimators} for details of
the estimators. In Design 1, $U_{i}$ follows a standard normal distribution that is
independent of $X_{i}$, and $\mu \left( \cdot \right) $ is quadratic as in (%
\ref{treat_fn}). Design 2 differs from Design 1 only in relaxing the
exclusion restriction.}
\end{sidewaystable}

\begin{sidewaystable}[htb]
\caption{MTE estimators considered in the simulation}
\label{table:estimators}
\bigskip \renewcommand\arraystretch{1.8}
\resizebox{\linewidth}{!}{
\begin{tabular}{C{4cm}cL{4.8cm}cL{9.6cm}}
\hline\hline
Estimator & Step 1 & & Step 2 & \\ \hline
Parametric MTE & \textit{Case 1.3}: & Probit regression of $D_{i}$ on $%
X_{i}^{C}$, $\left( X_{i}^{C}\right) ^{2}$, and $X_{ji}^{D}$. & \textit{Case
2.2.2}: & Linear regression of $Y_{i}$ on $X_{i}$ and $\left. \phi \left(
\Phi ^{-1}\left( \hat{P}_{i}\right) \right) \right/ \left( \hat{P}%
_{i}-1+d\right) $ for each group $d=0,1$. \\ 
Normal MTE & \textit{Case 1.1}: & Nonparametric regression of $D_{i}$ on $%
X_{i}^{C}$ and $X_{ji}^{D}$. & \textit{Case 2.2.2}: & Linear regression of $%
Y_{i}$ on $X_{i}$ and $\left. \phi \left( \Phi ^{-1}\left( \hat{P}%
_{i}\right) \right) \right/ \left( \hat{P}_{i}-1+d\right) $ for each group $%
d=0,1$. \\ 
Linear MTE & \textit{Case 1.1}: & Nonparametric regression of $D_{i}$ on $%
X_{i}^{C}$ and $X_{ji}^{D}$. & \textit{Case 2.2.1}: & Linear regression of $%
Y_{i}$ on $X_{i}$ and $\hat{P}_{i}$ for each group. \\ 
Quadratic MTE & \textit{Case 1.1}: & Nonparametric regression of $D_{i}$ on $%
X_{i}^{C}$ and $X_{ji}^{D}$. & \textit{Case 2.2.1}: & Linear regression of $%
Y_{i}$ on $X_{i}$, $\hat{P}_{i}$, and $\hat{P}_{i}^{2}$ for each group. \\ 
Semiparametric MTE & \textit{Case 1.1}: & Nonparametric regression of $D_{i}$
on $X_{i}^{C}$ and $X_{ji}^{D}$. & \textit{Case 2.1}: & Partially linear
regression of $Y_{i}$ on $X_{i}$ and $g_{d}\left( \hat{P}_{i}\right) $ by (%
\ref{PDLS}). \\ 
Semiparametric IV-MTE & \textit{Case 1.1}: & Nonparametric regression of $%
D_{i}$ on $X_{i}^{C}$ and $X_{ji}^{D}$. & \textit{Case 2.1}: & Partially
linear regression of $Y_{i}$ on $X_{ji}^{D}$ and $g_{d}\left( \hat{P}%
_{i}\right) $ by (\ref{PDLS}) with $X_{i}^{C}$ excluded. \\ 
Correctly specified semiparametric MTE & \textit{Case 1.1}: & Nonparametric
regression of $D_{i}$ on $X_{i}^{C}$ and $X_{ji}^{D}$. & \textit{Case 2.1}: & 
Partially linear regression of $Y_{i}$ on $X_{i}$, $g_{d}\left( \hat{P}%
_{i}\right) $, and the true nonlinear term in Designs 13-16. \\ \hline\hline
\end{tabular}
}\medskip \newline
{\small Notes: The italics refer to the cases presented in Table \ref{table:implement}%
. The semiparametric IV-MTE estimator is considered only in Designs 1-4,
while the correctly specified semiparametric MTE estimator is considered
only in Designs 13-16.}
\end{sidewaystable}

\clearpage

\begin{figure}[tbh]
\caption{Simulated MTE curves for Designs 1-4}
\label{Fig:Design1-4}
\bigskip \setlength{\abovecaptionskip}{-0.1cm} 
\begin{minipage}[t]{0.48\linewidth}
  \centerline{\includegraphics[width=7.626cm,height=5.55cm]{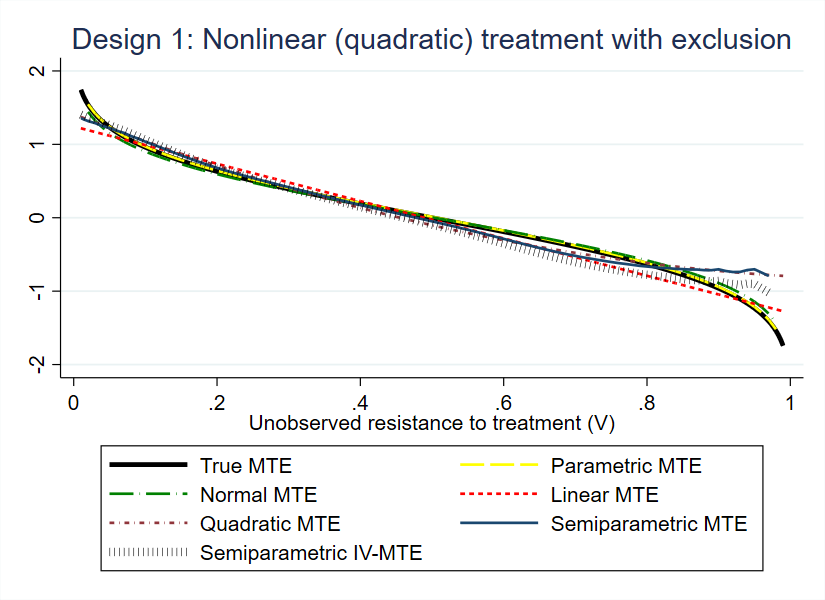}}
\end{minipage}
\begin{minipage}[t]{0.48\linewidth}
  \centerline{\includegraphics[width=7.626cm,height=5.55cm]{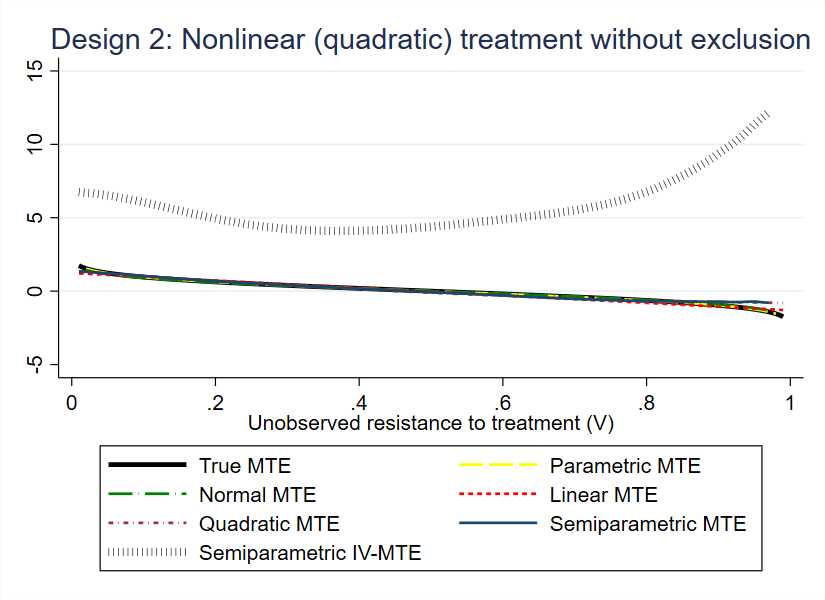}}
\end{minipage}
\begin{minipage}[t]{0.48\linewidth}
  \centerline{\includegraphics[width=7.626cm,height=5.55cm]{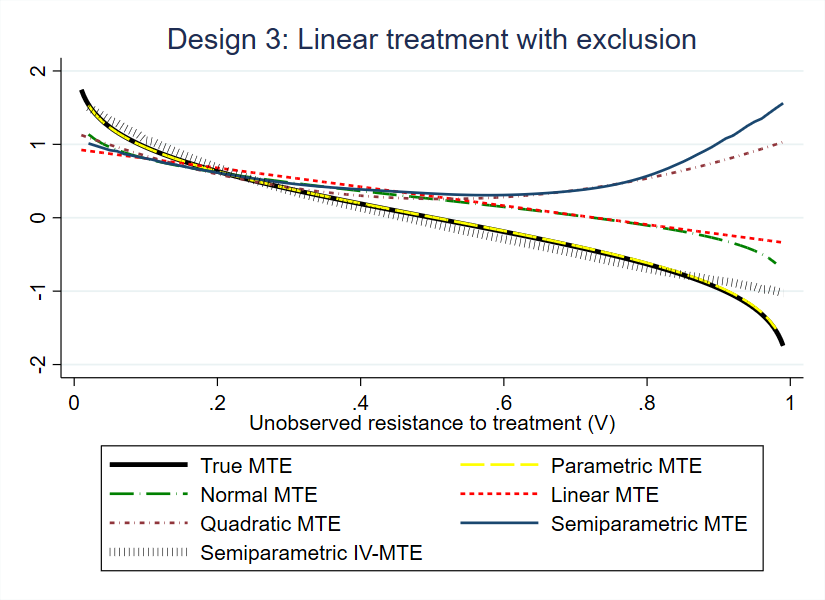}}
\end{minipage}
\begin{minipage}[t]{0.48\linewidth}
  \centerline{\includegraphics[width=7.626cm,height=5.55cm]{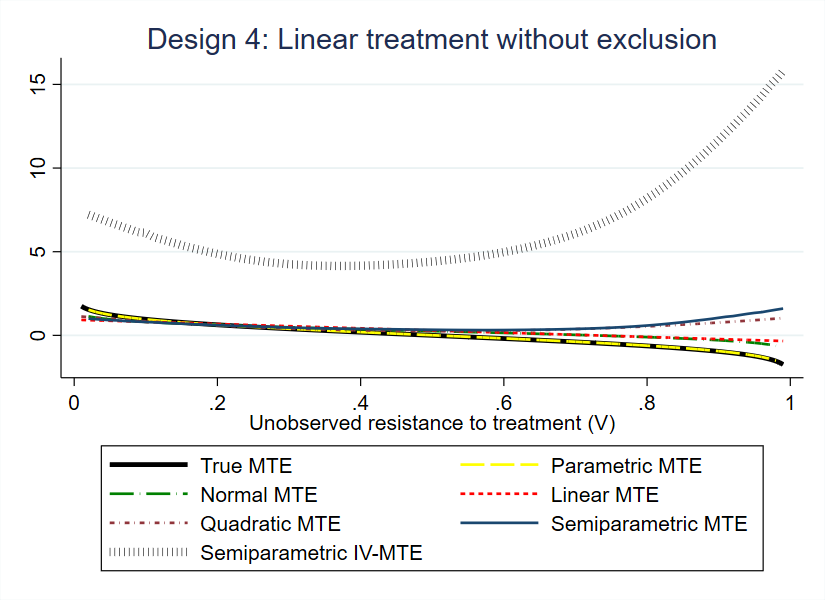}}
\end{minipage}
{\small Notes: The parametric MTE estimator specifies both parametric first
step and parametric second step. The normal, linear, quadratic MTE
estimators specify nonparametric first step and parametric second step. The
semiparametric MTE estimator specifies nonparametric first step and
partially linear (semiparametric) second step. The above five MTE estimators
impose no exclusion restriction and base the identification on different
functional forms between the treatment and outcome equations, while the
semiparametric IV-MTE estimator imposes the exclusion restriction $\beta
_{1}^{C}=\beta _{0}^{C}=0$ but otherwise is the same as the semiparametric
MTE. Please refer to Table \ref{table:estimators} for more details about the
estimators}
\end{figure}

\clearpage

\begin{sidewaystable}[htb]
\caption{Simulation results for Designs 1-4}
\label{table:Design1-4}
\medskip \renewcommand\arraystretch{1.1}
\resizebox{\linewidth}{!}{
\begin{tabular}{ccccccccccccccccc}
\hline\hline
&  & \multicolumn{3}{c}{Design 1} &  & \multicolumn{3}{c}{Design 2} &  & 
\multicolumn{3}{c}{Design 3} &  & \multicolumn{3}{c}{Design 4} \\ 
\cline{3-5}\cline{7-9}\cline{11-13}\cline{15-17}
&  & Bias & RMSE & Coverage &  & Bias & RMSE & Coverage &  & Bias & RMSE & 
Coverage &  & Bias & RMSE & Coverage \\ \hline
\multicolumn{17}{l}{\underline{Parametric estimator}} \\ 
MTE$\left( v=0.5\right) $ &  & 0.014 & 0.101 & 0.962 &  & -0.004 & 0.101 & 
0.958 &  & 0.008 & 0.214 & 0.945 &  & 0.008 & 0.214 & 0.946 \\ 
ATE &  & 0.051 & 0.107 & 0.927 &  & 0.033 & 0.101 & 0.942 &  & 0.045 & 0.206
& 0.948 &  & 0.045 & 0.207 & 0.944 \\ 
$\hat{\beta}_{1}^{C}-\hat{\beta}_{0}^{C}$ &  & -0.001 & 0.054 & 0.948 &  & -0.001 & 
0.054 & 0.948 &  & -0.003 & 0.125 & 0.949 &  & -0.003 & 0.125 & 0.949 \\ 
\multicolumn{17}{l}{\underline{Normal estimator}} \\ 
MTE$\left( v=0.5\right) $ &  & 0.006 & 0.118 & 0.954 &  & -0.003 & 0.119 & 
0.951 &  & 0.257 & 0.309 & 0.674 &  & 0.256 & 0.309 & 0.677 \\ 
ATE &  & 0.041 & 0.119 & 0.931 &  & 0.032 & 0.117 & 0.939 &  & 0.278 & 0.323
& 0.601 &  & 0.278 & 0.323 & 0.601 \\ 
$\hat{\beta}_{1}^{C}-\hat{\beta}_{0}^{C}$ &  & -0.035 & 0.067 & 0.909 &  & -0.035 & 
0.067 & 0.909 &  & -0.168 & 0.193 & 0.574 &  & -0.168 & 0.193 & 0.574 \\ 
\multicolumn{17}{l}{\underline{Linear estimator}} \\ 
MTE$\left( v=0.5\right) $ &  & -0.027 & 0.123 & 0.951 &  & -0.037 & 0.126 & 
0.949 &  & 0.293 & 0.351 & 0.669 &  & 0.293 & 0.351 & 0.665 \\ 
ATE &  & -0.015 & 0.118 & 0.951 &  & -0.024 & 0.121 & 0.954 &  & 0.300 & 
0.355 & 0.652 &  & 0.300 & 0.355 & 0.651 \\ 
$\hat{\beta}_{1}^{C}-\hat{\beta}_{0}^{C}$ &  & -0.024 & 0.064 & 0.924 &  & -0.024 & 
0.064 & 0.924 &  & -0.188 & 0.215 & 0.576 &  & -0.188 & 0.215 & 0.576 \\ 
\multicolumn{17}{l}{\underline{Quadratic estimator}} \\ 
MTE$\left( v=0.5\right) $ &  & -0.102 & 0.176 & 0.899 &  & -0.111 & 0.182 & 
0.891 &  & 0.257 & 0.323 & 0.734 &  & 0.257 & 0.323 & 0.738 \\ 
ATE &  & 0.043 & 0.142 & 0.934 &  & 0.034 & 0.141 & 0.937 &  & 0.538 & 0.584
& 0.350 &  & 0.537 & 0.584 & 0.349 \\ 
$\hat{\beta}_{1}^{C}-\hat{\beta}_{0}^{C}$ &  & -0.028 & 0.067 & 0.923 &  & -0.028 & 
0.067 & 0.923 &  & -0.261 & 0.284 & 0.363 &  & -0.261 & 0.284 & 0.363 \\ 
\multicolumn{17}{l}{\underline{Semiparametric estimator}} \\ 
MTE$\left( v=0.5\right) $ &  & -0.049 & 0.309 & 0.947 &  & -0.056 & 0.312 & 
0.947 &  & 0.329 & 0.450 & 0.809 &  & 0.336 & 0.455 & 0.807 \\ 
ATE &  & 0.041 & 0.168 & 0.940 &  & 0.044 & 0.194 & 0.947 &  & 0.621 & 0.684
& 0.441 &  & 0.629 & 0.693 & 0.440 \\ 
$\hat{\beta}_{1}^{C}-\hat{\beta}_{0}^{C}$ &  & -0.029 & 0.068 & 0.917 &  & -0.029 & 
0.068 & 0.922 &  & -0.282 & 0.307 & 0.358 &  & -0.287 & 0.311 & 0.345 \\ 
\multicolumn{17}{l}{\underline{Semiparametric IV estimator}} \\ 
MTE$\left( v=0.5\right) $ &  & -0.094 & 0.307 & 0.939 &  & 4.394 & 4.427 & 0
&  & -0.082 & 0.264 & 0.940 &  & 4.410 & 4.424 & 0 \\ 
ATE &  & -0.016 & 0.119 & 0.947 &  & 6.324 & 6.330 & 0 &  & 0.005 & 0.108 & 
0.955 &  & 6.868 & 6.872 & 0 \\ \hline\hline
\end{tabular}
}\medskip \newline
{\small Notes: RMSE stands for the root mean squared error, and Coverage stands for
the coverage probability of asymptotic 95\% confidence interval. For the semiparametric IV
estimator, $\beta _{1}^{C}$ and $\beta _{0}^{C}$ are assumed to be zero
and thus not estimated.}
\end{sidewaystable}

\clearpage

\begin{figure}[tbh]
\caption{Simulated MTE curves for Designs 5-8}
\label{Fig:Design5-8}
\bigskip \setlength{\abovecaptionskip}{-0.1cm} 
\begin{minipage}[t]{0.48\linewidth}
  \centerline{\includegraphics[width=7.626cm,height=5.55cm]{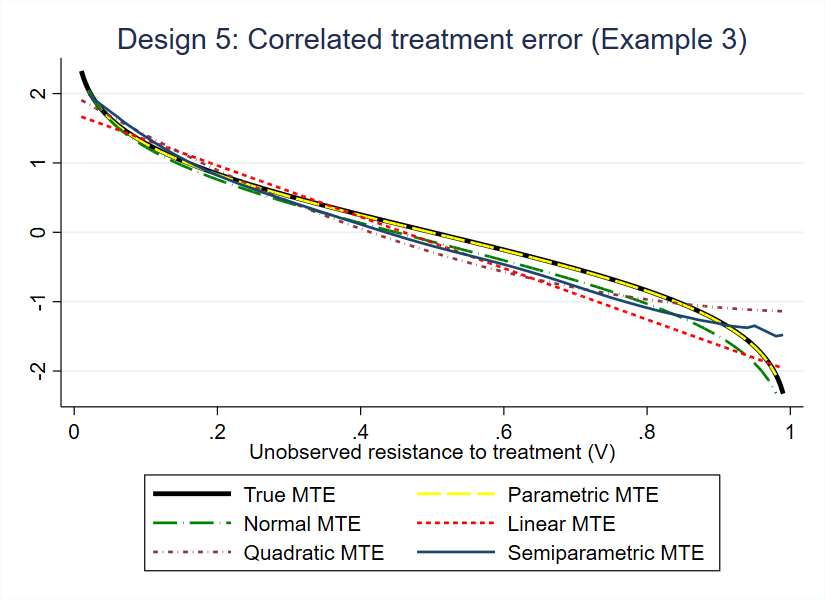}}
\end{minipage}
\begin{minipage}[t]{0.48\linewidth}
  \centerline{\includegraphics[width=7.626cm,height=5.55cm]{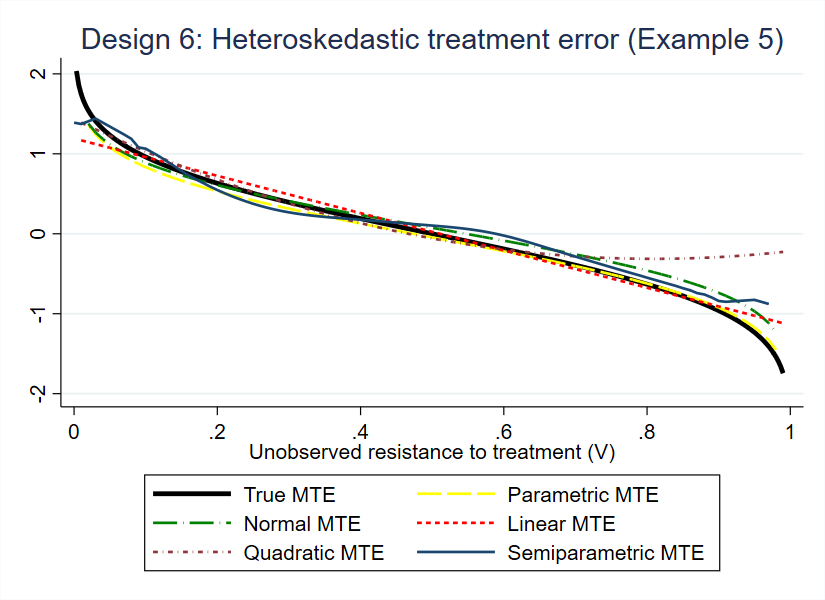}}
\end{minipage}
\begin{minipage}[t]{0.48\linewidth}
  \centerline{\includegraphics[width=7.626cm,height=5.55cm]{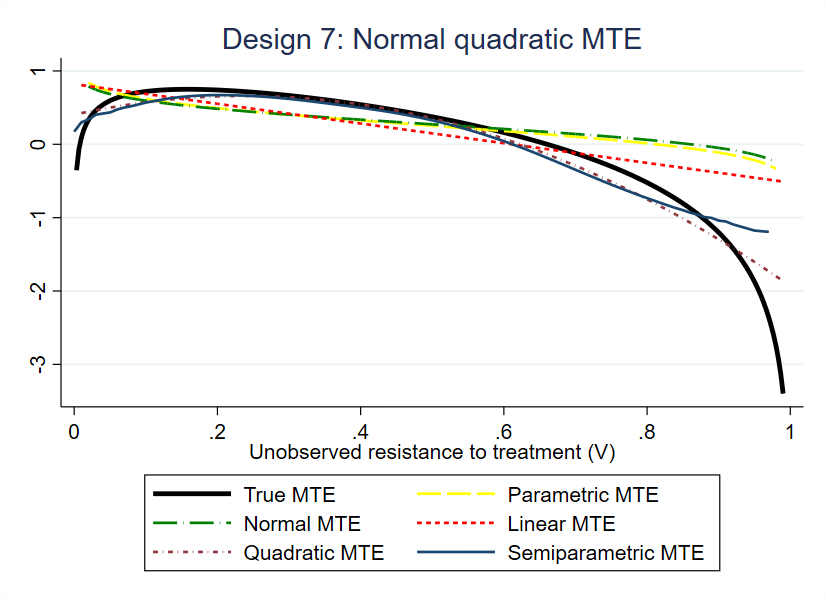}}
\end{minipage}
\begin{minipage}[t]{0.48\linewidth}
  \centerline{\includegraphics[width=7.626cm,height=5.55cm]{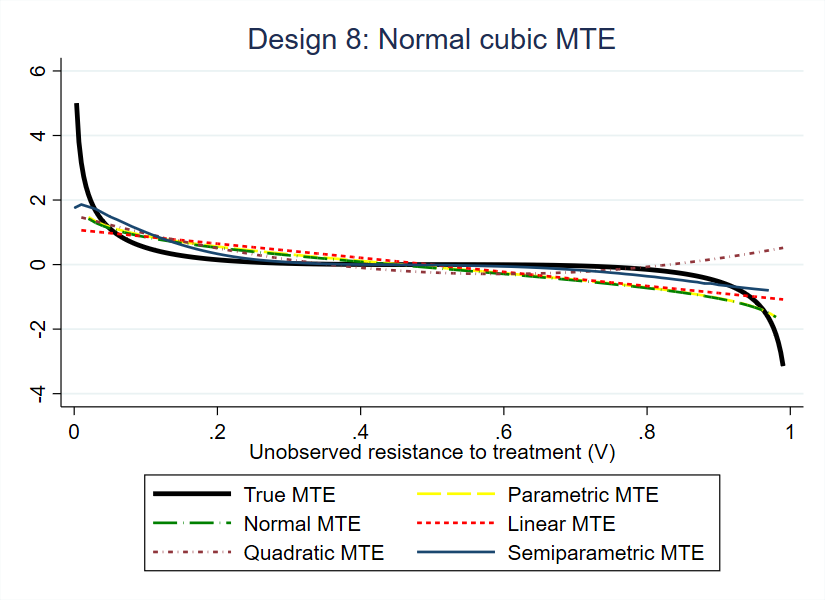}}
\end{minipage}
{\small Notes: The parametric MTE estimator specifies both parametric first
step and parametric second step. The normal, linear, quadratic MTE
estimators specify nonparametric first step and parametric second step. And
the semiparametric MTE estimator specifies nonparametric first step and
partially linear (semiparametric) second step.}
\end{figure}

\clearpage

\begin{sidewaystable}[htb]
\caption{Simulation results for Designs 5-8}
\label{table:Design5-8}
\medskip \renewcommand\arraystretch{1.2}
\resizebox{\linewidth}{!}{
\begin{tabular}{ccccccccccccccccc}
\hline\hline
&  & \multicolumn{3}{c}{Design 5} &  & \multicolumn{3}{c}{Design 6} &  & 
\multicolumn{3}{c}{Design 7} &  & \multicolumn{3}{c}{Design 8} \\ 
\cline{3-5}\cline{7-9}\cline{11-13}\cline{15-17}
&  & Bias & RMSE & Coverage &  & Bias & RMSE & Coverage &  & Bias & RMSE & 
Coverage &  & Bias & RMSE & Coverage \\ \hline
\multicolumn{17}{l}{\underline{Parametric estimator}} \\ 
MTE$\left( v=0.5\right) $ &  & -0.007 & 0.081 & 0.947 &  & -0.045 & 0.198 & 
0.347 &  & -0.122 & 0.152 & 0.725 &  & -0.076 & 0.133 & 0.891 \\ 
ATE &  & 0.042 & 0.087 & 0.911 &  & -0.011 & 0.183 & 0.347 &  & 0.267 & 0.280
& 0.124 &  & -0.039 & 0.110 & 0.932 \\ 
$\hat{\beta}_{1}^{C}-\hat{\beta}_{0}^{C}$ &  & -0.002 & 0.038 & 0.949 &  & 0.064 & 
0.098 & 0.315 &  & -0.062 & 0.078 & 0.724 &  & 0.038 & 0.070 & 0.899 \\ 
\multicolumn{17}{l}{\underline{Normal estimator}} \\ 
MTE$\left( v=0.5\right) $ &  & -0.135 & 0.164 & 0.701 &  & 0.075 & 0.137 & 
0.912 &  & -0.102 & 0.146 & 0.833 &  & -0.102 & 0.164 & 0.876 \\ 
ATE &  & -0.083 & 0.121 & 0.846 &  & 0.106 & 0.152 & 0.844 &  & 0.285 & 0.301
& 0.176 &  & -0.066 & 0.138 & 0.916 \\ 
$\hat{\beta}_{1}^{C}-\hat{\beta}_{0}^{C}$ &  & 0.005 & 0.038 & 0.954 &  & -0.013 & 
0.054 & 0.949 &  & -0.083 & 0.097 & 0.628 &  & 0.018 & 0.065 & 0.943 \\ 
\multicolumn{17}{l}{\underline{Linear estimator}} \\ 
MTE$\left( v=0.5\right) $ &  & -0.147 & 0.176 & 0.671 &  & 0.025 & 0.117 & 
0.949 &  & -0.226 & 0.249 & 0.421 &  & -0.008 & 0.130 & 0.952 \\ 
ATE &  & -0.128 & 0.159 & 0.735 &  & 0.037 & 0.118 & 0.943 &  & 0.156 & 0.187
& 0.670 &  & 0.002 & 0.127 & 0.957 \\ 
$\hat{\beta}_{1}^{C}-\hat{\beta}_{0}^{C}$ &  & 0.007 & 0.039 & 0.950 &  & -0.013 & 
0.054 & 0.956 &  & -0.033 & 0.061 & 0.890 &  & -0.024 & 0.066 & 0.932 \\ 
\multicolumn{17}{l}{\underline{Quadratic estimator}} \\ 
MTE$\left( v=0.5\right) $ &  & -0.288 & 0.311 & 0.321 &  & -0.058 & 0.149 & 
0.920 &  & -0.021 & 0.126 & 0.952 &  & -0.245 & 0.289 & 0.632 \\ 
ATE &  & -0.044 & 0.121 & 0.936 &  & 0.168 & 0.222 & 0.795 &  & 0.000 & 0.124
& 0.959 &  & 0.181 & 0.234 & 0.792 \\ 
$\hat{\beta}_{1}^{C}-\hat{\beta}_{0}^{C}$ &  & 0.003 & 0.039 & 0.950 &  & -0.005 & 
0.055 & 0.956 &  & -0.020 & 0.055 & 0.933 &  & -0.041 & 0.074 & 0.892 \\ 
\multicolumn{17}{l}{\underline{Semiparametric estimator}} \\ 
MTE$\left( v=0.5\right) $ &  & -0.197 & 0.334 & 0.890 &  & 0.103 & 0.321 & 
0.922 &  & -0.049 & 0.266 & 0.944 &  & -0.014 & 0.313 & 0.958 \\ 
ATE &  & -0.076 & 0.163 & 0.917 &  & 0.130 & 0.209 & 0.887 &  & 0.035 & 0.155
& 0.943 &  & 0.042 & 0.179 & 0.951 \\ 
$\hat{\beta}_{1}^{C}-\hat{\beta}_{0}^{C}$ &  & 0.006 & 0.039 & 0.946 &  & -0.013 & 
0.059 & 0.941 &  & -0.027 & 0.058 & 0.914 &  & -0.018 & 0.066 & 0.947 \\ 
\hline\hline
\end{tabular}
}\medskip \newline
{\small Notes: RMSE stands for the root mean squared error, and Coverage stands for
the coverage probability of asymptotic 95\% confidence interval.}
\end{sidewaystable}

\clearpage

\begin{figure}[tbh]
\caption{Simulated MTE curves for Designs 9-12}
\label{Fig:Design9-12}
\bigskip \setlength{\abovecaptionskip}{-0.1cm} 
\begin{minipage}[t]{0.48\linewidth}
  \centerline{\includegraphics[width=7.626cm,height=5.55cm]{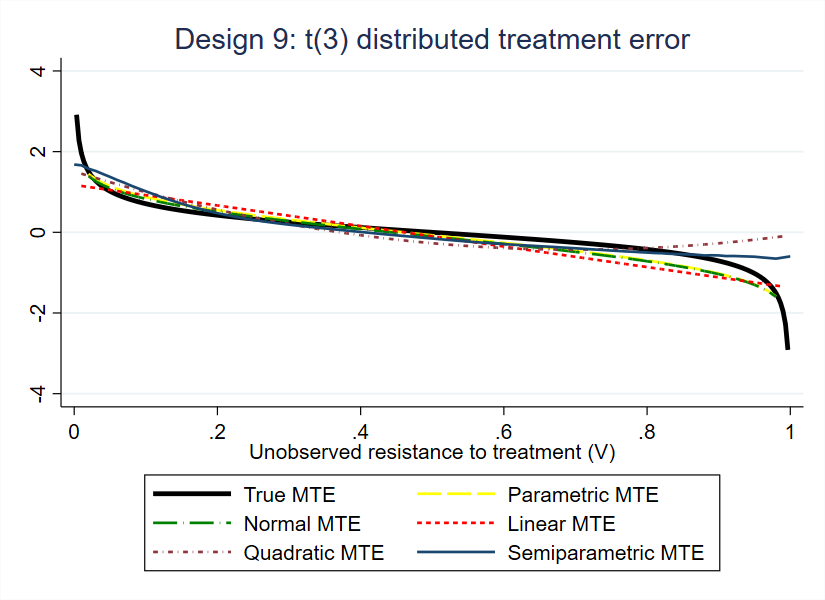}}
\end{minipage}
\begin{minipage}[t]{0.48\linewidth}
  \centerline{\includegraphics[width=7.626cm,height=5.55cm]{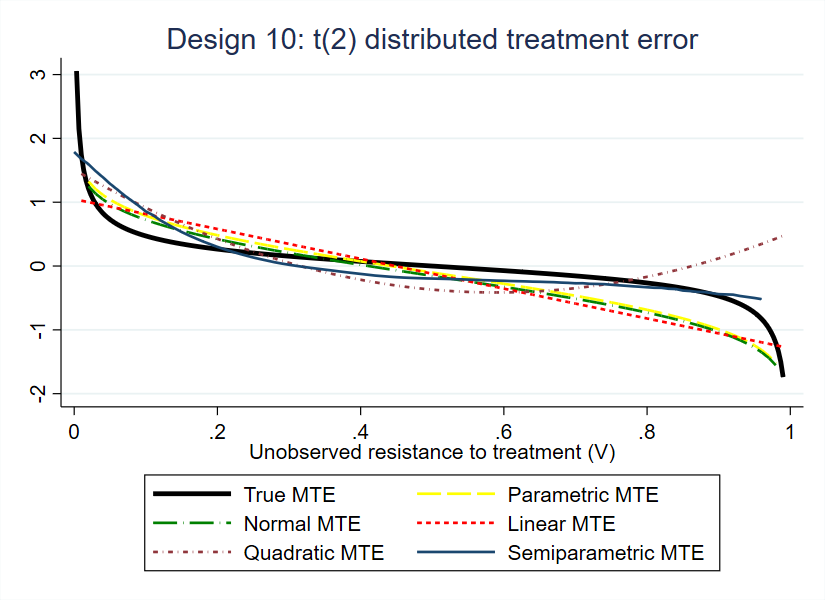}}
\end{minipage}
\begin{minipage}[t]{0.48\linewidth}
  \centerline{\includegraphics[width=7.626cm,height=5.55cm]{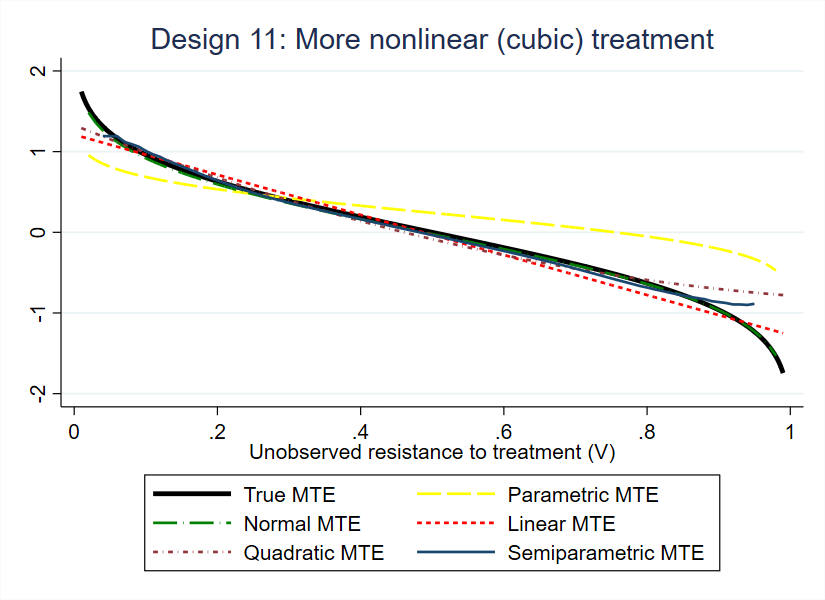}}
\end{minipage}
\begin{minipage}[t]{0.48\linewidth}
  \centerline{\includegraphics[width=7.626cm,height=5.55cm]{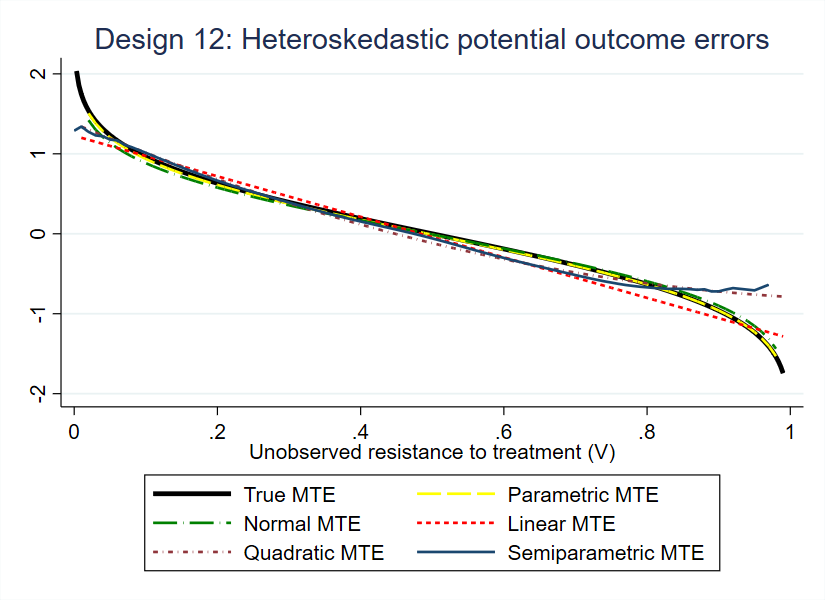}}
\end{minipage}
{\small Notes: The parametric MTE estimator specifies both parametric first
step and parametric second step. The normal, linear, quadratic MTE
estimators specify nonparametric first step and parametric second step. And
the semiparametric MTE estimator specifies nonparametric first step and
partially linear (semiparametric) second step.}
\end{figure}

\newpage

\begin{figure}[tbh]
\caption{Simulated MTE curves for Designs 13-16}
\label{Fig:Design13-16}
\bigskip \setlength{\abovecaptionskip}{-0.1cm} 
\begin{minipage}[t]{0.48\linewidth}
  \centerline{\includegraphics[width=7.626cm,height=5.55cm]{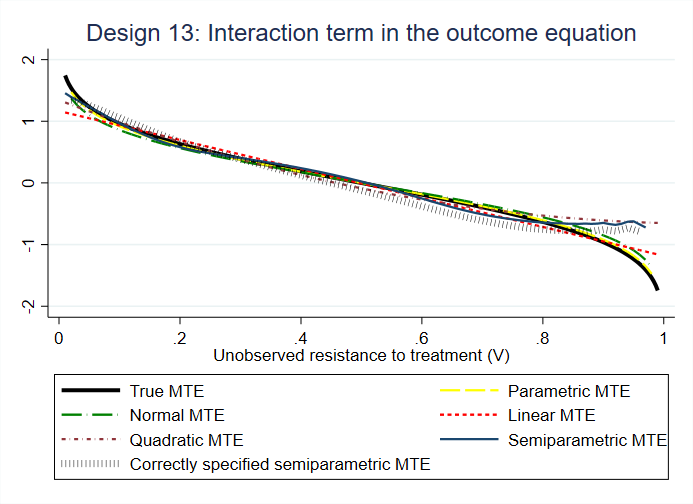}}
\end{minipage}
\begin{minipage}[t]{0.48\linewidth}
  \centerline{\includegraphics[width=7.626cm,height=5.55cm]{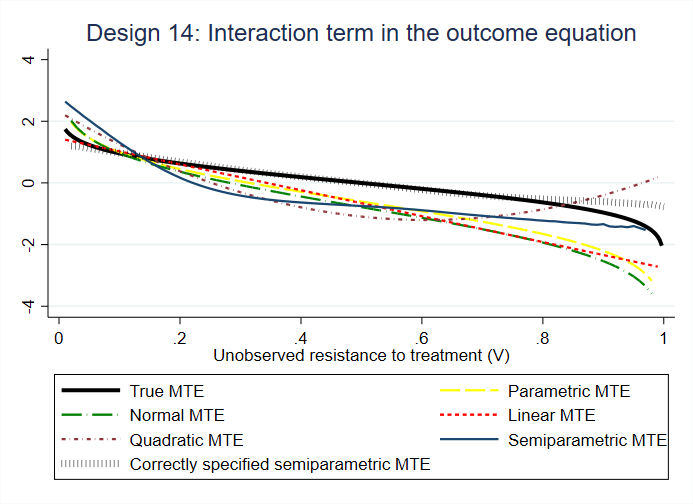}}
\end{minipage}
\begin{minipage}[t]{0.48\linewidth}
  \centerline{\includegraphics[width=7.626cm,height=5.55cm]{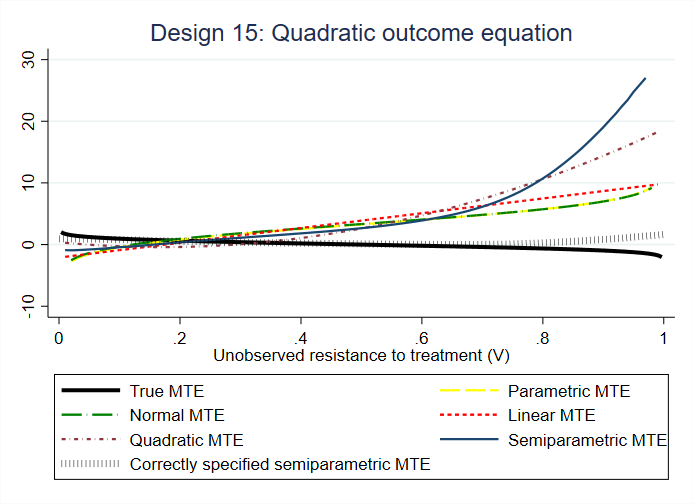}}
\end{minipage}
\begin{minipage}[t]{0.48\linewidth}
  \centerline{\includegraphics[width=7.626cm,height=5.55cm]{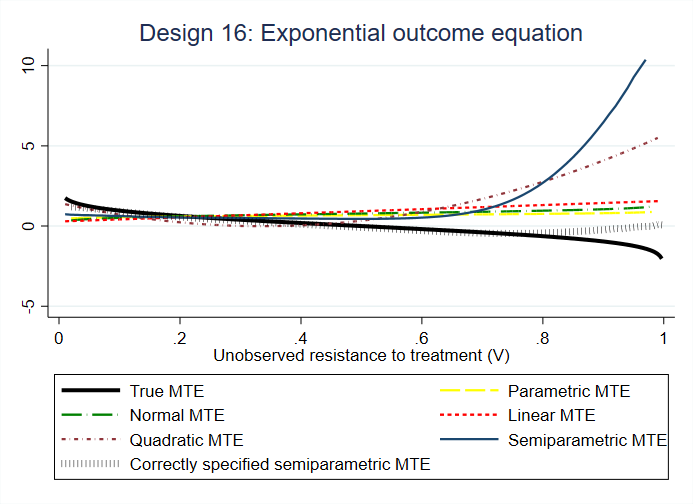}}
\end{minipage}
{\small Notes: The parametric, normal, linear, quadratic, and semiparametric
MTE estimators are defined as before without including any nonlinear term in
the outcome equation, while the correctly specified semiparametric MTE
estimator includes the true nonlinear term in the outcome equation.}
\end{figure}

\end{appendix}

\end{document}